\documentclass[sigconf]{acmartm}

\AtBeginDocument{%
  }

\usepackage{xcolor}
\usepackage{xspace}
\usepackage{bbm}
\usepackage{graphicx}
\usepackage{hyperref}
\usepackage{cleveref}
\usepackage{stmaryrd}
\usepackage[all]{xy}

\AtEndPreamble{%
  \theoremstyle{acmdefinition}
  \newtheorem{remark}[theorem]{Remark}
}

\crefname{section}{Sec.}{Sec.}
\Crefname{section}{Sec.}{Sec.}
\crefname{example}{Ex.}{Ex.}
\Crefname{example}{Ex.}{Ex.}
\crefname{remark}{Rem.}{Rem.}
\Crefname{remark}{Rem.}{Rem.}
\crefname{lemma}{Lem.}{Lem.}
\Crefname{lemma}{Lem.}{Lem.}
\crefname{proposition}{Prop.}{Prop.}
\Crefname{proposition}{Prop.}{Prop.}
\crefname{definition}{Def.}{Def.}
\Crefname{definition}{Def.}{Def.}
\crefname{theorem}{Thm.}{Thms.}
\Crefname{theorem}{Thm.}{Thms.}
\newcommand{\caseref}[1]{case~\ref{#1}}

\usepackage{mathpartir}
\newcommand{\indrulename}[1]{\texttt{#1}}
\newcommand{\indrule}[3]{\infer{#2}{#3}{\indrulename{#1}}}
\newcommand{\inderiv}[2]{\begin{array}{c}#1\\#2\end{array}}

\renewcommand{\emptyset}{\varnothing}
\newcommand{\Nat}{\mathbb{N}}

\newcommand{\NotNow}[1]{}

\newcommand{\HS}{\hspace{.5cm}}
\newcommand{\ST}{\ |\ }
\newcommand{\ie}{{\em i.e.}\xspace}
\newcommand{\eg}{{\em e.g.}\xspace}

\newcommand{\cf}{{\em cf.}\xspace}
\newcommand{\ih}{i.h.\xspace}
\newcommand{\set}[1]{\{#1\}}

\newcommand{\eqdef}{\,\mathrel{:=}\,}
\newcommand{\iffdef}{\,\mathrel{\overset{\mathrm{def}}{\iff}}\,}

\newcommand{\eqby}[1]{\,=_{#1}\,}
\newcommand{\sub}[2]{\{#1:=#2\}}

\newcommand{\fv}[1]{\mathsf{fv}(#1)}

\newcommand{\dom}[1]{\mathsf{dom}(#1)}
\newcommand{\defn}[1]{\emph{#1}}
\newcommand{\ARG}{(\cdot)}
\newcommand{\powerset}[1]{\mathcal{P}(#1)}
\newcommand{\minof}[2]{\min\{#1,#2\}}

\renewcommand{\theenumi}{\arabic{enumi}}
\renewcommand{\theenumii}{\arabic{enumii}}
\renewcommand{\theenumiii}{\arabic{enumiii}}
\renewcommand{\theenumiv}{\arabic{enumiv}}
\makeatletter
\renewcommand\p@enumii{\theenumi.}
\renewcommand\p@enumiii{\theenumi.\theenumii.}
\renewcommand\p@enumiv{\theenumi.\theenumii.\theenumiii.}
\makeatother

\newcommand{\lambdaST}{{\lambda}^{\to}}
\newcommand{\lambdaF}{{\lambda}^{\mathsf{F}}}
\newcommand{\lambdaFw}{{\lambda}^{\mathsf{F\omega}}}
\newcommand{\lambdaCheck}{\lambda^{\to}_\mathfrak{m}}
\newcommand{\lambdaCheckF}{\lambda^{\mathsf{F}}_{\mathfrak{m}}}
\newcommand{\lambdaCheckFw}{{\lambda^{\mathsf{F\omega}}_{\mathfrak{m}}}}

\newcommand{\origlam}[2]{\lambda #1.\,#2}
\newcommand{\lam}[2]{\origlam{#1}{#2}}
\newcommand{\lamt}[2]{#1^{?}#2}

\newcommand{\tos}{\mathrel{\twoheadrightarrow}}
\newcommand{\symWeakHead}{w}
\newcommand{\tow}{\to_{\symWeakHead}}   
\newcommand{\tows}{\twoheadrightarrow_{\symWeakHead}}
\newcommand{\tof}{\tow}
\newcommand{\tofs}{\tows}

\newcommand{\toinv}{\leftarrow}
\newcommand{\tosym}{\leftrightarrow}
\newcommand{\toinvs}{\twoheadleftarrow}
\newcommand{\conveq}{\twoheadleftarrow\!\!\!\twoheadrightarrow}

\newcommand{\ctxhole}{\Box}
\newcommand{\gctx}{\mathtt{C}}
\newcommand{\wctx}{\mathtt{W}}
\newcommand{\ctxof}[2]{#1[#2]}

\newcommand{\BB}[1]{\textcolor{blue}{#1}}
\newcommand{\RR}[1]{\textcolor{red}{#1}}
\newcommand{\gen}[1]{\BB{\maltese}_{#1}}
\newcommand{\verif}[2]{\RR{\langle}#1\RR{\rangle}#2}
\newcommand{\genteig}[1][\teig]{\gen{#1}}
\newcommand{\verifteig}[2][\teig]{\verif{#1}{#2}}

\newcommand{\iunit}{\star}
\newcommand{\eunit}[2]{#1{?}#2}

\newcommand{\arrow}{\Rightarrow}

\newcommand{\bki}{\kappa}

\newcommand{\ki}{K}
\newcommand{\kitwo}{J}
\newcommand{\kithree}{L}

\newcommand{\tenv}{\xi}
\newcommand{\tenvtwo}{\tenv'}
\newcommand{\typesize}[1]{|#1|}

\newcommand{\origvar}{x}
\newcommand{\origvartwo}{y}
\newcommand{\origvarthree}{z}
\newcommand{\var}{\origvar}
\newcommand{\vartwo}{\origvartwo}
\newcommand{\varthree}{\origvarthree}

\newcommand{\utm}{M}
\newcommand{\utmtwo}{N}
\newcommand{\utmthree}{P}
\newcommand{\utmfour}{Q}
\newcommand{\utmfive}{R}
\newcommand{\utmsix}{S}
\newcommand{\utmProg}{P}

\newcommand{\uhd}{H}
\newcommand{\uneu}{\utm^\uparrow}
\newcommand{\uneutwo}{\utmtwo^\uparrow}
\newcommand{\unf}{\utm^\downarrow}
\newcommand{\unftwo}{\utmtwo^\downarrow}
\newcommand{\unfthree}{\utmthree^\downarrow}
\newcommand{\argu}{I}
\newcommand{\argutwo}{J}
\newcommand{\argunf}{\argu^\downarrow}
\newcommand{\argunftwo}{\argutwo^\downarrow}
\newcommand{\argus}{\vec{\argu}}
\newcommand{\argutwos}{\vec{\argutwo}}
\newcommand{\ttms}{\vec{\ttm}}
\newcommand{\ttmtwos}{\vec{\ttmtwo}}
\newcommand{\ttmthrees}{\vec{\ttmthree}}
\newcommand{\ttmfours}{\vec{\ttmfour}}
\newcommand{\ttmfives}{\vec{\ttmfive}}

\newcommand{\origtm}{t}
\newcommand{\origtmtwo}{s}
\newcommand{\origtmthree}{u}
\newcommand{\tm}{\origtm}
\newcommand{\tmtwo}{\origtmtwo}
\newcommand{\tmthree}{\origtmthree}

\newcommand{\origjudg}[3]{#1\vdash#2:#3}
\newcommand{\judg}[3]{\origjudg{#1}{#2}{#3}}

\newcommand{\subs}[2]{#1^{#2}}
\newcommand{\origsubst}{\sigma}
\newcommand{\subst}{\origsubst}
\newcommand{\gensubst}[1]{\underline{#1}}
\newcommand{\compat}[2]{#1\vDash#2}
\newcommand{\extsub}[2]{[#1\mapsto#2]}

\newcommand{\Prop}{\mathtt{Prop}}
\newcommand{\imp}{\Rightarrow}
\newcommand{\all}[3]{\forall{#1^{#2}}.\,#3}

\newcommand{\tlam}[3]{\lambda{#1^{#2}}.\,#3}

\newcommand{\teig}{\mathbf{a}}
\newcommand{\teigtwo}{\mathbf{b}}
\newcommand{\teigthree}{\mathbf{c}}
\newcommand{\teigP}{\mathbf{p}}

\newcommand{\tvar}{a}
\newcommand{\tvartwo}{b}
\newcommand{\tvarthree}{c}

\newcommand{\ttm}{A}
\newcommand{\ttmtwo}{B}
\newcommand{\ttmthree}{C}
\newcommand{\ttmfour}{D}
\newcommand{\ttmfive}{E}

\newcommand{\tobeta}{\mathrel{\to_\beta}}
\newcommand{\toeta}{\mathrel{\to_\eta}}

\newcommand{\tobetas}{\mathrel{\tos_\beta}}
\newcommand{\toetas}{\mathrel{\tos_\eta}}

\newcommand{\eqbeta}{\mathrel{=_\beta}}

\newcommand{\tobetainvs}{\mathrel{\toinvs_{\beta}}}

\newcommand{\tjudg}[3]{#1\vdash#2:#3}
\newcommand{\pjudgenv}[2]{#1\vdash#2\,\texttt{env}}
\newcommand{\pjudg}[4]{#1;#2\vdash#3:#4}

\newcommand{\pvar}{x}
\newcommand{\pvartwo}{y}

\newcommand{\plam}[2]{\lambda#1.\,#2}
\newcommand{\plamt}[3]{\plam{#1^{#2}}{#3}}
\newcommand{\palli}[3]{\Lambda{#1^{#2}}.\,#3}

\newcommand{\fresh}[3]{\nu{#1^{#2}}.\,#3}

\newcommand{\ptm}{t}
\newcommand{\ptmtwo}{s}
\newcommand{\ptmthree}{u}

\newcommand{\penv}{\Gamma}
\newcommand{\penvtwo}{\Delta}

\newcommand{\wkjudg}[1]{\mathtt{wk}(#1)}
\newcommand{\preq}[2]{#1 \equiv #2}

\newcommand{\tsubst}{\sigma} 
\newcommand{\psubst}{\gamma} 
\newcommand{\psubstext}{{\psubst*}}

\newcommand{\size}[1]{\lvert#1\rvert}
\newcommand{\sig}[3][\penv]{\lvert#2\rvert_{#1\vdash#3}}

\newcommand{\minimalLogic}{minimal logic\xspace}
\newcommand{\typeEnvironment}{environment\xspace}
\newcommand{\atomicProposition}{atomic type\xspace}
\newcommand{\atomicPropositions}{atomic types\xspace}

\newcommand{\unitConstructor}{success\xspace}
\newcommand{\unitEliminator}{guard\xspace}
\newcommand{\kind}{kind\xspace}
\newcommand{\kinds}{kinds\xspace}

\newcommand{\Kinds}{Kinds\xspace}
\newcommand{\logicalVariable}{type-variable\xspace}
\newcommand{\logicalVariables}{type-variables\xspace}
\newcommand{\LogicalVariable}{Type-variable\xspace}
\newcommand{\LogicalVariables}{Type-variables\xspace}
\newcommand{\logicalEigenvariable}{eigenvariable\xspace}
\newcommand{\logicalEigenvariables}{eigenvariables\xspace}
\newcommand{\LogicalEigenvariable}{Eigenvariable\xspace}
\newcommand{\LogicalEigenvariables}{Eigenvariables\xspace}
\newcommand{\logicalTerm}{type\xspace}
\newcommand{\logicalTerms}{types\xspace}
\newcommand{\LogicalTerms}{Types\xspace}
\newcommand{\logicalAbs}{type-abs.\xspace}
\newcommand{\logicalAbstraction}{type-abstraction\xspace}
\newcommand{\logicalAbstractions}{type-abstractions\xspace}

\newcommand{\logicalApp}{type-app.\xspace}
\newcommand{\logicalApplication}{type-application\xspace}

\newcommand{\logicalEnvironment}{kind-environment\xspace}

\newcommand{\logicalSubstitution}{type-substitution\xspace}

\newcommand{\logicalBetaEquivalence}{$\tau$-equivalence\xspace}
\newcommand{\proofVariable}{variable\xspace}
\newcommand{\proofVariables}{variables\xspace}
\newcommand{\proofAbs}{term-abs.\xspace}
\newcommand{\proofAbstraction}{term-abstraction\xspace}

\newcommand{\proofApp}{term-app.\xspace}
\newcommand{\proofApplication}{term-application\xspace}

\newcommand{\proofTerm}{term\xspace}
\newcommand{\proofTerms}{terms\xspace}
\newcommand{\ProofTerms}{Terms\xspace}
\newcommand{\headForm}{head form\xspace}
\newcommand{\headForms}{head forms\xspace}

\newcommand{\HeadForms}{Head forms\xspace}
\newcommand{\proofEnvironment}{environment\xspace}
\newcommand{\proofSubstitution}{substitution\xspace}
\newcommand{\ProofSubstitutions}{Substitutions\xspace}
\newcommand{\metaterm}{metaterm\xspace}
\newcommand{\metaterms}{metaterms\xspace}
\newcommand{\Metaterms}{Metaterms\xspace}
\newcommand{\argument}{input\xspace}
\newcommand{\arguments}{inputs\xspace}
\newcommand{\Argument}{Input\xspace}
\newcommand{\verificationCandidate}{verification candidate\xspace}
\newcommand{\verificationCandidates}{verification candidates\xspace}
\newcommand{\VerificationCandidates}{Verification candidates\xspace}
\newcommand{\candidateAssignment}{candidate-assignment\xspace}
\newcommand{\candidateAssignments}{candidate-assignments\xspace}
\newcommand{\CandidateAssignments}{Candidate-assignments\xspace}
\newcommand{\kindAssignment}{kind-assignment\xspace}
\newcommand{\candidateFamily}{candidate-family\xspace}
\newcommand{\candidateFamilies}{candidate-families\xspace}
\newcommand{\CandidateFamilies}{Candidate-families\xspace}

\newcommand{\judgf}[3]{#1\vdash#2:#3}
\newcommand{\plamf}[2]{\lambda#1.\,#2}
\newcommand{\allf}[2]{\forall#1.\,#2}
\newcommand{\pallfi}[2]{\Lambda{#1}.\,#2}
\newcommand{\freshf}[2]{\nu{#1}.\,#2}
\newcommand{\vcset}{X}
\newcommand{\VC}[1]{\mathcal{VC}_{#1}}
\newcommand{\asig}{\eta}
\newcommand{\asigext}{{\eta*}}
\newcommand{\semf}[2]{\llbracket#1\rrbracket#2}

\newcommand{\compatf}[3][\asig]{#2\vDash^{#1}#3}

\newcommand{\sto}{\rhd} 
\newcommand{\stoinv}{\lhd}
\newcommand{\stovert}{\triangledown}

\newcommand{\kiasig}{\Xi}
\newcommand{\ValidMetaterms}{\mathbb{M}}
\newcommand{\ValidKinds}{\mathbb{K}}
\newcommand{\ValidTypes}{\mathbb{T}}
\newcommand{\ValidTypesK}[1]{\ValidTypes_{#1}}
\newcommand{\GValidTypes}{\ValidTypes^0}
\newcommand{\GValidTypesK}[1]{\ValidTypes^0_{#1}}
\newcommand{\iI}{{i \in I}}
\newcommand{\Prod}[1]{\prod_{#1}}
\newcommand{\singleElem}{\diamondsuit}
\newcommand{\semg}[3]{\llbracket#1:#2\rrbracket#3}
\newcommand{\semgbase}[1]{\llparenthesis#1\rrparenthesis}
\newcommand{\VCK}[2]{\VC{#1:#2}}
\newcommand{\Verif}[1]{\mathsf{Verif}_{#1}}
\newcommand{\VerifK}[2]{\Verif{#1:#2}}

\newcommand{\match}[4]{#1\,/\,#2:#3\Rrightarrow#4}
\newcommand{\coherent}{\frown}
\newcommand{\placeholder}{\text{...}}

\newcommand{\LambdaFInters}{\mathcal{T}}
\newcommand{\linearTypes}{linear-types\xspace}
\newcommand{\LinearTypes}{Linear-types\xspace}
\newcommand{\linearMultitype}{linear-multitype\xspace}
\newcommand{\linearMultitypes}{linear-multitypes\xspace}
\newcommand{\LinearMultitypes}{Linear-multitypes\xspace}
\newcommand{\linearEnvironment}{linear-environment\xspace}
\newcommand{\linearEnvironments}{linear-environments\xspace}
\newcommand{\derivs}[2]{#1\rhd#2}
\newcommand{\deriv}{\Phi}
\newcommand{\derivtwo}{\Psi}
\newcommand{\derivthree}{\Pi}
\newcommand{\limp}{\multimap}
\newcommand{\mset}[1]{[#1]}
\newcommand{\allimp}{\mathop{\rotatebox{90}{$\forall$}}}
\newcommand{\genv}{\mathbf{E}}
\newcommand{\genvtwo}{\mathbf{F}}
\newcommand{\ltyp}{\mathbb{T}}
\newcommand{\ltyptwo}{\mathbb{S}}
\newcommand{\mtyp}{\mathbb{M}}
\newcommand{\mtyptwo}{\mathbb{N}}
\newcommand{\jul}[3]{#1\vdash#2:#3}
\newcommand{\jum}[3]{#1\Vdash#2:#3}
\newcommand{\ruleLVar}{\indrulename{Lvar}}
\newcommand{\ruleLImpI}{\indrulename{Llam}}
\newcommand{\ruleLImpE}{\indrulename{Lapp}}
\newcommand{\ruleLAllI}{\indrulename{Llamt}}
\newcommand{\ruleLAllE}{\indrulename{Lappt}}
\newcommand{\ruleLUnitI}{\indrulename{L$\iunit$}}
\newcommand{\ruleLUnitE}{\indrulename{L$\eunit{}{}$}}
\newcommand{\ruleLFresh}{\indrulename{L$\nu$}}
\newcommand{\ruleLVerifTeig}{\indrulename{L$\langle\rangle\teig$}}
\newcommand{\ruleLVerifImp}{\indrulename{L$\langle\rangle{\imp}$}}
\newcommand{\ruleLVerifAll}{\indrulename{L$\langle\rangle{\forall}$}}
\newcommand{\ruleLGenTeig}{\indrulename{L$\maltese\teig$}}
\newcommand{\ruleLGenImp}{\indrulename{L$\maltese{\imp}$}}
\newcommand{\ruleLGenAll}{\indrulename{L$\maltese{\forall}$}}
\newcommand{\ruleLMulti}{\indrulename{Lmulti}}

\newcommand{\semr}[1]{\llbracket#1\rrbracket}
\newcommand{\realizes}{\Vdash}
\newcommand{\pair}[2]{\langle#1,#2\rangle}

\newcommand{\epistemic}{epistemic\xspace}
\newcommand{\Epistemic}{Epistemic\xspace}
\newcommand{\ruleNum}[1]{\text{\footnotesize{(#1)}}}
\newcommand{\itemNum}[1]{\,\,\textbf{(#1)}\,}
\newcommand{\Case}[2][Case ]{\medskip\noindent\textbf{#1#2}}
\newcommand{\tool}[1]{\textsf{#1}}

\begin{document}

\title[Verifiers and Generators: \Epistemic Semantics for Intuitionistic Logic]{Verifiers and Generators: \\ \Epistemic Semantics for Intuitionistic Logic \\ (Long Version)}

\author[P.~Barenbaum]{Pablo Barenbaum}
\email{pbarenbaum@dc.uba.ar}
\orcid{0009-0003-2494-3345}
\affiliation{%
  \institution{Universidad Nacional de Quilmes (CONICET)}
  \city{Bernal}
  \country{Argentina}
}
\affiliation{%
  \institution{Departamento de Computación, Facultad de Ciencias Exactas y Naturales, Universidad de Buenos Aires}
  \city{CABA}
  \country{Argentina}
}

\settopmatter{printacmref=false}
\setcopyright{none}

\begin{abstract}

This paper explores \emph{\epistemic realizability},
a form of realizability in which the property that a piece of data
constitutes evidence for a logical proposition is semi-decidable.
In this framework, each proposition $\ttm$ is assigned a \emph{verifier}
program $\verif{\ttm}{}$ that checks whether a datum $X$ is a realizer for $\ttm$,
and a dual \emph{generator} program $\gen{\ttm}$ that behaves as a generic
realizer for $\ttm$.
We propose \epistemic realizability interpretations for \minimalLogic,
second-order intuitionistic logic, and higher-order intuitionistic
logic, proving that each system is sound and complete under the
proposed semantics.

\end{abstract}

\received{20 February 2007}
\received[revised]{12 March 2009}
\received[accepted]{5 June 2009}

\maketitle

\section{Introduction}
\label{sec:intro}

Verificationism explains the meaning of logical connectives by means
of the notion of \emph{evidence} rather than the notion of \emph{truth}.
Under this view,
the meaning of a proposition $\ttm$ is given by a \emph{verification procedure}
to recognize evidence for $\ttm$~\cite{Dummett1991}.
This work explores the problem of providing a formal semantics for logical
connectives in an ``epistemically acceptable'' way.
Formally, this means that verification procedures must be \emph{effectively checkable}:
the meaning of a proposition $\ttm$ must be a program $P$
that computes a \emph{partial computable predicate},
such that $P$ accepts $X$ if it constitutes evidence
for $\ttm$, and either rejects $X$ or remains undefined otherwise.
We call this ``\epistemic semantics'', following Nelson~\cite[Sec.~6]{NelsonIntuitionism}.

The primary challenge lies in providing \textbf{\epistemic semantics for logical
connectives}. Standard semantic frameworks typically do not provide effective
means of recognizing evidence.
For example, classical model-theoretic semantics defines the meaning of
a universal quantifier $\forall x. P(x)$ by declaring that a structure satisfies the formula
if and only if it satisfies $P(a)$ for every element $a$ in the universe.
While this definition is mathematically acceptable
---assuming a prior understanding of universal quantification in the metalanguage---
it is unsatisfactory from the perspective of \epistemic semantics.
Rather than providing an effective verification procedure
to determine whether an object is evidence for $\forall x. P(x)$,
the meaning of the universal quantifier is displaced from the object language to the metatheory.

\paragraph{Realizability interpretations.}
Many constructive approaches to semantics take
BHK-style\footnote{BHK stands for Brouwer--Heyting--Kolmogorov.}
interpretations as their starting point~\cite{TroelstraVanDalen1988}.
These interpretations align with verificationism
by defining the meaning of a proposition
through its \emph{realizers} (also referred to as \emph{constructions},
\emph{proofs}, \emph{evidence}, etc.).
One formal incarnation of these ideas
is Kleene's realizability~\cite{Kleene1945},
in which the meaning of conjunction is defined as:
\[
 \pair{n}{m} \realizes (\ttm \land \ttmtwo)
 \iffdef
 n \realizes \ttm \text{ and } m \realizes \ttmtwo
\]
where $\pair{n}{m}$ represents Gödel's pairing function
and $n \realizes \ttm$ denotes that the natural number $n$ is a realizer for $\ttm$.
This definition declares that a realizer for a conjunction $\ttm \land \ttmtwo$
is given by a pair containing a realizer for $\ttm$ and a realizer for $\ttmtwo$.
Moreover, it is \textbf{epistemically acceptable}, in the sense that it provides
a decision procedure: to verify that
a pair $\pair{n}{m}$ is a realizer for $\ttm \land \ttmtwo$, one can recursively verify
that $n$ is a realizer for $\ttm$ and $m$ is a realizer for $\ttmtwo$.

\paragraph{Troublesome connectives.}
Unfortunately, BHK-style interpretations do not provide epistemically acceptable
definitions for all connectives. In particular, the semantics of ``troublesome''
operators ---implication, negation, universal quantification--- is
expressed in terms of \emph{meta-level} universal quantification.
For instance, the meaning of implication in Kleene's realizability
is defined as follows:
\[
  n \realizes (\ttm \imp \ttmtwo)
  \iffdef
  \forall m \in \Nat.\,(m \realizes \ttm
               \implies
               \varphi_n(m) \realizes \ttmtwo)
\]
where $\varphi_n$ denotes the $n$-th partial computable function.
This definition declares that a realizer for an implication $\ttm \imp \ttmtwo$
is given by a partial computable function $\varphi_n$
that transforms \emph{every} realizer for $\ttm$ into a realizer for $\ttmtwo$.
This is a priori \textbf{not epistemically acceptable}, because
to verify that $n$ is a realizer for $\ttm \imp \ttmtwo$, one would have to
semi-decide a universal quantification $\forall m \in \Nat.(...)$
over the natural numbers, which is in general an impossible task\footnote{%
  The meaning of implication in the object language depends
  on the meaning of universal quantification in the metatheory.
  Note that there are ways to understand the universal formula
  $\psi \equiv \forall m \in \Nat.(...)$ in an epistemically acceptable way;
  for example, as the existence
  of a derivation of $\psi$ in Peano arithmetic. (See the subsequent comments on formalism).
}.

\paragraph{The formalist escape route.}
As discussed above, the intuitive ``meaning'' of certain connectives,
specifically \emph{implication}, \emph{negation}, and \emph{universal quantification},
seems difficult to reconcile with an \epistemic semantics.
One way to achieve this reconciliation is through \emph{formalism}.
Formalism can be synthesized as the view that evidence for a proposition $\ttm$ is
provided by a derivation of $\ttm$ in a specific formal proof system.
Any proof system that has the property of \emph{semi-decidable proof-checking}
immediately provides a trivial \epistemic semantics for propositions.
Indeed, the verification procedure is given by the proof-checker itself,
and the string of symbols encoding the derivation serves as the required
existential witness.
Formalism allows us to understand mathematical activity as a
combinatorial game of producing derivations, but leaves open the question
of whether propositions can be given any other meaning beyond the symbols
themselves. For example, in the absence of semantics, proof systems appear
as entirely arbitrary constructions, and it is impossible to even wonder
whether they are sound or complete.

\paragraph{\Epistemic realizability.}
This work studies a \textbf{\epistemic realizability} semantics
for logical propositions, inspired by BHK-style
interpretations and traditional realizability,
but formulated from the perspective of \epistemic semantics.
Early works on realizability were formulated in the language of arithmetic,
so data types other than integers had to be handled through explicit
Gödel codings.
In this paper we use 
---as other works also do---
extensions of the \emph{untyped $\lambda$-calculus}
as the language
to express both programs and data\footnote{This choice is not
essential but quite convenient.}.

The interpretations we propose
center on the dual notions of \emph{verifiers} and \emph{generators}.
We begin by interpreting each proposition $\ttm$ as a
\emph{program} called its \textbf{verifier} and written $\verif{\ttm}{}$.
The verifier for $\ttm$ is an untyped $\lambda$-term
that computes a partial computable predicate.
We say that $X$ is a \emph{realizer} for $\ttm$
if, by definition, $\verif{\ttm}{(X)}$ succeeds.
Success, in this framework, is defined by the property that
$\verif{\ttm}{(X)} \tos \iunit$,
where $\tos$ denotes reduction in zero or more steps
and $\iunit$ is a constant representing successful termination.
Note that the input $X$ is itself expressed using an untyped
$\lambda$-term.

For example, the meaning of a conjunction $\ttm\land\ttmtwo$
is a program $\verif{\ttm\land\ttmtwo}{}$
such that $\verif{\ttm\land\ttmtwo}(X)$ succeeds
if both $\verif{\ttm}{(\pi_1(X))}$ and $\verif{\ttmtwo}{(\pi_2(X))}$
succeed, where $\pi_1(X)$ and $\pi_2(X)$ are the first and second
projections of a pair, in accordance with the BHK interpretation.

\paragraph{Verifiers and generators.}
A significant difficulty arises when attempting to define the meaning
of ``troublesome'' connectives. For example, implication is challenging
due to the reversal of \emph{polarity} in the antecedent.
According to the BHK interpretation,
the verifier $\verif{\ttm\imp\ttmtwo}$ must take an input $X$
and determine whether, for every $Y$ such that $\verif{\ttm}(Y)$
succeeds, the application $\verif{\ttmtwo}{(X\,Y)}$ also succeeds. This verification
is not decidable a priori. Indeed, the predicate:
\[
  P(X) \,\,\,\,\equiv\,\,\,\,
  \forall Y. (
    \verif{\ttm}(Y)\tos\iunit \implies \verif{\ttmtwo}{(X\,Y)}\tos\iunit
  )
\]
starts with a universal quantifier. This is not a priori semi-decidable,
so it is unclear how to define a program $\verif{\ttm\imp\ttmtwo}$ that computes $P$.

To resolve this, we refine the interpretation: now each proposition $\ttm$
is interpreted by both its \textbf{verifier} $\verif{\ttm}$
and a second program, called its \textbf{generator} and written $\gen{\ttm}$.
The generator for $\ttm$ is an untyped $\lambda$-term that
behaves as a generic realizer for $\ttm$.
For example, in the case of conjunction,
$\gen{\ttm\land\ttmtwo}$ behaves as the pair $(\gen{\ttm},\gen{\ttmtwo})$,
also in accordance with the BHK interpretation.

From a computer science perspective,
a verifier $\verif{\ttm}{}$ is a program that \emph{tests}
whether its argument complies with a specification $\ttm$.
Conversely, a generator $\gen{\ttm}{}$ serves as a \emph{stub},
or a ``mock implementation'' that behaves as if it were compliant
with said specification.
One basic property of this system is \emph{correctness}:
a verifier must accept its corresponding generator as a valid realizer,
that is, $\verif{\ttm}{(\gen{\ttm})} \tos \iunit$.

Aided by the presence of generators, we can define the verifier
for an implication in such a way that
$\verif{\ttm\imp\ttmtwo}{(X)}$ succeeds if and only if
$\verif{\ttmtwo}(X\,\gen{\ttm})$ succeeds.
That is, to determine whether $X$ is a realizer for $\ttm\imp\ttmtwo$,
we check whether $X$ produces a realizer for $\ttmtwo$ 
when supplied with the realizer for $\ttm$ produced by the generator.

To complete the picture, we need to define the
generator for an implication $\ttm\imp\ttmtwo$.
Verifiers and generators are defined in a mutually inductive fashion.
Indeed, the generator $\gen{\ttm\imp\ttmtwo}$
is a program that takes an arbitrary input $X$
and returns $\gen{\ttmtwo}$ provided that $\verif{\ttm}{(X)}$ succeeds.
Formally, this can be written as
$\gen{\ttm\imp\ttmtwo} = \lam{\var}{(\eunit{\verif{\ttm}{\var}}{\gen{\ttmtwo}})}$.
Here $\lam{\var}{(...)}$ is the standard $\lambda$-notation for anonymous functions,
and $\eunit{\utm}{\utmtwo}$ is an operator that proceeds to evaluate $\utmtwo$
provided that the evaluation of $\utm$ terminates successfully.

\paragraph{Scope of the work and propositions-as-types.}
We study \epistemic realizability across three incrementally more
expressive systems of intuitionistic logic:
1. \textbf{\minimalLogic};
2. \textbf{second-order intuitionistic logic};
and 3. \textbf{higher-order intuitionistic logic}.
As is well-known, derivations of sequents in natural deduction systems
correspond to terms in typed $\lambda$-calculi, through the propositions-as-types
correspondence (\cf for instance~\cite{Girard1989,Sorensen2006}).
The formal deductive systems that we study are three standard typed
$\lambda$-calculi corresponding to each of the logics above, namely:
1.~\textbf{simply typed $\lambda$-calculus}~($\lambdaST$);
2.~\textbf{polymorphic $\lambda$-calculus}~(System~F~\cite{GirardThesis,Reynolds1974}); and
3.~\textbf{higher-order polymorphic $\lambda$-calculus}~(System~F$\omega$~\cite{GirardThesis}).

\paragraph{Contributions.}
For each of the three aforementioned typed $\lambda$-calculi,
we define a corresponding \emph{untyped} $\lambda$-calculus,
called the \textbf{metacalculus}.
For example, for the simply typed $\lambda$-calculus ($\lambdaST$)
we define a metacalculus ($\lambdaCheck$)
that allows to express both verifiers $\verif{\ttm}$ and generators $\gen{\ttm}$
for every proposition $\ttm$ in the language of \minimalLogic.
For each system, the main result is that the logic is
both \textbf{sound} and \textbf{complete} with respect to the \epistemic realizability
interpretation. This establishes that $\ttm$ is a theorem in the
logical theory if and only if there exists evidence for $\ttm$ in the metacalculus,
that is, an untyped term $X$ such that $\verif{\ttm}{(X)}$ succeeds.
For completeness to hold, we further require that the evidence $X$ must
be a \emph{normalizing} and \emph{pure} term; that is, it must have a
normal form and not contain extraneous elements such as generators $\gen{\ttm}$.

While relating typed $\lambda$-calculi and \epistemic realizability 
may appear straightforward, proving soundness and completeness presents
significant technical challenges.
For example, proving soundness of the \emph{modus ponens} rule,
requires showing that if both $\verif{\ttm\imp\ttmtwo}{(X)}$ and $\verif{\ttm}{(Y)}$
succeed then $\verif{\ttmtwo}{(X\,Y)}$ must also succeed.
In traditional realizability, this is precisely the definition
of implication, leaving nothing to prove.
In our ``\epistemic'' setting, however,
the condition that $\verif{\ttm\imp\ttmtwo}{(X)}$ succeeds
only guarantees that $X$ behaves correctly when applied to the \emph{generator}
$\gen{\ttm}$. Consequently, one must prove that
$\verif{\ttmtwo}{(X\,\gen{\ttm})} \tos \iunit$
and $\verif{\ttm}{Y} \tos \iunit$
imply $\verif{\ttmtwo}{(X\,Y)} \tos \iunit$.
Establishing this property is a non-trivial task, as it requires to show
that in some sense $X$ cannot distinguish between the actual realizer $Y$
and the ``mock'' realizer produced by the generator $\gen{\ttm}$.

\paragraph{Structure of the paper.}
In \cref{sec:preliminaries} we review some preliminary notions and fix
some of the notations used in the rest of the paper.
In~\cref{sec:simply_typed} we study \epistemic realizability for
\minimalLogic.
In~\cref{sec:second_order} we extend the results to second-order logic.
This requires to generalize the notions of verifier and generator to deal
with second-order universal quantifiers. The proof techniques also have
to be adapted, to address the usual difficulties posed by impredicativity.
In~\cref{sec:higher_order} we discuss the extension to higher-order logic.
Finally, in~\cref{sec:conclusion} we conclude,
and discuss some potential applications in programming languages and
proof assistants, as well as related and future work.

\section{Preliminaries}
\label{sec:preliminaries}

In this section we recall some standard notions, to fix nomenclature and
notation.

We work with various typed and untyped $\lambda$-calculi,
all of them instances of \defn{abstract rewriting systems} (ARSs).
An ARS is a pair $(X,\to)$, where $X$ is a set of elements called \emph{objects},
and $\to$ is a binary relation ${\to} \subseteq X \times X$.
Standard references are~\cite{BaaderNipkow,Terese} for rewriting systems
and~\cite{Barendregt1984,Barendregt2013} for the $\lambda$-calculus.
If $R \subseteq X \times Y$ and $S \subseteq Y \times Z$ are binary relations,
we write $R \circ S$ (or just $R\,S$) for their composition, defined by
$\set{(x,z) \in X \times Z \ST \exists y \in Y.\, x \mathrel{R} y \land y \mathrel{R} z}$.
If $(X,\to)$ is an ARS, we write $\to^0$ for the diagonal relation
and ${\to^{n+1}} \eqdef (\to \circ \to^n)$.
We write $\to^=$ for the reflexive closure
(\ie ${\to^=} = {\to^0} \cup {\to}$);
$\to^{-1}$ or $\toinv$ for the inverse relation
(\ie ${\toinv} = \set{(y,x) \ST x \to y}$);
$\to^*$ or $\tos$ for the reflexive--transitive closure
(\ie ${\tos} = \cup_{n \geq 0} \to^n$);
$\tosym$ for the symmetric closure
(\ie ${\tosym} = {\toinv} \cup {\to}$);
and $\conveq$ for reflexive--symmetric--transitive closure
(\ie ${\conveq} = {\tosym^*}$).
Given a relation $\to_X$, we write $=_X$ for $\conveq_X$,
and call it \emph{$X$-convertibility}.
Given relations $\to_X$ and $\to_Y$,
we write $\to_{XY}$ for their union.

An object $x$ is $\to$-\emph{irreducible} or a $\to$-\emph{normal form}
if there is no $y$ such that $x \to y$.
An object $x$ is $\to$-\emph{normalizing} if there exists a normal form $y$
such that $x \tos y$.
An ARS is \emph{deterministic}
if for all objects $x,y,y'$ such that $x \to y$ and $x \to y'$
one has that $y = y'$.
An ARS has the \emph{diamond property} if ${\toinv\,\to} \subseteq {\to\,\toinv}$,
it is \emph{subcommutative} if ${\toinv\,\to} \subseteq {\to^=\,\toinv^=}$,
and it is \emph{confluent} or \emph{Church--Rosser}
if ${\toinvs\,\tos} \subseteq {\tos\,\toinvs}$.
Determinism, the diamond property, and subcommutativity are
each stronger properties than confluence.
The relations $\to_1$ and $\to_2$ \emph{commute} if
${\toinvs_1\,\tos_2} \subseteq {\tos_2\,\toinvs_1}$.

In the concrete rewriting systems that we use, objects are \emph{expressions}
($\utm,\utmtwo,\ttm,\ttmtwo,\tm,\tmtwo,\hdots$) of various sorts
(\proofTerms, \logicalTerms, etc.).
Expressions may contain occurrences of \emph{variables} ($\var,\vartwo,\tvar,\tvartwo,\hdots$)
of various sorts (\proofVariables, \logicalVariables, etc.).
Some syntactic constructors may bind variables, \eg
a $\lambda$-abstraction $\lam{\var}{\tm}$ binds the occurrences of $\var$ in
$\tm$. The notions of free and bound occurrences of a variable are defined
as expected in each case. We work implicitly up to $\alpha$-renaming of bound variables.
We write $\fv{E}$ to denote the set of free variables of the
expression $E$. We write $E_1\sub{\var}{E_2}$ for the expression that results
from performing the capture-avoiding substitution of all the free occurrences
of $\var$ in $E_1$ by $E_2$.
Sometimes we work with substitutions $\subst$, which are functions mapping
variables to expressions.
We write $\subst\extsub{\var}{E}$ for the extension of the substitution $\subst$
with the mapping $\var \mapsto E$, \ie
the substitution $\subst'$ such that
$\subst'(\vartwo) = E$ if $\var = \vartwo$,
and $\subst'(\vartwo) = \subst(\vartwo)$ otherwise.
We write $\subs{E}{\subst}$ to denote the expression
that results from the parallel substitution of each free occurrence of each
free variable $\var$ in $E$ by $\subst(\var)$, avoiding capture.
The usual properties of substitution apply in these cases,
\eg $\subs{E_1\sub{\var}{E_2}}{\subst} = \subs{E_1}{\subst}\sub{\var}{\subs{E_2}{\subst}}$
holds assuming that $\var$ is \emph{safe} for $\subst$,
\ie $\var\notin\fv{\subst(\vartwo)}$ for every $\vartwo$.
Sometimes we use these standard syntactical facts without providing
explicit proofs.

In these systems we usually consider \defn{contexts},
which are expressions with a single occurrence of a distinguished free
variable $\ctxhole$, called the \defn{hole}.
If $\gctx$ is a context and $E$ is an expression, we write
$\ctxof{\gctx}{E}$ for the expression that results from ``plugging'' $E$
into the hole of $\gctx$, possibly capturing variables.

\section{\Epistemic Realizability for Minimal Logic}
\label{sec:simply_typed}

In this section, we develop a \epistemic realizability semantics
for \textbf{\minimalLogic}, the intuitionistic system
that includes implication as the single connective.
The main results in this section are \textbf{soundness}~(\cref{thm:soundness})
and \textbf{completeness}~(\cref{thm:completeness}).

\paragraph{The simply typed $\lambda$-calculus ($\lambdaST$).}
We begin by recalling the typed $\lambda$-calculus that corresponds
to \minimalLogic
from the perspective of the propositions-as-types correspondence,
the (Curry-style) simply typed $\lambda$-calculus.
We assume given a denumerable set of \defn{\atomicPropositions} ($\teig,\teigtwo,\hdots$)
and \defn{variables} ($\var,\vartwo,\hdots$).
The sets of \defn{\logicalTerms} and \defn{terms} of $\lambdaST$ are given by:
\[
  \begin{array}{lrcl}
  \textsc{\LogicalTerms} &
  \ttm,\ttmtwo,\hdots & ::= & \teig \mid \ttm\arrow\ttmtwo
  \\
  \textsc{Terms} &
  \tm, \tmtwo, \hdots & ::= & \var \mid \lam{\var}{\tm} \mid \tm\,\tmtwo
  \end{array}
\]
An \defn{\typeEnvironment}, ranged over by $\penv,\penvtwo,\hdots$,
is a finite set of type assignments $\set{\var_1:\ttm_1,\hdots,\var_n:\ttm_n}$,
where $\var_1,\hdots,\var_n$ are distinct variables.
We write $\dom{\penv}$ for the set of variables that appear in $\penv$.
\defn{Typing judgements} are of the form $\judg{\penv}{\tm}{\ttm}$,
and defined inductively by the following rules:
\[
  \indrule{}{
  }{
    \judg{\penv,\var:\ttm}{\var}{\ttm}
  }
  \,\,\,\,\,\,\,\,
  \indrule{}{
    \judg{\penv,\var:\ttm}{\tm}{\ttmtwo}
  }{
    \judg{\penv}{\lam{\var}{\tm}}{\ttm\arrow\ttmtwo}
  }
  \,\,\,\,\,\,\,\,
  \indrule{}{
    \judg{\penv}{\tm}{\ttm\arrow\ttmtwo}
    \HS
    \judg{\penv}{\tmtwo}{\ttm}
  }{
    \judg{\penv}{\tm\,\tmtwo}{\ttmtwo}
  }
\]
The binary relations of $\beta$ and $\eta$-reduction
are defined by the following reduction rules, closed under arbitrary contexts:
\[
  \begin{array}{r@{\,}c@{\,}l@{\HS\HS}r@{\,}c@{\,}ll}
    (\lam{\var}{\tm})\,\tmtwo & \tobeta & \tm\sub{\var}{\tmtwo}
    &
    \lam{\var}{\tm\,\var} & \toeta & \tm & \text{if $\var\notin\fv{\tm}$}\\
  \end{array}
\]
Both notions of reduction are defined over untyped \proofTerms,
but they can be restricted to typable \proofTerms, using the well-known
property that they preserve types (\ie \emph{subject reduction} holds).

\paragraph{The metacalculus ($\lambdaCheck$).}
Next, we define the \textbf{metacalculus}
in which we will be able to formulate the notions of verifiers, generators,
realizers, and so on, for \minimalLogic.
The metacalculus $\lambdaCheck$ is an extension of the \emph{untyped}
$\lambda$-calculus.
It would be possible to formulate the semantics using the (unextended)
$\lambda$-calculus directly, but working in an extension is
clearer, and more convenient, as it avoids the accidental complexities
of encodings.

To emphasize their specific role, and to avoid confusion
with $\lambdaST$-terms, terms of the metacalculus are called \emph{\metaterms}.
The set of \defn{\metaterms} of $\lambdaCheck$ is given by the grammar:
\[
  \begin{array}{llrll}
  \textsc{\Metaterms}
  &
    \utm,\utmtwo,\hdots
  & ::=  & \var
         & \text{variable}
  \\
  &
  & \mid & \lam{\var}{\utm}
         & \text{abstraction}
  \\
  &
  & \mid & \utm\,\utmtwo
         & \text{application}
  \\
  &
  & \mid & \iunit
         & \text{\unitConstructor}
  \\
  &
  & \mid & \eunit{\utm}{\utmtwo}
         & \text{\unitEliminator}
  \\
  &
  & \mid & \genteig
         & \text{generator}
  \\
  &
  & \mid & \verifteig{\utm}
         & \text{verifier}
  \end{array}
\]
For every \atomicProposition~$\teig$
there are \metaterms $\genteig$ and $\verifteig{\utm}$.
In the case of an arbitrary \logicalTerm $\ttm$,
the \metaterms $\gen{\ttm}$ and $\verif{\ttm}{\utm}$
shall be defined below~(\cref{def:generators_and_verifiers}).

We write $\wctx,\wctx',\hdots$ to stand for the set of \defn{weak head contexts},
defined by the following grammar:
\[
  \wctx ::= \ctxhole
       \mid \wctx\,\utm
       \mid \eunit{\wctx}{\utm}
       \mid \verifteig{\wctx}
\]
We define \defn{reduction} in $\lambdaCheck$, written $\to$, as the union
of the three reduction rules below, closed by arbitrary contexts.
Moreover, \defn{weak head reduction}, written $\tow$,
is defined as the union of the same three reduction rules,
but closed by weak head contexts.
\[
  \begin{array}{rll}
    (\lam{\var}{\utm})\,\utmtwo
    & \to &
    \utm\sub{\var}{\utmtwo}
  \\
    \eunit{\iunit}{\utm}
    & \to &
    \utm
  \\
    \verifteig{\genteig}
    & \to &
    \iunit
  \end{array}
\]
We use the constant $\iunit$ to represent that a computation \emph{succeeds}.
The \unitEliminator $\eunit{\utm}{\utmtwo}$ can be used to check whether $\utm$
succeeds and return $\utmtwo$ if this is the case.
A \metaterm of the form $\genteig$
represents a generic realizer for the \atomicProposition $\teig$,
while a \metaterm of the form $\verifteig{\utm}$ intuitively ``checks''
whether $\utm$ is a realizer for the \atomicProposition $\teig$.

We adopt the convention that the \unitEliminator has higher precedence than
the abstraction, but lower precedence than the application and
the verifier, so for example
$\lam{\var}{\eunit{\verifteig{\utm}}{\utmtwo\,\utmthree}}$
is equivalent to
$\lam{\var}{(\eunit{(\verifteig{\utm})}{(\utmtwo\,\utmthree)})}$.

A \metaterm is \defn{pure} if it only involves variables,
abstractions, and applications.
Terms of $\lambdaST$ are identified with pure \metaterms of $\lambdaCheck$.
Note that reduction in $\lambdaCheck$ over pure terms is just
standard $\tobeta$-reduction.
A \metaterm is \defn{good} if it is pure and $\tobeta$-normalizing.

\paragraph{Interpretation of \logicalTerms.}
The metacalculus $\lambdaCheck$ includes generators and verifiers
\emph{only for \atomicPropositions}.
We extend them to arbitrary \logicalTerms $\ttm$
by providing mutually recursive definitions of the generator
$\gen{\ttm}$ that produces a generic realizer of $\ttm$,
and the verifier $\verif{\ttm}{\utm}$ that checks whether
its argument $\utm$ is a realizer for $\ttm$.

\begin{definition}[Generators and verifiers for $\lambdaCheck$]
\label{def:generators_and_verifiers}
Let $\ttm$ be a \logicalTerm.
The \defn{generator} for $\ttm$, written $\gen{\ttm}$,
and the \defn{verifier} for $\ttm$ with argument $\utm$, written $\verif{\ttm}{\utm}$,
are \metaterms defined as follows by induction on $\ttm$.
In the base case, when $\ttm = \teig$ is an \atomicProposition,
the definition coincides with the \metaterm constructors in the grammar
($\genteig$ and $\verifteig{\utm}$).
In the recursive case:
\[
  \begin{array}{rcl}
    \gen{\ttm\arrow\ttmtwo}
  & \eqdef &
    \lam{\var}{(\eunit{\verif{\ttm}{\var}}{\gen{\ttmtwo}})}
  \\
    \verif{\ttm\arrow\ttmtwo}{\utm}
  & \eqdef &
    \verif{\ttmtwo}{(\utm\,\gen{\ttm})}
  \end{array}
\]
\end{definition}

The definition of $\gen{\ttm\arrow\ttmtwo}$ intuitively means that,
to generate a generic realizer for $\ttm\arrow\ttmtwo$,
we may construct an
abstraction $\lam{\var}{(\eunit{\verif{\ttm}{\var}}{\gen{\ttmtwo}})}$
that checks whether its argument is a realizer for $\ttm$,
and, if this succeeds, generates a realizer for $\ttmtwo$.
The definition of $\verif{\ttm\arrow\ttmtwo}{\utm}$ intuitively means that,
to verify whether $\utm$ is a realizer for $\ttm\arrow\ttmtwo$,
we provide $\utm$ with a generic realizer for $\ttm$, and check whether the
result is a realizer for $\ttmtwo$.

Formally, we say that $\utm$ is a \defn{realizer} for $\ttm$
if $\verif{\ttm}{\utm} \tos \iunit$, and we say that $\ttm$ is
\defn{realizable} if it has a good realizer.

\paragraph{Note on syntax.}
The syntax for verifiers requires $\verif{\ttm}{}$ to always
be accompanied by its argument $\utm$. This is just a syntactic
constraint that helps simplify some proofs.
If the verifier needs to be treated as a stand-alone operator
it can be defined via an $\eta$-like expansion
as $\verif{\ttm}{} \eqdef \lam{\var}{\verif{\ttm}{\var}}$.

\begin{example}[The $K$ and $S$ combinators realize Hilbert-style axioms]
The following reduction shows that the $K$ combinator
($\lam{\var}{\lam{\vartwo}{\var}}$)
realizes the first axiom ($\teig\imp\teigtwo\imp\teig$)
in the Hilbert-style system
for \minimalLogic:
\[
  \begin{array}{lll}
    \verif{\teig\imp\teigtwo\imp\teig}{(\lam{\var}{\lam{\vartwo}{\var}})}
  & = &
    \verif{\teigtwo\imp\teig}{((\lam{\var}{\lam{\vartwo}{\var}})\,\gen{\teig})}
  \\
  & = &
    \verif{\teig}{((\lam{\var}{\lam{\vartwo}{\var}})\,\gen{\teig}\,\gen{\teigtwo})}
  \\
  & \tos &
    \verif{\teig}{\gen{\teig}} \to \iunit
  \end{array}
\]
The $S$ combinator
($\lam{\var}{\lam{\vartwo}{\lam{\varthree}{\var\,\varthree\,(\vartwo\,\varthree)}}}$),
in turn, realizes the second axiom
$((\teig\imp\teigtwo\imp\teigthree) \imp (\teig\imp\teigtwo) \imp \teig\imp\teigthree)$:
\[
  \begin{array}{ll}
  &
    \verif{
      (\teig\imp\teigtwo\imp\teigthree)
       \imp (\teig\imp\teigtwo)
       \imp \teig\imp\teigthree
    }{
      (\lam{\var}{\lam{\vartwo}{\lam{\varthree}{\var\,\varthree\,(\vartwo\,\varthree)}}})
    }
  \\
  = &
    \verif{
      \teigthree
    }{
      (
      (\lam{\var}{\lam{\vartwo}{\lam{\varthree}{\var\,\varthree\,(\vartwo\,\varthree)}}})
      \,
      \gen{\teig\imp\teigtwo\imp\teigthree}
      \,
      \gen{\teig\imp\teigtwo}
      \,
      \gen{\teig}
      )
    }
  \\
  \tos &
    \verif{
      \teigthree
    }{
      (
      \gen{\teig\imp\teigtwo\imp\teigthree}\,\gen{\teig}\,
      (\gen{\teig\imp\teigtwo}\,\gen{\teig})
      )
    }
  \\
  = &
    \verif{
      \teigthree
    }{
      (
      (\lam{\var}{\eunit{\verif{\teig}{\var}}{\gen{\teigtwo\imp\teigthree}}})\,\gen{\teig}\,
      (\gen{\teig\imp\teigtwo}\,\gen{\teig})
      )
    }
  \\
  \to &
    \verif{
      \teigthree
    }{
      (
      (\eunit{\verif{\teig}{\gen{\teig}}}{\gen{\teigtwo\imp\teigthree}})\,
      (\gen{\teig\imp\teigtwo}\,\gen{\teig})
      )
    }
  \\
  \tos &
    \verif{
      \teigthree
    }{
      (
      \gen{\teigtwo\imp\teigthree}\,
      (\gen{\teig\imp\teigtwo}\,\gen{\teig})
      )
    }
  \\
  \tos &
    \verif{
      \teigthree
    }{
      (
      (\lam{\var}{\eunit{\verif{\teigtwo}{\var}}{\gen{\teigthree}}})
      (\gen{\teig\imp\teigtwo}\,\gen{\teig})
      )
    }
  \\
  \tos &
    \verif{
      \teigthree
    }{
      (
      \eunit{\verif{\teigtwo}{(\gen{\teig\imp\teigtwo}\,\gen{\teig})}}{\gen{\teigthree}}
      )
    }
    \tos \hdots \to \iunit
  \end{array}
\]
\end{example}

The two following examples illustrate applications of \epistemic realizability.
Note that these are well-known results that can be established through other
methods, such as standard realizability interpretations.

\begin{example}[Consistency]
It straightforward to show that \emph{not every \logicalTerm is realizable}.
Indeed, let $\teig$ be an arbitrary \atomicProposition
and suppose that $\verif{\teig}{\utm} \tos \iunit$ for a pure
\metaterm $\utm$.
Inspecting the rewriting rules of $\lambdaCheck$,
it is clear that this could only be possible if there were a reduction
$\utm \tos \gen{\teig}$. But $\utm$ is pure and contains no generators,
and it is straighforward to show that reduction cannot create generators,
so we obtain a contradiction.
Combining this with soundness~(\cref{thm:soundness}), we conclude that not
every formula is a theorem in \minimalLogic; that is, the logic is \emph{consistent}.
\end{example}

\begin{example}[Specification results]
The interpretation can be used to characterize the inherent
computational behaviour of the realizers of a given \logicalTerm.
For example, we claim that if $\utm$ is a pure term that realizes $\teig\imp\teig$,
then $\utm \eqbeta \lam{\var}{\var}$.
To see this (informally), suppose that $\verif{\teig\imp\teig}{\utm} \tos \iunit$.
By definition, we have $\verif{\teig}{(\utm\,\gen{\teig})} \tos \iunit$.
This can only happen if $\utm\,\gen{\teig} \tos \gen{\teig}$.
But $\utm$ is a pure term without verifiers, and reduction cannot spontaneously
create a verifier $\verif{\teig}{}$. So $\gen{\teig}$ can be understood as
just an ``inert'' symbol that does not interact with $\utm$,
and it can be replaced by a fresh variable $\var$
to obtain a reduction $\utm\,\var \tos \var$.
Since reduction in $\lambdaCheck$ over pure terms
is $\tobeta$-reduction, it follows that $\utm\,\var \eqbeta \var$,
which is easily seen to imply $\utm \eqbeta \lam{\var}{\var}$.
Combining this with soundness~(\cref{thm:soundness}), one obtains that
\emph{the unique realizer of $\teig\imp\teig$ is the identity}.
This is a classic result, dating back at least to Reynolds~\cite{Reynolds1983}.
\end{example}

\paragraph{Rewriting properties of $\lambdaCheck$.}
Before going on, we mention some basic syntactic results that are necessary
to show that the interpretation is sound and complete.
First, it is straightforward to show that there is at most
one weak head redex~(see~\cref{a:sec:determinism}). Hence:
\begin{proposition}
\label{prop:determinism}
Weak head reduction $\tow$ is deterministic.
\end{proposition}

Moreover, $\lambdaCheck$ is an
\emph{orthogonal higher-order rewriting system} in the sense of
Nipkow~\cite{Nipkow1991}, which entails:
\begin{proposition}
\label{prop:confluence}
Reduction in $\lambdaCheck$ is confluent.
\end{proposition}

A further rewriting property that we need is that
reductions ending in $\iunit$ can be \emph{standardized}
to weak head reductions.
Compare this property with the well-known fact that the head
reduction strategy is head normalizing in the
$\lambda$-calculus~(see \eg~\cite[Coro.~11.4.8]{Barendregt1984}).
We do not include an explicit proof of this result, but \cref{sec:second_order}
proves an analogous fact for the (more difficult) second-order case:

\begin{proposition}[Standardization]
\label{prop:weak_standardization}
If $\utm \tos \iunit$ then $\utm \tows \iunit$.
\end{proposition}

\subsection{Soundness of $\lambdaST$ with respect to $\lambdaCheck$}

Our goal is to show that the typed calculus $\lambdaST$ is sound
and complete with respect to the verification semantics induced
by the metacalculus $\lambdaCheck$.
In this section we focus on \textbf{soundness}, which states that
all theorems of \minimalLogic are realizable.
More precisely, whenever $\tm$ is a \emph{proof} of $\ttm$,
\ie a closed \proofTerm $\tm$ such that $\judg{}{\tm}{\ttm}$ holds,
then $\tm$ is also a \emph{realizer}\footnote{
  Recall that all \proofTerms of the typed calculus $\lambdaST$
  can be understood as pure \metaterms of the metacalculus $\lambdaCheck$.
} for $\ttm$, that is, $\verif{\ttm}{\tm} \tos \iunit$.
Moreover, $\tm$ is clearly \emph{good}:
it is pure by construction, and the standard fact that simply typed
terms are strongly normalizing implies that $\tm$ is $\beta$-normalizing.

\paragraph{Stating soundness for open terms.} 
To be able to reason inductively over the structure of \proofTerms,
we are compelled to generalize the statement of soundness
to \emph{open} \proofTerms, \ie \proofTerms that contain free variables.
The generalized version states that if
$\judg{\var_1:\ttm_1,\hdots,\var_n:\ttm_n}{\tm}{\ttmtwo}$
holds, then to obtain a realizer for $\ttmtwo$
one should replace each \proofVariable $\var_i$ by the generator $\gen{\ttm_i}$.
Formally, one has that 
$\verif{\ttmtwo}{\tm\sub{\var_1,\hdots,\var_n}{\gen{\ttm_1},\hdots,\gen{\ttm_n}}}
 \tos \iunit$.

\paragraph{Substitutions and generative substitutions.}
A \emph{substitution} ($\subst,\subst',\hdots$)
is a function mapping variables to \metaterms.
Recall that we write $\subs{\utm}{\subst}$ for the result
of applying the substition $\subst$ on each free variable of $\utm$,
avoiding capture.

The \defn{generative substitution} for an \typeEnvironment $\penv$
is written $\gensubst{\penv}$
and defined as the substitution
such that $\gensubst{\penv}(\var) = \gen{\ttm}$ holds for every $\var:\ttm \in \penv$,
and $\gensubst{\penv}(\var) = \var$ holds for every $\var \notin \dom{\penv}$.
By abuse of notation, sometimes we write just $\penv$ for the generative
substitution, rather than $\gensubst{\penv}$.
For example, if $\penv = (\var:\ttm,\vartwo:\ttmtwo)$ then:
\[
  \subs{((\lam{\var}{\var\,\vartwo\,\vartwo})\,\var)}{\penv}
  = (\lam{\var}{\var\,\gen{\ttmtwo}\,\gen{\ttmtwo}})\,\gen{\ttm}
\]
Using generative substitutions, soundness
can be refined to state that $\judg{\penv}{\tm}{\ttmtwo}$
implies $\verif{\ttmtwo}{\subs{\tm}{\penv}} \tos \iunit$.

\paragraph{Correctness and universality.}
The proofs of soundness and completeness that we provide in this section
are ``syntactic'' in nature, as they rely only on rewriting techniques and
elementary structural induction over the size of expressions.
We anticipate that the transition to second and higher-order logic in
the following sections shall require a different approach.

The proof of soundness of \minimalLogic relies on the two following
key properties:
\begin{enumerate}
\item
  \textbf{Correctness.}
  The generator $\gen{\ttm}$ is accepted by the verifier $\verif{\ttm}{}$,
  that is, $\verif{\ttm}{\gen{\ttm}} \tos \iunit$.
\item
  \textbf{Universality.}
  The generator $\gen{\ttm}$ is ``universal''
  with respect to the verifier $\verif{\ttm}$.
  More precisely, if $\utmProg$ and $\utm$ are \metaterms
  such that $\utmProg$ accepts the generator $\gen{\ttm}$,
  and the verifier $\verif{\ttm}{}$ accepts $\utm$,
  then $\utmProg$ accepts $\utm$.
  Schematically:
  \[
    \indrule{}{
      \utmProg\,\gen{\ttm} \tos \iunit
      \HS
      \verif{\ttm}{\utm} \tos \iunit
    }{
      \utmProg\,\utm \tos \iunit
    }
  \]
\end{enumerate}
These two conditions together state that verifier/generator
pairs enjoy a sort of ``universal property''\footnote{%
  It could be interesting to understand verifier/generator pairs from a categorical perspective.
}.
Correctness is easy to show:

\begin{lemma}[Correctness]
\label{lem:correctness}
$\verif{\ttm}{\gen{\ttm}} \tos \iunit$
\end{lemma}
\begin{proof}
By induction on $\ttm$.
If $\ttm = \teig$
then $\verifteig{\genteig} \to \iunit$.
If $\ttm = (\ttmtwo \arrow \ttmthree)$, then:
  \[
    \begin{array}{rl@{\,\,}llll}
    &&
      \verif{\ttmtwo\arrow\ttmthree}{\gen{\ttmtwo\arrow\ttmthree}}
    =
      \verif{\ttmthree}{(\gen{\ttmtwo\to\ttmthree}\,\gen{\ttmtwo})}
    \\
    & = &
      \verif{\ttmthree}{((\lam{\var}{\eunit{\verif{\ttmtwo}{\var}}{\gen{\ttmthree}}})\,\gen{\ttmtwo})}
    \to
      \verif{\ttmthree}{(\eunit{\verif{\ttmtwo}{\gen{\ttmtwo}}}{\gen{\ttmthree}})}
    \\
    & \tos_{\text{\ih}} &
      \verif{\ttmthree}{(\eunit{\iunit}{\gen{\ttmthree}})}
      \to
      \verif{\ttmthree}{\gen{\ttmthree}}
      \tos_{\text{\ih}}
      \iunit
      \qedhere
    \end{array}
  \]
\end{proof}
The property of universality is comparatively much subtler.
We first observe that if $\utmProg\,\gen{\ttm} \tos \iunit$
and $\utmProg$ is $\lambda$-abstraction that consumes its argument,
for example $\utmProg = \lam{\var}{\utmProg'}$,
then $\utmProg\,\gen{\ttm} \to \utmProg'\sub{\var}{\gen{\ttm}} \tos \iunit$.
Indeed, to reason inductively we need to consider the following generalization
of the statement of universality, in which occurrences of generators
$\gen{\ttm_1},\hdots,\gen{\ttm_n}$ inside a term may be replaced by
terms $\utm_1,\hdots,\utm_n$ that are accepted by the corresponding verifiers.
The property of universality becomes:
\[
  \indrule{}{
    \utmProg\sub{\var_1,\hdots,\var_n}{\gen{\ttm_1},\hdots,\gen{\ttm_n}} \tos \iunit
    \HS
    (\verif{\ttm_i}{\utm_i} \tos \iunit)_{1 \leq i \leq n}
  }{
    \utmProg\sub{\var_1,\hdots,\var_n}{\utm_1,\hdots,\utm_n} \tos \iunit
  }
\]
The first premise of this ``rule'' can be rephrased,
using a generative substitution, as
$\subs{\utmProg}{\penv} \tos \iunit$,
assuming that $\penv$ is the \typeEnvironment
$\penv = (\var_1:\ttm_1,\hdots,\var_n:\ttm_n)$.
We can also write the conclusion as
of the form $\subs{\utmProg}{\subst} \tos \iunit$
where $\subst$ is the substitution mapping each $\var_i$ to $\utm_i$.
More precisely:

\begin{definition}[Compatibility]
Let $\penv$ be an \typeEnvironment and $\subst$ be a substitution.
We say that $\penv$ is \defn{compatible} with $\subst$,
written $\compat{\penv}{\subst}$,
if for every $\var:\ttm \in \penv$ we have that $\var \in \dom{\subst}$
and $\verif{\ttm}{\subs{\var}{\subst}} \tos \iunit$.
\end{definition}

\begin{lemma}[Universality]
\label{lem:universality}
Let $\compat{\penv}{\subst}$
and $\subs{\utm}{\penv} \tos \iunit$.
Then $\subs{\utm}{\subst} \tos \iunit$.
\end{lemma}
\begin{proof}
We sketch the proof here; see \cref{a:sec:universality} for details.
The idea is first to standardize $\subs{\utm}{\penv} \tos \iunit$
using \cref{prop:weak_standardization}
to obtain $\subs{\utm}{\penv} \tows \iunit$.
Let $n$ be the length of this weak head reduction.
Write $\typesize{\ttm}$ for the size\footnote{Defined by
  $\typesize{\teig} \eqdef 1$ and
  $\typesize{\ttm\imp\ttmtwo} \eqdef 1 + \typesize{\ttm} + \typesize{\ttmtwo}$.}
of the \logicalTerm $\ttm$,
and $\typesize{\penv}$ for the maximum size of a \logicalTerm occurring in $\penv$,
\ie $\typesize{\penv} \eqdef \max\set{\typesize{\ttm} \ST \exists \var.\, (\var:\ttm \in \penv)}$.
We show that $\subs{\utm}{\subst} \tows \iunit$ holds
by induction on the lexicographic measure $(\typesize{\penv}, n)$.
The interesting case is when $\utm$ is of the form
$\utm = \ctxof{\wctx}{\var\,\utmtwo}$ where $\subs{\var}{\penv}$ is an abstraction.
Then $\var : (\ttm\arrow\ttmtwo) \in \penv$
for certain \logicalTerms $\ttm,\ttmtwo$, and the weak head reduction is of the form:
\[
  \begin{array}{rlll}
  &&
    \subs{\utm}{\penv}
    =
    \ctxof{\subs{\wctx}{\penv}}{\gen{\ttm\arrow\ttmtwo}\,\subs{\utmtwo}{\penv}}
    =
    \ctxof{\subs{\wctx}{\penv}}{(\lam{\var}{\eunit{\verif{\ttm}{\var}}{\gen{\ttmtwo}}})\,\subs{\utmtwo}{\penv}}
  \\
  & \tow &
    \ctxof{\subs{\wctx}{\penv}}{\eunit{\verif{\ttm}{\subs{\utmtwo}{\penv}}}{\gen{\ttmtwo}}}
    \tows
    \iunit
  \end{array}
\]
The tail of the reduction sequence
$\ctxof{\subs{\wctx}{\penv}}{\eunit{\verif{\ttm}{\subs{\utmtwo}{\penv}}}{\gen{\ttmtwo}}}
 \tows \iunit$
can be split into two reductions $\verif{\ttm}{\subs{\utmtwo}{\penv}} \tows \iunit$
and $\ctxof{\subs{\wctx}{\penv}}{\gen{\ttmtwo}} \tows \iunit$,
each shorter than the original one.

Consider the \metaterm $\ctxof{\wctx}{\vartwo}$,
where $\vartwo$ is fresh, and 
$\penv_1 := \penv\cup\set{\vartwo:\ttmtwo}$,
and $\subst_1 := \subst\extsub{\vartwo}{\subs{\var}{\subst}\,\gen{\ttm}}$.
The \ih may be applied since $\compat{\penv_1}{\subst_1}$ 
and $\typesize{\penv_1} = \typesize{\penv}$
and $\subs{\ctxof{\wctx}{\vartwo}}{\penv_1} = \ctxof{\subs{\wctx}{\penv}}{\gen{\ttmtwo}} \tows \iunit$
is shorter than the original reduction.
So we obtain that $\subs{\ctxof{\wctx}{\vartwo}}{\subst_1} \tows \iunit$,
which means that
\textbf{(I)}: $\ctxof{\subs{\wctx}{\subst}}{\subs{\var}{\subst}\,\gen{\ttm}} \tows \iunit$. 

Recall that $\verif{\ttm}{\subs{\utmtwo}{\penv}} \tows \iunit$.
The \ih may be applied a second time since $\compat{\penv}{\subst}$
and $\verif{\ttm}{\subs{\utmtwo}{\penv}} \tows \iunit$ is shorter than
the original reduction.
So we obtain \textbf{(II)}: $\verif{\ttm}{\subs{\utmtwo}{\subst}} \tows \iunit$.

Finally, consider the \metaterm $\ctxof{\subs{\wctx}{\subst}}{\subs{\var}{\subst}\,\varthree}$
where $\varthree$ is fresh, and $\penv_2 := \set{\varthree:\ttm}$,
and $\subst_2 := \extsub{\varthree}{\subs{\utmtwo}{\subst}}$.
The \ih may be applied a third time, since $\compat{\penv_2}{\subst_2}$ holds by \textbf{(II)},
and $\typesize{\penv_1} > \typesize{\penv}$,
and $\ctxof{\subs{\wctx}{\subst}}{\subs{\var}{\subst}\,\gen{\ttm}} \tows \iunit$
holds by \textbf{(I)}.
We obtain that
$\subs{\ctxof{\subs{\wctx}{\subst}}{\subs{\var}{\subst}\,\varthree}}{\subst_2}
 \tows \iunit$,
but
$\subs{\ctxof{\subs{\wctx}{\subst}}{\subs{\var}{\subst}\,\varthree}}{\subst_2}
 = \ctxof{\subs{\wctx}{\subst}}{\subs{\var}{\subst}\,\subs{\utmtwo}{\subst}}
 = \subs{\utm}{\subst}$
so we are done.
\end{proof}

Finally, we can show:

\begin{theorem}[Soundness]
\label{thm:soundness}
If $\judg{\penv}{\tm}{\ttm}$ then $\verif{\ttm}{\subs{\tm}{\penv}} \tos \iunit$.
\end{theorem}
\begin{proof}
By induction on the derivation of the judgment $\judg{\penv}{\tm}{\ttm}$.
In the \emph{variable} case, suppose that
$\judg{\penv}{\var}{\ttm}$ holds, so $\var:\ttm\in\penv$.
Then
  $\verif{\ttm}{\subs{\tm}{\penv}}
  = \verif{\ttm}{\subs{\var}{\penv}}
  = \verif{\ttm}{\gen{\ttm}}
  \tos \iunit$
  by correctness~(\cref{lem:correctness}).
\\
\indent
In the \emph{abstraction} case,
  we have that $\tm = \lam{\var}{\tmtwo}$
  and $\ttm = (\ttmtwo\arrow\ttmthree)$,
  and $\judg{\penv}{\lam{\var}{\tmtwo}}{\ttmtwo\arrow\ttmthree}$
  is derived from $\judg{\penv,\var:\ttmtwo}{\tmtwo}{\ttmthree}$.
  By $\alpha$-conversion, we may assume that $\var \notin \dom{\penv}$.
  Then we have
  $
    \verif{\ttmtwo\arrow\ttmthree}{\subs{(\lam{\var}{\tmtwo})}{\penv}}
    =
    \verif{\ttmthree}{((\lam{\var}{\subs{\tmtwo}{\penv}})\,\gen{\ttmtwo})}
    \to
    \verif{\ttmthree}{(\subs{\tmtwo}{\penv}\sub{\var}{\gen{\ttmtwo}})}
    =
    \verif{\ttmthree}{\subs{\tmtwo}{\penv,\var:\ttmtwo}}
    \tos_{\text{\ih}} \iunit
  $.
\\
\indent
Finally, in the \emph{application} case,
  we have that $\tm = \tmtwo\,\tmthree$
  and $\judg{\penv}{\tmtwo\,\tmthree}{\ttm}$
  is derived from $\judg{\penv}{\tmtwo}{\ttmtwo\arrow\ttm}$
  and $\judg{\penv}{\tmthree}{\ttmtwo}$.
  By \ih on the first premise,
  $\verif{\ttmtwo\arrow\ttm}{\subs{\tmtwo}{\penv}} \tos \iunit$,
  which by definition of $\verif{\ttmtwo\arrow\ttm}{}$
  means that
  $\verif{\ttm}{(\subs{\tmtwo}{\penv}\,\gen{\ttmtwo})} \tos \iunit$.
  By \ih on the second premise, we have that
  $\verif{\ttmtwo}{\subs{\tmthree}{\penv}} \tos \iunit$.
  Let $\vartwo$ be a fresh variable,
  consider the \metaterm
  $\verif{\ttm}{(\subs{\tmtwo}{\penv}\,\vartwo)}$,
  the typing context $\penvtwo := \set{\vartwo:\ttmtwo}$,
  and the substitution $\subst := \set{\vartwo \mapsto \subs{\tmthree}{\penv}}$.
  Remark that
  $\compat{\penvtwo}{\subst}$
  so by universality (\cref{lem:universality})
  we have that
  $\subs{(\verif{\ttm}{(\subs{\tmtwo}{\penv}\,\vartwo)})}{\penvtwo} \tos \iunit$
  implies
  $
    \subs{(\verif{\ttm}{(\subs{\tmtwo}{\penv}\,\vartwo)})}{\subst} \tos \iunit
  $.
  Performing the substitutions, this means that
  $
    \verif{\ttm}{(\subs{\tmtwo}{\penv}\,\gen{\ttmtwo})} \tos \iunit
  $
  implies
  $
    \verif{\ttm}{(\subs{\tmtwo}{\penv}\,\subs{\tmthree}{\penv})} \tos \iunit
  $. The antecedent of the implication holds, so
  $\verif{\ttm}{\subs{\tm}{\penv}}
   = \verif{\ttm}{(\subs{\tmtwo}{\penv}\,\subs{\tmthree}{\penv})}
   \tos \iunit$.
\end{proof}

\subsection{Completeness of $\lambdaST$ with respect to $\lambdaCheck$}

Completeness is the statement that \emph{realizable \logicalTerms are provable},
\ie that if $\utm$
is a good \metaterm such that $\verif{\ttm}{\utm} \tos \iunit$
then $\vdash \ttm$ holds in \minimalLogic.
As we did for soundness, we have to generalize the statement in
order to account for possibly open terms.

The generalized statement of completeness says that if
$\utm$ is a good \metaterm
such that $\verif{\ttm}{\subs{\utm}{\penv}} \tos \iunit$
then there exists a \proofTerm $\tm$ such that $\judg{\penv}{\tm}{\ttm}$.
Moreover, $\tm$ is such that $\utm \eqbeta \tm$.
In fact, $\tm$ can be chosen to be the $\beta$-normal form of $\utm$.

\paragraph{Measuring the size of realizers.}
The proof of completeness proceeds by induction on an integer
$\sig{\utm}{\ttm}$ called the \defn{proof size}.
It is defined for any pure \metaterm $\utm$,
\textbf{even if it is untypable}, assuming that $\utm$ has a $\beta$-normal form.
Recall the well-known fact that the following grammar of \emph{normal \proofTerms}
characterizes exactly the set of \proofTerms in $\beta$-normal form
(which coincides with the set of \emph{pure \metaterms} in normal form):
\begin{definition}[Neutral and normal \proofTerms]
\label{def:grammar_of_nfs}
\[
  \begin{array}{lrcl}
  \textsc{Neutral \proofTerms} & \uneu & ::= & \,\, \var\,\unf_1\hdots\unf_n \HS (n \geq 0) \\
  \textsc{Normal \proofTerms} & \unf & ::= & \,\, \uneu \,\,\mid\,\, \lam{\var}{\unf}
  \end{array}
\]
\end{definition}

\begin{definition}[Proof size]
\label{def:proof_size}
The \defn{proof size} of a normal \proofTerm $\unf$ with respect to a \proofEnvironment $\penv$
and a \logicalTerm $\ttm$
is a non-negative integer written $\sig{\unf}{\ttm}$
and defined as follows, by induction on the shape of $\unf$,
according to \cref{def:grammar_of_nfs}:
\[
  \begin{array}{l@{\,\,}l@{\,\,}ll}
    \sig{\var\,\unf_1\hdots\unf_n}{\ttm}
    & \eqdef &
    \multicolumn{2}{l}{\!\!\!
          \minof{\size{\ttm}}{\size{\ttm'}}
        + \sum_{i=1}^n\set{{\sig{\unf_i}{\ttm_i}}}
    }
  \\
    &&
    \multicolumn{2}{l}{
    \HS\text{(if $\var: (\ttm_1\imp\hdots\imp\ttm_n\imp\ttm') \in \penv$)}
    }
  \\
      \sig{\var\,\unf_1\hdots\unf_n}{\ttm}
    & \eqdef &
      0
    & \text{(otherwise)}
  \\
      \sig{\lam{\var}{\unf}}{\ttm}
    & \eqdef &
      \sig[\penv,\var:\ttmtwo]{\unf}{\ttmthree}
    & \text{(if $\ttm = (\ttmtwo\imp\ttmthree)$)}
  \\
      \sig{\lam{\var}{\unf}}{\ttm}
    & \eqdef &
      0
    & \text{(otherwise)}
  \end{array}
\]
The definition is generalized for any pure \metaterm $\utm$
as long as it has a normal form $\unf$,
defining $\sig{\utm}{\ttm} \eqdef \sig{\unf}{\ttm}$.
\end{definition}

One key observation is that
$\sig{\utm}{\ttm\imp\ttmtwo} \geq \sig[\penv,\pvar:\ttm]{\utm\,\var}{\ttmtwo}$.
Another key observation is that
if $\var:(\ttm_1\imp\hdots\imp\ttm_n\imp\ttm')\in\penv$
then $\sig{\var\utm_1\hdots\utm_n}{\ttm'} > \sig{\utm_i}{\ttm_i}$
for every $1 \leq i \leq n$
(see \cref{lemma:size_normalize,lemma:sizeVar_argument} in the
appendix for details).

\begin{theorem}[Completeness]
\label{thm:completeness}
If $\verif{\ttm}{\subs{\utm}{\penv}} \tos \iunit$
and $\utm$ is good, then its $\beta$-normal form
is a \proofTerm $\utm^\downarrow$ such that $\judg{\penv}{\utm^\downarrow}{\ttm}$.
\end{theorem}
\begin{proof}
We sketch the proof here; see \cref{a:sec:completeness} for details.
Let $\unf$ be the $\beta$-normal form of $\utm$, which is a normal \proofTerm.
We proceed by induction on $\sig{\unf}{\ttm}$.
Write $\ttm = (\ttm_1 \arrow \hdots \ttm_n \arrow \teig)$.
Note that
$\verif{\ttm}{\subs{\utm}{\penv}}
 \tos \verif{\teig}{(\subs{\utm}{\penv}\,\gen{\ttm_1}\hdots\gen{\ttm_n})}
 \tos \iunit$.
Let $\var_1,\hdots,\var_n$ be fresh variables,
and consider the \metaterm $\utm\,\var_1\hdots\var_n$
and the \typeEnvironment $\penvtwo := (\penv,\var_1:\ttm_1,\hdots,\var_n:\ttm_n)$.
Then
$\verif{\teig}{\subs{(\utm\,\var_1\hdots\var_n)}{\penvtwo}}
= \verif{\teig}{(\subs{\utm}{\penv}\,\gen{\ttm_1}\hdots\gen{\ttm_n})}
\tos \iunit$.
Let $\utmtwo^\downarrow$ be the $\beta$-normal form of the application
$\utm^\downarrow\,\var_1\hdots\var_n$.
By confluence (\cref{prop:confluence}),
$\verif{\teig}{\subs{(\utmtwo^\downarrow)}{\penvtwo}} \tos \iunit$.
Note that
$\sig{\unf}{\ttm}
 = \sig{\unf}{\ttm_1\imp\hdots\ttm_n\imp\teig}
 \geq \sig[\penv,\var_1:\ttm_1,\hdots,\var_n:\ttm_n]{\unftwo}{\teig}
 = \sig[\penvtwo]{\unftwo}{\teig}$.
We consider two cases, depending on whether $\utmtwo^\downarrow$
is an abstraction or a neutral \proofTerm.

If $\utmtwo^\downarrow$ is an \textbf{abstraction},
  \ie $\utmtwo^\downarrow = \lam{\vartwo}{\utmthree}$,
  then we have that
  $\verifteig{\lam{\vartwo}{\subs{\utmthree}{\penvtwo}}}
  = \verifteig{\subs{(\lam{\vartwo}{\utmthree})}{\penvtwo}}
  = \verifteig{\subs{(\utmtwo^\downarrow)}{\penvtwo}}
  \tos \iunit$,
  which is impossible,
  because the reducts of an abstraction are always abstractions,
  so in particular $\lam{\vartwo}{\subs{\utmthree}{\penvtwo}}$
  cannot have $\genteig$ as a reduct.

If $\utmtwo^\downarrow$ is a \textbf{neutral term},
  then $\utmtwo^\downarrow = \var\,\utmtwo_1\hdots\utmtwo_m$,
  and we know that
  $\verifteig{(\subs{\var}{\penvtwo}\,\subs{\utmtwo_1}{\penvtwo}\hdots\subs{\utmtwo_m}{\penvtwo})}
  \tos \iunit$.
  Note that $\var\in\dom{\penvtwo}$
  because otherwise we would have
  $\verifteig{(\var\,\subs{\utmtwo_1}{\penvtwo}\hdots\subs{\utmtwo_m}{\penvtwo})}
   \not{\tos} \iunit$, a contradiction.

  Suppose then that $\var:\ttmtwo\in\penvtwo$ for some $\ttmtwo$.
  Without loss of generality, write
  $\ttmtwo = (\ttmtwo_1\arrow\hdots\arrow\ttmtwo_k\arrow\teigtwo)$.
  We write $\lamt{\ttmthree}{\utmthree}$ to abbreviate $\lam{\var}{(\eunit{\verif{\ttmthree}{\var}}{\utmthree})}$
  assuming that $\var \notin \fv{\utmthree}$.
  Note that:
  \[
    \begin{array}{l@{\,\,}l@{\,\,}lll}
      \verifteig{((\lamt{\ttmtwo_1}{\hdots\lamt{\ttmtwo_k}{\gen{\teigtwo}}})\,\subs{\utmtwo_1}{\penvtwo}\hdots\subs{\utmtwo_m}{\penvtwo})}
    & = &
      \verifteig{(\gen{\ttmtwo}\,\subs{\utmtwo_1}{\penvtwo}\hdots\subs{\utmtwo_m}{\penvtwo})}
    \\
    & = &
      \verifteig{(\subs{\var}{\penvtwo}\,\subs{\utmtwo_1}{\penvtwo}\hdots\subs{\utmtwo_m}{\penvtwo})}
      \,\,\tos\,\,\iunit
    \end{array}
  \]
  The cases $k < m$ and $k > m$ are impossible because the reduction
  of $\verifteig{((\lamt{\ttmtwo_1}{\hdots\lamt{\ttmtwo_k}{\gen{\teigtwo}}})\,\subs{\utmtwo_1}{\penvtwo}\hdots\subs{\utmtwo_m}{\penvtwo})}$ would become ``stuck''
  (see \cref{a:thm:completeness} in the appendix for a detailed proof).
  Hence $k = m$
  and we must have that $\verif{\ttmtwo_i}{\subs{\utmtwo_i}{\penvtwo}} \tos \iunit$
  for every $1 \leq i \leq k$;
  otherwise reduction would also become ``stuck''.
  Then we have that
  $
    \verifteig{((\lamt{\ttmtwo_1}{\hdots\lamt{\ttmtwo_k}{\gen{\teigtwo}}})\,\subs{\utmtwo_1}{\penvtwo}\hdots\subs{\utmtwo_k}{\penvtwo})}
    \tos
    \verifteig{\gen{\teigtwo}}
    \to \iunit
  $, and the only way this can happen is if $\teig = \teigtwo$.

  For each $1 \leq i \leq k$
  note that
  $
    \sig{\utm^\downarrow}{\ttm}
    \geq
    \sig[\penvtwo]{\utmtwo^\downarrow}{\teig}
    >
    \sig[\penvtwo]{\utmtwo_i}{\ttmtwo_i}
  $
  holds,
  and note also that $\verif{\ttmtwo_i}{\subs{\utmtwo_i}{\penvtwo}} \tos \iunit$,
  so by \ih
  $\judg{\penvtwo}{\utmtwo^\downarrow_i}{\ttmtwo_i}$.
  Since $\var:(\ttmtwo_1\arrow\hdots\arrow\ttmtwo_k\arrow\teig) \in \penv \subseteq \penvtwo$,
  we can derive the judgment
  $\judg{\penvtwo}{\var\,\unftwo_1\hdots\unftwo_k}{\teig}$.
  Take $\tm := \lam{\var_1\hdots\var_n}{\var\,\utmtwo^\downarrow_1\hdots\utmtwo^\downarrow_k}$.
  Note that:
  \[
    \begin{array}{llll}
      \tm
    & = &
      \lam{\var_1\hdots\var_n}{\var\,\utmtwo^\downarrow_1\hdots\utmtwo^\downarrow_k}
    \\
    & \eqbeta &
      \lam{\var_1\hdots\var_n}{\var\,\utmtwo_1\hdots\utmtwo_k}
    \\
    & = &
      \lam{\var_1\hdots\var_n}{\utmtwo^\downarrow}
    \\
    & \eqbeta &
       \lam{\var_1\hdots\var_n}{\utm^\downarrow\,\var_1\hdots\var_n}
      & \text{as $\utm^\downarrow\,\var_1\hdots\var_n \tobetas \utmtwo^\downarrow$}
    \\
    & \toetas &
       \utm^\downarrow
       & \text{using $n$ steps of $\eta$-reduction}
    \end{array}
  \]
  Thus $\tm \eqbeta\toetas \utm^\downarrow \tobetainvs \utm$.
  By confluence of $\tobeta$, we have
  $\tm \tobetas \tobetainvs \toetas \utm^\downarrow \tobetainvs \utm$.
  Moreover, since $\tobeta$ and $\toeta$ commute and $\utm^\downarrow$ is in $\beta$-normal form,
  the situation is $\tm \tobetas \toetas \utm^\downarrow \tobetainvs \utm$.

  Finally, since $\penvtwo = (\penv,\var_1:\ttm_1,\hdots,\var_n:\ttm_n)$,
  we have that
  $\judg{\penv}{\lam{\var_1\hdots\var_n}{\var\,\utmtwo^\downarrow_1\hdots\utmtwo^\downarrow_k}}{\ttm_1\arrow\hdots\ttm_n\arrow\teig}$,
  that is 
  $\judg{\penv}{\tm}{\ttm}$.
  Since $\tm \tobetas\toetas \utm^\downarrow$,
  by subject reduction we conclude $\judg{\penv}{\utm^\downarrow}{\ttm}$.
\end{proof}

\section{Extension to Second-Order Logic}
\label{sec:second_order}

One way to extend the \epistemic realizability interpretation
from the previous section
would be to define verifiers and generators for additional propositional
connectives.
For example, one might define the verifier for conjunction as
$\verif{\ttm\land\ttmtwo}{\utm} =
 \eunit{\verif{\ttm}{\pi_1(\utm)}}{\verif{\ttmtwo}{\pi_2(\utmtwo)}}$,
and the generator as
$\gen{\ttm\land\ttmtwo} = (\gen{\ttm},\gen{\ttmtwo})$.
However, it is not immediately obvious how to define verifiers and generators
for \emph{positive} connectives such as disjunction ($\ttm\lor\ttmtwo$).

We take a different route, and extend the interpretation from \minimalLogic
to \textbf{second-order logic}, incorporating second-order universal
quantification. It is well-known that other logical connectives
(including the intuitionistic propositional connectives),
as well as inductive datatypes (natural numbers, lists, trees, etc.)
can be expressed in second-order logic, \eg using Böhm--Berarducci encodings~\cite{Bohm85}.

The main results in this section are \textbf{soundness}~(\cref{thm:f:soundness})
and \textbf{completeness}~(\cref{thm:f:completeness}).
While the proof of completeness is an adaptation of the one for \minimalLogic,
the proof of soundness requires a fundamentally different approach.

\paragraph{The second-order $\lambda$-calculus ($\lambdaF$).}
We begin by recalling the typed $\lambda$-calculus corresponding
to second-order logic, namely Girard--Reynolds System~F~\cite{GirardThesis,Reynolds1974}.
We assume given disjoint denumerable sets
of \defn{\logicalVariables} ($\tvar,\tvartwo,\hdots$),
and \defn{\logicalEigenvariables} ($\teig,\teigtwo,\hdots$).
The sets of \emph{\logicalTerms} and \emph{\proofTerms}
are given by the grammar:
\[
  \begin{array}{lrcl}
  \textsc{\LogicalTerms} &
  \ttm, \ttmtwo, \hdots
    & ::= &
           \tvar
      \mid \teig
      \mid \ttm\imp\ttmtwo
      \mid \allf{\tvar}{\ttm}
  \\
  \textsc{\ProofTerms} &
  \ptm,\ptmtwo,\hdots
    & ::= &
           \pvar
      \mid \plamf{\pvar}{\ptm}
      \mid \ptm\,\ptmtwo
      \mid \pallfi{\tvar}{\ptm}
      \mid \ptm\,\ttm
  \end{array}
\]
The (second-order) universal quantifier $\allf{\tvar}{\ttm}$
binds the occurrences of $\tvar$ in $\ttm$.
\LogicalEigenvariables are similar to \logicalVariables, but they
are assumed to never occur bound in \logicalTerms or \proofTerms
(although it will be possible that they occur bound in \metaterms).
A \logicalTerm is \defn{atomic} if it is either a \logicalVariable
or an \logicalEigenvariable.

From the propositions-as-types point of view,
$\plamf{\pvar}{\ptm}$ and $\ptm\,\ptmtwo$
encode introduction and elimination of the implication,
while $\pallfi{\tvar}{\ptm}$ and $\ptm\,\ttm$
encode introduction and elimination of the second-order universal quantifier.
Typing judgments are defined with the three typing rules for $\lambdaST$ given
in the previous section, plus the two following rules:
\[
  \indrule{}{
    \judgf{\penv}{\ptm}{\ttm}
    \HS
    \tvar\notin\fv{\penv}
  }{
    \judgf{\penv}{\pallfi{\tvar}{\ptm}}{\allf{\tvar}{\ttm}}
  }
  \HS
  \indrule{}{
    \judgf{\penv}{\ptm}{\allf{\tvar}{\ttm}}
  }{
    \judgf{\penv}{\ptm\,\ttmtwo}{\ttm\sub{\tvar}{\ttmtwo}}
  }
\]
We define \defn{$\beta$} and $\eta$-reduction as the union of the following
rules, closed by arbitrary contexts:
\[
  \begin{array}{r@{\,\,}l@{\,\,}l@{\HS}r@{\,\,}l@{\,\,}l}
    (\plamf{\pvar}{\ptm})\,\ptmtwo
    & \tobeta &
    \ptm\sub{\pvar}{\ptmtwo}
  &
    (\pallfi{\tvar}{\ptm})\,\ttm
    & \tobeta &
    \ptm\sub{\tvar}{\ttm}
  \\
    \plamf{\pvar}{\ptm\,\pvar}
    & \toeta &
    \ptm
    \,\,\,\,\text{if $\pvar \notin \fv{\ptm}$}
  &
    \pallfi{\tvar}{\ptm\,\tvar}
    & \toeta &
    \ptm
    \,\,\,\,\text{if $\tvar \notin \fv{\ptm}$}
  \end{array}
\]
Both notions of reduction are defined over arbitrary \proofTerms,
even if they are not typable, but they can also be restricted to
typable \proofTerms, using the well-known property that they preserve
types.

\paragraph{The second-order metacalculus ($\lambdaCheckF$).}
The set of \metaterms of $\lambdaCheckF$ is given by the following grammar:
\[
  \begin{array}{lrlllll}
  \multicolumn{7}{l}{\textsc{Metaterms} \HS \utm,\utmtwo,\hdots ::=} \\
  &      & \pvar
         & \text{\proofVariable}
  & \mid & \iunit
         & \text{\unitConstructor}
  \\
  & \mid & \plamf{\pvar}{\utm}
         & \text{\proofAbs}
  & \mid & \eunit{\utm}{\utmtwo}
         & \text{\unitEliminator}
  \\
  & \mid & \utm\,\utmtwo
         & \text{\proofApp}
  & \mid & \gen{\ttm}
         & \text{generator}
  \\
  & \mid & \pallfi{\tvar}{\utm}
         & \text{\logicalAbs}
  & \mid & \verif{\ttm}{\utm}
         & \text{verifier}
  \\
  & \mid & \utm\,\ttm
         & \text{\logicalApp}
  & \mid & \freshf{\teig}{\utm}
         & \text{fresh \logicalEigenvariable}
  \end{array}
\]
A \metaterm is \defn{pure} if it is built using only the five
productions on the left-hand side of the grammar.
\ProofTerms of $\lambdaF$ are identified with pure \metaterms of $\lambdaCheckF$.
A \metaterm is \defn{good} if it is pure and $\tobeta$-normalizing.

\paragraph{Syntactical differences between $\lambdaCheckF$ and $\lambdaCheck$.}
The relationship between $\lambdaCheckF$ and $\lambdaF$
is analogous to that of $\lambdaCheck$ and $\lambdaST$.
We mention some of the differences between $\lambdaCheckF$ and $\lambdaCheck$.
First, besides the usual abstraction and application,
  $\lambdaCheckF$
  also includes abstractions $\pallfi{\tvar}{\utm}$
  and applications $\utm\,\ttm$,
  which abstract over \logicalTerms.
Second,
  the syntax of the metacalculus for \minimalLogic $\lambdaCheck$
  of the previous section
  included only generators and verifiers of the forms
  $\gen{\teig}$ and $\verif{\teig}{}$,
  where $\teig$ was an atomic \logicalTerm.
  Generators and verifiers for more complex
  \logicalTerms were \emph{defined} operators
  (\cref{def:generators_and_verifiers}).
  Instead, the second-order metacalculus $\lambdaCheckF$
  includes generators and verifiers for \emph{arbitrary} \logicalTerms
  ($\gen{\ttm}$ and $\verif{\ttm}{}$) directly on the syntax.
  This is necessary because in $\lambdaCheckF$ \logicalVariables may become
  instantiated \emph{dynamically},
  for example $(\pallfi{\tvar}{\gen{\tvar}})\,\ttm \to \gen{\ttm}$.
Third, 
  $\lambdaCheckF$ includes a constructor $\freshf{\teig}{\utm}$
  that binds a local \logicalEigenvariable $\teig$,
  representing a \emph{symbolic} or \emph{constant} \logicalTerm,
  in the scope of $\utm$.
  This construct is used to verify universal quantifiers.

\paragraph{Reduction in $\lambdaCheckF$.}
\emph{Weak head contexts} are given by the grammar:
\[
  \begin{array}{l@{\,\,\,\,}l@{\,\,}l@{\,\,}l}
    \textsc{Weak head contexts} &
    \wctx & ::= & \ctxhole
              \mid \wctx\,\utm
              \mid \wctx\,\ttm
              \mid \eunit{\wctx}{\utm}
              \mid \verif{\ttm}{\wctx}
              \mid \freshf{\teig}{\wctx}
  \end{array}
\]

The \defn{reduction} relation in $\lambdaCheckF$, written $\to$,
is given by the union of the rewriting rules below,
closed by arbitrary contexts.
The binary relation of \defn{weak head reduction}, written $\tof$,
is also given by the union of the rewriting rules below,
but closed by weak head contexts.
\[
  \begin{array}{lrcll}
    \ruleNum{1} &
    (\plamf{\pvar}{\utm})\,\utmtwo
    & \to &
    \utm\sub{\pvar}{\utmtwo}
  \\
    \ruleNum{2} &
    (\pallfi{\tvar}{\utm})\,\ttm
    & \to &
    \utm\sub{\tvar}{\ttm}
  \\
    \ruleNum{3} &
    \eunit{\iunit}{\utm}
    & \to &
    \utm
  \\
    \ruleNum{4} &
    \verif{\teig}{\gen{\teig}}
    & \to &
    \iunit
  \\
    \ruleNum{5} &
    \gen{\ttm\imp\ttmtwo}
    & \to &
    \plamf{\pvar}{(\eunit{\verif{\ttm}{\pvar}}{\gen{\ttmtwo}})}
  \\
    \ruleNum{6} &
    \verif{\ttm\imp\ttmtwo}{\utm}
    & \to &
    \verif{\ttmtwo}{(\utm\,\gen{\ttm})}
  \\
    \ruleNum{7} &
    \gen{\allf{\tvar}{\ttm}}
    & \to &
    \pallfi{\tvar}{\gen{\ttm}}
  \\
    \ruleNum{8} &
    \verif{\allf{\tvar}{\ttm}}{\utm}
    & \to &
    \freshf{\teig}{\verif{\ttm\sub{\tvar}{\teig}}{(\utm\,\teig)}}
    & \text{($\teig$ fresh)}
  \\
    \ruleNum{9} &
    \freshf{\teig}{\utm}
    & \to &
    \utm
    & \text{(if $\teig\notin\fv{\utm}$)}
  \end{array}
\]
Rules~1 and~2 are $\beta$-reduction rules for \proofAbstraction
and \logicalAbstraction.
Rule~3 is the same as in $\lambdaCheck$.
Rule~4 is the same as in $\lambdaCheck$,
but only applies when $\teig$ is an \emph{\logicalEigenvariable},
not a \logicalVariable.
Rules~5 and~6 implement verification and generation of the implication,
with the same behavior as in \minimalLogic~(\cref{def:generators_and_verifiers})
but through reduction rules instead of definitional equations.
Rules~7 and~8 are the most interesting ones, as they
implement verification and generation of the second-order universal quantifier,
which is the new connective.
Rule~7 states that the generic realizer for $\allf{\tvar}{\ttm}$
is a \logicalAbstraction that receives a type $\tvar$ and returns the
generic realizer $\gen{\ttm}$, where $\ttm$ may depend on $\tvar$.
Rule~8 states that to verify whether $\utm$ is a realizer for $\ttm$
one has to check whether the application $\utm\,\teig$
is a realizer for $\verif{\ttm\sub{\tvar}{\teig}}$, where $\teig$ is
a fresh \logicalEigenvariable.
\LogicalEigenvariables represent ``symbolic'' or ``constant'' types.
Finally, rule~9 removes the fresh \logicalEigenvariable introduction
when it becomes unnecessary. 

\begin{example}[Encoding of conjunction]
Consider the usual second-order encoding of conjunction,
$\ttm\land\ttmtwo \eqdef \allf{\tvarthree}{(\ttm\imp\ttmtwo\imp\tvarthree)\imp\tvarthree}$.
Then $\plamf{\var}{\plamf{\vartwo}{\pallfi{\tvarthree}{\plamf{f}{f\,\var\,\vartwo}}}}$
realizes $\land$-introduction:
\[
  \begin{array}{ll}
  &
    \verif{\teig\imp\teigtwo\imp(\teig\land\teigtwo)}{
      (\plamf{\var}{\plamf{\vartwo}{\pallfi{\tvarthree}{\plamf{f}{f\,\var\,\vartwo}}}})
    }
  \\
  \tos &
    \verif{\teig\land\teigtwo}{
      (\pallfi{\tvarthree}{\plamf{f}{f\,\gen{\teig}\,\gen{\teigtwo}}})
    }
  \\
  \to &
    \freshf{\teigthree}{
      \verif{(\teig\imp\teigtwo\imp\teigthree)\imp\teigthree}{
        ((\pallfi{\tvarthree}{\plamf{f}{f\,\gen{\teig}\,\gen{\teigtwo}}})\,\teigthree)
      }
    }
  \\
  \to &
    \freshf{\teigthree}{
      \verif{(\teig\imp\teigtwo\imp\teigthree)\imp\teigthree}{
        (\plamf{f}{f\,\gen{\teig}\,\gen{\teigtwo}})
      }
    }
  \\
  \to &
    \freshf{\teigthree}{
      \verif{\teigthree}{
        ((\plamf{f}{f\,\gen{\teig}\,\gen{\teigtwo}})\,\gen{\teig\imp\teigtwo\imp\teigthree})
      }
    }
  \\
  \to &
    \freshf{\teigthree}{
      \verif{\teigthree}{
        (\gen{\teig\imp\teigtwo\imp\teigthree}\,\gen{\teig}\,\gen{\teigtwo})
      }
    }
  \\
  \tos &
    \freshf{\teigthree}{\verif{\teigthree}{\gen{\teigthree}}}
    \,\, \tos \,\,
    \freshf{\teigthree}{\iunit}
    \,\, \tos \,\,
    \iunit
  \end{array}
\]
and $\plamf{p}{p\,\teig\,(\plamf{\pvar}{\plamf{\pvartwo}{\pvar}})}$
realizes $\land$-elimination of the first conjunct:
\[
  \begin{array}{ll}
  &
    \verif{(\teig\land\teigtwo)\imp\teig}{
      (\plamf{p}{p\,\teig\,(\plamf{\pvar}{\plamf{\pvartwo}{\pvar}})})
    }
  \\
  \tos &
    \verif{\teig}{
      (\gen{\teig\land\teigtwo}\,\teig\,(\plamf{\pvar}{\plamf{\pvartwo}{\pvar}}))
    }
  \\
  \to &
    \verif{\teig}{
      ((\pallfi{\tvarthree}{\gen{(\teig\imp\teigtwo\imp\tvarthree)\imp\tvarthree}})\,\teig\,(\plamf{\pvar}{\plamf{\pvartwo}{\pvar}}))
    }
  \\
  \to &
    \verif{\teig}{
      (\gen{(\teig\imp\teigtwo\imp\teig)\imp\teig}\,(\plamf{\pvar}{\plamf{\pvartwo}{\pvar}}))
    }
  \\
  \to &
    \verif{\teig}{
      ((\lam{\varthree}{
          (\eunit{(\verif{\teig\imp\teigtwo\imp\teig}{\varthree})}{\gen{\teig}})
        })\,(\plamf{\pvar}{\plamf{\pvartwo}{\pvar}}))
    }
  \\
  \to &
    \verif{\teig}{
      (\eunit{(\verif{\teig\imp\teigtwo\imp\teig}{(\plamf{\pvar}{\plamf{\pvartwo}{\pvar}})})}{\gen{\teig}})
    }
  \\
  \tos &
    \verif{\teig}{(\eunit{\iunit}{\gen{\teig}})}
    \,\, \tos \,\,
    \verif{\teig}{\gen{\teig}}
    \,\, \tos \,\,
    \iunit
  \end{array}
\]
\end{example}

\paragraph{Rewriting properties of $\lambdaCheckF$.}

Unlike weak head reduction in the simply typed case,
weak head reduction is \emph{not} deterministic in the second-order case.
For example,
if $\utm \tow \utm'$
then
$\verif{\ttm\imp\ttmtwo}{\utm} \tof \verif{\ttmtwo}{(\utm\,\gen{\ttm})}$
and $\verif{\ttm\imp\ttmtwo}{\utm} \tof \verif{\ttm\imp\ttmtwo}{\utm'}$.
However, it is easy to show that weak head reduction in $\lambdaCheckF$ enjoys
a strong form of confluence (see~\cref{a:sec:determinism} for details):

\begin{proposition}
\label{prop:f:subcommutative}
Weak head reduction $\tow$ is subcommutative.
\end{proposition}

Regarding confluence of reduction, note that
the fourth rewriting rule of $\lambdaCheckF$ is \emph{not left-linear}.
However, because \logicalEigenvariables cannot be instantiated,
confluence can still be established following standard techniques:
defining a notion of \emph{simultaneous reduction} ($\sto$) such that
${\to} \subseteq {\sto} \subseteq {\tos}$
and such that $\sto$ satisfies the diamond property.
(A detailed proof is provided in~\cref{a:sec:f:confluence}):

\begin{proposition}
\label{prop:f:confluence}
Reduction in $\lambdaCheckF$ is confluent.
\end{proposition}

As in the case of $\lambdaCheck$,
reductions to $\iunit$ in $\lambdaCheckF$
can be standardized to weak head reductions.
We omit the proof of standardization due to lack of space,
but see~\cref{a:sec:f:standardization} for details.
The proof relies on an \emph{intersection type system}~\cite{Coppo80},
in particular with \emph{non-idempotent intersection}~\cite{Gardner94,deCarvalho2007}.
These type systems can be used to \emph{characterize} dynamic properties
of programs such as head, weak, and strong normalization~\cite{Bucciarelli17},
and it is well known they can be used to prove standardization results~\cite{Kesner16}.

\begin{proposition}[Standardization]
\label{prop:f:standardization}
If $\utm \tos \iunit$ then $\utm \tofs \iunit$.
\end{proposition}
\begin{proof}
See~\cref{a:sec:f:standardization} for the proof.
The idea is to propose a type system $\LambdaFInters$
based on \emph{non-idempotent intersection types},
with two key properties.
First, \textbf{weak head subject reduction} ensures that
  if $\utm \tow \utm'$
  and $\utm$ is typable in $\LambdaFInters$
  then $\utm'$ is typable in $\LambdaFInters$,
  where moreover the \emph{size} of the derivation tree for
  $\utm'$ is strictly smaller than the one for $\utm$.
Second, \textbf{subject expansion} ensures that
  if $\utm \to \utm'$
  and $\utm'$ is typable in $\LambdaFInters$
  then $\utm$ is typable in $\LambdaFInters$.

Suppose that $\utm \tos \iunit$.
The constant $\iunit$ is trivially typable in $\LambdaFInters$,
so by subject expansion $\utm$ is also typable.
Let $\utm \tow \utm_1 \tow \utm_2 \hdots$ be a maximal sequence
of weak head reduction steps.
The sequence must be finite because weak head subject reduction ensures
that the size of the typing derivation decreases.
So there exists $n$ such that $\utm \tows \utm_n$
and $\utm_n$ is a $\tow$-normal form.
Since $\utm \tos \utm_n$ and $\utm \tos \iunit$,
by confluence~(\cref{prop:f:confluence}) one concludes that
$\utm_n = \iunit$.
\end{proof}

\subsection{Soundness of $\lambdaF$ with respect to $\lambdaCheckF$}

We start by mentioning that the \textbf{correctness} lemma is easy to
prove for $\lambdaCheckF$ by induction on the size of $\ttm$
(see \cref{a:lem:f:correctness} in the appendix for the proof):

\begin{lemma}[Correctness]
\label{lem:f:correctness}
If $\ttm$ has no free \logicalVariables then
$\verif{\ttm}{\gen{\ttm}} \tofs \iunit$.
\end{lemma}
Note that correctness only holds for \logicalTerms
\emph{without free \logicalVariables}.
For example, $\verif{\tvar}{\gen{\tvar}}$
is irreducible (since $\tvar$ is a \logicalVariable),
while $\verif{\teig}{\gen{\teig}} \to \iunit$ holds
for an \logicalEigenvariable $\teig$.

Unfortunately, it is not clear how to extend the syntactic
proof of soundness for minimal logic~(\cref{thm:soundness})
to second-order logic.
The kind of difficulties that arise are quite typical in \emph{impredicative}
settings.
In particular, the proof of \textbf{universality}
for the metacalculus for \minimalLogic $\lambdaCheck$~(\cref{lem:universality})
proceeded by induction using a well-founded measure that depended on the size
of \logicalTerms.
This measure cannot be easily adapted to the second-order metacalculus
$\lambdaCheckF$, because contracting a redex $(\allf{\tvar}{\utm})\,\ttm$
results in a \metaterm $\utm\sub{\tvar}{\ttm}$ in which, for example,
a generator $\gen{\ttmtwo}$ occurring inside $\utm$
becomes a generator $\gen{\ttmtwo\sub{\tvar}{\ttm}}$ for a potentially
larger \logicalTerm.

\paragraph{Verification candidates.}
To prove soundness, we use a notion we call \emph{\verificationCandidates},
inspired by reducibility candidates in Girard's proof of strong normalization
for System~F~\cite{GirardThesis}.
A \emph{\verificationCandidate} is a set of \metaterms enjoying certain
closure properties.
The proof technique essentially works by mapping each \logicalTerm
$\ttm$ to a \verificationCandidate $\semf{\ttm}{}$.
A first key lemma ensures that any \metaterm $\utm \in \semf{\ttm}{}$
is accepted by the verifier for $\ttm$,
\ie that $\verif{\ttm}{\utm} \tofs \iunit$.
A second key lemma ensures that each typable \proofTerm of \logicalTerm $\ttm$ 
belongs to $\semf{\ttm}{}$.
As usual in these cases,
the statements of all results have to be generalized
in order to consider the interpretation $\semf{\ttm}{\asig}$
of a \logicalTerm $\ttm$ under an appropriate mapping $\asig$,
called a \emph{\candidateAssignment},
that provides an interpretation for free \logicalVariables.

\begin{definition}[\VerificationCandidates]
Let $\ttm$ be a \logicalTerm.
A set of \metaterms $\vcset$ is an \defn{$\ttm$-\verificationCandidate}
if the three following conditions hold:
\itemNum{1} $\gen{\ttm} \in \vcset$,
\itemNum{2} if $\utm \in \vcset$, then $\verif{\ttm}\utm \tofs \iunit$,
\itemNum{3} if $\utm$, $\utm'$ are \metaterms such that $\utm \tof \utm'$, then $\utm \in \vcset$ if and only if $\utm' \in \vcset$.

We write $\VC{\ttm}$ for the set of all $\ttm$-\verificationCandidates.
\end{definition}

\paragraph{Interpreting \logicalTerms as \verificationCandidates.}
A \defn{\logicalSubstitution} is a function $\tsubst$ mapping each \logicalVariable $\tvar$
to a \logicalTerm $\tsubst(\tvar)$.
A \defn{\candidateAssignment} is a function $\asig$ mapping each \logicalVariable $\tvar$
to a pair $(\ttm,\vcset)$, where $\ttm$ is a \emph{closed} \logicalTerm, \ie
a \logicalTerm without free \logicalVariables,
and $\vcset$ is an $\ttm$-\verificationCandidate.
We write $\asig\extsub{\tvar}{(\ttm,\vcset)}$ for the \candidateAssignment such that
$(\asig\extsub{\tvar}{(\ttm,\vcset)})(\tvar) = (\ttm,\vcset)$
and
$(\asig\extsub{\tvar}{(\ttm,\vcset)})(\tvartwo) = \asig(\tvartwo)$
for any \logicalVariable $\tvartwo$ other than $\tvar$.
For each \candidateAssignment $\asig$, we define its \emph{first projection}
as the \logicalSubstitution $\tilde{\asig}$ such that
$\tilde{\asig}(\tvar) = \ttm$ whenever $\asig(\tvar) = (\ttm,\vcset)$.
By abuse of notation, sometimes we write $\asig$ for the first projection of
the \candidateAssignment $\asig$ instead of $\tilde{\asig}$
(\eg $\subs{\ttm}{\asig}$ stands for $\subs{\ttm}{\tilde{\asig}}$). 

\begin{definition}[Candidate interpretation of \logicalTerms]
\label{def:f:interpreation_of_types}
Let $\ttm$ be a \logicalTerm and $\asig$ be a \candidateAssignment.
We define a set of \metaterms $\semf{\ttm}{\asig}$ by recursion on $\ttm$
as follows:
\begin{enumerate}
\item
  $\semf{\tvar}{\asig} \eqdef \vcset$ if $\asig(\tvar) = (\ttm,\vcset)$.
\item
  $\semf{\teig}{\asig}$
  is the set of \metaterms $\utm$ such that $\utm \tofs \gen{\teig}$.
\item
  $\semf{\ttm\imp\ttmtwo}{\asig}$
  is the set of \metaterms $\utm$ such that: \\
  \hphantom{\HS}
    $\verif{\subs{(\ttm\imp\ttmtwo)}{\asig}}\utm \tofs \iunit$, and
  \\
  \hphantom{\HS}
    for every $\utmtwo\in\semf{\ttm}{\asig}$
    we have that $\utm\,\utmtwo\in\semf{\ttmtwo}{\asig}$.
\item
  $\semf{\allf{\tvar}{\ttm}}{\asig}$
  is the set of \metaterms $\utm$
  such that: \\
  \hphantom{\HS}
    $\verif{\subs{(\allf{\tvar}{\ttm})}{\asig}}\utm \tofs \iunit$, and
  \\
  \hphantom{\HS}
    for every closed \logicalTerm $\ttmtwo$ and every $\vcset\in\VC{\ttmtwo}$, \\
  \hphantom{\HS\HS}
    we have that $\utm\,\ttmtwo\in\semf{\ttm}{\asig\extsub{\tvar}{(\ttmtwo,\vcset)}}$.
\end{enumerate}
\end{definition}

Below, we mention three key properties of the interpretation of types.
We have omitted the detailed proofs but they can be found in
\cref{a:sec:f:soundness}.
First, the interpretation of a type $\ttm$ always yields
an $\ttm$-\verificationCandidate.
More precisely, the following \emph{adequacy} lemma holds
(see \cref{a:lem:f:adequacy_for_propositions} for the proof):

\begin{lemma}[Adequacy for \logicalTerms]
\label{lem:f:adequacy_for_propositions}
$\semf{\ttm}{\asig} \in \VC{\subs{\ttm}{\asig}}$
\end{lemma}

Second, the two lemmas below are typical in reducibility interpretations
and related techniques, such as logical relations.
The first lemma expresses an \emph{irrelevance} property,
and the second one expresses a \emph{substitution} property.
(See \cref{a:lem:f:irrelevance} and \cref{a:lem:f:substitution}
for the proofs):

\begin{lemma}
\label{lem:f:irrelevance}
If $\asig_1$ and $\asig_2$ agree on the free \logicalVariables of $\ttm$,
then $\semf{\ttm}{\asig_1} = \semf{\ttm}{\asig_2}$.
\end{lemma}

\begin{lemma}
\label{lem:f:substitution}
$\semf{\ttm\sub{\tvar}{\ttmtwo}}{\asig}
 = \semf{\ttm}{\asig\extsub{\tvar}{(\subs{\ttmtwo}{\asig},\semf{\ttmtwo}{\asig})}}$
\end{lemma}

\paragraph{Proving soundness through verification candidates.}
Before stating and proving the main result, we require a few more definitions.
A \defn{\proofSubstitution} is a function $\psubst$
mapping each \proofVariable to a \metaterm.
The \defn{generative substitution} for an \typeEnvironment $\penv$
is written $\gensubst{\penv}$
and defined as the substitution
such that $\gensubst{\penv}(\var) = \gen{\ttm}$ holds for every $\var:\ttm \in \penv$,
and $\gensubst{\penv}(\var) = \var$ holds for every $\var \notin \dom{\penv}$.
By abuse of notation, sometimes we write just $\penv$ for the generative
substitution, rather than $\gensubst{\penv}$.
If $\asig$ is a \candidateAssignment, $\penv$ is an \proofEnvironment,
and $\psubst$ is a \proofSubstitution,
we say that $\psubst$ is \defn{compatible} with $\penv$ under $\asig$,
written $\compatf{\penv}{\psubst}$,
if for every \proofVariable $\pvar:\ttm \in \penv$
we have that $\subs{\pvar}{\psubst} \in \semf{\ttm}{\asig}$.

The main lemma to prove soundness roughly states that
all terms of type $\ttm$ are elements of the $\ttm$-\verificationCandidate
given by $\semf{\ttm}{}$. The precise statement is more general,
to account for the presence of potentially free \proofVariables and
\logicalVariables:

\begin{lemma}[Adequacy for terms]
\label{lem:f:adequacy_for_terms}
Let $\judgf{\penv}{\ptm}{\ttm}$ and $\compatf{\penv}{\psubst}$.
Then $\subs{\ptm}{\asig\psubst} \in \semf{\ttm}{\asig}$.
\end{lemma}
\begin{proof}
By induction on the derivation of the typing judgment $\judgf{\penv}{\ptm}{\ttm}$.
See \cref{a:lem:f:adequacy_for_terms} for the proof.
\end{proof}

\begin{theorem}[Soundness]
\label{thm:f:soundness}
Let $\penv,\tm,\ttm$ be such that they have no free \logicalVariables\footnote{%
  If they have free \logicalVariables,
  they can be substituted by fresh \logicalEigenvariables to apply the theorem.
  For example the judgment $\judg{\pvar:\tvar}{\pvar}{\tvar}$ holds
  and $\verif{\tvar}{\gen{\tvar}} \not{\tos} \iunit$,
  but substituting $\tvar$ by a fresh \logicalEigenvariable $\teig$
  we have that $\judg{\pvar:\teig}{\pvar}{\teig}$ 
  and $\verif{\teig}{\gen{\teig}} \tos \iunit$.
}.
If $\judg{\penv}{\tm}{\ttm}$
then $\verif{\ttm}{\subs{\tm}{\penv}} \tos \iunit$.
\end{theorem}
\begin{proof}
Let $\teig$ be a fixed \logicalEigenvariable,
and let $\semf{\teig}{}$ be the set of \metaterms $\utm$
such that $\utm \tows \gen{\teig}$.
Observe that $\asig(\tvar) = (\teig,\semf{\teig}{})$ is a \candidateAssignment.
Note that $\compatf{\penv}{\gensubst{\penv}}$,
where $\gensubst{\penv}$ is the generative substitution for $\penv$,
because if $\pvar:\ttm\in\penv$ then $\verif{\ttm}{\gen{\ttm}} \tos \iunit$
by correctness~(\cref{lem:f:correctness}).
Then by \cref{lem:f:adequacy_for_terms}
we have $\subs{\tm}{\asig\penv} \in \semf{\ttm}{\asig}$.
Finally, by \cref{lem:f:adequacy_for_propositions}
$\semf{\ttm}{\asig}$ is a $\subs{\ttm}{\asig}$-\verificationCandidate,
so $\verif{\subs{\ttm}{\asig}}{\subs{\tm}{\asig\penv}} \tos \iunit$.
Since $\ttm$ and $\tm$ have no free \logicalVariables
this means that $\verif{\ttm}{\subs{\tm}{\penv}} \tos \iunit$.
\end{proof}

\subsection{Completeness of $\lambdaF$ with respect to $\lambdaCheckF$}
\label{sec:second_order:completeness}

Fortunately, the technique to show completeness for \minimalLogic (\cref{thm:completeness})
can be extended to the second-order case.
The details of the completeness proof can be found in \cref{a:sec:completeness}.
Here we only provide some of the key definitions,
and specifically the \emph{proof size} (generalizing \cref{def:proof_size}).

\paragraph{Characterization of normal forms.}
An \defn{\argument} ($\argu,\argutwo,\hdots$)
is either a \metaterm or a \logicalTerm
(given by the grammar $\argu ::= \utm \mid \ttm$).
If $\argus = (\argu_1,\hdots,\argu_n)$ is a sequence of \arguments,
we write $\utm\,\argus$ or $\utm\,\argu_1\hdots\argu_n$
for the application of $\utm$ to the $n$ \arguments.
The sets of \emph{pseudo-redexes}, \emph{\headForms}, \emph{normal \arguments},
\emph{neutral \proofTerms}, and \emph{normal \proofTerms}
are given mutually inductively by the grammar:
\[
  \begin{array}{lr@{\,\,}c@{\,\,}l}
  \textsc{Pseudo-redexes} &
    \Psi, \hdots & ::= & (\plamf{\pvar}{\unf})\,\ttm
            \,\,\mid\,\, (\pallfi{\tvar}{\unf})\,\unftwo
  \\
  \textsc{\HeadForms} &
    \uhd,\uhd',\hdots & ::= & \pvar \mid \Psi
  \\
  \textsc{Normal \arguments} &
    \argunf,\argunftwo,\hdots
           & ::= & \unf
      \,\,\mid\,\, \ttm
  \\
  \textsc{Neutral \proofTerms} &
    \uneu,\uneutwo,\hdots
           & ::= & \uhd\,\argunf_1\hdots\argunf_n
  \\
  \textsc{Normal \proofTerms} &
    \unf,\unftwo,\hdots
           & ::= & \uneu
      \,\,\mid\,\, \plamf{\pvar}{\unf}
      \,\,\mid\,\, \pallfi{\tvar}{\unf}
  \end{array}
\]
It is a routine exercise to show that pure \metaterms
in $\beta$-normal form correspond exactly to normal \proofTerms.

\paragraph{\Argument matching.}
In the case of \minimalLogic, it is clear at a glance
that a \proofTerm of type $\ttm$
can take as arguments a list of $n$ \proofTerms $(\utm_1,\hdots,\utm_n)$
only if $\ttm$ is of the form $\ttm_1\imp\hdots\ttm_n\imp\ttm'$.
In such case, we write
$\match{\ttm}{(\utm_1,\hdots,\utm_n)}{(\ttm_1,\hdots,\ttm_n)}{\ttm'}$.
We generalize this predicate to the second-order case,
writing $\match{\ttm}{\argus}{\ttms}{\ttm'}$
where $\ttm,\ttm'$ are \logicalTerms, $\argus$ is a list of \arguments,
and $\ttms$ a list of \logicalTerms. Inductively,
we declare that $\match{\ttm}{\emptyset}{\emptyset}{\ttm}$ holds,
and:
\[
  \indrule{}{
    \match{\ttmtwo}{\argus}{\ttmtwos}{\ttmtwo'}
  }{
    \match{(\ttm\imp\ttmtwo)}{(\utm,\argus)}{(\ttm,\ttmtwos)}{\ttmtwo'}
  }
  \HS
  \indrule{}{
    \match{\ttm\sub{\tvar}{\ttmtwo}}{\argus}{\ttmthrees}{\ttmthree'}
  }{
    \match{\allf{\tvar}{\ttm}}{(\ttmtwo,\argus)}{(\ttmtwo,\ttmthrees)}{\ttmthree'}
  }
\]
For example,
$\match{(\allf{\tvar}{(\tvar\imp\tvar)\imp\tvar})}{(\ttmtwo,\utm)}{(\ttmtwo,(\ttmtwo\imp\ttmtwo))}{\ttmtwo}$.

\begin{definition}[Proof size of normal \proofTerms]
\label{def:f:proof_size}
The \defn{proof size} of a normal \proofTerm
$\unf$ with respect to a \proofEnvironment $\penv$
and a \logicalTerm $\ttm$ is a non-negative integer written $\sig{\unf}{\ttm}$
and defined as follows, by induction on the shape of the normal form $\unf$.
The definition is extended for normal \arguments ($\sig{\argunf}{\ttm}$)
by declaring that $\sig{\argunf}{\ttm} = 0$ when the \argument $\argunf$ is a \logicalTerm.
\[
  \begin{array}{l@{\,\,\,\,}l}
    \multicolumn{2}{l}{
      \sig{\pvar\,\argunf_1\hdots\argunf_n}{\ttm} \eqdef
            \minof{\size{\ttm}}{\size{\ttm'}}
          + \sum_{i=1}^n\set{{\sig{\argunf_i}{\ttmtwo_i}}}
    }
    \\
    & \text{if $\pvar:\ttmtwo \in \penv$ and $\match{\ttmtwo}{\argus}{\ttmtwos}{\ttm'}$}
    \\
    & \text{$\argus = (\argunf_1,\hdots,\argunf_n)$,
            $\ttmtwos = (\ttmtwo_1,\hdots,\ttmtwo_n)$}
  \\
    \sig{\plamf{\pvar}{\unf}}{\ttm} \eqdef
        \sig[\penv,\pvar:\ttmtwo]{\unf}{\ttmthree}
        & \text{if $\ttm = (\ttmtwo\imp\ttmthree)$}
  \\
    \sig{\pallfi{\tvar}{\unf}}{\ttm} \eqdef
        \sig[\penv]{\unf}{\ttmtwo}
        & \text{if $\ttm = (\allf{\tvar}{\ttmtwo})$, assuming $\tvar\notin\fv{\penv}$}
  \\
    \sig{\unf}{\ttm} \eqdef
        0
        & \text{in any other case}
  \end{array}
\]
Note in particular that $\sig{\uhd\,\argunf_1\hdots\argunf_n}{\ttm} = 0$
if $\uhd$ is a pseudo-redex rather than a variable.
\end{definition}

\begin{theorem}[Completeness]
\label{thm:f:completeness}
If $\verif{\ttm}{\subs{\utm}{\penv}} \tos \iunit$
and $\utm$ is good,
then its $\beta$-normal form is a \proofTerm $\unf$ such
that $\judgf{\penv}{\unf}{\ttm}$.
\end{theorem}
\begin{proof}
The details of the proof may be found in \cref{a:thm:f:completeness}.
The proof proceeds by induction on the lexicographic pair
$(\sig{\unf}{\ttm},\size{\ttm})$, where $\size{\ttm}$ denotes
the size of the \logicalTerm $\ttm$.
\end{proof}

\section{Extension to Higher-Order Logic}
\label{sec:higher_order}

In this section, we discuss how to extend the \epistemic realizability
interpretations of the previous sections to higher-order logic.
Recall that higher-order logic generalizes systems such as first and
second-order logic by allowing universal quantifiers to bind variables
of arbitrarily high order. \LogicalVariables may represent not
only elements of the universe of discourse, but also predicates on
these elements, predicates on these predicates, and so on.
Due to its expressivity, it has been used as the basis of
functional programming languages such as \textsc{Haskell},
and proof assistants such as \tool{HOL Light} or \tool{Isabelle/HOL}.
Many formalizations of higher-order logic can be found in the literature,
starting with Church's simple theory of types~\cite{Church1940}.
Our work is based on Girard's System F$\omega$~\cite{GirardThesis},
one of the vertices of Barendregt's $\lambda$-cube~\cite{Barendregt1991}.
System F$\omega$ is a \emph{two-layered} typed $\lambda$-calculus with
three sorts of expressions: \emph{\kinds}, \emph{\logicalTerms}, and \emph{\proofTerms}.

The main results in this section are again
\textbf{soundness}~(\cref{thm:hol:soundness})
and \textbf{completeness}~(\cref{thm:hol:completeness}).

\paragraph{The higher-order $\lambda$-calculus ($\lambdaFw$).}
We recall the typed $\lambda$-calculus corresponding
to higher-order logic, namely Girard's System~F$\omega$~\cite{GirardThesis,Reynolds1974}.
The set of \defn{\kinds} is given by the following grammar, including 
base kinds ($\bki,\bki',\hdots$) and a constant $\Prop$:
\[
  \textsc{\Kinds}
  \HS
  \ki,\kitwo,\hdots ::= \bki \mid \Prop \mid \ki\to\kitwo
\]
As in the second-order case,
we assume given disjoint sets of \defn{\logicalVariables} ($\tvar,\tvartwo,\hdots$)
and \defn{\logicalEigenvariables} ($\teig,\teigtwo,\hdots$).
The set of \defn{\logicalTerms} is given by the grammar:
\[
  \textsc{\LogicalTerms}
  \HS
  \ttm, \ttmtwo, \hdots ::= \tvar
                       \mid \teig
                       \mid \tlam{\tvar}{\ki}{\ttm}
                       \mid \ttm\,\ttmtwo
                       \mid \ttm\imp\ttmtwo
                       \mid \all{\tvar}{\ki}{\ttm}
\]
A \defn{\logicalEnvironment} ($\tenv,\tenvtwo,\hdots$)
is a finite set of assignments
$\tvar_1:\ki_1,\hdots,\tvar_n:\ki_n,\teig_1:\ki_{n+1},\hdots,\teig_m:\ki_{n+m}$,
mapping \logicalVariables and \logicalEigenvariables to \kinds.
Kinding judgments are of the form $\tjudg{\tenv}{\ttm}{\ki}$,
meaning that the \logicalTerm $\ttm$ is well-kinded of \kind $\ki$ under the \logicalEnvironment $\tenv$.
The typing rules are the expected ones.
In particular,
$\tjudg{\tenv}{\ttm\imp\ttmtwo}{\Prop}$
holds if $\tjudg{\tenv}{\ttm}{\Prop}$ and $\tjudg{\tenv}{\ttmtwo}{\Prop}$ hold;
and $\tjudg{\tenv}{\all{\tvar}{\ki}{\ttm}}{\Prop}$
holds if $\tjudg{\tenv,\tvar:\ki}{\ttm}{\Prop}$ holds.
We define \defn{$\beta$-convertibility} as the least congruence
generated by the usual $\beta$-reduction rule:
\[
  (\tlam{\tvar}{\ki}{\ttm})\,\ttmtwo
  \tobeta
  \ttm\sub{\tvar}{\ttmtwo}
\]
\LogicalTerms of \kind $\Prop$ are called \emph{propositions}.
For example, the \logicalTerm
$\all{P}{\bki\to\Prop}{\all{f}{\bki\to\bki}{\all{\tvar}{\bki}{(P\,\tvar \imp P\,(f\,\tvar))}}}$
is a proposition.
Note that \logicalVariables bound by universal quantifiers may
play the role of predicate symbols (\eg~$P$)
or function symbols (\eg~$f$).

\paragraph{\ProofTerms and typing judgments.}
Given a denumerable set of \defn{\proofVariables} ($\pvar,\pvartwo,\hdots$),
the set of \defn{\proofTerms} is given by the grammar:
\[
  \textsc{\ProofTerms}
  \HS
  \ptm,\ptmtwo,\hdots
       ::= \pvar
      \mid \plam{\pvar}{\ptm}
      \mid \ptm\,\ptmtwo
      \mid \palli{\tvar}{\ki}{\ptm}
      \mid \ptm\,\ttm
\]
Note that the grammar of \proofTerms is similar to the one for
the second-order case, with the difference that \logicalVariables
bound by \logicalAbstractions are decorated with their \kind.
We write $\plam{\pvar}{\ptm}$ instead of $\plamf{\pvar}{\ptm}$
to avoid confusion with the abstraction at the level of \logicalTerms.

An \defn{\proofEnvironment} ($\penv,\penv',\hdots$) is a finite set of
$\set{\pvar_1:\ttm_1,\hdots,\pvar_n:\ttm_n}$ that maps distinct
\proofVariables to \logicalTerms.
We introduce two further forms of judgment:
\itemNum{1}
   $\pjudgenv{\tenv}{\penv}$, meaning that the \proofEnvironment
  $\penv$ is well-formed under the \logicalEnvironment $\tenv$
  (\ie it maps \proofVariables to \emph{propositions}),
and \itemNum{2}
  $\pjudg{\tenv}{\penv}{\ptm}{\ttm}$,
  meaning that the \proofTerm $\ptm$ is a proof of the proposition $\ttm$
  under the environments $\tenv$ and $\penv$.
Valid judgments are given by:
\[
  \indrule{}{
    (\tjudg{\tenv}{\ttm_i}{\Prop})_{i=1}^{n}
  }{
    \pjudgenv{\tenv}{(\pvar_1:\ttm_1,\hdots,\pvar_n:\ttm_n)}
  }
  \HS
  \indrule{}{
    \pjudgenv{\tenv}{(\penv,\pvar:\ttm)}
  }{
    \pjudg{\tenv}{\penv,\pvar:\ttm}{\pvar}{\ttm}
  }
\]
\[
  \indrule{}{
    \pjudg{\tenv}{\penv,\pvar:\ttm}{\ptm}{\ttmtwo}
  }{
    \pjudg{\tenv}{\penv}{\plamt{\pvar}{\ttm}{\ptm}}{\ttm\imp\ttmtwo}
  }
  \HS
  \indrule{}{
    \pjudg{\tenv}{\penv}{\ptm}{\ttm\imp\ttmtwo}
    \HS
    \pjudg{\tenv}{\penv}{\ptmtwo}{\ttm}
  }{
    \pjudg{\tenv}{\penv}{\ptm\,\ptmtwo}{\ttmtwo}
  }
\]
\[
  \indrule{}{
    \pjudg{\tenv,\tvar:\ki}{\penv}{\ptm}{\ttm}
    \HS
    \tvar\notin\fv{\penv}
  }{
    \pjudg{\tenv}{\penv}{\palli{\tvar}{\ki}{\ptm}}{\all{\tvar}{\ki}{\ttm}}
  }
  \HS
  \indrule{}{
    \pjudg{\tenv}{\penv}{\ptm}{\all{\tvar}{\ki}{\ttm}}
    \HS
    \tjudg{\tenv}{\ttmtwo}{\ki}
  }{
    \pjudg{\tenv}{\penv}{\ptm\,\ttmtwo}{\ttm\sub{\tvar}{\ttmtwo}}
  }
\]
\[
  \indrule{}{
    \pjudg{\tenv}{\penv}{\ptm}{\ttm}
    \HS
    \tjudg{\tenv}{\ttm}{\Prop}
    \HS
    \tjudg{\tenv}{\ttmtwo}{\Prop}
    \HS
    \ttm \eqbeta \ttmtwo
  }{
    \pjudg{\tenv}{\penv}{\ptm}{\ttmtwo}
  }
\]
The last rule is called the \emph{conversion} rule.
We define \defn{$\beta$} and $\eta$-reduction as the union of the following
rules, closed by arbitrary contexts:
\[
  \begin{array}{r@{\,\,}l@{\,\,}l@{\HS}r@{\,\,}l@{\,\,}l}
    (\plam{\pvar}{\ptm})\,\ptmtwo
    & \tobeta &
    \ptm\sub{\pvar}{\ptmtwo}
  &
    (\palli{\tvar}{\ki}{\ptm})\,\ttm
    & \tobeta &
    \ptm\sub{\tvar}{\ttm}
  \\
    \plam{\pvar}{\ptm\,\pvar}
    & \toeta &
    \ptm
    \,\,\,\,\text{if $\pvar \notin \fv{\ptm}$}
  &
    \palli{\tvar}{\ki}{\ptm\,\tvar}
    & \toeta &
    \ptm
    \,\,\,\,\text{if $\tvar \notin \fv{\ptm}$}
  \end{array}
\]

\paragraph{The higher-order metacalculus ($\lambdaCheckFw$)}

Terms of the $\lambdaCheckFw$ calculus (called ``\metaterms'') are defined
as \textbf{equivalence classes of pre-\metaterms} modulo a suitable equivalence relation.
We start by presenting the syntax of pre-\metaterms of $\lambdaCheckFw$:
\[
  \begin{array}{lrlllll}
  \multicolumn{7}{l}{\textsc{Pre-\metaterms} \HS \utm,\utmtwo,\hdots ::=} \\
  &      & \pvar
         & \text{\proofVariable}
  & \mid & \iunit
         & \text{\unitConstructor}
  \\
  & \mid & \plam{\pvar}{\utm}
         & \text{\proofAbs}
  & \mid & \eunit{\utm}{\utmtwo}
         & \text{\unitEliminator}
  \\
  & \mid & \utm\,\utmtwo
         & \text{\proofApp}
  & \mid & \gen{\ttm}
         & \text{generator}
  \\
  & \mid & \palli{\tvar}{\ki}{\utm}
         & \text{\logicalAbs}
  & \mid & \verif{\ttm}{\utm}
         & \text{verifier}
  \\
  & \mid & \utm\,\ttm
         & \text{\logicalApp}
  & \mid & \fresh{\teig}{\ki}{\utm}
         & \text{fresh \logicalEigenvariable}
  \end{array}
\]
The notions of \defn{pure} and \defn{good} \metaterms are defined as
for $\lambdaCheckF$.

The higher-order metacalculus $\lambdaCheckFw$
is similar to second-order metacalculus $\lambdaCheckF$,
except for the fact that $\palli{\tvar}{\ki}{\utm}$
and $\fresh{\teig}{\ki}{\utm}$ include a \kind decoration.
This was not necessary in $\lambdaCheckF$
because all \logicalVariables were morally always of \kind $\Prop$.
The syntax of $\lambdaCheckFw$ includes generators $\gen{\ttm}$
and verifiers $\verif{\ttm}{}$ for arbitrary \logicalTerms $\ttm$,
which shall be restricted to be \emph{propositions}.

\paragraph{Global \kindAssignment.}
\Kinds and \logicalTerms of expressions in $\lambdaCheckFw$ depend
on a \logicalEnvironment $\tenv$.
To alleviate the notation, we stick to the following well-established
convention~\cite{GirardThesis,Gallier1990}.
We assume that there is a global \emph{\kindAssignment} $\kiasig$,
\ie a function mapping each \logicalVariable $\tvar$
and each \logicalEigenvariable $\teig$ to its \kind.
We assume that the global \kindAssignment has ``enough''
variables of each kind,
\ie for each \kind $\ki$ there are denumerably infinite sets of \logicalVariables
and \logicalEigenvariables of \kind $\ki$.
Whenever we have an expression like $\palli{\tvar}{\ki}{\utm}$
that binds a \logicalVariable of \kind $\ki$,
and in case we need to reason about the body $\utm$,
we assume (by $\alpha$-renaming if necessary)
that $\tvar$ has been chosen in such a way that $\kiasig(\tvar) = \ki$,
and similarly for \logicalEigenvariables.
We write $\tenv \subseteq \kiasig$
if the \logicalEnvironment $\tenv$ is \emph{compatible}
with $\kiasig$,
\ie it assigns the same \kinds as $\kiasig$ to \logicalVariables
and \logicalEigenvariables. 
We write $\tjudg{\kiasig}{\ttm}{\ki}$ to mean
that $\tjudg{\tenv}{\ttm}{\ki}$ holds for some $\tenv \subseteq \kiasig$.
Similarly, we may write $\pjudgenv{\kiasig}{\penv}$
to mean that $\pjudgenv{\tenv}{\penv}$ holds,
and $\pjudg{\kiasig}{\penv}{\ptm}{\ttm}$
to mean that $\pjudg{\tenv}{\penv}{\ptm}{\ttm}$ holds,
for some $\tenv \subseteq \kiasig$

\paragraph{Well-kindedness for pre-\metaterms.}
The $\lambdaCheck$ and $\lambdaCheckF$ calculi of previous sections
were completely untyped.
In contrast, in $\lambdaCheckFw$ we shall impose an invariant enforcing that
pre-\metaterms be \emph{well-kinded}.
Specifically, well-kindedness imposes the restriction that,
in all occurrences of
$\gen{\ttm}$ and $\verif{\ttm}{\utm}$ the \logicalTerm $\ttm$ is a proposition,
and that in each \logicalApplication $\utm\,\ttm$
the \logicalTerm $\ttm$ is well-kinded.
For example, $\fresh{\teig}{\Prop}{\gen{\teig}}$ is well-kinded,
while $\fresh{\teig}{\ki\to\Prop}{\gen{\teig}}$
(in which $\gen{\teig}$ generates a ``realizer'' for a type that
is not a proposition) is \emph{not} well-kinded.
The metacalculus is still untyped, in that \metaterms are not required to be
well-typed; for example, $\var\,\var$ is well-kinded.

Formally, we say that a pre-\metaterm $\utm$ is \emph{well-kinded}
if the predicate $\wkjudg{\utm}$ holds,
according to the following inductive rules:
\[
  \indrule{}{
  }{
    \wkjudg{\pvar}
  }
  \,\,\,\,
  \indrule{}{
    \wkjudg{\utm}
  }{
    \wkjudg{\plam{\pvar}{\utm}}
  }
  \,\,\,\,
  \indrule{}{
    \wkjudg{\utm}
    \HS
    \wkjudg{\utmtwo}
  }{
    \wkjudg{\utm\,\utmtwo}
  }
  \,\,\,\,
  \indrule{}{
    \wkjudg{\utm}
    \HS
    \kiasig(\tvar) = \ki
  }{
    \wkjudg{\palli{\tvar}{\ki}{\utm}}
  }
\]
\[
  \indrule{}{
    \wkjudg{\utm}
    \HS
    \tjudg{\kiasig}{\ttm}{\ki}
  }{
    \wkjudg{\utm\,\ttm}
  }
  \,\,\,\,
  \indrule{}{
  }{
    \wkjudg{\iunit}
  }
  \,\,\,\,
  \indrule{}{
    \wkjudg{\utm}
    \HS
    \wkjudg{\utmtwo}
  }{
    \wkjudg{\eunit{\utm}{\utmtwo}}
  }
\]
\[
  \indrule{}{
    \tjudg{\kiasig}{\ttm}{\Prop}
  }{
    \wkjudg{\gen{\ttm}}
  }
  \,\,\,\,
  \indrule{}{
    \tjudg{\kiasig}{\ttm}{\Prop}
    \HS
    \wkjudg{\utm}
  }{
    \wkjudg{\verif{\ttm}{\utm}}
  }
  \,\,\,\,
  \indrule{}{
    \wkjudg{\utm}
    \HS
    \kiasig(\teig) = \ki
  }{
    \wkjudg{\fresh{\teig}{\ki}{\utm}}
  }
\]
Note that the pre-\metaterm $\palli{\tvartwo}{\bki\to\bki}{(\palli{\tvar}{\Prop}{\gen{\tvar}})\,\tvartwo}$
is well-kinded but naive $\beta$-conversion
would yield $\palli{\tvartwo}{\bki\to\bki}{\gen{\tvartwo}}$,
which is not well-kinded.
Reduction in $\lambdaCheckFw$ shall be restricted to
ensure that well-kindedness is preserved.

\paragraph{\Metaterms and \logicalBetaEquivalence.}
The set of \emph{\metaterms} is obtained by quotienting the set of well-kinded pre-\metaterms
by the equivalence relation that identifies them if they are equal up to
$\beta$-conversion of the \logicalTerms occurring in them.
Formally, \emph{\logicalBetaEquivalence}
is written $\preq{\utm}{\utm'}$ and obtained by the following inductive
definition, which assumes that $\utm$ and $\utm'$ are well-kinded:
\[
  \indrule{}{
    \preq{\utm}{\utm'}
  }{
    \preq{\plam{\pvar}{\utm}}{\plam{\pvar}{\utm'}}
  }
  \,\,\,\,
  \indrule{}{
    \preq{\utm}{\utm'}
    \HS
    \preq{\utmtwo}{\utmtwo'}
  }{
    \preq{\utm\,\utmtwo}{\utm'\,\utmtwo'}
  }
  \,\,\,\,
  \indrule{}{
    \preq{\utm}{\utm'}
    \HS
    \kiasig(\tvar) = \ki
  }{
    \preq{\palli{\tvar}{\ki}{\utm}}{\palli{\tvar}{\ki}{\utm'}}
  }
\]
\[
  \indrule{}{
  }{
    \preq{\pvar}{\pvar}
  }
  \HS
  \indrule{}{
    \preq{\utm}{\utm'}
    \HS
    \ttm\eqbeta\ttm'
  }{
    \preq{\utm\,\ttm}{\utm'\,\ttm'}
  }
  \HS
  \indrule{}{
  }{
    \preq{\iunit}{\iunit}
  }
  \HS
  \indrule{}{
    \preq{\utm}{\utm'}
    \HS
    \preq{\utmtwo}{\utmtwo'}
  }{
    \preq{\eunit{\utm}{\utmtwo}}{\eunit{\utm'}{\utmtwo'}}
  }
\]
\[
  \indrule{}{
    \ttm\eqbeta\ttm'
  }{
    \preq{\gen{\ttm}}{\gen{\ttm'}}
  }
  \HS
  \indrule{}{
    \ttm\eqbeta\ttm'
    \HS
    \preq{\utm}{\utm'}
  }{
    \preq{\verif{\ttm}{\utm}}{\verif{\ttm'}{\utm'}}
  }
  \HS
  \indrule{}{
    \preq{\utm}{\utm'}
    \HS
    \kiasig(\teig) = \ki
  }{
    \preq{\fresh{\teig}{\ki}{\utm}}{\fresh{\teig}{\ki}{\utm'}}
  }
\]
We work with \metaterms by writing pre-\metaterms,
treating them \emph{implicitly} modulo \logicalBetaEquivalence.
As usual in other settings (\eg higher-order rewriting systems),
we say that a \logicalVariable or \logicalEigenvariable occurs
\emph{free} in a \metaterm $\utm$ if it occurs free in \emph{every} representative.
From now on, expressions $\utm,\utmtwo,\hdots$ denote \metaterms rather
than pre-\metaterms.
We shall not explicitly state lemmas showing that constructions can
be lifted from pre-\metaterms to \metaterms, but these lifting properties
are easy due to the standard fact that well-kinded \logicalTerms have
unique $\beta$-normal forms.

\paragraph{Reduction in $\lambdaCheckFw$.}
The set of \emph{weak head contexts} is given by:
\[
  \wctx ::= \ctxhole
         \mid \wctx\,\utm
         \mid \wctx\,\ttm
         \mid \eunit{\wctx}{\utm}
         \mid \verif{\ttm}{\wctx}
         \mid \fresh{\teig}{\ki}{\wctx}
\]
The \defn{reduction} relation, written $\to$,
is given by the union of the reduction rules below,
closed by congruence under arbitrary contexts.
The \defn{weak head reduction} relation, written $\tow$,
is given by the union of the reduction rules below,
closed by congruence under weak head contexts.
\[
  \begin{array}{lrcll}
  \ruleNum{1}
  &
    (\plam{\pvar}{\utm})\,\utmtwo
    & \to &
    \utm\sub{\pvar}{\utmtwo}
  \\
  \ruleNum{2}
  &
    (\palli{\tvar}{\ki}{\utm})\,\ttm
    & \to &
    \utm\sub{\tvar}{\ttm}
    & \text{(if $\tjudg{\kiasig}{\ttm}{\ki}$)}
  \\
  \ruleNum{3}
  &
    \eunit{\iunit}{\utm}
    & \to &
    \utm
  \\
  \ruleNum{4}
  &
    \verif{\teig\vec{\ttm}}{\gen{\teig\vec{\ttm}}}
    & \to &
    \iunit
  \\
  \ruleNum{5}
  &
    \verif{\ttm\imp\ttmtwo}{\utm}
    & \to &
    \verif{\ttmtwo}{(\utm\,\gen{\ttm})}
  \\
  \ruleNum{6}
  &
    \verif{\all{\tvar}{\ki}{\ttm}}{\utm}
    & \to &
      \fresh{\teig}{\ki}{\verif{\ttm\sub{\tvar}{\teig}}{(\utm\,\teig)}}
    & \text{($\teig$ fresh)}
  \\
  \ruleNum{7}
  &
    \gen{\ttm\imp\ttmtwo}
    & \to &
    \plam{\pvar}{(\eunit{\verif{\ttm}{\pvar}}{\gen{\ttmtwo}})}
  \\
  \ruleNum{8}
  &
    \gen{\all{\tvar}{\ki}{\ttm}}
    & \to &
    \palli{\tvar}{\ki}{\gen{\ttm}}
  \\
  \ruleNum{9}
  &
    \fresh{\teig}{\ki}{\utm}
    & \to &
    \utm
    & \text{(if $\teig\notin\fv{\utm}$)}
  \end{array}
\]
The rewriting rules are almost exactly the same as to those of the
second-order metacalculus $\lambdaCheckF$.
There are two noteworthy differences.
First, rule~2 requires that the redex is \emph{well-kinded} to proceed,
\ie that the actual argument $\ttm$
has the same \kind as the formal parameter $\tvar$.
For example, the step
  $(\palli{\tvar}{\Prop}{\gen{\tvar\imp\tvar}})\,\teigtwo
  \to
  \gen{\teigtwo\imp\teigtwo}$
is allowed if and only if $\teigtwo:\Prop\in\kiasig$.
This side condition allows to show that \emph{reduction
preserves well-kindedness}.

Second, note that if $\ttm$ is a well-kinded proposition without
free \logicalVariables,
it must be exactly of one of three possible forms up to $\beta$-convertibility:
an applied \logicalEigenvariable $(\teig\,\ttm_1\hdots\ttm_n)$,
an implication $(\ttmtwo\imp\ttmthree)$,
or a universal quantification $(\all{\tvar}{\ki}{\ttmtwo})$.
The reduction rule~4 treats the case for $\teig\,\ttm_1\hdots\ttm_n$,
written $\teig\vec{\ttm}$ in the rule.
Contrast this with the rule for $\lambdaCheckF$,
which was just $\verif{\teig}{\gen{\teig}} \to \iunit$.
In the second-order case, all \logicalEigenvariables were morally
propositions (of \kind $\Prop$), and could not be applied,
while in $\lambdaFw$ they can also play the role of \emph{predicate symbols}
(of \kind $\ki_1\to\hdots\to\ki_n\to\Prop$).

\begin{example}[Symmetry of Leibniz equality]
Let $\ttm,\ttmtwo$ be \logicalTerms of kind $\ki$,
and let \emph{Leibniz equality} be defined by its
usual encoding
$(\ttm =_\ki \ttmtwo) :=
 \all{P}{\ki\to\Prop}{(P\,\ttm \imp P\,\ttmtwo)}$.
Let $Q_P := \tlam{c}{\ki}{(P\,c \imp P\,\ttm)}$.
Then the following reduction shows that \emph{symmetry}
is realizable:
\[
  \begin{array}{l@{\,\,}l}
  &
    \verif{(\ttm =_\ki \ttmtwo) \imp (\ttmtwo =_\ki \ttm)}{
      (\plam{e}{
        \palli{P}{\ki\to\Prop}{
          \plam{\pvar}{
            e\,Q_P\,(\plam{\pvartwo}{\pvartwo})\,\pvar
          }
        }
      })
    }
  \\
  \tos &
    \verif{\ttmtwo =_\ki \ttm}{
      (
        \palli{P}{\ki\to\Prop}{
          \plam{\pvar}{
            \gen{\ttm =_\ki \ttmtwo}\,Q_P\,(\plam{\pvartwo}{\pvartwo})\,\pvar
          }
        }
      )
    }
  \\
  \tos &
    \fresh{\teigP}{\ki\to\Prop}{
      \verif{\teigP\ttmtwo\imp\teigP\ttm}{
        (
          \plam{\pvar}{
            \gen{\ttm =_\ki \ttmtwo}\,Q_{\teigP}\,(\plam{\pvartwo}{\pvartwo})\,\pvar
          }
        )
      }
    }
  \\
  \tos &
    \fresh{\teigP}{\ki\to\Prop}{
      \verif{\teigP\ttm}{
        (
          \gen{\ttm =_\ki \ttmtwo}\,Q_{\teigP}\,(\plam{\pvartwo}{\pvartwo})\,\gen{(\teigP\ttmtwo)}
        )
      }
    }
  \\
  \tos &
    \fresh{\teigP}{\ki\to\Prop}{
      \verif{\teigP\ttm}{
        (
          \gen{(Q_{\teigP}\ttm \imp Q_{\teigP}\ttmtwo)}\,
          (\plam{\pvartwo}{\pvartwo})\,\gen{\teigP\ttmtwo}
        )
      }
    }
  \\
  \equiv &
    \fresh{\teigP}{\ki\to\Prop}{
      \verif{\teigP\ttm}{
        (
          \gen{(\teigP\ttm\imp\teigP\ttm)\imp\teigP\ttmtwo\imp\teigP\ttm}\,
          (\plam{\pvartwo}{\pvartwo})\,\gen{\teigP\ttmtwo}
        )
      }
    } \tos \iunit
  \end{array}
\]
\end{example}

\paragraph{Soundness and completeness of $\lambdaFw$ with respect to $\lambdaCheckFw$.}
The rest of the development follows closely what we have done in
the second-order case~(\cref{sec:second_order}).
In particular, the notion of \emph{generative substitution}
$\gensubst{\penv}$ is defined as before,
and $\subs{\utm}{\penv}$ denotes $\subs{\utm}{\gensubst{\penv}}$.
The proofs of soundness and completeness
are extensions of those for the second-order case
(\cref{thm:f:soundness} and \cref{thm:f:completeness}).
The proof of completeness is a straightforward extension, 
while the proof of soundness requires to parameterize the interpretation
with respect to an \emph{indexed family of \verificationCandidates},
using \logicalTerms as indices, to be able to interpret higher-kinded \logicalTerms.
This follows well-established techniques~\cite{GirardThesis,Gallier1990}.
Details for the proof of soundness can be found in \cref{a:sec:hol:soundness}.

\begin{theorem}[Soundness]
\label{thm:hol:soundness}
Let $\penv,\tm,\ttm$ be such that they have no free \logicalVariables.
If $\pjudg{\kiasig}{\penv}{\tm}{\ttm}$
then $\verif{\ttm}{\subs{\tm}{\penv}} \tos \iunit$.
\end{theorem}

\begin{theorem}[Completeness]
\label{thm:hol:completeness}
Let $\tjudg{\kiasig}{\ttm}{\Prop}$
and $\wkjudg{\utm}$
and $\pjudgenv{\kiasig}{\penv}$.
If $\verif{\ttm}{\subs{\utm}{\penv}} \tos \iunit$
and $\utm$ is good,
then its $\beta$-normal form is a \proofTerm $\unf$ such
that $\pjudg{\kiasig}{\penv}{\unf}{\ttm}$.
\end{theorem}

\section{Conclusion}
\label{sec:conclusion}

This paper formulates \epistemic realizability interpretations
for \textbf{minimal}, \textbf{second}, and \textbf{higher-order logic},
showing that each of the three systems is \textbf{sound}
(\cref{thm:soundness,thm:f:soundness,thm:hol:soundness})
and \textbf{complete}
(\cref{thm:completeness,thm:f:completeness,thm:hol:completeness}).

\paragraph{Related work.}
Our main motivation is philosophical: to explain the meaning of propositions
from a verificationist perspective, but imposing the condition that
the semantics is ``\epistemic'', \ie that recognizing evidence for a
proposition is semi-decidable.
Similar concerns underly several foundational programs, including
Kreisel's theory of constructions~\cite{Kreisel1962},
Prawitz's proof-theoretic semantics~\cite{Prawitz1971},
Dummett's anti-realist theory of meaning~\cite{Dummett1991},
and Martin-L\"of's meaning explanations~\cite{MartinLof1984}.
In those works, the meaning of a proposition is defined
by the conditions for making a judgment that asserts it,
while in our case the meaning of a proposition is defined \emph{compositionally},
without referring to the notion of judgment.
For example, the meaning of $\ttm\imp\ttmtwo$ is a pair of
programs $(\verif{\ttm\imp\ttmtwo}{},\gen{\ttm\imp\ttmtwo})$
which are defined in terms of the meanings of $\ttm$ and $\ttmtwo$,
\ie in terms of the pairs of programs
$(\verif{\ttm}{},\gen{\ttm})$ and $(\verif{\ttmtwo}{},\gen{\ttmtwo})$,
without making reference to a judgment like $\ttm\vdash\ttmtwo$.

Our motivation is perhaps closest to that of
Girard's Ludics~\cite{Girard2001,Faggian2002,Faggian2011,Girard2011,Terui2011}, in which the meaning of
propositions is given by the interaction between proofs and
\emph{counter-proofs} (related to \emph{verifiers})
which may include \emph{daimons} (related to \emph{generators}).

The notions of \emph{correctness} and \emph{universality} in this paper
resemble the notion of \emph{harmony} of introduction and elimination
rules~\cite{Dummett1991,Zeilberger2008}, and specifically the conditions
of \emph{local soundness} and \emph{local completeness}~\cite{Davies2001,Pfenning2001}.

The motivation behind realizability interpretations is to extract
\emph{computational content} from proofs.
Although some realizability interpretations extract $\lambda$-terms from proofs
(\eg Kreisel's),
as far as we know,
none of these intepretations express the verification procedure as a program.
An active field of research is that of Krivine realizability~\cite{Krivine2009},
as a means to extract computational content from
classical and
set theoretical reasoning principles~\cite{Miquel2011,Guillermo2016,Miquey2018,Dinis2021,Fontanella2024}.

\paragraph{Limitations.}
Our completeness theorems require that the realizer is ``good'':
both \emph{pure} and \emph{$\beta$-normalizing}. We conjecture that
they can be strengthened to avoid relying on $\beta$-normalization.
One strategy could be
to show that $\verif{\ttm}{\subs{\utm}{\penv}} \tos \iunit$
implies that $\utm$ is $\beta$-normalizing,
through a non-idempotent intersection type
system (\cf~\cref{prop:f:standardization}).

It is unclear how to extend the semantics to more expressive
systems, \eg allowing dependent elimination on datatypes.
Inductive and coinductive datatypes can be
encoded in type theories with \emph{impredicative}
universes~\cite{Awodey2018,Geuvers2025}. While these encodings can be
reproduced in higher-order logic, the resulting datatypes
only support \emph{non-dependent} elimination.

\paragraph{Potential application: ``typeless types''.}
One interesting aspect is that verifiers and
generators can be implemented directly \emph{by the user}
in an untyped programming language
that supports features such as first-class functions.
For example, the operators of \cref{sec:second_order}
can be easily implemented in a language such as \tool{Python} or a \tool{Lisp} dialect.
The verifier $\verif{\ttm}{}$ can then be used to \emph{dynamically}
check whether an \emph{untyped} program $\utm$ behaves as a term
of type $\ttm$, without inspecting the source code\footnote{%
  This assumes, of course, that $\utm$ has no side-effects.
}.
This could be a way to integrate static and dynamic type
systems, as in gradual typing~\cite{Siek2006}.

In a similar spirit, \epistemic realizability could potentially become
the basis of proof assistants in which proof checking is implemented
not by statically analyzing the proofs, but by dynamically running the
verifier. These tools would not require type annotations on the
source proofs, and they would not need to perform complex operations such
as elaboration or higher-order unification\footnote{Error reporting, however,
could be challenging.}.

\paragraph{Potential as a foundation for ultrafinitism.}
Ultrafinitism is a branch of constructivism concerned with
\emph{feasibility}~\cite{Parikh1971,YesseninVolpin1975,Nelson1986,Sazonov1995}.
An ultrafinitist may deny that an expression like $2^{1000}$ denotes a natural
number, on the grounds that computational resources are
bounded: a computation exceeding our available time is effectively
indistinguishable from a non-terminating one.
The question of providing logical foundations for ultrafinitism has
not been definitely settled,
as attested by some recent and ongoing works~\cite{Gajda2023,Mannucci2023,Hamkins2025}.

The \epistemic realizability semantics proposed in this paper could perhaps
provide the underlying semantics for an ultrafinitistic logic.
For example, taking the usual second-order encoding of natural numbers,
$\Nat := \allf{\tvar}{(\tvar\imp(\tvar\imp\tvar)\imp\tvar)}$
and the usual definition of exponentiation in this setting,
it is feasible to check that the function $\plamf{n}{2^n}$
is a realizer for $\Nat\imp\Nat$,
as it amounts to a ``short'' calculation
$\verif{\Nat\imp\Nat}{(\plamf{n}{2^n})} \tos \verif{\Nat}{(2^{\gen{\Nat}})} \tos \iunit$.
On the other hand, it is unfeasible to check that $(\plamf{n}{2^n})\,1000$
is a realizer for $\Nat$,
as it would amount to a calculation
$\verif{\Nat}{((\plamf{n}{2^n})\,1000)}
 \tos \verif{\Nat}{2^{1000}}
 \tos \iunit$ that exceeds our computational resources\footnote{%
  Note that our soundness and completeness results entail that
  the typing rule for application is valid, so from
  $\verif{\Nat\imp\Nat}{(\plamf{n}{2^n})} \tos \iunit$
  and $\verif{\Nat}{1000} \tos \iunit$
  one may conclude that
  $\verif{\Nat}{((\plamf{n}{2^n})\,1000)} \tos \iunit$.
  A radical ultrafinitist would likely reject
  the validity of the meta-theoretical reasoning used to justify this.
}.
This seems consistent with the positions of some ultrafinitists,
including Sazonov's proposal that only normal (cut-free) proofs should
be allowed~\cite{Sazonov1995},
and Zeilberger's view that general theorems about ``all'' integers
are actually statements about \emph{symbolic} integers~\cite{Zeilberger2004}.

\paragraph{Future work.}
It would be interesting to provide a semantics for
positive connectives directly, without resorting to their second-order encodings.
For example, the generator $\gen{\ttm\lor\ttmtwo}$
should behave non-deterministically as either
$\texttt{inj}_1(\gen{\ttm})$ or $\texttt{inj}_2(\gen{\ttmtwo})$.

Epistemic realizability relates typability ($\vdash\ptm:\ttm$)
and successful termination ($\verif{\ttm}{\ptm} \tos \iunit$).
This may provide a new strategy to connect simple types
and intersection types, which is the topic of recent work~\cite{Pautasso2023}.

As part of ongoing work, we attempt to extend the semantics
to classical propositional logic,
using a variant of Parigot's $\lambda\mu$-calculus~\cite{Parigot1992}.
We expect that in linear and classical calculi based on sequent calculus,
such as $\bar{\lambda}\mu\tilde{\mu}$~\cite{Curien2000},
the duality between verifiers and generators should become more apparent,
with the verifier for $\ttm$ corresponding to the generator for $\neg\ttm$.

An alternative way to provide a \emph{\epistemic} semantics for \minimalLogic
is through an idea, due to Scott~\cite{Scott1980}, that each proposition $\ttm$
can be interpreted as an \emph{idempotent operator} $\semr{\ttm}$
and each typable term $\vdash \ptm : \ttm$
as a \emph{fixed point} $\semr{\ptm}$ of the operator $\semr{\ttm}$.
The problem is whether this idea can be extended to second and higher-order logic.

\bibliographystyle{ACM-Reference-Format}
\bibliography{biblio}

@article{Church1940,
  title={A Formulation of the Simple Theory of Types},
  journal={The Journal of Symbolic Logic},
  author={Church, Alonzo},
  volume={8},
  number={2},
  pages={56--68},
  url={https://books.google.com.ar/books?id=hP_sGwAACAAJ},
  year={1940}
}

@phdthesis{GirardThesis,
    Author = {Jean-Yves Girard},
    Title = {Interpr\'etation fonctionnelle et \'elimination des coupures de l'arithm\'etique d'ordre sup\'erieur},
    School = {Universit\'e Paris 7},
    Type = {Ph{D} thesis},
    Year = {1972}
}

@incollection{Gallier1990,
  author    = {Jean H. Gallier},
  title     = {On {G}irard's ``{C}andidats de {R}eductibilit\'e''},
  booktitle = {Logic and Computer Science},
  editor    = {Piergiorgio Odifreddi},
  series    = {APIC Studies in Data Processing},
  volume    = {31},
  pages     = {123--203},
  publisher = {Academic Press},
  year      = {1990},
  address   = {London, UK}
}

@inproceedings{Nipkow1991,
  title={Higher-order critical pairs},
  author={Nipkow, Tobias},
  booktitle={Proceedings 1991 Sixth Annual IEEE Symposium on Logic in Computer Science},
  pages={342--343},
  year={1991},
  organization={IEEE Computer Society}
}

@Book{Barendregt1984,
  author = {Henk Barendregt},
  title = {The Lambda Calculus: Its Syntax and Semantics},
  publisher = {Elsevier},
  year = {1984},
  volume = {103}
}

@article{Barendregt1991,
  author={Barendregt, Henk},
  title={Introduction to generalized type systems},
  journal={Journal of Functional Programming},
  volume={1},
  number={2},
  pages={125--154},
  year={1991}
}

@book{Barendregt2013,
  title     = {Lambda Calculus with Types},
  author    = {Barendregt, Henk and Dekkers, Wil and Statman, Richard},
  year      = {2013},
  series    = {Perspectives in Logic},
  publisher = {Cambridge University Press}
}

@book{Girard2011,
  author    = {Jean-Yves Girard},
  title     = {The Blind Spot: Lectures on Logic},
  publisher = {European Mathematical Society},
  year      = {2011},
  address   = {Z{\"u}rich, Switzerland},
  isbn      = {978-3-03719-088-3},
  doi       = {10.4171/088}
}

@book{Dummett1991,
  author    = {Michael Dummett},
  title     = {The Logical Basis of Metaphysics},
  publisher = {Harvard University Press},
  year      = {1991},
  address   = {Cambridge, MA},
  isbn      = {978-0674537866}
}

@book{TroelstraVanDalen1988,
  author    = {A. S. Troelstra and D. van Dalen},
  title     = {Constructivism in Mathematics: An Introduction},
  volume    = {1},
  series    = {Studies in Logic and the Foundations of Mathematics},
  number    = {121},
  publisher = {North-Holland},
  address   = {Amsterdam},
  year      = {1988},
  isbn      = {978-0444705068}
}

@article{Kleene1945,
  title   = {On the interpretation of intuitionistic number theory},
  author  = {Kleene, Stephen Cole},
  journal = {Journal of Symbolic Logic},
  year    = {1945},
  volume  = {10},
  number  = {4},
  pages   = {109--124},
  doi     = {10.2307/2269016},
  publisher = {Association for Symbolic Logic}
}

@inproceedings{Kreisel1962,
  author    = {Georg Kreisel},
  title     = {Foundations of Intuitionistic Logic},
  booktitle = {Logic, Methodology and Philosophy of Science: Proceedings of the 1960 International Congress},
  editor    = {E. Nagel and P. Suppes and A. Tarski},
  publisher = {Stanford University Press},
  address   = {Stanford, CA},
  pages     = {198--210},
  year      = {1962}
}

@article{Prawitz1971,
  author    = {Dag Prawitz},
  title     = {Ideas and results in proof theory},
  journal   = {Proceedings of the Second Scandinavian Logic Symposium},
  volume    = {63},
  pages     = {235--307},
  year      = {1971},
  publisher = {North-Holland Publishing Company}
}

@book{MartinLof1984,
  author    = {Per Martin-L{\"o}f},
  title     = {Intuitionistic Type Theory},
  series    = {Studies in Proof Theory},
  volume    = {1},
  year      = {1984},
  publisher = {Bibliopolis, Naples},
  note      = {Notes by Giovanni Sambin of a series of lectures given in Padua, June 1980}
}

@article{Girard2001,
  author    = {Jean-Yves Girard},
  title     = {Locus Solum: From linear logic to ludics},
  journal   = {Mathematical Structures in Computer Science},
  volume    = {11},
  number    = {3},
  pages     = {301--506},
  year      = {2001},
  publisher = {Cambridge University Press}
}

@inproceedings{Reynolds1974,
  title={Towards a theory of type structure},
  author={Reynolds, John C},
  booktitle={Programming Symposium},
  pages={408--425},
  year={1974},
  organization={Springer}
}

@book{BaaderNipkow,
  title={Term Rewriting and All That},
  author={Baader, Franz and Nipkow, Tobias},
  isbn={9780521779203},
  lccn={97028286},
  url={https://books.google.com.ar/books?id=N7BvXVUCQk8C},
  year={1999},
  publisher={Cambridge University Press}
}

@book{Terese,
  author    = {Terese},
  title     = {Term Rewriting Systems},
  publisher = {Cambridge University Press},
  year      = 2003,
  volume    = 55,
  series    = {Cambridge Tracts in Theoretical Computer Science},
}

@inproceedings{Reynolds1983,
  author    = {Reynolds, John C.},
  title     = {Types, Abstraction and Parametric Polymorphism},
  booktitle = {Information Processing 83},
  pages     = {513--523},
  year      = {1983},
  publisher = {North-Holland}
}

@book{Sorensen2006,
  title={Lectures on the Curry-Howard isomorphism},
  author={S{\o}rensen, Morten Heine and Urzyczyn, Pawel},
  volume={149},
  year={2006},
  publisher={Elsevier}
}

@book{Girard1989,
  title     = {Proofs and Types},
  author    = {Girard, Jean-Yves and Lafont, Yves and Taylor, Paul},
  year      = {1989},
  publisher = {Cambridge University Press},
  series    = {Cambridge Tracts in Theoretical Computer Science},
  volume    = {7},
  address   = {Cambridge, UK},
  isbn      = {0-521-37181-3}
}

@article{Bohm85,
  author    = {Corrado B{\"{o}}hm and
               Alessandro Berarducci},
  title     = {Automatic Synthesis of Typed Lambda-Programs on Term Algebras},
  journal   = {Theor. Comput. Sci.},
  volume    = {39},
  pages     = {135--154},
  year      = {1985},
}

@inproceedings{Gardner94,
  title={Discovering needed reductions using type theory},
  author={Gardner, Philippa},
  booktitle={Theoretical Aspects of Computer Software},
  pages={555--574},
  year={1994},
  organization={Springer}
}

@phdthesis{deCarvalho2007,
  title={S{\'e}mantiques de la logique lin{\'e}aire et temps de calcul},
  author={Carvalho, Daniel de},
  year={2007},
  school={Ecole Doctorale Physique et Sciences de la Mati{\`e}re (Marseille)}
}

@article{Bucciarelli17,
  author       = {Antonio Bucciarelli and
                  Delia Kesner and
                  Daniel Ventura},
  title        = {Non-idempotent intersection types for the Lambda-Calculus},
  journal      = {Log. J. {IGPL}},
  volume       = {25},
  number       = {4},
  pages        = {431--464},
  year         = {2017},
}

@inproceedings{Kesner16,
  author       = {Delia Kesner},
  editor       = {Bart Jacobs and
                  Christof L{\"{o}}ding},
  title        = {Reasoning About Call-by-need by Means of Types},
  booktitle    = {Foundations of Software Science and Computation Structures - 19th
                  International Conference, {FOSSACS} 2016, Held as Part of the European
                  Joint Conferences on Theory and Practice of Software, {ETAPS} 2016,
                  Eindhoven, The Netherlands, April 2-8, 2016, Proceedings},
  series       = {Lecture Notes in Computer Science},
  volume       = {9634},
  pages        = {424--441},
  publisher    = {Springer},
  year         = {2016},
  url          = {https://doi.org/10.1007/978-3-662-49630-5\_25},
  doi          = {10.1007/978-3-662-49630-5\_25},
}

@article{Coppo80,
  author       = {Mario Coppo and
                  Mariangiola Dezani{-}Ciancaglini},
  title        = {An extension of the basic functionality theory for the {\(\lambda\)}-calculus},
  journal      = {Notre Dame J. Formal Log.},
  volume       = {21},
  number       = {4},
  pages        = {685--693},
  year         = {1980},
}

@inproceedings{Pautasso2023,
  author =  {Pautasso, Daniele and Ronchi Della Rocca, Simona},
  title =   {{A Quantitative Version of Simple Types}},
  booktitle =   {8th International Conference on Formal Structures for Computation and Deduction (FSCD 2023)},
  pages =   {29:1--29:21},
  series =  {Leibniz International Proceedings in Informatics (LIPIcs)},
  ISBN =    {978-3-95977-277-8},
  ISSN =    {1868-8969},
  year =    {2023},
  volume =  {260},
  editor =  {Gaboardi, Marco and van Raamsdonk, Femke},
  publisher =   {Schloss Dagstuhl -- Leibniz-Zentrum f{\"u}r Informatik},
  address = {Dagstuhl, Germany},
  URL =     {https://drops.dagstuhl.de/entities/document/10.4230/LIPIcs.FSCD.2023.29},
  URN =     {urn:nbn:de:0030-drops-180137},
  doi =     {10.4230/LIPIcs.FSCD.2023.29}
}

@inproceedings{Awodey2018,
  title={Impredicative encodings of (higher) inductive types},
  author={Awodey, Steve and Frey, Jonas and Speight, Sam},
  booktitle={Proceedings of the 33rd Annual ACM/IEEE Symposium on Logic in Computer Science},
  pages={76--85},
  year={2018}
}

@inproceedings{Geuvers2025,
  title={Impredicative Encodings of Inductive and Coinductive Types},
  author={Bronsveld, Steven and Geuvers, Herman and van der Weide, Niels},
  booktitle={10th International Conference on Formal Structures for Computation and Deduction},
  year={2025}
}

@inproceedings{Faggian2002,
  title={Travelling on designs: ludics dynamics},
  author={Faggian, Claudia},
  booktitle={International Workshop on Computer Science Logic},
  pages={427--441},
  year={2002},
  organization={Springer}
}

@article{Faggian2011,
  title={Ludics with repetitions (exponentials, interactive types and completeness)},
  author={Faggian, Claudia and Basaldella, Michele},
  journal={Logical Methods in Computer Science},
  volume={7},
  year={2011},
  publisher={Episciences. org}
}

@article{Terui2011,
  title={Computational ludics},
  author={Terui, Kazushige},
  journal={Theoretical Computer Science},
  volume={412},
  number={20},
  pages={2048--2071},
  year={2011},
  publisher={Elsevier}
}

@article{Zeilberger2008,
  title={On the unity of duality},
  author={Zeilberger, Noam},
  journal={Annals of pure and applied logic},
  volume={153},
  number={1-3},
  pages={66--96},
  year={2008},
  publisher={Elsevier}
}

@article{Davies2001,
  title={A modal analysis of staged computation},
  author={Davies, Rowan and Pfenning, Frank},
  journal={Journal of the ACM (JACM)},
  volume={48},
  number={3},
  pages={555--604},
  year={2001},
  publisher={ACM New York, NY, USA}
}

@article{Pfenning2001,
  title={A judgmental reconstruction of modal logic},
  author={Pfenning, Frank and Davies, Rowan},
  journal={Mathematical structures in computer science},
  volume={11},
  number={4},
  pages={511--540},
  year={2001},
  publisher={Cambridge University Press}
}

@inproceedings{Miquey2018,
  title={Realizability Interpretation and Normalization of Typed Call-by-Need$\backslash$lambda-calculus with Control.},
  author={Miquey, {\'E}tienne and Herbelin, Hugo},
  booktitle={FoSSaCS},
  pages={276--292},
  year={2018}
}

@article{Krivine2009,
  title={Realizability in classical logic},
  author={Krivine, Jean-Louis},
  journal={Panoramas et synth{\`e}ses},
  volume={27},
  pages={197--229},
  year={2009}
}

@article{Miquel2011,
  title={Existential witness extraction in classical realizability and via a negative translation},
  author={Miquel, Alexandre},
  journal={Logical Methods in Computer Science},
  volume={7},
  year={2011},
  publisher={Episciences. org}
}

@article{Guillermo2016,
  title={Specifying Peirce's law in classical realizability},
  author={Guillermo, Mauricio and Miquel, Alexandre},
  journal={Mathematical Structures in Computer Science},
  volume={26},
  number={7},
  pages={1269--1303},
  year={2016},
  publisher={Cambridge University Press}
}

@inproceedings{Dinis2021,
  title={Realizability with Stateful Computations for Nonstandard Analysis},
  author={Bruno Miguel Antunes Dinis and {\'E}tienne Miquey},
  booktitle={Annual Conference for Computer Science Logic},
  year={2021},
  url={https://api.semanticscholar.org/CorpusID:228902453}
}

@inproceedings{Siek2006,
  author    = {Siek, Jeremy G. and Taha, Walid},
  title     = {Gradual Typing for Functional Languages},
  booktitle = {Proceedings of the 2006 Scheme and Functional Programming Workshop},
  pages     = {81--92},
  year      = {2006},
  url       = {http://scheme2006.cs.uchicago.edu/13-siek.pdf}
}

@inproceedings{Fontanella2024,
  author =  {Fontanella, Laura and Geoffroy, Guillaume and Matthews, Richard},
  title =   {{Realizability Models for Large Cardinals}},
  booktitle =   {32nd EACSL Annual Conference on Computer Science Logic (CSL 2024)},
  pages =   {28:1--28:18},
  series =  {Leibniz International Proceedings in Informatics (LIPIcs)},
  ISBN =    {978-3-95977-310-2},
  ISSN =    {1868-8969},
  year =    {2024},
  volume =  {288},
  editor =  {Murano, Aniello and Silva, Alexandra},
  publisher =   {Schloss Dagstuhl -- Leibniz-Zentrum f{\"u}r Informatik},
  address = {Dagstuhl, Germany},
  URL =     {https://drops.dagstuhl.de/entities/document/10.4230/LIPIcs.CSL.2024.28},
  doi =     {10.4230/LIPIcs.CSL.2024.28}
}

@article{YesseninVolpin1975,
	author = {A. S. Yessenin{-}Volpin},
	doi = {10.2307/2272294},
	journal = {Journal of Symbolic Logic},
	number = {1},
	pages = {95--97},
	publisher = {Association for Symbolic Logic},
	title = {The Ultra-Intuitionistic Criticism and the Antitraditional Program for Foundations of Mathematics},
	volume = {40},
	year = {1975}
}

@inproceedings{Sazonov1995,
  author="Sazonov, Vladimir Yu.",
  editor="Leivant, Daniel",
  title="On feasible numbers",
  booktitle="Logic and Computational Complexity",
  year="1995",
  publisher="Springer Berlin Heidelberg",
  address="Berlin, Heidelberg",
  pages="30--51"
}

@book{Nelson1986,
	address = {Princeton, N.J.},
	author = {Edward Nelson},
	editor = {},
	publisher = {Princeton University Press},
	title = {Predicative Arithmetic},
	year = {1986}
}

@article{Parikh1971,
  title={Existence and feasibility in arithmetic},
  author={Rohit Parikh},
  journal={Journal of Symbolic Logic},
  year={1971},
  volume={36},
  pages={494 - 508},
  url={https://api.semanticscholar.org/CorpusID:6467866}
}

@InProceedings{Gajda2023,
  author =	{Gajda, Micha{\l} J.},
  title =	{{Consistent Ultrafinitist Logic}},
  booktitle =	{29th International Conference on Types for Proofs and Programs (TYPES 2023)},
  pages =	{5:1--5:20},
  series =	{Leibniz International Proceedings in Informatics (LIPIcs)},
  ISBN =	{978-3-95977-332-4},
  ISSN =	{1868-8969},
  year =	{2024},
  volume =	{303},
  editor =	{Kesner, Delia and Reyes, Eduardo Hermo and van den Berg, Benno},
  publisher =	{Schloss Dagstuhl -- Leibniz-Zentrum f{\"u}r Informatik},
  address =	{Dagstuhl, Germany},
  URL =		{https://drops.dagstuhl.de/entities/document/10.4230/LIPIcs.TYPES.2023.5},
  URN =		{urn:nbn:de:0030-drops-204833},
  doi =		{10.4230/LIPIcs.TYPES.2023.5}
}

@misc{Hamkins2025,
  title={A potentialist conception of ultrafinitism}, 
  author={Joel David Hamkins},
  year={2025},
  eprint={2512.06564},
  archivePrefix={arXiv},
  primaryClass={math.LO},
  url={https://arxiv.org/abs/2512.06564}, 
}

@misc{Mannucci2023,
      title={Model Theory of Ultrafinitism II: Deconstructing the Term Model (First Draft)}, 
      author={Mirco A. Mannucci},
      year={2023},
      eprint={2311.17931},
      archivePrefix={arXiv},
      primaryClass={math.LO},
      url={https://arxiv.org/abs/2311.17931}, 
}

@inproceedings{Parigot1992,
  author="Parigot, Michel",
  editor="Voronkov, Andrei",
  title="$\lambda$$\mu$-Calculus: An algorithmic interpretation of classical natural deduction",
  booktitle="Logic Programming and Automated Reasoning",
  year="1992",
  publisher="Springer Berlin Heidelberg",
  address="Berlin, Heidelberg",
  pages="190--201",
  isbn="978-3-540-47279-7"
}

@misc{Curien2000,
  author = {Pierre-Louis Curien and Hugo Herbelin},
  title = {The Duality of Computation},
  year = {2000}
}

@incollection{Scott1980,
  author    = {Scott, Dana S.},
  title     = {Relating Theories of the Lambda-calculus},
  booktitle = {To H. B. Curry: Essays on Combinatory Logic, Lambda Calculus and Formalism},
  editor    = {Hindley, J. Roger and Seldin, Jonathan P.},
  publisher = {Academic Press},
  year      = {1980},
  pages     = {403--450},
  address   = {London}
}

@inproceedings{Zeilberger2004,
  title={``{R}eal'' {A}nalysis is a {D}egenerate {C}ase of {D}iscrete {A}nalysis},
  author={Zeilberger, Doron},
  booktitle={New Progress in Difference Equations (Proceedings of the ICDEA 2001)},
  editor={Aulbach, Bernd and Elaydi, Saber and Ladas, Gerry},
  year={2004},
  publisher={Taylor \& Francis},
  address={London},
  pages={1--4}
}

@unpublished{NelsonIntuitionism,
  author       = {Edward Nelson},
  title        = {Understanding Intuitionism},
  institution  = {Princeton University, Department of Mathematics},
  howpublished = {Unpublished manuscript},
  year         = {n.d.},
  note         = {Available at \url{https://web.math.princeton.edu/~nelson/papers/int.pdf}},
  url          = {https://web.math.princeton.edu/~nelson/papers/int.pdf}
}

\newpage
\appendix

\section{Minimal Logic --- Properties of $\lambdaCheck$}

\subsection{Minimal Logic: Determinism of Weak Head Reduction (\cref{prop:determinism})}
  \label{a:sec:determinism}
  
To show that weak head reduction in $\lambdaCheck$ is deterministic,
\ie that if $\utmtwo \tow \utmtwo_1$ and $\utmtwo \tow \utmtwo_2$
then $\utmtwo_1 = \utmtwo_2$,
it suffices to prove the following lemma:

\begin{lemma}[Uniqueness of the weak head redex]
\label{lem:uniqueness_of_weak_head_redex}
Suppose that $\ctxof{\wctx_1}{\utm_1} = \ctxof{\wctx_2}{\utm_2}$
where $\utm_1$ and $\utm_2$ are redexes.
Then $\wctx_1 = \wctx_2$ and $\utm_1 = \utm_2$.
\end{lemma}
\begin{proof}
We proceed by induction on $\wctx_1$:
\begin{enumerate}
\item
  $\wctx_1 = \ctxhole$: then $\utm_1 = \ctxof{\wctx_2}{\utm_2}$.
  We claim that $\wctx_2 = \ctxhole$.
  Indeed, if $\wctx_2 \neq \ctxhole$
  then $\ctxof{\wctx_2}{\utm_2}$ does not contain a redex at the root,
  as can be seen by case analysis on the shape of the weak head
  context $\wctx_2$ and the redex $\utm_2$.
  Hence $\wctx_1 = \wctx_2 = \ctxhole$ and $\utm_1 = \utm_2$.
\item
  $\wctx_1 = \wctx'_1\,\utmtwo$: note that $\wctx_2$ cannot be empty,
  because this would imply that $\wctx_1$ is empty,
  using a symmetrical argument as in the base case.
  Hence $\wctx_1$ and $\wctx_2$ are both non-empty.
  Since $\ctxof{\wctx_1}{\utm_1} = \ctxof{\wctx_2}{\utm_2}$ by hypothesis,
  then $\wctx_2$ must be of the form $\wctx_2 = \wctx'_2\,\utmtwo$
  and $\ctxof{\wctx'_1}{\utm_1} = \ctxof{\wctx'_2}{\utm_2}$.
  By \ih, $\wctx'_1 = \wctx'_2$ and $\utm_1 = \utm_2$.
  To conclude, note that
  $\wctx_1 = \wctx'_1\,\utmtwo = \wctx'_2\,\utmtwo = \wctx_2$.
\item
  $\wctx_1 = \eunit{\wctx'_1}{\utmtwo}$:
  similar to the previous case.
\item
  $\wctx_1 = \verifteig{\wctx'_1}$:
  similar to the previous case.
\end{enumerate}
\end{proof}

\subsection{Minimal Logic: Universality (\cref{lem:universality})}
  \label{a:sec:universality}
  
To prove universality, we start by showing some auxiliary lemmas.

\begin{lemma}[Splitting]
\label{lem:splitting}
The following are equivalent:
\begin{enumerate}
\item
  $\ctxof{\wctx}{\eunit{\utm}{\utmtwo}} \tow^n \iunit$
\item
  There exist $n_1,n_2 \geq 0$ such that
  $\utm \tow^{n_1} \iunit$
  and $\ctxof{\wctx}{\utmtwo} \tow^{n_2} \iunit$
  and $n = 1 + n_1 + n_2$.
\end{enumerate}
\end{lemma}
\begin{proof}
($1 \implies 2$)
  Suppose that $\ctxof{\wctx}{\eunit{\utm}{\utmtwo}} \tow^n \iunit$.
  We proceed by induction on $n$.
  The reduction cannot be empty, so $n > 0$, \ie $n = 1 + m$ for some $m \geq 0$,
  and the reduction must be of the form
  $\ctxof{\wctx}{\eunit{\utm}{\utmtwo}} \tow \utmthree \tow^m \iunit$.
  We consider two subcases, depending on whether
  the redex contracted in the first step is $\eunit{\utm}{\utmtwo}$
  or internal to $\utm$:
  \begin{enumerate}
  \item
    If the contracted redex is $\eunit{\utm}{\utmtwo}$,
    then $\utm = \iunit$,
    and the reduction is of the form
    $\ctxof{\wctx}{\eunit{\iunit}{\utmtwo}}
     \tow \ctxof{\wctx}{\utmtwo}
     \tow^m \iunit$.
    Hence $\utm = \iunit \tow^0 \iunit$
    and $\ctxof{\wctx}{\utmtwo} \tow^m \iunit$,
    and taking $n_1 := 0$ and $n_2 := m$ we have that
    $n = 1 + m = 1 + n_1 + n_2$, as required.
  \item
    If the contracted redex is internal to $\utm$,
    then the reduction is of the form
    $\ctxof{\wctx}{\eunit{\utm}{\utmtwo}}
     \tow \ctxof{\wctx}{\eunit{\utm'}{\utmtwo}}
     \tow^m \iunit$
    where $\utm \tow \utm'$.
    By \ih there exist $m_1,m_2$
    such that $\utm' \tow^{m_1} \iunit$
    and $\ctxof{\wctx}{\utmtwo} \tow^{m_2} \iunit$
    and $m = 1 + m_1 + m_2$.
    Taking $n_1 := 1 + m_1$ and $n_2 := m_2$,
    we have that
    $\utm \tow \utm' \tows^{m_1} \iunit$
    where $n = 1 + m = 2 + m_1 + m_2 = 1 + n_1 + n_2$.
  \end{enumerate}
($2 \implies 1$)
  Suppose that $\utm \tow^{n_1} \iunit$ and $\ctxof{\wctx}{\utmtwo} \tow^{n_2} \iunit$.
  Then we have that
  $\ctxof{\wctx}{\eunit{\utm}{\utmtwo}}
  \tow^{n_1} \ctxof{\wctx}{\eunit{\iunit}{\utmtwo}}
  \tow \ctxof{\wctx}{\utmtwo}
  \tow^{n_2} \iunit$
  of length $n = 1 + n_1 + n_2$.
\end{proof}

Recall that a \defn{redex} is a reducible expression.
In this specific setting, a \metaterm $\utm$ is a redex
if it matches the left-hand side of one of the rewriting
rules in $\lambdaCheck$. Then:

\begin{lemma}[Substitution commutes with redexes]
\label{lem:subst_commute_redex}
\label{lem:tow_deterministic}
If $\utmfive \to \utmsix$ is a reduction step at the root
in $\lambdaCheck$,
then $\subs{\utmfive}{\subst} \to \subs{\utmsix}{\subst}$
is a reduction step at the root in $\lambdaCheck$. 
\end{lemma}
\begin{proof}
Straightforward, by case analysis on the possible shapes of the redex $\utmfive$.
\end{proof}

\begin{lemma}[Substitution commutes with weak head contexts]
\label{lem:subst_commute_wctx}
\quad
\begin{enumerate}
\item
  If $\wctx$ is a weak head context,
  then $\subs{\wctx}{\subst}$ is a weak head context. 
\item
  $\subs{\ctxof{\wctx}{\utm}}{\subst}
  = \ctxof{\subs{\wctx}{\subst}}{\subs{\utm}{\subst}}$
\end{enumerate}
\end{lemma}
\begin{proof}
Both items are straightforward by induction on $\wctx$.
The key remark is that the hole of a weak head context does not lie
inside an abstraction, so in the expression
$\ctxof{\subs{\wctx}{\subst}}{\subs{\utm}{\subst}}$ there is no
danger of variable capture.
\end{proof}

\begin{lemma}[Origin of weak head redexes]
\label{lem:subst_weak_head_redex}
If $\subs{\utm}{\penv}$ is of the form $\ctxof{\wctx}{\utmfive}$,
where $\utmfive$ is a redex, then exactly one of the following holds:
\begin{enumerate}
\item
  There exist $\wctx_0,\utmfive_0$ such that
  $\utm = \ctxof{\wctx_0}{\utmfive_0}$ 
  where $\subs{\wctx_0}{\penv} = \wctx$
  and $\subs{\utmfive_0}{\penv} = \utmfive$
  and $\utmfive_0$ is a redex.
\item
  There exist $\wctx_0,\var,\utmtwo,\vartwo,\utmthree$ such that
  $\utm = \ctxof{\wctx_0}{\var\,\utmtwo}$
  where $\subs{\wctx_0}{\penv} = \wctx$
  and $\subs{\var}{\penv} = \lam{\vartwo}{\utmthree}$.
\item
  There exist $\wctx_0,\var,\utmtwo$ such that
  $\utm = \ctxof{\wctx_0}{\eunit{\var}{\utmtwo}}$
  where $\subs{\wctx_0}{\penv} = \wctx$
  and $\subs{\var}{\penv} = \iunit$.
\item
  There exist $\wctx_0,\var$ such that
  $\utm = \ctxof{\wctx_0}{\verifteig{\var}}$
  where $\subs{\wctx_0}{\penv} = \wctx$
  and $\subs{\var}{\penv} = \genteig$.
\end{enumerate}
\end{lemma}
\begin{proof}
Straightforward by induction on $\utm$.
\end{proof}

\begin{lemma}[Universality]
\label{a:lem:universality}
Let $\compat{\penv}{\subst}$.
If $\subs{\utm}{\penv} \tos \iunit$
then $\subs{\utm}{\subst} \tos \iunit$.
\end{lemma}
\begin{proof}
Suppose that $\compat{\penv}{\subst}$.
By \cref{prop:weak_standardization},
it suffices to show that $\subs{\utm}{\penv} \tows \iunit$
implies $\subs{\utm}{\subst} \tows \iunit$.
Suppose that $\subs{\utm}{\penv} \tows \iunit$.
Let us write $\typesize{\ttm}$ for the size of the \logicalTerm $\ttm$,
and $\typesize{\penv}$ for the maximum size of a \logicalTerm occurring in $\penv$,
\ie $\typesize{\penv} \eqdef \max\set{\typesize{\ttm} \ST \exists \var.\, (\var:\ttm \in \penv)}$.
We show that $\subs{\utm}{\subst} \tows \iunit$
by induction on the lexicographic measure $(\typesize{\penv}, n)$,
where $n$ is the length of the reduction sequence $\subs{\utm}{\penv} \tow^n \iunit$.
We consider two cases, depending on whether this reduction sequence is empty or
not, \ie $n = 0$ or $n > 0$.

  \textbf{Empty ($n = 0$).}
  Then $\subs{\utm}{\penv} = \iunit$.
  We claim that $\utm = \iunit$.
  Indeed, note that $\utm$ cannot be a variable, because if
  $\utm = \var$ is a variable
  then either $\var:\ttm \in \penv$
  and $\subs{\utm}{\penv} = \gen{\ttm} \neq \iunit$,
  or $\var \notin \dom{\penv}$ and $\subs{\utm}{\penv} = \var \neq \iunit$.
  Moreover, if $\utm$ is not $\iunit$ nor a variable, then
  $\subs{\utm}{\penv} \neq \iunit$
  (for example, if $\utm$ is an application then
  $\subs{\utm}{\penv}$ is also an application).
  So we have that $\utm = \iunit$ and it is immediate to conclude that
  $\subs{\utm}{\subst} = \iunit \tows \iunit$.

  \textbf{Non-empty ($n = 1 + m$).}
  Since the reduction sequence $\subs{\utm}{\penv} \tow^n \iunit$ is non-empty,
  $\subs{\utm}{\penv}$ has a weak head redex.
  By \cref{lem:subst_weak_head_redex} there are four cases:
  either $\utm = \ctxof{\wctx}{\utmfive}$ where $\utmfive$ is a redex,
  or $\utm = \ctxof{\wctx}{\var\,\utmtwo}$ where $\subs{\var}{\penv}$ is an abstraction,
  or $\utm = \ctxof{\wctx}{\eunit{\var}{\utmtwo}}$ where $\subs{\var}{\penv} = \iunit$,
  or $\utm = \ctxof{\wctx}{\verifteig{\var}}$ where $\subs{\var}{\penv} = \genteig$.
  We consider each of these four cases individually:
  \begin{enumerate}
  \item
    If $\utm = \ctxof{\wctx}{\utmfive}$ where $\utmfive$ is a redex:
    consider the root reduction step $\utmfive \tow \utmsix$.
    First note that $\subs{\utmfive}{\penv} \tow \subs{\utmsix}{\penv}$
    is also a root reduction step by \cref{lem:subst_commute_redex}.
    Since head reduction is deterministic (\cref{lem:tow_deterministic}),
    the input reduction sequence must be of the form:
    \[
      \begin{array}{rlll}
      &&
        \subs{\utm}{\penv}
      \\
      & = &
        \ctxof{\subs{\wctx}{\penv}}{\subs{\utmfive}{\penv}}
        & \text{by \cref{lem:subst_commute_wctx}}
      \\
      & \tow &
        \ctxof{\subs{\wctx}{\penv}}{\subs{\utmsix}{\penv}}
        & \text{by \cref{lem:subst_commute_wctx}, using that $\subs{\utmfive}{\penv} \tow \subs{\utmsix}{\penv}$}
      \\
      & \tow^m &
        \iunit
      \end{array}
    \]
    The reduction sequence
    $\ctxof{\subs{\wctx}{\penv}}{\subs{\utmsix}{\penv}} \tow^m \iunit$
    is shorter than the original one ($m < n$).
    We may apply the \ih, using the same \typeEnvironment $\penv$,
    in such a way that the lexicographic measure decreases:
    $(\typesize{\penv},n) > (\typesize{\penv},m)$.
    By \ih we have that
    $\ctxof{\subs{\wctx}{\subst}}{\subs{\utmsix}{\subst}} \tows \iunit$.
    To conclude, note that
    $\subs{\utmfive}{\subst} \tow \subs{\utmsix}{\subst}$
    is a root reduction step by \cref{lem:subst_commute_redex},
    so:
    \[
      \begin{array}{rlll}
      &&
        \subs{\utm}{\subst}
      \\
      & = &
        \ctxof{\subs{\wctx}{\subst}}{\subs{\utmfive}{\subst}}
        & \text{by \cref{lem:subst_commute_wctx}}
      \\
      & \tow &
        \ctxof{\subs{\wctx}{\subst}}{\subs{\utmsix}{\subst}}
        & \text{by \cref{lem:subst_commute_wctx}, using that $\subs{\utmfive}{\subst} \tow \subs{\utmsix}{\subst}$}
      \\
      & \tows &
        \iunit
        & \text{as obtained from the \ih}
      \end{array}
    \]
  \item
    If $\utm = \ctxof{\wctx}{\var\,\utmtwo}$
    where $\subs{\var}{\penv}$ is an abstraction,
    \ie $\subs{\var}{\penv} = \lam{\vartwo}{\utmthree}$.

    First, we claim that $\var : \ttm\arrow\ttmtwo \in \penv$
    for certain \logicalTerms $\ttm,\ttmtwo$.
    Indeed, if $\var \notin \dom{\penv}$ then $\subs{\var}{\penv} = \var$
    is not an abstraction; so we must have that $\var \in \dom{\penv}$.
    Moreover if $\var : \teig \in \penv$, \ie if it is assigned an \atomicProposition,
    then $\subs{\var}{\penv} = \genteig$ is not an abstraction;
    so we must have that $\var$ occurs in $\penv$ with a \logicalTerm which
    is not an \atomicProposition.
    Thus $\var$ must occur in $\penv$ with an arrow type.

    Then the input reduction sequence must be of the form:
    \[
      \begin{array}{rlll}
      &&
        \subs{\utm}{\penv}
      \\
      & = &
        \subs{\ctxof{\wctx}{\var\,\utmtwo}}{\penv}
      \\
      & = &
        \ctxof{\subs{\wctx}{\penv}}{\subs{\var}{\penv}\,\subs{\utmtwo}{\penv}}
        & \text{by \cref{lem:subst_commute_wctx}}
      \\
      & = &
        \ctxof{\subs{\wctx}{\penv}}{\gen{\ttm\arrow\ttmtwo}\,\subs{\utmtwo}{\penv}}
        & \text{as $\var:\ttm\arrow\ttmtwo \in \penv$}
      \\
      & = &
        \ctxof{\subs{\wctx}{\penv}}{(\lam{\var}{\eunit{\verif{\ttm}{\var}}{\gen{\ttmtwo}}})\,\subs{\utmtwo}{\penv}}
      \\
      & \tow &
        \ctxof{\subs{\wctx}{\penv}}{\eunit{\verif{\ttm}{\subs{\utmtwo}{\penv}}}{\gen{\ttmtwo}}}
      \\
      & \tow^m &
        \iunit
      \end{array}
    \]
    Consider the tail of the reduction sequence:
    \[
      \ctxof{\subs{\wctx}{\penv}}{\eunit{\verif{\ttm}{\subs{\utmtwo}{\penv}}}{\gen{\ttmtwo}}}
       \tow^m \iunit
    \]
    By \cref{lem:splitting}, there exist $m_1,m_2 \geq 0$ such that
    $\verif{\ttm}{\subs{\utmtwo}{\penv}} \tow^{m_1} \iunit$
    and
    $\ctxof{\subs{\wctx}{\penv}}{\gen{\ttmtwo}} \tow^{m_2} \iunit$
    and $m = 1 + m_1 + m_2$.

    Since $\compat{\penv}{\subst}$ by hypothesis
    and $\var : \ttm\arrow\ttmtwo \in \penv$,
    we know that $\verif{\ttm\arrow\ttmtwo}{\subs{\var}{\subst}} \tos \iunit$.
    By definition of $\verif{\ttm\arrow\ttmtwo}{\ARG}$, we have:
    \[
      \verif{\ttmtwo}{(\subs{\var}{\subst}\,\gen{\ttm})}
      = \verif{\ttm\arrow\ttmtwo}{\subs{\var}{\subst}}
      \tos \iunit
    \]

    Let $\vartwo$ be a fresh variable, and consider
    the \metaterm $\ctxof{\wctx}{\vartwo}$,
    the \typeEnvironment $\penv_1 := \penv\cup\set{\vartwo:\ttmtwo}$,
    and the substitution $\subst_1 := \subst\extsub{\vartwo}{\subs{\var}{\subst}\,\gen{\ttm}}$.
    Note that $\compat{\penv_1}{\subst_1}$
    because for $\vartwo:\ttmtwo\in\penv_1$
    we have that
    $\verif{\ttmtwo}{(\vartwo^{\subst_1})}
    = \verif{\ttmtwo}{(\var^{\subst}\,\gen{\ttm})}
    \tows \iunit$,
    and for any variable other than $\vartwo$ the property
    follows from the fact that $\compat{\penv}{\subst}$.
    Recall that $\var:\ttm\arrow\ttmtwo \in \penv$
    so $\typesize{\penv} \geq \typesize{\ttm\arrow\ttmtwo} > \typesize{\ttmtwo}$,
    which means that
    $\typesize{\penv_1}
     = \typesize{\penv,\vartwo:\ttmtwo}
     = \max\set{\typesize{\penv},\typesize{\ttmtwo}}
     = \typesize{\penv}$.
    Also, note that:
    \[
      \begin{array}{rlll}
      &&
        \subs{\ctxof{\wctx}{\vartwo}}{\penv_1}
      \\
      & = &
        \ctxof{\subs{\wctx}{\penv_1}}{\subs{\vartwo}{\penv_1}}
        & \text{by \cref{lem:subst_commute_wctx}}
      \\
      & = &
        \ctxof{\subs{\wctx}{\penv}}{\subs{\vartwo}{\penv_1}}
        & \text{since $\vartwo$ is fresh for $\wctx$}
      \\
      & = &
        \ctxof{\subs{\wctx}{\penv}}{\gen{\ttmtwo}}
        & \text{by definition of $\subs{\vartwo}{\penv_1}$}
      \\
      & \tow^{m_2} &
        \iunit
        & \text{as already shown}
      \end{array}
    \]
    Furthermore, note that the lexicographic measure decreases:
    $(\typesize{\penv},n)
     = (\typesize{\penv_1},n)
     > (\typesize{\penv_1},m_2)$,
    so we may apply the \ih to obtain that
    $\subs{\ctxof{\wctx}{\vartwo}}{\subst_1} \tows \iunit$.
    Hence:
    \[
      \begin{array}{rlll}
      &&
        \ctxof{\subs{\wctx}{\subst}}{\subs{\var}{\subst}\,\gen{\ttm}}
      \\
      & = &
        \ctxof{\subs{\wctx}{\subst}}{\subs{\vartwo}{\subst_1}}
        & \text{by definition of $\subs{\vartwo}{\subst_1}$}
      \\
      & = &
        \ctxof{\subs{\wctx}{\subst_1}}{\subs{\vartwo}{\subst_1}}
        & \text{since $\vartwo$ is fresh for $\wctx$}
      \\
      & = &
        \subs{\ctxof{\wctx}{\vartwo}}{\subst_1}
        & \text{by \cref{lem:subst_commute_wctx}}
      \\
      & \tows &
        \iunit
        & \text{as just shown}
      \end{array}
    \]

    Recall that we have $\verif{\ttm}{\subs{\utmtwo}{\penv}} \tow^{m_1} \iunit$.
    Taking the environment $\penv$ and the substitution $\subst$,
    we can note that the lexicographic measure decreases:
    $(\typesize{\penv},n) > (\typesize{\penv},m_1)$,
    so may apply the \ih to obtain that
    $\verif{\ttm}{\subs{\utmtwo}{\subst}} \tows \iunit$.

    Now let $\varthree$ be a fresh variable, and consider
    the \metaterm $\ctxof{\subs{\wctx}{\subst}}{\subs{\var}{\subst}\,\varthree}$,
    the \typeEnvironment $\penv_2 := \set{\varthree:\ttm}$,
    and the substitution $\subst_2 := \extsub{\varthree}{\subs{\utmtwo}{\subst}}$.
    Note that $\compat{\penv_2}{\subst_2}$
    because $\verif{\ttm}{\subs{\utmtwo}{\subst}} \tows \iunit$, as we have
    just proved.
    Moreover, we have that:
    \[
      \begin{array}{rlll}
      &&
        \subs{\ctxof{\subs{\wctx}{\subst}}{\subs{\var}{\subst}\,\varthree}}{\penv_2}
      \\
      & = &
        \ctxof{\subs{(\subs{\wctx}{\subst})}{\penv_2}}{\subs{(\subs{\var}{\subst})}{\penv_2}\,\subs{\varthree}{\penv_2}}
        & \text{by \cref{lem:subst_commute_wctx}}
      \\
      & = &
        \ctxof{\subs{\wctx}{\subst}}{\subs{\var}{\subst}\,\subs{\varthree}{\penv_2}}
        & \text{since $\varthree$ is fresh for $\wctx,\var,\subst$}
      \\
      & = &
        \ctxof{\subs{\wctx}{\subst}}{\subs{\var}{\subst}\,\gen{\ttm}}
        & \text{since $\varthree : \ttm \in \penv_2$}
      \\
      & \tows &
        \iunit
        & \text{as already shown}
      \end{array}
    \]
    Let the length of this reduction sequence be called $n'$.
    Note that the lexicographic measure decreases:
    $(\typesize{\penv},n) > (\typesize{\penv_2},n')$
    because the first component decreases, \ie
    $\typesize{\penv}
     \geq \typesize{\ttm\arrow\ttmtwo}
     > \typesize{\ttm}
     = \typesize{\penv_2}$,
    so we may apply the \ih to conclude that
    $\subs{\ctxof{\subs{\wctx}{\subst}}{\subs{\var}{\subst}\,\varthree}}{\subst_2}
     \tows \iunit$.
    Finally, note that:
    \[
      \begin{array}{rlll}
      &&
        \subs{\utm}{\subst}
      \\
      & = &
        \subs{\ctxof{\wctx}{\var\,\utmtwo}}{\subst}
      \\
      & = &
        \ctxof{\subs{\wctx}{\subst}}{\subs{\var}{\subst}\,\subs{\utmtwo}{\subst}}
        & \text{by \cref{lem:subst_commute_wctx}}
      \\
      & = &
        \ctxof{\subs{\wctx}{\subst}}{\subs{\var}{\subst}\,\subs{\varthree}{\subst_2}}
        & \text{by definition of $\subst_2$}
      \\
      & = &
        \ctxof{\subs{(\subs{\wctx}{\subst})}{\subst_2}}{\subs{(\subs{\var}{\subst})}{\subst_2}\,\subs{\varthree}{\subst_2}}
        & \text{since $\varthree$ is fresh for $\wctx,\var,\subst$}
      \\
      & = &
        \subs{\ctxof{\subs{\wctx}{\subst}}{\subs{\var}{\subst}\,\varthree}}{\subst_2}
        & \text{by \cref{lem:subst_commute_wctx}}
      \\
      & \tows &
        \iunit
        & \text{as just shown}
      \end{array}
    \]
  \item
    If $\utm = \ctxof{\wctx}{\eunit{\var}{\utmtwo}}$
    where $\subs{\var}{\penv} = \iunit$:
    we claim that this case is impossible. 
    Indeed, since $\subs{\var}{\penv} = \iunit$,
    we know that $\var \in \dom{\penv}$,
    but if $\var : \ttm \in \dom{\penv}$
    then $\subs{\var}{\penv} = \gen{\ttm} \neq \iunit$,
    contradicting the fact that $\subs{\var}{\penv} = \iunit$.
  \item
    If $\utm = \ctxof{\wctx}{\verifteig{\var}}$
    where $\subs{\var}{\penv} = \genteig$:
    then it must be the case that $\var : \teig \in \penv$.
    Since head reduction is deterministic (\cref{lem:tow_deterministic}),
    the input reduction sequence must be of the form:
    \[
      \begin{array}{rlll}
      &&
        \subs{\utm}{\penv}
      \\
      & = &
        \subs{\ctxof{\wctx}{\verifteig{\var}}}{\penv}
      \\
      & = &
        \ctxof{\subs{\wctx}{\penv}}{\verifteig{\subs{\var}{\penv}}}
        & \text{by \cref{lem:subst_commute_wctx}}
      \\
      & = &
        \ctxof{\subs{\wctx}{\penv}}{\verifteig{\genteig}}
        & \text{since $\var : \teig \in \penv$}
      \\
      & \tow &
        \ctxof{\subs{\wctx}{\penv}}{\iunit}
      \\
      & = &
        \subs{\ctxof{\wctx}{\iunit}}{\penv}
        & \text{by \cref{lem:subst_commute_wctx}}
      \\
      & \tow^m &
        \iunit
      \end{array}
    \]
    Note that the lexicographic measure decreases:
    $(\typesize{\penv},n) > (\typesize{\penv},m)$,
    so by \ih we have that
    $\subs{\ctxof{\wctx}{\iunit}}{\subst} \tows \iunit$.
    Hence:
    \[
      \begin{array}{rlll}
      &&
        \subs{\utm}{\subst}
      \\
      & = &
        \subs{\ctxof{\wctx}{\verifteig{\var}}}{\subst}
      \\
      & = &
        \ctxof{\subs{\wctx}{\subst}}{\verifteig{\subs{\var}{\subst}}}
        & \text{by \cref{lem:subst_commute_wctx}}
      \\
      & = &
        \ctxof{\subs{\wctx}{\subst}}{\iunit}
        & \text{because $\verifteig{\subs{\var}{\subst}} \tows \iunit$ by $\compat{\penv}{\subst}$}
      \\
      & = &
        \subs{\ctxof{\wctx}{\iunit}}{\subst}
        & \text{by \cref{lem:subst_commute_wctx}}
      \\
      & \tows &
        \iunit
        & \text{as just shown}
      \end{array}
    \]
  \end{enumerate}
\end{proof}

\subsection{Minimal Logic: Completeness (\cref{thm:completeness})}
  \label{a:sec:completeness}
  
\begin{remark}
$\sig[\penv,\var:\ttm]{\var}{\ttm'} = \minof{\size{\ttm}}{\size{\ttm'}}$.
\end{remark}

\begin{lemma}
\label{lemma:size_weakening}
If $\var \notin \fv{\unf}$
then $\sig[\penv]{\unf}{\ttm} = \sig[\penv,\var:\ttmtwo]{\unf}{\ttm}$
\end{lemma}
\begin{proof}
Straightforward by induction on $\unf$.
\end{proof}

\begin{lemma}[Proof size of application to \proofVariables]
\label{lemma:size_normalize}
\quad
\begin{enumerate}
\item
  Let $\unf$ be a normal \proofTerm.
  Then the application $\unf\,\var$ $\beta$-reduces to a $\beta$-normal form
  $\unftwo$ in at most one step,
  and $\sig{\unf}{\ttm\imp\ttmtwo} \geq \sig[\penv,\pvar:\ttm]{\unftwo}{\ttmtwo}$.
\item
  Let $\unf$ be a normal \proofTerm.
  Then $\unf\,\var_1\hdots\var_n$ $\beta$-reduces to a $\beta$-normal form $\utmtwo$
  in at most $n$ steps,
  and $\sig{\unf}{\ttm_1\imp\hdots\ttm_n\imp\ttmtwo} \geq \sig[\penv,\var_1:\ttm_1,\hdots,\var_n:\ttm_n]{\unftwo}{\ttmtwo}$.
\end{enumerate}
\end{lemma}
\begin{proof}
The second item is immediate by induction on $n$, resorting to the first item
in each step.
To show the first item, proceed by induction on the shape of the
normal \proofTerm $\unf$ according to~\cref{def:grammar_of_nfs}.

In the \textbf{abstraction} case,
  we have that $\unf = \lam{\var}{\unftwo}$.
  Note that $\unf\,\var = (\lam{\var}{\unftwo})\,\var \to \unftwo$
  and:
  \[
    \sig{\unf}{\ttm\imp\ttmtwo}
    =
    \sig{\lam{\var}{\unftwo}}{\ttm\imp\ttmtwo}
    =
    \sig[\penv,\var:\ttm]{\unftwo}{\ttmtwo}
  \]

In the \textbf{neutral \proofTerm} case,
  we have that $\unf = \vartwo\,\unf_1\hdots\unf_n$.
  Take $\unftwo := \vartwo\,\unf_1\hdots\unf_n\,\var$.
  Let $\mathcal{P}$ and $\mathcal{Q}$ be the following conditions:
  \[
    \begin{array}{lll}
      \mathcal{P} & \equiv &
          \text{there exist $\ttmthree,\ttmthree_1,\hdots,\ttmthree_n,\ttmfour$} \\
        &&\text{such that $\vartwo:\ttmthree \in \penv$
                and $\ttmthree = (\ttmthree_1\imp\hdots\ttmthree_n\imp\ttmfour)$}
    \\
      \mathcal{Q} & \equiv &
          \text{there exist $\ttmthree,\ttmthree_1,\hdots,\ttmthree_n,\ttm',\ttmtwo'$} \\
        &&\text{such that $\vartwo:\ttmthree \in \penv$
                and $\ttmthree = (\ttmthree_1\imp\hdots\ttmthree_n\imp\ttm'\imp\ttmtwo')$}
    \end{array}
  \]
  If $\mathcal{Q}$ does not hold, we are done,
  because
    $\sig[\penv,\var:\ttm]{\vartwo\,\unf_1\hdots\unf_n\,\var}{\ttmtwo} = 0$,
  so
    $\sig{\unf}{\ttm\imp\ttmtwo} \geq 0 = \sig[\penv,\var:\ttm]{\vartwo\,\unf_1\hdots\unf_n\,\var}{\ttmtwo}$.

  Assume that $\mathcal{Q}$ holds, and note that $\mathcal{P}$ also holds.
  More precisely,
    there exist $\ttmthree,\ttmthree_1,\hdots,\ttmthree_n,\ttm',\ttmtwo'$
    such that $\vartwo:\ttmthree \in \penv$
    and $\ttmthree = (\ttmthree_1\imp\hdots\ttmthree_n\imp\ttm'\imp\ttmtwo')$.
  Then:
  \[
    \begin{array}{@{}r@{\,\,}cll}
    &&
      \sig{\unf}{\ttm\imp\ttmtwo}
    \\
    & = &
      \sig{\vartwo\,\unf_1\hdots\unf_n}{\ttm\imp\ttmtwo}
    \\
    & = &
      \minof{\size{\ttm\imp\ttmtwo}}{\size{\ttm'\imp\ttmtwo'}}
      + \sum_{i=1}^n \sig{\unf_i}{\ttmthree_i}
    \\
    & = &
      \minof{1+\size{\ttm}+\size{\ttmtwo}}{1+\size{\ttm'}+\size{\ttmtwo'}}
      + \sum_{i=1}^n \sig{\unf_i}{\ttmthree_i}
    \\
    & > &
      \minof{\size{\ttm}}{\size{\ttm'}}
      + \minof{\size{\ttmtwo}}{\size{\ttmtwo'}}
      + \sum_{i=1}^n \sig{\unf_i}{\ttmthree_i}
    \\
    & = &
      \minof{\size{\ttmtwo}}{\size{\ttmtwo'}} + \sum_{i=1}^n \sig{\unf_i}{\ttmthree_i} + \minof{\size{\ttm}}{\size{\ttm'}}
    \\
    & = &
      \minof{\size{\ttmtwo}}{\size{\ttmtwo'}}
          + \sum_{i=1}^n \sig{\unf_i}{\ttmthree_i}
          + \sig[\penv,\var:\ttm]{\var}{\ttm'}
    \\
    & = &
      \minof{\size{\ttmtwo}}{\size{\ttmtwo'}}
          + \sum_{i=1}^n \sig[\penv,\var:\ttm]{\unf_i}{\ttmthree_i}
          + \sig[\penv,\var:\ttm]{\var}{\ttm'}
      \,\,\text{by \cref{lemma:size_weakening}}
    \\
    & = &
      \sig[\penv,\var:\ttm]{\vartwo\,\unf_1\hdots\unf_n\,\var}{\ttmtwo}
    \end{array}
  \]
\end{proof}

\begin{lemma}[Proof size of arguments]
\label{lemma:sizeVar_argument}
Let $\var\utm_1\hdots\utm_n$ be a \metaterm headed by a variable
and suppose that $\var:(\ttmtwo_1\imp\hdots\imp\ttmtwo_n\imp\ttm)\in\penv$.
Then for all $1 \leq i \leq n$
we have that
$\sig{\var\utm_1\hdots\utm_n}{\ttm} > \sig{\utm_i}{\ttmtwo_i}$.
\end{lemma}
\begin{proof}
Immediate, since by definition
$\sig{\var\utm_1\hdots\utm_n}{\ttm}
= \size{\ttm} + \sum_{i=1}^{n} \sig{\utm_i}{\ttmtwo_i}
> \sig{\utm_i}{\ttmtwo_i}$.
\end{proof}

\begin{theorem}[Completeness]
\label{a:thm:completeness}
If $\verif{\ttm}{\subs{\utm}{\penv}} \tos \iunit$
and $\utm$ is good,
then its $\beta$-normal form is a term $\utm^\downarrow$
such that $\judg{\penv}{\utm^\downarrow}{\ttm}$.
\end{theorem}
\begin{proof}
Let $\unf$ be the $\beta$-normal form of $\utm$, which is a normal \proofTerm.
We proceed by induction on the natural number $\sig{\unf}{\ttm}$.
Without loss of generality, let us write
$\ttm = (\ttm_1 \arrow \hdots \ttm_n \arrow \teig)$.
Note that
$\verif{\ttm}{\subs{\utm}{\penv}}
 \tos \verif{\teig}{(\subs{\utm}{\penv}\,\gen{\ttm_1}\hdots\gen{\ttm_n})}
 \tos \iunit$.
Let $\var_1,\hdots,\var_n$ be fresh variables,
and consider the \metaterm $\utm\,\var_1\hdots\var_n$
and the \typeEnvironment $\penvtwo := (\penv,\var_1:\ttm_1,\hdots,\var_n:\ttm_n)$.
Note that $\subs{\utm}{\penv} = \subs{\utm}{\penvtwo}$ because $\var_1,\hdots,\var_n$
are fresh.
Then we have that
$\verif{\teig}{\subs{(\utm\,\var_1\hdots\var_n)}{\penvtwo}}
= \verif{\teig}{(\subs{\utm}{\penvtwo}\,\gen{\ttm_1}\hdots\gen{\ttm_n})}
= \verif{\teig}{(\subs{\utm}{\penv}\,\gen{\ttm_1}\hdots\gen{\ttm_n})}
\tos \iunit$.
Let $\utmtwo^\downarrow$ be the $\beta$-normal form of $\utm^\downarrow\,\var_1\hdots\var_n$.
Then by confluence (\cref{prop:confluence})
we have that
$\verif{\teig}{\subs{(\utmtwo^\downarrow)}{\penvtwo}} \tos \iunit$.
Note that $\utmtwo^\downarrow$ is pure and
that
$\sig{\unf}{\ttm}
 = \sig{\unf}{\ttm_1\imp\hdots\ttm_n\imp\teig}
 \geq \sig[\penv,\var_1:\ttm_1,\hdots,\var_n:\ttm_n]{\unftwo}{\teig}
 = \sig[\penvtwo]{\unftwo}{\teig}$
holds by \cref{lemma:size_normalize}.

Since $\utmtwo^\downarrow$ is a pure \metaterm in $\beta$-normal form, we consider two
cases, depending on whether it is an abstraction or a neutral \proofTerm.

\medskip
First,
  if $\utmtwo^\downarrow$ is an \textbf{abstraction},
  \ie $\utmtwo^\downarrow = \lam{\vartwo}{\utmthree}$,
  then we have that
  $\verifteig{\lam{\vartwo}{\subs{\utmthree}{\penvtwo}}}
  = \verifteig{\subs{(\lam{\vartwo}{\utmthree})}{\penvtwo}}
  = \verifteig{\subs{(\utmtwo^\downarrow)}{\penvtwo}}
  \tos \iunit$,
  which is impossible,
  because the reducts of an abstraction are always abstractions,
  so in particular $\lam{\vartwo}{\subs{\utmthree}{\penvtwo}}$
  cannot have $\genteig$ as a reduct.
\medskip

Second, if $\utmtwo^\downarrow$ is a \textbf{neutral term},
  then $\utmtwo^\downarrow = \var\,\utmtwo_1\hdots\utmtwo_m$,
  and we have that
  $\verifteig{(\subs{\var}{\penvtwo}\,\subs{\utmtwo_1}{\penvtwo}\hdots\subs{\utmtwo_m}{\penvtwo})}
  = \verifteig{\subs{(\var\,\utmtwo_1\hdots\utmtwo_m)}{\penvtwo}}
  = \verifteig{\subs{(\utmtwo^\downarrow)}{\penvtwo}}
  \tos \iunit$.
  We claim that $\var\in\dom{\penvtwo}$.
  Indeed, if $\var\notin\dom{\penvtwo}$, then $\subs{\var}{\penvtwo} = \var$,
  and we have that
  $\verifteig{(\var\,\subs{\utmtwo_1}{\penvtwo}\hdots\subs{\utmtwo_m}{\penvtwo})}
   \tos \iunit$,
  which is impossible,
  because the reducts of a \metaterm headed by a variable $\var$
  are always headed by $\var$,
  so in particular $\var\,\subs{\utmtwo_1}{\penvtwo}\hdots\subs{\utmtwo_m}{\penvtwo}$
  cannot have $\genteig$ as a reduct.

  Then $\var \in \dom{\penvtwo}$. Suppose then that $\var:\ttmtwo\in\penvtwo$.
  Without loss of generality, $\ttmtwo$ is of the form
  $\ttmtwo = (\ttmtwo_1\arrow\hdots\arrow\ttmtwo_k\arrow\teigtwo)$.
  For an arbitrary \metaterm $\utm$,
  and assuming $\var \notin \fv{\utm}$,
  let us write $\lamt{\ttm}{\utm}$ to abbreviate $\lam{\var}{(\eunit{\verif{\ttm}{\var}}{\utm})}$.
  Note that:
  \[
    \begin{array}{rlll}
    &&
      \verifteig{((\lamt{\ttmtwo_1}{\hdots\lamt{\ttmtwo_k}{\gen{\teigtwo}}})\,\subs{\utmtwo_1}{\penvtwo}\hdots\subs{\utmtwo_m}{\penvtwo})}
    \\
    & = &
      \verifteig{(\gen{\ttmtwo}\,\subs{\utmtwo_1}{\penvtwo}\hdots\subs{\utmtwo_m}{\penvtwo})}
    \\
    & = &
      \verifteig{(\subs{\var}{\penvtwo}\,\subs{\utmtwo_1}{\penvtwo}\hdots\subs{\utmtwo_m}{\penvtwo})}
    \\
    & \tos &
      \iunit
    \end{array}
  \]
  We first claim that $k = m$.
  Let us discard the cases $k < m$ and $k > m$:
  
  The case $k < m$ is impossible, because if $k < m$ we obtain a contradiction.
    We claim that
    $\verifteig{((\lamt{\ttmtwo_1}{\hdots\lamt{\ttmtwo_k}{\gen{\teigtwo}}})\,\subs{\utmtwo_1}{\penvtwo}\hdots\subs{\utmtwo_m}{\penvtwo})}$
    cannot reduce to $\iunit$.
    To show this, we consider two cases:
    \begin{enumerate}
    \item
      \label{thm:completeness:case_k_lt_m_stuck}
      If for every index $1 \leq i \leq m$ we have that
      $\verif{\ttmtwo_i}{\subs{\utmtwo_i}{\penvtwo}} \tos \iunit$,
      then:
      $
        \verifteig{((\lamt{\ttmtwo_1}{\hdots\lamt{\ttmtwo_k}{\gen{\teigtwo}}})\,\subs{\utmtwo_1}{\penvtwo}\hdots\subs{\utmtwo_m}{\penvtwo})}
        \tos
        \verifteig{(\gen{\teigtwo}\,\subs{\utmtwo_{k+1}}{\penvtwo}\hdots\subs{\utmtwo_m}{\penvtwo})}
      $.
      The reducts of $\gen{\teigtwo}\,\subs{\utmtwo_{k+1}}{\penvtwo}\hdots\subs{\utmtwo_m}{\penvtwo}$
      always have $\gen{\teigtwo}$ at the head and exactly $m - k > 0$ arguments,
      so the argument of $\verifteig{\ARG}$ does not reduce to $\genteig$,
      and the \metaterm does not reduce to $\iunit$.
    \item
      Otherwise, let $1 \leq i \leq m$ be the least index such that
      $\verif{\ttmtwo_i}{\subs{\utmtwo_i}{\penvtwo}}$ does not reduce to $\iunit$.
      Then:
      \[
        \begin{array}{ll}
        &
        \verifteig{((\lamt{\ttmtwo_1}{\hdots\lamt{\ttmtwo_k}{\gen{\teigtwo}}})\,\subs{\utmtwo_1}{\penvtwo}\hdots\subs{\utmtwo_m}{\penvtwo})}
        \\
        \tos &
        \verifteig{((\eunit{\verif{\ttmtwo_i}{\subs{\utmtwo_i}{\penvtwo}}}{(\lamt{\ttmtwo_{i+1}}{\hdots\lamt{\ttmtwo_k}{\gen{\teigtwo}}})})\,\subs{\utmtwo_{i+1}}{\penvtwo}\hdots\subs{\utmtwo_m}{\penvtwo})}
        \end{array}
      \]
      Since $\verif{\ttmtwo_i}{\subs{\utmtwo_i}{\penvtwo}}$ does not reduce to $\iunit$,
      the unit eliminator
      $(\eunit{\verif{\ttmtwo_i}{\subs{\utmtwo_i}{\penvtwo}}}{\ARG})$
      does not reduce to a redex,
      so the argument of $\verifteig{\ARG}$ does not reduce to $\genteig$,
      and the \metaterm does not reduce to $\iunit$.
    \end{enumerate}

  The case $k > m$ is impossible, because if $k > m$ we obtain a contradiction.
    We claim that
    $\verifteig{((\lamt{\ttmtwo_1}{\hdots\lamt{\ttmtwo_k}{\gen{\teigtwo}}})\,\subs{\utmtwo_1}{\penvtwo}\hdots\subs{\utmtwo_m}{\penvtwo})}$
    cannot reduce to $\iunit$.
    To show this, note first that
    if there is an index $1 \leq i \leq k$ such that
    $\verif{\ttmtwo_i}{\subs{\utmtwo_i}{\penvtwo}}$ does not reduce to $\iunit$,
    the situation is similar as for \cref{thm:completeness:case_k_lt_m_stuck} above.
    If for every index $1 \leq i \leq k$ we have that
    $\verif{\ttmtwo_i}{\subs{\utmtwo_i}{\penvtwo}} \tos \iunit$,
    then:
    \[
      \verifteig{((\lamt{\ttmtwo_1}{\hdots\lamt{\ttmtwo_k}{\gen{\teigtwo}}})\,\subs{\utmtwo_1}{\penvtwo}\hdots\subs{\utmtwo_m}{\penvtwo})}
      \tos
      \verifteig{(\lamt{\ttmtwo_{m+1}}{\hdots\lamt{\ttmtwo_k}{\gen{\teigtwo}}})}
    \]
    Note that the reducts of $\lamt{\ttmtwo_{m+1}}{\ARG}$ always start with an abstraction,
    so the argument of $\verifteig{\ARG}$ does not reduce to $\genteig$,
    and the \metaterm does not reduce to $\iunit$.

  Hence $k = m$.
  If there is an index $1 \leq i \leq k$ such that
  $\verif{\ttmtwo_i}{\subs{\utmtwo_i}{\penvtwo}}$ does not reduce to $\iunit$,
  the situation is similar as for \cref{thm:completeness:case_k_lt_m_stuck} above.
  Hence $\verif{\ttmtwo_i}{\subs{\utmtwo_i}{\penvtwo}} \tos \iunit$
  for every index $1 \leq i \leq k$, and we have that:
  \[
    \verifteig{((\lamt{\ttmtwo_1}{\hdots\lamt{\ttmtwo_k}{\gen{\teigtwo}}})\,\subs{\utmtwo_1}{\penvtwo}\hdots\subs{\utmtwo_m}{\penvtwo})}
    \tos
    \verifteig{\gen{\teigtwo}}
  \]
  so $\verifteig{\gen{\teigtwo}} \tos \iunit$,
  which means that $\teig = \teigtwo$.

  Recall that $\utmtwo^\downarrow$ is a pure \metaterm in normal form,
  and that $\utmtwo^\downarrow = \var\,\utmtwo_1\hdots\utmtwo_k$,
  so $\utmtwo_1,\hdots,\utmtwo_k$ are also pure \metaterms in normal form.
  Note also that:
  \[
    \sig{\utm^\downarrow}{\ttm}
    \geq
    \sig[\penvtwo]{\utmtwo^\downarrow}{\teig}
    >
    \sig[\penvtwo]{\utmtwo_i}{\ttmtwo_i}
  \]
  holds for all $1 \leq i \leq k$, by resorting to \cref{lemma:sizeVar_argument}.
  Since, for each $1 \leq i \leq k$,
  we know that $\verif{\ttmtwo_i}{\subs{\utmtwo_i}{\penvtwo}} \tos \iunit$,
  by \ih we obtain that
  $\judg{\penvtwo}{\utmtwo^\downarrow_i}{\ttmtwo_i}$.
  Since $\var:(\ttmtwo_1\arrow\hdots\arrow\ttmtwo_k\arrow\teig) \in \penv \subseteq \penvtwo$,
  we can derive the judgment
  $\judg{\penvtwo}{\var\,\utmtwo^\downarrow_1\hdots\utmtwo^\downarrow_k}{\teig}$.
  Take $\tm := \lam{\var_1\hdots\var_n}{\var\,\utmtwo^\downarrow_1\hdots\utmtwo^\downarrow_k}$.
  Note that:
  \[
    \begin{array}{llll}
      \tm
    & = &
      \lam{\var_1\hdots\var_n}{\var\,\utmtwo^\downarrow_1\hdots\utmtwo^\downarrow_k}
    \\
    & \eqbeta &
      \lam{\var_1\hdots\var_n}{\var\,\utmtwo_1\hdots\utmtwo_k}
    \\
    & = &
      \lam{\var_1\hdots\var_n}{\utmtwo^\downarrow}
    \\
    & \eqbeta &
       \lam{\var_1\hdots\var_n}{\utm^\downarrow\,\var_1\hdots\var_n}
      & \text{as $\utm^\downarrow\,\var_1\hdots\var_n \tobetas \utmtwo^\downarrow$}
    \\
    & \toetas &
       \utm^\downarrow
       & \text{using $n$ steps of $\eta$-reduction}
    \end{array}
  \]
  Thus $\tm \eqbeta\toetas \utm^\downarrow \tobetainvs \utm$.
  By confluence of $\tobeta$~\cite[Theorem~3.2.8]{Barendregt1984}, we have
  $\tm \tobetas \tobetainvs \toetas \utm^\downarrow \tobetainvs \utm$.
  Moreover, since $\tobeta$ and $\toeta$ commute~\cite[Lemma~3.3.8]{Barendregt1984}
  and $\utm^\downarrow$ is in $\beta$-normal form,
  the situation is $\tm \tobetas \toetas \utm^\downarrow \tobetainvs \utm$.

  Finally, since $\penvtwo = (\penv,\var_1:\ttm_1,\hdots,\var_n:\ttm_n)$,
  we have that
  $\judg{\penv}{\lam{\var_1\hdots\var_n}{\var\,\utmtwo^\downarrow_1\hdots\utmtwo^\downarrow_k}}{\ttm_1\arrow\hdots\ttm_n\arrow\teig}$,
  that is 
  $\judg{\penv}{\tm}{\ttm}$.
  Since $\tm \tobetas\toetas \utm^\downarrow$,
  by subject reduction~\cite[Proposition~1B.6]{Barendregt2013}
  we conclude that $\judg{\penv}{\utm^\downarrow}{\ttm}$.
\end{proof}

\section{Second-Order Logic --- Properties of $\lambdaCheckF$}

\subsection{Second-Order: Subcommutativity of Weak Head Reduction (\cref{prop:f:subcommutative})}
  \label{a:sec:f:subcommutativity}
  
\begin{proposition}
\label{a:prop:f:subcommutative}
Weak head reduction $\tow$ is subcommutative.
\end{proposition}
\begin{proof}
Suppose that $\utm_1 \tow \utm_2$ and $\utm_1 \tow \utm_3$,
and let us show that there exists a \metaterm $\utm_4$
such that $\utm_2 \tow^= \utm_4$ and $\utm_3 \tow^= \utm_4$.
We proceed by induction on $\utm_1$.

Note that the cases in which $\utm_1$ is a \textbf{\proofVariable} ($\utm_1 = \pvar$),
a \textbf{\proofAbstraction} ($\utm_1 = \plamf{\pvar}{\utm'_1}$),
a \textbf{\logicalAbstraction} ($\utm_1 = \pallfi{\tvar}{\utm'_1}$),
and a \textbf{success} ($\utm_1 = \iunit$)
are impossible, because \metaterms of these shapes
are $\tow$-irreducible.

If $\utm_1$ is a \textbf{\proofApplication}, $\utm_1 = \utmtwo_1\,\utmthree_1$:
then there are three subcases, depending on whether the
steps $\utm_1 \tow \utm_2$ and $\utm_1 \tow \utm_3$
are both at the root, both internal to $\utmtwo_1$,
or one at the root and one internal:
\begin{enumerate}
\item
  If both steps are at the root:
  note that there is only one rewriting rule in $\lambdaCheckF$ that involves
  an application at the root.
  Hence $\utmtwo_1 = \plamf{\pvar}{\utmtwo'_1}$
  and the two steps are of the form
  $(\plamf{\pvar}{\utmtwo'_1})\,\utmthree_1
   \tow
   \utmtwo'_1\sub{\pvar}{\utmthree_1}$,
  so $\utm_2 = \utm_3 = \utmtwo'_1\sub{\pvar}{\utmthree_1}$
  and we are done.
\item
  If both steps are internal to $\utmtwo_1$:
  then the step $\utm_1 \tow \utm_2$
  is of the form $\utmtwo_1\,\utmthree_1 \tow \utmtwo_2\,\utmthree_1$
  with $\utmtwo_1 \tow \utmtwo_2$.
  Similarly, the step $\utm_1 \tow \utm_2$
  is of the form $\utmtwo_1\,\utmthree_1 \tow \utmtwo_3\,\utmthree_1$
  with $\utmtwo_1 \tow \utmtwo_3$.
  By \ih, there exists a term $\utmtwo_4$ such that
  $\utmtwo_2 \tow^= \utmtwo_4$ and $\utmtwo_3 \tow^= \utmtwo_4$.
  Taking $\utm_4 := \utmtwo_4\,\utmthree_1$
  we have that
  $\utm_2 = \utmtwo_2\,\utmthree_1 \tow \utmtwo_4\,\utmthree_1 = \utm_4$
  and
  $\utm_3 = \utmtwo_3\,\utmthree_1 \tow \utmtwo_3\,\utmthree_1 = \utm_4$.
\item
  If one step is at the rule, and one internal to $\utmtwo_1$:
  we claim that this case is impossible.
  Without loss of generality,
  suppose that the first step $\utm_1 \tow \utm_2$ is at the root
  and the second step $\utm_1 \tow \utm_3$ is internal to $\utmtwo_1$.
  Then $\utmtwo_1 = \lam{\var}{\utmtwo'_1}$
  and the first step is of the form
  $\utm_1 = (\plamf{\pvar}{\utmtwo'_1})\,\utmthree_1
          \tow \utmtwo'_1\sub{\pvar}{\utmthree_1} = \utm_2$.
  Then the second step is of the form
  $\utm_1 = (\plamf{\pvar}{\utmtwo'_1})\,\utmthree_1
          \tow \utmtwo_2\,\utmthree_1$
  with $\plamf{\pvar}{\utmtwo'_1} \tow \utmtwo_2$,
  but there are no rules that allow weak head reduction $\tow$
  on a $\lambda$-abstraction.
\end{enumerate}

If $\utm_1$ is a \textbf{\logicalApplication} ($\utm_1 = \utmtwo_1\,\ttm$)
or a \textbf{guard} ($\utm_1 = \eunit{\utmtwo_1}{\utmtwo_2}$),
the proof is similar to that of the \proofAbstraction.

If $\utm_1$ is a \textbf{generator}, $\utm_1 = \gen{\ttm}$,
the proof is immediate because there is at most one way to reduce $\utm_1$.

If $\utm_1$ is a \textbf{verifier}, $\utm_1 = \verif{\ttm}{\utmtwo_1}$
there are four cases, depending on the shape of $\ttm$:
\begin{enumerate}
\item \LogicalVariable ($\ttm = \tvar$):
  Then the only way to reduce $\utm_1$ is with a reduction internal to
  $\utmtwo_1$ and we conclude resorting to the \ih.
\item \LogicalEigenvariable ($\ttm = \teig$):
  Then there are two cases.
  If $\utmtwo_1 = \gen{\teig}$,
  then $\utm_1$ reduces at the root,
  so $\utm_1 = \verif{\teig}{\gen{\teig}} \tow \iunit$
  and this is the only way to reduce $\utm_1$.
  Otherwise $\utmtwo_1 \neq \gen{\teig}$,
  so both steps $\utm_1 \to \utm_2$ and $\utm_1 \to \utm_3$
  must be internal to $\utmtwo_1$ and we conclude by
  resorting to the \ih.
\item Implication ($\ttm = (\ttmtwo \imp \ttmthree)$):
  If both steps are at the root or both are internal to $\utmtwo_1$,
  it is easy to conclude,
  similarly as in the case of the \proofAbstraction.
  Without loss of generality, suppose that
  $\utm_1 \tow \utm_2$ is at the root
  and
  $\utm_1 \tow \utm_3$ is internal to $\utmtwo_1$.
  Then the situation is the following,
  where $\utmtwo_1 \tow \utmtwo_2$:
  \[
    \xymatrix{
      \utm_1 = \verif{\ttm\imp\ttmtwo}{\utmtwo_1}
      \ar[r]
      \ar[d]
    &
      \verif{\ttmtwo}{(\utmtwo_1\,\gen{\ttm})} = \utm_2
      \ar@{.>}[d]
    \\
      \utm_3 = \verif{\ttm\imp\ttmtwo}{\utmtwo_2}
      \ar@{.>}[r]
    &
      \verif{\ttmtwo}{(\utmtwo_2\,\gen{\ttm})} = \utm_4
    }
  \]
\item Universal quantification ($\ttm = (\allf{\tvar}{\ttmtwo})$):
  If both steps are at the root or both are internal to $\utmtwo_1$,
  it is easy to conclude,
  similarly as in the case of the \proofAbstraction.
  Without loss of generality, suppose that
  $\utm_1 \tow \utm_2$ is at the root
  and
  $\utm_1 \tow \utm_3$ is internal to $\utmtwo_1$.
  Then the situation is the following,
  where $\utmtwo_1 \tow \utmtwo_2$:
  \[
    \xymatrix{
      \utm_1 = \verif{\allf{\tvar}{\ttmtwo}}{\utmtwo_1}
      \ar[r]
      \ar[d]
    &
      \freshf{\teig}{\verif{\ttmtwo\sub{\tvar}{\teig}}{(\utmtwo_1\,\teig)}}
      = \utm_2
      \ar@{.>}[d]
    \\
      \utm_3 = \verif{\allf{\tvar}{\ttmtwo}}{\utmtwo_2}
      \ar@{.>}[r]
    &
      \freshf{\teig}{\verif{\ttmtwo\sub{\tvar}{\teig}}{(\utmtwo_2\,\teig)}}
      = \utm_4
    }
  \]
\end{enumerate}

If $\utm_1$ is a \textbf{fresh \logicalEigenvariable introduction}
($\utm_1 = \freshf{\teig}{\utmtwo_1})$,
  then if both steps are at the root or both are internal to $\utmtwo_1$,
  it is easy to conclude, similarly as in the case of the \proofAbstraction.
  Without loss of generality, suppose that
  $\utm_1 \tow \utm_2$ is at the root
  and
  $\utm_1 \tow \utm_3$ is internal to $\utmtwo_1$.
  Then the situation is the following,
  where $\utmtwo_1 \tow \utmtwo_2$:
  \[
    \xymatrix{
      \utm_1 = \freshf{\teig}{\utmtwo_1}
      \ar[r]
      \ar[d]
    &
      \utmtwo_1 = \utm_2
      \ar@{.>}[d]
    \\
      \utm_3 = \freshf{\teig}{\utmtwo_2}
      \ar@{.>}[r]
    &
      \utmtwo_2 = \utm_4
    }
  \]
  Note that the step on the top of the diagram is such that
  $\teig \notin \fv{\utmtwo_1}$
  and reduction does not create free \logicalEigenvariables so
  the fact that $\teig \notin \fv{\utmtwo_2}$
  allows us to perform the step on the bottom of the diagram.
\end{proof}

\subsection{Second-Order: Confluence (\cref{prop:f:confluence})}
  \label{a:sec:f:confluence}
  
\begin{definition}[Simultaneous reduction]
The relation $\utm \sto \utmtwo$,
called \emph{simultaneous reduction},
is defined as follows between \metaterms
of $\lambdaCheckF$:
\[
  \indrule{var}{
  }{
    \pvar \sto \pvar
  }
  \HS
  \indrule{lam}{
    \utm \sto \utm'
  }{
    \plamf{\pvar}{\utm} \sto \plamf{\pvar}{\utm'}
  }
\]
\[
  \indrule{app1}{
    \utm \sto \utm'
    \HS
    \utmtwo \sto \utmtwo'
  }{
    \utm\,\utmtwo \sto \utm'\,\utmtwo'
  }
  \HS
  \indrule{app2}{
    \utm \sto \utm'
    \HS
    \utmtwo \sto \utmtwo'
  }{
    (\plamf{\pvar}{\utm})\,\utmtwo \sto \utm'\sub{\pvar}{\utmtwo'}
  }
\]
\[
  \indrule{tlam}{
    \utm \sto \utm'
  }{
    \pallfi{\tvar}{\utm} \sto \pallfi{\tvar}{\utm'}
  }
\]
\[
  \indrule{tapp1}{
    \utm \sto \utm'
  }{
    \utm\,\ttm \sto \utm'\,\ttm
  }
  \HS
  \indrule{tapp2}{
    \utm \sto \utm'
  }{
    (\pallfi{\tvar}{\utm})\,\ttm \sto \utm'\sub{\tvar}{\ttm}
  }
\]
\[
  \indrule{success}{
  }{
    \iunit \sto \iunit
  }
\]
\[
  \indrule{guard1}{
    \utm \sto \utm'
    \HS
    \utmtwo \sto \utmtwo'
  }{
    \eunit{\utm}{\utmtwo} \sto \eunit{\utm'}{\utmtwo'}
  }
  \HS
  \indrule{guard2}{
    \utm \sto \utm'
  }{
    \eunit{\iunit}{\utm} \sto \utm'
  }
\]
\[
  \indrule{g-refl}{
  }{
    \gen{\ttm} \sto \gen{\ttm}
  }
  \HS
  \indrule{g-imp}{
  }{
    \gen{\ttm\imp\ttmtwo} \sto \plamf{\var}{\eunit{\verif{\ttm}{\pvar}}{\gen{\ttmtwo}}}
  }
\]
\[
  \indrule{g-all}{
  }{
    \gen{\allf{\tvar}{\ttm}} \sto \pallfi{\tvar}{\gen{\ttm}}
  }
\]
\[
  \indrule{v-cong}{
    \utm \sto \utm'
  }{
    \verif{\ttm}{\utm} \sto \verif{\ttm}{\utm'}
  }
  \HS
  \indrule{v-eig}{
  }{
    \verif{\teig}{\gen{\teig}} \sto \iunit
  }
\]
\[
  \indrule{v-imp}{
    \utm \sto \utm'
  }{
    \verif{\ttm\imp\ttmtwo}{\utm} \sto \verif{\ttmtwo}{(\utm'\,\gen{\ttm})}
  }
\]
\[
  \indrule{v-all}{
    \utm \sto \utm'
    \HS
    \teig\notin\fv{\ttm,\utm}
  }{
    \verif{\allf{\tvar}{\ttm}}{\utm} \sto \freshf{\teig}{\verif{\ttm\sub{\tvar}{\teig}}{(\utm'\,\teig)}}
  }
\]
\[
  \indrule{fresh1}{
    \utm \sto \utm'
  }{
    \freshf{\teig}{\utm} \sto \freshf{\teig}{\utm'}
  }
  \HS
  \indrule{fresh2}{
    \utm \sto \utm'
    \HS
    \teig\notin\fv{\utm}
  }{
    \freshf{\teig}{\utm} \sto \utm'
  }
\]
\end{definition}

\begin{remark}
\quad
\begin{enumerate}
\item
  If $\plamf{\pvar}{\utm} \sto \utmtwo$
  then $\utmtwo$ is of the form
  $\utmtwo = \lam{\var}{\utm'}$ and $\utm \sto \utm'$.
\item
  If $\pallfi{\tvar}{\utm} \sto \utmtwo$
  then $\utmtwo$ is of the form
  $\utmtwo = \pallfi{\tvar}{\utm'}$ and $\utm \sto \utm'$.
\item
  If $\iunit \sto \utmtwo$ then $\utmtwo = \iunit$.
\item
  If $\gen{\teig} \sto \utmtwo$ then $\utmtwo = \gen{\teig}$.
\end{enumerate}
\end{remark}

\begin{lemma}[Reflexivity]
\label{lem:f:sto_refl}
$\utm \sto \utm$
\end{lemma}
\begin{proof}
Straightforward by induction on $\utm$.
\end{proof}

\begin{lemma}[Substitution]
\label{lem:f:sto_substitution}
If $\utm \sto \utm'$ and $\utmtwo \sto \utmtwo'$
then $\utm\sub{\var}{\utmtwo} \sto \utm'\sub{\var}{\utmtwo'}$.
\end{lemma}
\begin{proof}
Straightforward by induction on the derivation of $\utm \sto \utm'$,
resorting to reflexivity~(\cref{lem:f:sto_refl}) in the \indrulename{var} case.
\end{proof}

\begin{lemma}[Type substitution]
\label{lem:f:sto_type_substitution}
If $\utm \sto \utm'$
then $\utm\sub{\tvar}{\ttm} \sto \utm'\sub{\tvar}{\ttm}$.
\end{lemma}
\begin{proof}
Straightforward by induction on the derivation of $\utm \sto \utm'$.
Note that in the \indrulename{v-eig} case
we have that $\utm = \verif{\teig}{\gen{\teig}} \sto \iunit = \utm'$
where $\teig$ is a \logicalEigenvariable
so $\teig \neq \tvar$.
\end{proof}

\begin{lemma}[Tait--Martin-Löf technique]
\label{lem:f:tait_martin_lof}
\quad
\begin{enumerate}
\item ${\to} \subseteq {\sto}$
\item ${\sto} \subseteq {\tos}$
\item Diamond property: $\stoinv\,\sto \subseteq \sto\,\stoinv$
\end{enumerate}
\end{lemma}
\begin{proof}
For the first item, suppose that $\utm \to \utmtwo$.
It is straightforward to show that $\utm \sto \utmtwo$
by induction on $\utm$
and resorting to reflexivity~(\cref{lem:f:sto_refl}).
For the second item, suppose that $\utm \sto \utmtwo$.
It is straightforward to show that $\utm \tos \utmtwo$
by induction on the derivation of $\utm \sto \utmtwo$.

For the third item, suppose that
$\utm_1 \sto \utm_2$ and $\utm_1 \sto \utm_3$
and let us show that there exists a term $\utm_4$
such that $\utm_2 \sto \utm_4$ and $\utm_3 \sto \utm_4$.
We proceed by induction on $\utm_4$.

If $\utm_1$ is a \textbf{variable} ($\utm_1 = \pvar$)
or a \textbf{success} ($\utm_1 = \iunit$),
it is immediate because $\utm_1 = \utm_2 = \utm_3$ and it suffices
to take $\utm_4 := \utm_1$.

If $\utm_1$ is a \textbf{\proofAbstraction} ($\utm_1 = \plamf{\pvar}{\utmtwo_1}$)
or a \textbf{\logicalAbstraction} ($\utm_1 = \pallfi{\tvar}{\utmtwo_1}$),
because the reduction can only be internal to $\utmtwo_1$,
and we may conclude by resorting to the \ih.

If $\utm_1$ is a \textbf{\proofApplication} ($\utm_1 = \utmtwo_1\,\utmthree_1$)
there are four cases (two of them being symmetric),
depending on whether
both steps are derived using \indrulename{app1},
both steps are derived using \indrulename{app2},
or one is derived using \indrulename{app1} and the other one using \indrulename{app2}:
\begin{enumerate}
\item
  \indrulename{app1}/\indrulename{app1}:
  Straightforward resorting to the \ih.
\item
  \indrulename{app2}/\indrulename{app2}:
  Then $\utmtwo_1 = \plamf{\pvar}{\utmtwo'_1}$
  where $\utmtwo'_1 \sto \utmtwo'_2$ and $\utmtwo'_1 \sto \utmtwo'_3$
  and $\utmthree_1 \sto \utmthree_2$ and $\utmthree_1 \sto \utmthree_3$.
  By \ih there exists $\utmtwo'_4$
  such that $\utmtwo'_2 \sto \utmtwo'_4$ and $\utmtwo'_3 \sto \utmtwo'_4$
  and there exists $\utmthree_4$
  such that $\utmthree_2 \sto \utmthree_4$
  and $\utmthree_3 \sto \utmthree_4$.
  Taking $\utm_4 := \utmtwo'_4\sub{\pvar}{\utmthree'_4}$
  and applying \cref{lem:f:sto_substitution}, we have that:
  \[
    \begin{array}{ccc}
      (\plamf{\pvar}{\utmtwo'_1})\,\utmthree_1
      & \sto &
      \utmtwo'_2\sub{\pvar}{\utmthree_2}
    \\ \stovert & & \stovert \\
      \utmtwo'_3\sub{\pvar}{\utmthree_3}
      & \sto &
      \utmtwo'_4\sub{\pvar}{\utmthree_4}
    \end{array}
  \]
\item
  \indrulename{app1}/\indrulename{app2}:
  Then $\utmtwo_1 = \plamf{\pvar}{\utmtwo'_1}$
  where $\utmtwo'_1 \sto \utmtwo'_2$ and $\utmtwo'_1 \sto \utmtwo'_3$
  and $\utmthree_1 \sto \utmthree_2$ and $\utmthree_1 \sto \utmthree_3$.
  By \ih there exists $\utmtwo'_4$
  such that $\utmtwo'_2 \sto \utmtwo'_4$ and $\utmtwo'_3 \sto \utmtwo'_4$
  and there exists $\utmthree_4$
  such that $\utmthree_2 \sto \utmthree_4$
  and $\utmthree_3 \sto \utmthree_4$.
  Taking $\utm_4 := \utmtwo'_4\sub{\pvar}{\utmthree_4}$
  and applying \cref{lem:f:sto_substitution}, we have that:
  \[
    \begin{array}{ccc}
      (\plamf{\pvar}{\utmtwo'_1})\,\utmthree_1
      & \sto &
      (\plamf{\pvar}{\utmtwo'_2})\,\utmthree_2
    \\ \stovert & & \stovert \\
      \utmtwo'_3\sub{\pvar}{\utmthree_3}
      & \sto &
      \utmtwo'_4\sub{\pvar}{\utmthree_4}
    \end{array}
  \]
\end{enumerate}

If $\utm_1$ is a \textbf{\logicalApplication} ($\utm_1 = \utmtwo_1\,\ttm$)
there are four cases (two of them being symmetric),
depending on whether
both steps are derived using \indrulename{tapp1},
both steps are derived using \indrulename{tapp2},
or one is derived using \indrulename{tapp1} and the other one using \indrulename{tapp2}:
\begin{enumerate}
\item
  \indrulename{tapp1}/\indrulename{tapp1}:
  Straightforward resorting to the \ih.
\item
  \indrulename{tapp2}/\indrulename{tapp2}:
  Then $\utmtwo_1 = \pallfi{\tvar}{\utmtwo'_1}$
  where $\utmtwo'_1 \sto \utmtwo'_2$ and $\utmtwo'_1 \sto \utmtwo'_3$
  By \ih there exists $\utmtwo'_4$
  such that $\utmtwo'_2 \sto \utmtwo'_4$
  and $\utmtwo'_3 \sto \utmtwo'_4$.
  Taking $\utm_4 := \utmtwo'_4\sub{\tvar}{\ttm}$
  and applying \cref{lem:f:sto_type_substitution} we have that:
  \[
    \begin{array}{ccc}
      (\pallfi{\tvar}{\utmtwo'_1})\,\ttm
      & \sto &
      \utmtwo'_2\sub{\tvar}{\ttm}
    \\ \stovert & & \stovert \\
      \utmtwo'_3\sub{\tvar}{\ttm}
      & \sto &
      \utmtwo'_4\sub{\tvar}{\ttm}
    \end{array}
  \]
\item
  \indrulename{tapp1}/\indrulename{tapp2}:
  Then $\utmtwo_1 = \pallfi{\tvar}{\utmtwo'_1}$
  where $\utmtwo'_1 \sto \utmtwo'_2$ and $\utmtwo'_1 \sto \utmtwo'_3$.
  By \ih there exists $\utmtwo'_4$
  such that $\utmtwo'_2 \sto \utmtwo'_4$
  and $\utmtwo'_3 \sto \utmtwo'_4$.
  Taking $\utm_4 := \utmtwo'_4\sub{\tvar}{\ttm}$
  and applying \cref{lem:f:sto_type_substitution}, we have that:
  \[
    \begin{array}{ccc}
      (\pallfi{\tvar}{\utmtwo'_1})\,\ttm
      & \sto &
      (\pallfi{\tvar}{\utmtwo'_2})\,\ttm
    \\ \stovert & & \stovert \\
      \utmtwo'_3\sub{\tvar}{\ttm}
      & \sto &
      \utmtwo'_4\sub{\tvar}{\ttm}
    \end{array}
  \]
\end{enumerate}

If $\utm_1$ is a \textbf{guard} ($\utm_1 = \eunit{\utmtwo_1}{\utmthree_1}$)
there are four cases (two of them being symmetric),
depending on whether
both steps are derived using \indrulename{guard1},
both steps are derived using \indrulename{guard2},
or one is derived using \indrulename{guard1} and the other one using \indrulename{guard2}:
\begin{enumerate}
\item
  \indrulename{guard1}/\indrulename{guard1}:
  Straightforward resorting to the \ih.
\item
  \indrulename{guard2}/\indrulename{guard2}:
  Then $\utmtwo_1 = \iunit$
  where $\utmtwo_2 = \utmtwo_3 = \iunit$
  and $\utmthree_1 \sto \utmthree_2$ and $\utmthree_1 \sto \utmthree_3$.
  By \ih there exists $\utmthree_4$
  such that $\utmthree_2 \sto \utmthree_4$
  and $\utmthree_3 \sto \utmthree_4$.
  Taking $\utm_4 := \utmthree_4$ we have that:
  \[
    \begin{array}{ccc}
      \eunit{\iunit}{\utmthree_1}
      & \sto &
      \utmthree_2
    \\ \stovert & & \stovert \\
      \utmthree_3
      & \sto &
      \utmthree_4
    \end{array}
  \]
\item
  \indrulename{guard1}/\indrulename{guard2}:
  Then $\utmtwo_1 = \iunit$
  where $\utmtwo_2 = \utmtwo_3 = \iunit$
  and $\utmthree_1 \sto \utmthree_2$ and $\utmthree_1 \sto \utmthree_3$.
  By \ih there exists $\utmthree_4$
  such that $\utmthree_2 \sto \utmthree_4$
  and $\utmthree_3 \sto \utmthree_4$.
  Taking $\utm_4 := \utmthree_4$, we have that:
  \[
    \begin{array}{ccc}
      \eunit{\iunit}{\utmthree_1}
      & \sto &
      \eunit{\iunit}{\utmthree_2}
    \\ \stovert & & \stovert \\
      \utmthree_3
      & \sto &
      \utmthree_4
    \end{array}
  \]
\end{enumerate}

If $\utm_1$ is a \textbf{generator} ($\utm_1 = \gen{\ttm}$),
note first that if both steps are derived using the same rule,
it is immediate to conclude
because $\utm_2 = \utm_3$ and it suffices to take $\utm_4 := \utm_2$
and resort to reflexivity~(\cref{lem:f:sto_refl}).
Moreover,
note that it is not possible to derive one step using the \indrulename{g-imp} rule
and the other one using the \indrulename{g-all} rule,
because we would simultaneously have $\ttm = (\ttmtwo\imp\ttmthree)$
and $\ttm = (\allf{\tvar}{\ttmfour})$, which is impossible.
The remaining cases are when
one step is derived using the \indrulename{g-refl} rule
and the other one is derived using either the \indrulename{g-imp}
or the \indrulename{g-all} rule. These cases are immediate
because if $\utm_1 \sto \utm_2$ is derived using the \indrulename{g-refl}
rule then $\utm_2 = \utm_1$, and it suffices to take $\utm_4 := \utm_3$.

If $\utm_1$ is a \textbf{verifier} ($\utm_1 = \verif{\ttm}{\utmtwo_1}$),
note first that if both steps are derived using the same rule
it is immediate to conclude by resorting to the \ih. 
Moreover, note that the rules \indrulename{v-eig}, \indrulename{v-imp}, 
and \indrulename{v-all} are mutually exclusive,
because $\ttm$ cannot be simultaneously of two shapes
(\eg an \logicalEigenvariable and an implication).
The remaining cases are when
one step is derived using the \indrulename{v-cong} rule
and the other one is derived using either the \indrulename{v-eig},
the \indrulename{v-imp}, or the \indrulename{v-all} rule:
\begin{enumerate}
\item
  \indrulename{v-cong}/\indrulename{v-eig}:
  Then $\utm_1 = \verif{\teig}{\gen{\teig}}$.
  Suppose that the step $\utm_1 \sto \utm_2$
  is derived from the \indrulename{v-cong} rule.
  Then necessarily $\utm_2 = \verif{\teig}{\gen{\teig}}$,
  and the situation is:
  \[
    \begin{array}{ccc}
      \verif{\teig}{\gen{\teig}}
      & \sto &
      \verif{\teig}{\gen{\teig}}
    \\ \stovert & & \stovert \\
      \iunit
      & \sto &
      \iunit
    \end{array}
  \]
\item
  \indrulename{v-cong}/\indrulename{v-imp}:
  Then $\ttm = (\ttmtwo\imp\ttmthree)$.
  Suppose that the step $\utm_1 \sto \utm_2$
  is derived from the \indrulename{v-cong} rule
  and $\utm_1 \sto \utm_2$
  is derived from the \indrulename{v-imp} rule.
  Then $\utm_2 = \verif{\ttmtwo\imp\ttmthree}{\utmtwo_2}$
  where $\utmtwo_1 \sto \utmtwo_2$,
  and $\utm_3 = \verif{\ttmthree}{(\utmtwo_3\,\gen{\ttmtwo})}$
  where $\utmtwo_1 \sto \utmtwo_3$.
  By \ih there exists $\utmtwo_4$
  such that $\utmtwo_2 \sto \utmtwo_4$
  and $\utmtwo_3 \sto \utmtwo_4$,
  and the situation is:
  \[
    \begin{array}{ccc}
      \verif{\ttmtwo\imp\ttmthree}{\utmtwo_1}
      & \sto &
      \verif{\ttmtwo\imp\ttmthree}{\utmtwo_2}
    \\ \stovert & & \stovert \\
      \verif{\ttmthree}{(\utmtwo_3\,\gen{\ttmtwo})}
      & \sto &
      \verif{\ttmthree}{(\utmtwo_4\,\gen{\ttmtwo})}
    \end{array}
  \]
\item
  \indrulename{v-cong}/\indrulename{v-all}:
  Then $\ttm = (\allf{\tvar}{\ttmtwo})$.
  Suppose that the step $\utm_1 \sto \utm_2$
  is derived from the \indrulename{v-cong} rule
  and $\utm_1 \sto \utm_2$
  is derived from the \indrulename{v-all} rule.
  Then $\utm_2 = \verif{\allf{\tvar}{\ttmtwo}}{\utmtwo_2}$
  where $\utmtwo_1 \sto \utmtwo_2$,
  and $\utm_3 = \freshf{\teig}{\verif{\ttmtwo\sub{\tvar}{\teig}}{(\utmtwo_3\,\teig)}}$
  where $\utmtwo_1 \sto \utmtwo_3$.
  By \ih there exists $\utmtwo_4$
  such that $\utmtwo_2 \sto \utmtwo_4$
  and $\utmtwo_3 \sto \utmtwo_4$,
  and the situation is:
  \[
    \begin{array}{ccc}
      \verif{\allf{\tvar}{\ttmtwo}}{\utmtwo_1}
      & \sto &
      \verif{\allf{\tvar}{\ttmtwo}}{\utmtwo_2}
    \\ \stovert & & \stovert \\
      \freshf{\teig}{\verif{\ttmtwo\sub{\tvar}{\teig}}{(\utmtwo_3\,\teig)}}
      & \sto &
      \freshf{\teig}{\verif{\ttmtwo\sub{\tvar}{\teig}}{(\utmtwo_4\,\teig)}}
    \end{array}
  \]
\end{enumerate}

If $\utm_1$ is a \textbf{fresh \logicalEigenvariable introduction} ($\utm_1 = \freshf{\teig}{\utmtwo_1}$)
there are four cases (two of them being symmetric),
depending on whether
both steps are derived using \indrulename{fresh1},
both steps are derived using \indrulename{fresh2},
or one is derived using \indrulename{fresh1} and the other one using \indrulename{fresh2}:
\begin{enumerate}
\item
  \indrulename{fresh1}/\indrulename{fresh1}:
  Straightforward resorting to the \ih.
\item
  \indrulename{fresh2}/\indrulename{fresh2}:
  Straightforward resorting to the \ih.
\item
  \indrulename{fresh1}/\indrulename{fresh2}:
  Then $\utmtwo_1 \sto \utmtwo_2$
  and $\utmtwo_1 \sto \utmtwo_3$
  and $\teig \notin \fv{\utmtwo_1}$.
  By \ih there exists $\utmtwo_4$
  such that $\utmtwo_2 \sto \utmtwo_4$
  and $\utmtwo_3 \sto \utmtwo_4$,
  and the situation is:
  \[
    \begin{array}{ccc}
      \freshf{\teig}{\utmtwo_1}
      & \sto &
      \freshf{\teig}{\utmtwo_2}
    \\ \stovert & & \stovert \\
      \utmtwo_3
      & \sto &
      \utmtwo_4
    \end{array}
  \]
  To be able to perform the reduction on the right,
  observe that $\teig \notin \fv{\utmtwo_2}$
  because $\teig \notin \fv{\utmtwo_1}$ and $\utmtwo_1 \sto \utmtwo_2$,
  and reduction does not create free \logicalEigenvariables.
\end{enumerate}
\end{proof}

\begin{proposition}
\label{a:prop:f:confluence}
Reduction in $\lambdaCheckF$ is confluent.
\end{proposition}
\begin{proof}
An immediate consequence of \cref{lem:f:tait_martin_lof}
by abstract results \cite[Prop.~1.1.11]{Terese}.
\end{proof}

\subsection{Second-Order: Standardization (\cref{prop:f:standardization})}
  \label{a:sec:f:standardization}
  
We use a non-idempotent intersection type system $\LambdaFInters$ to prove
standardization~(\cref{prop:f:standardization}).
\medskip

The sets of \defn{\linearTypes} and \defn{\linearMultitypes} of $\LambdaFInters$
are given mutually inductively by the grammar:
\[
  \begin{array}{lrcl}
    \textsc{\LinearTypes} &
    \ltyp,\ltyptwo,\hdots
        & ::=  & \iunit
    \\
       && \mid & \gen{\tvar}
    \\
       && \mid & \mtyp\limp\ltyp
    \\
       && \mid & \ttm\allimp\ltyp
  \\
    \textsc{\LinearMultitypes} &
    \mtyp,\mtyptwo,\hdots
      & ::= & \mset{\ltyp_i}_\iI
  \end{array}
\]
\LinearMultitypes are finite multisets of \linearTypes.
The notation $\mset{\ltyp_i}_\iI$ stands for a multiset of elements
indexed by a finite set $I$, that sometimes we may also write
more explicitly as $\mset{\ltyp_1,\hdots,\ltyp_n}$.
The \defn{size} of a \linearMultitype $\mtyp$ is written $\size{\mtyp}$
and defined as its cardinality as a multiset (\ie number of elements,
counting their multiplicities).
We write $\mtyp+\mtyptwo$ for the union of \linearMultitypes.
A \defn{\linearEnvironment} is a function $\genv$ mapping each
\proofVariable to a \linearMultitype. We assume that \linearEnvironments
have \emph{finite support}, \ie that $\genv(\pvar) \neq \mset{\,}$
for only finitely many \proofVariables.
If $\genv,\genvtwo$ are \linearEnvironments, we write $\genv + \genvtwo$
for the \linearEnvironment such that
$(\genv + \genvtwo)(\pvar) = \genv(\pvar)  + \genvtwo(\pvar)$
for each \proofVariable~$\pvar$.

We define two forms of typing judgments, $\jul{\genv}{\utm}{\ltyp}$ 
and $\jum{\genv}{\utm}{\mtyp}$.
Valid judgments in system $\LambdaFInters$ are defined inductively by
the following rules:
\[
  \indrule{\ruleLVar}{
  }{
    \jul{\pvar:\mset{\ltyp}}{\pvar}{\ltyp}
  }
  \HS
  \indrule{\ruleLImpI}{
    \jul{\genv,\pvar:\mtyp}{\utm}{\ltyp}
  }{
    \jul{\genv}{\plamf{\pvar}{\utm}}{\mtyp\limp\ltyp}
  }
\]
\[
  \indrule{\ruleLImpE}{
    \jul{\genv_1}{\utm}{\mtyp\limp\ltyp}
    \HS
    \jum{\genv_2}{\utmtwo}{\mtyp}
  }{
    \jul{\genv_1+\genv_2}{\utm\,\utmtwo}{\ltyp}
  }
\]
\[
  \indrule{\ruleLAllI}{
    \jul{\genv}{\utm\sub{\tvar}{\ttm}}{\ltyp}
  }{
    \jul{\genv}{\pallfi{\tvar}{\utm}}{\ttm\allimp\ltyp}
  }
  \HS
  \indrule{\ruleLAllE}{
    \jul{\genv}{\utm}{\ttm\allimp\ltyp}
  }{
    \jul{\genv}{\utm\,\ttm}{\ltyp}
  }
\]
\[
  \indrule{\ruleLUnitI}{
  }{
    \jul{}{\iunit}{\iunit}
  }
  \HS
  \indrule{\ruleLUnitE}{
    \jul{\genv_1}{\utm}{\iunit}
    \HS
    \jul{\genv_2}{\utmtwo}{\ltyp}
  }{
    \jul{\genv_1+\genv_2}{\eunit{\utm}{\utmtwo}}{\ltyp}
  }
\]
\[
  \indrule{\ruleLFresh}{
    \jul{\genv}{\utm}{\ltyp}
    \HS
    \teig\notin\fv{\genv,\ltyp}
  }{
    \jul{\genv}{\freshf{\teig}{\utm}}{\ltyp}
  }
  \HS
  \indrule{\ruleLVerifTeig}{
    \jul{\genv}{\utm}{\gen{\teig}}
  }{
    \jul{\genv}{\verif{\teig}{\utm}}{\iunit}
  }
\]
\[
  \indrule{\ruleLVerifImp}{
    \jul{\genv}{\verif{\ttmtwo}{(\utm\,\gen{\ttm})}}{\iunit}
  }{
    \jul{\genv}{\verif{\ttm\imp\ttmtwo}{\utm}}{\iunit}
  }
\]
\[
  \indrule{\ruleLVerifAll}{
    \jul{\genv}{\freshf{\teig}{\verif{\ttm\sub{\tvar}{\teig}}{(\utm\,\teig)}}}{\iunit}
    \HS
    \teig\notin\fv{\ttm,\utm}
  }{
    \jul{\genv}{\verif{\allf{\tvar}{\ttm}}{\utm}}{\iunit}
  }
\]
\[
  \indrule{\ruleLGenTeig}{
  }{
    \jul{}{\gen{\teig}}{\gen{\teig}}
  }
  \HS
  \indrule{\ruleLGenImp}{
    \jul{\genv}{\plamf{\pvar}{\eunit{\verif{\ttm}{\pvar}}{\gen{\ttmtwo}}}}{\ltyp}
  }{
    \jul{\genv}{\gen{\ttm\imp\ttmtwo}}{\ltyp}
  }
\]
\[
  \indrule{\ruleLGenAll}{
    \jul{\genv}{\pallfi{\tvar}{\gen{\ttm}}}{\ltyp}
  }{
    \jul{\genv}{\gen{\allf{\tvar}{\ttm}}}{\ltyp}
  }
  \HS
  \indrule{\ruleLMulti}{
    (\jul{\genv_i}{\utm}{\ltyp_i})_\iI
  }{
    \jum{+_\iI\genv_i}{\utm}{\mset{\ltyp_i}_\iI}
  }
\]
Derivation trees (or just \emph{derivations})
in this system are ranged over by $\deriv,\derivtwo,\hdots$.
We write $\derivs{\deriv}{\jul{\genv}{\utm}{\ltyp}}$
to mean that $\deriv$ is a derivation tree with conclusion $\jul{\genv}{\utm}{\ltyp}$.
The \defn{size} of a derivation $\deriv$ is written $\size{\deriv}$
and defined as the number of inference rules in $\deriv$,
without counting instances of the $\ruleLMulti$ rule.

\begin{lemma}[Weighted substitution lemma]
\label{lem:lin:weighted_substitution_lemma}
If $\derivs{\deriv}{\jul{\genv,\pvar:\mtyp}{\utm}{\ltyp}}$
and $\derivs{\derivtwo}{\jum{\genvtwo}{\utmtwo}{\mtyp}}$
where $\pvar\notin\fv{\utmtwo}$,
then there exists a derivation $\deriv'$
such that $\derivs{\deriv'}{\jul{\genv+\genvtwo}{\utm\sub{\pvar}{\utmtwo}}{\ltyp}}$
and $\size{\deriv'} = \size{\deriv} - \size{\mtyp} + \size{\derivtwo}$.
\end{lemma}
\begin{proof}
By induction on $\deriv$:

\Case{$\ruleLVar$:}
  Let $\derivs{\deriv}{\jul{\pvartwo:\mset{\ltyp}}{\pvartwo}{\ltyp}}$.
  We consider two subcases, depending on whether $\pvartwo = \pvar$ or not:
  \begin{enumerate}
  \item
    If $\utm = \pvartwo = \pvar$,
    the derivation $\deriv$
    is of the form $\derivs{\deriv}{\jul{\pvar:\mset{\ltyp}}{\pvar}{\ltyp}}$,
    so $\genv = \emptyset$
    and $\mtyp = \mset{\ltyp}$.
    Moreover, the derivation $\derivtwo$ is of the form
    $\derivs{\derivtwo}{\jum{\genvtwo}{\utmtwo}{\mset{\ltyp}}}$
    and it must be derived by $\ruleLMulti$ from a single premise
    $\derivs{\derivtwo'}{\jul{\genvtwo}{\utmtwo}{\ltyp}}$.
    Our goal is to construct a derivation with conclusion
    $\jul{\genv+\genvtwo}{\utm\sub{\pvar}{\utmtwo}}{\ltyp}$,
    that is to say, $\jul{\genvtwo}{\utmtwo}{\ltyp}$.
    It suffices to take $\deriv' := \derivtwo'$.
    Note that $\size{\deriv} = 1$ and $\size{\mtyp} = 1$, so
    $\size{\deriv'}
     = \size{\derivtwo'}
     = \size{\derivtwo}
     = \size{\deriv} - \size{\mtyp} + \size{\derivtwo}$.
  \item
    If $\utm = \pvartwo \neq \pvar$, 
    the derivation $\deriv$
    is of the form $\derivs{\deriv}{\jul{\pvartwo:\mset{\ltyp}}{\pvartwo}{\ltyp}}$,
    so $\genv = (\pvartwo:\mset{\ltyp})$
    and $\mtyp = \mset{\,}$.
    Moreover, the derivation $\derivtwo$
    is of the form $\derivs{\derivtwo}{\jum{\genvtwo}{\utmtwo}{\mset{\,}}}$
    and it must be derived by $\ruleLMulti$ with the empty set of premises,
    so in particular $\genvtwo = \emptyset$
    and $\size{\derivtwo} = 0$.
    Our goal is to construct a derivation with conclusion
    $\jul{\genv+\genvtwo}{\utm\sub{\pvar}{\utmtwo}}{\ltyp}$,
    that is to say, $\jul{\pvartwo:\mset{\ltyp}}{\pvartwo}{\ltyp}$.
    It suffices to take $\deriv' := \deriv$.
    Note that $\size{\mtyp} = 0$, so
    $\size{\deriv'}
     = \size{\deriv}
     = \size{\deriv} - \size{\mtyp} + \size{\derivtwo}
    $.
  \end{enumerate}

\Case{$\ruleLImpI$:}
  Let $\deriv$ be of the form:
  \[
    \indrule{\ruleLImpI}{
      \inderiv{\deriv_1}{\jul{\genv,\pvar:\mtyp,\pvartwo:\mtyp'}{\utmthree}{\ltyp'}}
    }{
      \jul{\genv,\pvar:\mtyp}{\plamf{\pvartwo}{\utmthree}}{\mtyp'\limp\ltyp'}
    }
  \]
  Assume by $\alpha$-conversion that $\pvartwo \notin \fv{\pvar,\utmtwo}$.
  By \ih there exists a derivation
  $\derivs{\deriv'_1}{\jul{(\genv,\pvar:\mtyp)+\genvtwo}{\utmthree\sub{\pvar}{\utmtwo}}{\ltyp'}}$
  such that $\size{\deriv'_1} = \size{\deriv_1} - \size{\mtyp} + \size{\derivtwo}$.
  Note that $((\genv,\pvar:\mtyp)+\genvtwo) = ((\genv+\genvtwo),\pvar:\mtyp)$.
  Applying $\ruleLImpI$
  we obtain a derivation
  $\derivs{\deriv'}{\jul{\genv+\genvtwo}{\plamf{\pvartwo}{\utmthree\sub{\pvar}{\utmtwo}}}{\mtyptwo'\limp\ltyp'}}$.
  Moreover
  $\size{\deriv'}
  = 1 + \size{\deriv'_1}
  = 1 + \size{\deriv_1} - \size{\mtyp} + \size{\derivtwo}
  = \size{\deriv} - \size{\mtyp} + \size{\derivtwo}$.

\Case{$\ruleLImpE$:}
  Let $\deriv$ be derived using the $\ruleLImpE$ rule.
  Then it is necessarily of the form:
  \[
    \indrule{\ruleLImpE}{
      \inderiv{\deriv_0}{
        \jul{\genv_0,\pvar:\mtyp_0}{\utmthree}{\mset{\ltyptwo_i}_\iI\limp\ltyp}
      }
      \HS
      \derivthree
    }{
      \jul{\genv_0+_\iI\genv_i,\pvar:\mtyp}{\utmthree\,\utmfour}{\mtyp'\limp\ltyp'}
    }
  \]
  with
  \[
    \derivthree =
      \indrule{\ruleLMulti}{
        \left(
        \inderiv{\deriv_i}{
          \jul{\genv_i,\pvar:\mtyp_i}{\utmfour}{\ltyptwo_i}
        }
        \right)_\iI
      }{
        \jul{+_\iI\genv_i,\pvar:+_\iI\mtyp_i}{\utmfour}{\mset{\ltyptwo_i}_\iI}
      }
  \]
  where $I$ is a finite set,
  and we assume that there is an index $0 \notin I$,
  and moreover $\utm = \utmthree\,\utmfour$
  and $\genv = \genv_0+_\iI\genv_i$
  and $\mtyp = \mtyp_0+_\iI\mtyp_i$.
  \\
  The derivation $\derivs{\derivtwo}{\jul{\genvtwo}{\utmtwo}{\mtyp}}$
  is derived from the $\ruleLMulti$ rule,
  and thus it can be split in $|I| + 1$ derivations
  $\derivs{\derivtwo_i}{\jul{\genvtwo_i}{\utmtwo}{\mtyp_i}}$,
  one for each index $i \in \set{0} \cup I$,
  in such a way that
  $\genvtwo = \genvtwo_0+_\iI\genvtwo_i$
  and $\size{\derivtwo} = \size{\derivtwo_0} +_\iI \size{\derivtwo_i}$.
  \\
  Applying the \ih on $\deriv_0$ and $\derivtwo_0$
  we have that there exists a derivation
    $\derivs{\deriv'_0}{
       \jul{\genv_0+\genvtwo_0}{\utmthree\sub{\pvar}{\utmtwo}}{\mset{\ltyptwo_i}_\iI\limp\ltyp}
     }$
  such that
    $\size{\deriv'_0} = \size{\deriv_0} - \size{\mtyp_0} + \size{\derivtwo_0}$.
  \\
  Applying the \ih for each $i \in I$ on $\deriv_i$ and $\derivtwo_i$
  we have that there exists a derivation
    $\derivs{\deriv'_i}{
       \jul{\genv_i+\genvtwo_i}{\utmfour\sub{\pvar}{\utmtwo}}{\ltyptwo_i}
     }$
  such that
    $\size{\deriv'_i} = \size{\deriv_i} - \size{\mtyp_i} + \size{\derivtwo_i}$.
  \\
  Hence we may construct $\deriv'$ as follows:
  \[
    \indrule{\ruleLImpE}{
      \inderiv{\deriv'_0}{
        \jul{\genv_0+\genvtwo_0}{\utmthree\sub{\pvar}{\utmtwo}}{\mset{\ltyptwo_i}_\iI\limp\ltyp}
      }
      \HS
      \derivthree'
    }{
      \jul{\genv_0+_\iI\genv_i+\genvtwo_0+_\iI\genvtwo_i}{\utmthree\sub{\pvar}{\utmtwo}\,\utmfour\sub{\pvar}{\utmtwo}}{\mtyp'\limp\ltyp'}
    }
  \]
  with 
  \[
    \derivthree' :=
      \indrule{\ruleLMulti}{
        \left(
        \inderiv{\deriv'_i}{
          \jul{\genv_i+\genvtwo_i}{\utmfour\sub{\pvar}{\utmtwo}}{\ltyptwo_i}
        }
        \right)_\iI
      }{
        \jul{+_\iI\genv_i+_\iI\genvtwo_i}{\utmfour\sub{\pvar}{\utmtwo}}{\mset{\ltyptwo_i}_\iI}
      }
  \]

  To conclude, note that $\size{\deriv} = 1 + \size{\deriv_0} +_\iI \size{\deriv_i}$,
  so
  $\size{\deriv'}
  = 1 + \size{\deriv'_0} +_\iI \size{\deriv'_i}
  = 1 + (\size{\deriv_0} - \size{\mtyp_0} + \size{\derivtwo_0})
      +_\iI (\size{\deriv_i} - \size{\mtyp_i} + \size{\derivtwo_i})
  = 1 +  \size{\deriv_0} +_\iI \size{\deriv_i}
      - (\size{\mtyp_0} +_\iI \size{\mtyp_i})
      + (\size{\derivtwo_0} +_\iI \size{\derivtwo_i})
  = \size{\deriv} - \size{\mtyp} + \size{\derivtwo}$.

\Case{$\ruleLAllI$:}
  Let $\deriv$ be of the form:
  \[
    \indrule{\ruleLAllI}{
      \inderiv{\deriv_1}{\jul{\genv,\pvar:\mtyp}{\utmthree\sub{\tvar}{\ttm}}{\ltyp'}}
    }{
      \jul{\genv,\pvar:\mtyp}{\pallfi{\tvar}{\utmthree}}{\ttm\allimp\ltyp'}
    }
  \]
  where by $\alpha$-conversion we may assume that $\tvar \notin \fv{\utmtwo}$.
  By \ih there exists a derivation
  $\derivs{\deriv'_1}{\jul{\genv+\genvtwo}{\utmthree\sub{\tvar}{\ttm}\sub{\pvar}{\utmtwo}}{\ltyp'}}$
  such that
  $\size{\deriv'_1} = \size{\deriv_1} - \size{\mtyp} + \size{\derivtwo}$.
  Note that
  $\utmthree\sub{\tvar}{\ttm}\sub{\pvar}{\utmtwo}
  = \utmthree\sub{\pvar}{\utmtwo}\sub{\tvar}{\ttm}$
  because $\tvar \notin \fv{\utmtwo}$.
  Applying $\ruleLAllI$
  we obtain a derivation
  $\derivs{\deriv'}{\jul{\genv+\genvtwo}{\allf{\tvar}{\utmthree\sub{\pvar}{\utmtwo}}}{\ttm\allimp\ltyp'}}$.
  Moreover,
  $\size{\deriv'}
   = 1 + \size{\deriv'_1}
   = 1 + \size{\deriv_1} - \size{\mtyp} + \size{\derivtwo}
   = \size{\deriv} - \size{\mtyp} + \size{\derivtwo}$.

\Case{$\ruleLAllE$:}
  Let $\deriv$ be of the form:
  \[
    \indrule{\ruleLAllE}{
      \inderiv{\deriv_1}{\jul{\genv,\pvar:\mtyp}{\utmthree}{\ttm\allimp\ltyp}}
    }{
      \jul{\genv,\pvar:\mtyp}{\utmthree\,\ttm}{\ltyp}
    }
  \]
  By \ih there exists a derivation
  $\derivs{\deriv'_1}{\jul{\genv+\genvtwo}{\utmthree\sub{\pvar}{\utmtwo}}{\ttm\allimp\ltyp}}$
  such that
  $\size{\deriv'_1} = \size{\deriv_1} - \size{\mtyp} + \size{\derivtwo}$.
  Applying $\ruleLAllE$
  we obtain a derivation
  $\derivs{\deriv'}{\jul{\genv+\genvtwo}{\utmthree\sub{\pvar}{\utmtwo}\,\ttm}{\ltyp}}$.
  Moreover,
  $\size{\deriv'}
   = 1 + \size{\deriv'_1}
   = 1 + \size{\deriv_1} - \size{\mtyp} + \size{\derivtwo}
   = \size{\deriv} - \size{\mtyp} + \size{\derivtwo}$.

\Case{$\ruleLUnitI$:}
  Let $\deriv$ be derived from $\ruleLUnitI$.
  Then $\derivs{\deriv}{\jul{}{\iunit}{\iunit}}$
  where $\genv = \emptyset$
  and $\mtyp = \mset{\,}$
  and $\ltyp = \iunit$.
  Note that the derivation $\derivtwo$
  is of the form $\derivs{\derivtwo}{\jum{\genvtwo}{\utmtwo}{\mset{\,}}}$
  and it must be derived by $\ruleLMulti$ with the empty set of premises,
  so $\genvtwo = \emptyset$ and $\size{\derivtwo} = 0$.
  Our goal is to construct a derivation with conclusion
  $\jul{\genv+\genvtwo}{\utm\sub{\pvar}{\utmtwo}}{\ltyp}$,
  that is to say, $\jul{}{\iunit}{\iunit}$.
  It suffices to take $\deriv' := \deriv$,
  noting that
  $\size{\deriv'}
   = \size{\deriv}
   = \size{\deriv} - \size{\mtyp} + \size{\derivtwo}
  $.

\Case{$\ruleLUnitE$:}
  Let $\deriv$ be of the form:
  \[
    \indrule{\ruleLUnitE}{
      \inderiv{\deriv_1}{\jul{\genv_1,\pvar:\mtyp_1}{\utmthree}{\iunit}}
      \HS
      \inderiv{\deriv_2}{\jul{\genv_2,\pvar:\mtyp_2}{\utmfour}{\ltyp}}
    }{
      \jul{\genv_1+\genv_2,\pvar:\mtyp_1+\mtyp_2}{\eunit{\utmthree}{\utmfour}}{\ltyp}
    }
  \]
  where $\genv = \genv_1+\genv_2$
  and $\mtyp = \mtyp_1+\mtyp_2$
  and $\utm = \eunit{\utmthree}{\utmfour}$.
  Note that the derivation $\derivs{\derivtwo}{\jul{\genvtwo}{\utmtwo}{\mtyp}}$
  is derived from the $\ruleLMulti$ rule,
  so it can be split into
  $\derivs{\derivtwo_1}{\jul{\genvtwo_1}{\utmtwo}{\mtyp_1}}$
  and
  $\derivs{\derivtwo_2}{\jul{\genvtwo_2}{\utmtwo}{\mtyp_2}}$,
  where
  $\genvtwo = \genvtwo_1+\genvtwo_2$
  and $\size{\derivtwo} = \size{\derivtwo_1} + \size{\derivtwo_2}$.
  By \ih on $\deriv_1$,
  there exists a derivation
    $\derivs{\deriv'_1}{\jul{\genv_1+\genvtwo_1}{\utmthree\sub{\pvar}{\utmtwo}}{\iunit}}$
  such that  
    $\size{\deriv'_1} = \size{\deriv_1} - \size{\mtyp_1} + \size{\derivtwo_1}$.
  By \ih on $\deriv_2$,
  there exists a derivation
    $\derivs{\deriv'_2}{\jul{\genv_2+\genvtwo_2}{\utmfour\sub{\pvar}{\utmtwo}}{\ltyp}}$
  such that  
    $\size{\deriv'_2} = \size{\deriv_2} - \size{\mtyp_2} + \size{\derivtwo_2}$.
  Applying $\ruleLUnitE$ on $\deriv'_1$ and $\deriv'_2$
  we obtain a derivation
    $\derivs{\deriv'}{\jul{\genv_1+\genv_2+\genvtwo_1+\genvtwo_2}{
       \eunit{\utmthree\sub{\pvar}{\utmtwo}}{\utmfour\sub{\pvar}{\utmtwo}}
     }{\ltyp}}$.
  To conclude, note that
    $\size{\deriv'}
     = 1 + \size{\deriv'_1} + \size{\deriv'_2}
     = 1 + (\size{\deriv_1} - \size{\mtyp_1} + \size{\derivtwo_1})
         + (\size{\deriv_2} - \size{\mtyp_2} + \size{\derivtwo_2})
     = 1 + \size{\deriv_1} + \size{\deriv_2}
         - (\size{\mtyp_1} + \size{\mtyp_2})
         + (\size{\derivtwo_1} + \size{\derivtwo_2})
     = \size{\deriv} - \size{\mtyp} + \size{\derivtwo}$.

\Case{$\ruleLFresh$:}
  Let $\deriv$ be of the form:
  \[
    \indrule{\ruleLFresh}{
      \inderiv{\deriv_1}{\jul{\genv,\pvar:\mtyp}{\utmthree}{\ltyp}}
    }{
      \jul{\genv,\pvar:\mtyp}{\freshf{\teig}{\utmthree}}{\ltyp}
    }
  \]
  where $\teig\notin\fv{\genv,\mtyp,\ltyp}$.
  By $\alpha$-conversion, we may assume also that $\teig \notin \fv{\genvtwo,\utmtwo}$.
  By \ih there exists a derivation
    $\derivs{\deriv'_1}{\jul{\genv+\genvtwo}{\utmthree\sub{\pvar}{\utmtwo}}{\ltyp}}$
  such that
    $\size{\deriv'_1} = \size{\deriv_1} - \size{\mtyp} + \size{\derivtwo}$.
  Applying $\ruleLFresh$
  we obtain a derivation
    $\derivs{\deriv'}{\jul{\genv+\genvtwo}{\freshf{\teig}{\utmthree\sub{\pvar}{\utmtwo}}}{\ltyp}}$.
  Moreover,
  $\size{\deriv'}
   = 1 + \size{\deriv'_1}
   = 1 + \size{\deriv_1} - \size{\mtyp} + \size{\derivtwo}
   = \size{\deriv} - \size{\mtyp} + \size{\derivtwo}$.

\Case{$\ruleLVerifTeig$:}
  Let $\deriv$ be of the form:
  \[
    \indrule{\ruleLVerifTeig}{
      \inderiv{\deriv_1}{\jul{\genv,\pvar:\mtyp}{\utmthree}{\gen{\teig}}}
    }{
      \jul{\genv,\pvar:\mtyp}{\verif{\teig}{\utmthree}}{\iunit}
    }
  \]
  By \ih there exists a derivation
    $\derivs{\deriv'_1}{\jul{\genv+\genvtwo}{\utmthree\sub{\pvar}{\utmtwo}}{\gen{\teig}}}$
  such that
    $\size{\deriv'_1} = \size{\deriv_1} - \size{\mtyp} + \size{\derivtwo}$.
  Applying $\ruleLVerifTeig$
  we obtain a derivation
    $\derivs{\deriv'}{\jul{\genv+\genvtwo}{\verif{\teig}{\utmthree\sub{\pvar}{\utmtwo}}}{\ltyp}}$.
  Moreover,
  $\size{\deriv'}
   = 1 + \size{\deriv'_1}
   = 1 + \size{\deriv_1} - \size{\mtyp} + \size{\derivtwo}
   = \size{\deriv} - \size{\mtyp} + \size{\derivtwo}$.

\Case{$\ruleLVerifImp$:}
  Let $\deriv$ be of the form:
  \[
    \indrule{\ruleLVerifImp}{
      \inderiv{\deriv_1}{\jul{\genv,\pvar:\mtyp}{\verif{\ttmtwo}{(\utmthree\,\gen{\ttm})}}{\iunit}}
    }{
      \jul{\genv,\pvar:\mtyp}{\verif{\ttm\imp\ttmtwo}{\utmthree}}{\iunit}
    }
  \]
  By \ih there exists a derivation
    $\derivs{\deriv'_1}{\jul{\genv+\genvtwo}{\verif{\ttmtwo}{(\utmthree\sub{\pvar}{\utmtwo}\,\gen{\ttm})}}{\iunit}}$
  such that
    $\size{\deriv'_1} = \size{\deriv_1} - \size{\mtyp} + \size{\derivtwo}$.
  Applying $\ruleLVerifImp$
  we obtain a derivation
    $\derivs{\deriv'}{\jul{\genv+\genvtwo}{\verif{\ttm\imp\ttmtwo}{\utmthree\sub{\pvar}{\utmtwo}}}{\iunit}}$
  Moreover,
  $\size{\deriv'}
   = 1 + \size{\deriv'_1}
   = 1 + \size{\deriv_1} - \size{\mtyp} + \size{\derivtwo}
   = \size{\deriv} - \size{\mtyp} + \size{\derivtwo}$.

\Case{$\ruleLVerifAll$:}
  Let $\deriv$ be of the form:
  \[
    \indrule{\ruleLVerifAll}{
      \jul{\genv,\pvar:\mtyp}{\freshf{\teig}{\verif{\ttm\sub{\tvar}{\teig}}{(\utmthree\,\teig)}}}{\iunit}
    }{
      \jul{\genv,\pvar:\mtyp}{\verif{\allf{\tvar}{\ttm}}{\utmthree}}{\iunit}
    }
  \]
  where $\teig\notin\fv{\ttm,\utmthree}$.
  By \ih there exists a derivation
    $\derivs{\deriv'_1}{\jul{\genv+\genvtwo}{\freshf{\teig}{\verif{\ttm\sub{\tvar}{\teig}}{(\utmthree\sub{\pvar}{\utmtwo}\,\teig)}}}{\iunit}}$
  such that
    $\size{\deriv'_1} = \size{\deriv_1} - \size{\mtyp} + \size{\derivtwo}$.
  Applying $\ruleLVerifAll$
  we obtain a derivation
    $\derivs{\deriv'}{\jul{\genv+\genvtwo}{\verif{\allf{\tvar}{\ttm}}{\utmthree\sub{\pvar}{\utmtwo}}}{\iunit}}$.
  Moreover,
  $\size{\deriv'}
   = 1 + \size{\deriv'_1}
   = 1 + \size{\deriv_1} - \size{\mtyp} + \size{\derivtwo}
   = \size{\deriv} - \size{\mtyp} + \size{\derivtwo}$.

\Case{$\ruleLGenTeig$:}
  Let $\deriv$ be derived from $\ruleLGenTeig$.
  Then $\derivs{\deriv}{\jul{}{\gen{\teig}}{\gen{\teig}}}$,
  where $\genv = \emptyset$
  and $\mtyp = \mset{\,}$
  and $\ltyp = \gen{\teig}$.
  Note that the derivation $\derivtwo$
  is of the form $\derivs{\derivtwo}{\jum{\genvtwo}{\utmtwo}{\mset{\,}}}$
  and it must be derived by $\ruleLMulti$ with the empty set of premises,
  so $\genvtwo = \emptyset$ and $\size{\derivtwo} = 0$.
  Our goal is to construct a derivation with conclusion
  $\jul{\genv+\genvtwo}{\utm\sub{\pvar}{\utmtwo}}{\ltyp}$,
  that is to say, $\jul{}{\gen{\teig}}{\gen{\teig}}$.
  It suffices to take $\deriv' := \deriv$,
  noting that
  $\size{\deriv'}
   = \size{\deriv}
   = \size{\deriv} - \size{\mtyp} + \size{\derivtwo}
  $.

\Case{$\ruleLGenImp$:}
  Let $\deriv$ be of the form:
  \[
    \indrule{\ruleLGenImp}{
      \inderiv{\deriv_1}{\jul{\genv,\pvar:\mtyp}{\plamf{\pvartwo}{\eunit{\verif{\ttm}{\pvartwo}}{\gen{\ttmtwo}}}}{\ltyp}}
    }{
      \jul{\genv,\pvar:\mtyp}{\gen{\ttm\imp\ttmtwo}}{\ltyp}
    }
  \]
  By \ih there exists a derivation
    $\derivs{\deriv'_1}{\jul{\genv+\genvtwo}{\plamf{\pvartwo}{\eunit{\verif{\ttm}{\pvartwo}}{\gen{\ttmtwo}}}}{\ltyp}}$
  such that
    $\size{\deriv'_1} = \size{\deriv_1} - \size{\mtyp} + \size{\derivtwo}$.
  Applying $\ruleLGenImp$ we obtain a derivation
    $\derivs{\deriv'}{\jul{\genv+\genvtwo}{\gen{\ttm\imp\ttmtwo}}{\ltyp}}$.
  Moreover,
  $\size{\deriv'}
   = 1 + \size{\deriv'_1}
   = 1 + \size{\deriv_1} - \size{\mtyp} + \size{\derivtwo}
   = \size{\deriv} - \size{\mtyp} + \size{\derivtwo}$.
  \\
  \emph{Note:} in this case, one can argue that necessarily
  $\genv = \emptyset$ and $\mtyp = \mset{\,}$,
  and consequently $\genvtwo = \emptyset$,
  but we do not need to make this remark for our proof to go through.

\Case{$\ruleLGenAll$:}
  Let $\deriv$ be of the form:
  \[
    \indrule{\ruleLGenAll}{
      \inderiv{\deriv_1}{\jul{\genv,\pvar:\mtyp}{\pallfi{\tvar}{\gen{\ttm}}}{\ltyp}}
    }{
      \jul{\genv,\pvar:\mtyp}{\gen{\allf{\tvar}{\ttm}}}{\ltyp}
    }
  \]
  By \ih there exists a derivation
    $\derivs{\deriv'_1}{\jul{\genv+\genvtwo}{\pallfi{\tvar}{\gen{\ttm}}}{\ltyp}}$
  such that
    $\size{\deriv'_1} = \size{\deriv_1} - \size{\mtyp} + \size{\derivtwo}$.
  Applying $\ruleLGenImp$ we obtain a derivation
    $\derivs{\deriv'}{\jul{\genv+\genvtwo}{\gen{\allf{\tvar}{\ttm}}}{\ltyp}}$.
  Moreover,
  $\size{\deriv'}
   = 1 + \size{\deriv'_1}
   = 1 + \size{\deriv_1} - \size{\mtyp} + \size{\derivtwo}
   = \size{\deriv} - \size{\mtyp} + \size{\derivtwo}$.
\end{proof}

\begin{lemma}[Weak head subject reduction]
\label{lem:lin:weak_head_subject_reduction}
Let $\utm \tof \utm'$ be a weak head reduction step.
If $\derivs{\deriv}{\jul{\genv}{\utm}{\ltyp}}$,
then there exists a derivation $\deriv'$
such that $\derivs{\deriv'}{\jul{\genv}{\utm'}{\ltyp}}$
and $\size{\deriv} > \size{\deriv'}$.
\end{lemma}
\begin{proof}
We proceed by induction on $\deriv$.

\Case{$\ruleLVar$:}
  Then $\utm = \pvar$ which is $\tof$-irreducible.

\Case{$\ruleLImpI$:}
  Then $\utm = \plamf{\pvar}{\utmtwo}$, which is $\tof$-irreducible.

\Case{$\ruleLImpE$:}
  Then $\utm = \utmtwo\,\utmthree$.
  If the reduction is internal to $\utmtwo$, it is immediate to conclude by \ih.
  So we may assume that the reduction is at the root.
  Then $\utmtwo = \plamf{\pvar}{\utmfour}$
  and the step is of the form
  $\utm = (\plamf{\pvar}{\utmfour})\,\utmthree
     \tof \utmfour\sub{\pvar}{\utmthree} = \utm'$
  and the derivation is of the form:
  \[
    \indrule{\ruleLImpE}{
      \indrule{\ruleLImpI}{
        \inderiv{\deriv_1}{\jul{\genv_1,\pvar:\mtyp}{\utmfour}{\ltyp}}
      }{
        \jul{\genv_1}{\plamf{\pvar}{\utmfour}}{\mtyp\limp\ltyp}
      }
      \HS
      \inderiv{\deriv_2}{\jum{\genv_2}{\utmthree}{\mtyp}}
    }{
      \jul{\genv_1+\genv_2}{(\plamf{\pvar}{\utmfour})\,\utmthree}{\ltyp}
    }
  \]
  where $\genv = \genv_1+\genv_2$
  and where we may assume by $\alpha$-conversion that $\pvar \notin \fv{\utmthree}$.
  By the weighted substitution lemma~(\cref{lem:lin:weighted_substitution_lemma})
  there exists a derivation $\deriv'$
  such that $\derivs{\deriv'}{\jul{\genv_1+\genv_2}{\utmfour\sub{\pvar}{\utmthree}}{\ltyp}}$
  where
    $\size{\deriv}
    = \size{\deriv_1} + \size{\deriv_2} + 2
    > \size{\deriv_1} - \size{\mtyp} + \size{\deriv_2}
    = \size{\deriv'}
    $.

\Case{$\ruleLAllI$:}
  Then $\utm = \pallfi{\tvar}{\utmtwo}$, which is $\tof$-irreducible.

\Case{$\ruleLAllE$:}
  Then $\utm = \utmtwo\,\ttm$.
  If the reduction is internal to $\utmtwo$, it is immediate to conclude by \ih.
  So we may assume that the reduction is at the root.
  Then $\utmtwo = \pallfi{\tvar}{\utmthree}$
  and the step is of the form
  $\utm = (\pallfi{\tvar}{\utmthree})\,\ttm
     \tof \utmthree\sub{\tvar}{\ttm} = \utm'$
  and the derivation is of the form:
  \[
    \indrule{\ruleLAllE}{
      \indrule{\ruleLAllI}{
        \inderiv{\deriv'}{\jul{\genv}{\utmthree\sub{\tvar}{\ttm}}{\ltyp}}
      }{
        \jul{\genv}{\pallfi{\tvar}{\utmthree}}{\ttm\allimp\ltyp}
      }
    }{
      \jul{\genv}{(\pallfi{\tvar}{\utmthree})\,\ttm}{\ltyp}
    }
  \]
  To conclude, note that
  $\derivs{\deriv'}{\jul{\genv}{\utmthree\sub{\tvar}{\ttm}}{\ltyp}}$
  and
  $\size{\deriv}
  = \size{\deriv'} + 2
  > \size{\deriv'}
  $.

\Case{$\ruleLUnitI$:}
  Then $\utm = \iunit$, which is $\tof$-irreducible.

\Case{$\ruleLUnitE$:}
  Then $\utm = \eunit{\utmtwo}{\utmthree}$.
  If the reduction is internal to $\utmtwo$, it is immediate to conclude by \ih.
  So we may assume that the reduction is at the root.
  Then $\utmtwo = \iunit$
  and the step is of the form
  $\utm = \eunit{\iunit}{\utmthree} \tof \utmthree = \utm'$
  and the derivation is of the form:
  \[
    \indrule{\ruleLUnitE}{
      \indrule{\ruleLUnitI}{
      }{
        \jul{}{\iunit}{\iunit}
      }
      \HS
      \inderiv{\deriv'}{\jul{\genv_2}{\utmthree}{\ltyp}}
    }{
      \jul{\genv_1+\genv_2}{\eunit{\iunit}{\utmthree}}{\ltyp}
    }
  \]
  where $\genv = \genv_1+\genv_2$.
  To conclude, note that
    $\derivs{\deriv'}{\jul{\genv_2}{\utmthree}{\ltyp}}$
  and
  $\size{\deriv}
  = \size{\deriv'} + 2
  > \size{\deriv'}
  $.

\Case{$\ruleLFresh$:}
  Then $\utm = \freshf{\teig}{\utmtwo}$.
  If the reduction is internal to $\utmtwo$, it is immediate to conclude by \ih.
  So we may assume that the reduction is at the root.
  Then $\teig\notin\fv{\utmtwo}$
  and the step is of the form
  $\utm = \freshf{\teig}{\utmtwo} \tof \utmtwo = \utm'$
  and the derivation is of the form:
  \[
    \indrule{\ruleLFresh}{
      \inderiv{\deriv'}{\jul{\genv}{\utmtwo}{\ltyp}}
      \HS
      \teig\notin\fv{\genv,\ltyp}
    }{
      \jul{\genv}{\freshf{\teig}{\utmtwo}}{\ltyp}
    }
  \]
  To conclude, note that
    $\derivs{\deriv'}{\jul{\genv}{\utmtwo}{\ltyp}}$
  and
  $\size{\deriv}
  = \size{\deriv'} + 1
  > \size{\deriv'}
  $.

\Case{$\ruleLVerifTeig$:}
  Then $\utm = \verif{\teig}{\utmtwo}$.
  If the reduction is internal to $\utmtwo$,
  it is immediate to conclude by \ih.
  So we may assume that the reduction is at the root.
  Then $\utmtwo = \gen{\teig}$
  and the step is of the form
  $\utm = \verif{\teig}{\utmtwo} \tof \iunit = \utm'$
  and the derivation is of the form:
  \[
    \indrule{\ruleLVerifTeig}{
      \indrule{\ruleLGenTeig}{
      }{
        \jul{}{\gen{\teig}}{\gen{\teig}}
      }
    }{
      \jul{}{\verif{\teig}{\gen{\teig}}}{\iunit}
    }
  \]
  where in particular the context $\genv$ must necessarily be empty.
  Define $\deriv'$ as the derivation constructed using the
  rule $\ruleLUnitI$,
  so that $\derivs{\deriv'}{\jul{}{\iunit}{\iunit}}$.
  To conclude, note that
  $\size{\deriv}
  = 2
  > 1
  = \size{\deriv'}
  $.

\Case{$\ruleLVerifImp$:}
  Then $\utm = \verif{\ttm\imp\ttmtwo}{\utmtwo}$.
  If the reduction is internal to $\utmtwo$,
  it is immediate to conclude by \ih.
  So we may assume that the reduction is at the root.
  Then the step is of the form
  $\utm = \verif{\ttm\imp\ttmtwo}{\utmtwo}
     \tof \verif{\ttmtwo}{(\utmtwo\,\gen{\ttm})} = \utm'$
  and the derivation is of the form:
  \[
    \indrule{\ruleLVerifTeig}{
      \inderiv{\deriv'}{\jul{\genv}{\verif{\ttmtwo}{(\utmtwo\,\gen{\ttm})}}{\iunit}}
    }{
      \jul{\genv}{\verif{\ttm\imp\ttmtwo}{\utmtwo}}{\iunit}
    }
  \]
  To conclude, note that
    $\derivs{\deriv'}{\jul{\genv}{\verif{\ttmtwo}{(\utmtwo\,\gen{\ttm})}}{\iunit}}$
  and
    $\size{\deriv}
    = 1 + \size{\deriv'}
    > \size{\deriv'}$.

\Case{$\ruleLVerifAll$:}
  Then $\utm = \verif{\allf{\tvar}{\ttm}}{\utmtwo}$.
  If the reduction is internal to $\utmtwo$,
  it is immediate to conclude by \ih.
  So we may assume that the reduction is at the root.
  Then the step is of the form
  $\utm = \verif{\allf{\tvar}{\ttm}}{\utmtwo}
     \tof \freshf{\teig}{\verif{\ttm\sub{\tvar}{\teig}}{(\utmtwo\,\teig)}} = \utm'$,
  where $\teig$ is assumed to be fresh, \ie $\teig \notin \fv{\ttm,\utmtwo}$,
  and the derivation is of the form:
  \[
    \indrule{\ruleLVerifAll}{
      \inderiv{\deriv'}{
        \jul{\genv}{\freshf{\teig}{\verif{\ttm\sub{\tvar}{\teig}}{(\utmtwo\,\teig)}}}{\iunit}
      }
      \HS
      \teig\notin\fv{\ttm,\utmtwo}
    }{
      \jul{\genv}{\verif{\allf{\tvar}{\ttm}}{\utmtwo}}{\iunit}
    }
  \]
  To conclude, note that
    $\derivs{\deriv'}{
       \jul{\genv}{\freshf{\teig}{\verif{\ttm\sub{\tvar}{\teig}}{(\utmtwo\,\teig)}}}{\iunit}
     }$
  and
    $\size{\deriv}
    = 1 + \size{\deriv'}
    > \size{\deriv'}$.

\Case{$\ruleLGenTeig$:}
  Then $\utm = \gen{\teig}$, which is $\tof$-irreducible.

\Case{$\ruleLGenImp$:}
  Then $\utm = \gen{\ttm\imp\ttmtwo}$,
  the step is of the form
  $\utm
   = \gen{\ttm\imp\ttmtwo}
   \tof \plamf{\pvar}{\eunit{\verif{\ttm}{\pvar}}{\gen{\ttmtwo}}}
   = \utm'
   $
  and the derivation is of the form:
  \[
    \indrule{\ruleLGenImp}{
      \inderiv{\deriv'}{
        \jul{\genv}{\plamf{\pvar}{\eunit{\verif{\ttm}{\pvar}}{\gen{\ttmtwo}}}}{\ltyp}
      }
    }{
      \jul{\genv}{\gen{\ttm\imp\ttmtwo}}{\ltyp}
    }
  \]
  To conclude, note that
  $\derivs{\deriv'}{\jul{\genv}{\plamf{\pvar}{\eunit{\verif{\ttm}{\pvar}}{\gen{\ttmtwo}}}}{\ltyp}}$
  and
  $\size{\deriv}
  = 1 + \size{\deriv'}
  > \size{\deriv'}$.

\Case{$\ruleLGenAll$:}
  Then $\utm = \gen{\allf{\tvar}{\ttm}}$,
  the step is of the form
  $\utm
   = \gen{\allf{\tvar}{\ttm}}
   \tof \pallfi{\tvar}{\gen{\ttm}}
   = \utm'
   $
  and the derivation is of the form:
  \[
    \indrule{\ruleLGenAll}{
      \inderiv{\deriv'}{\jul{\genv}{\pallfi{\tvar}{\gen{\ttm}}}{\ltyp}}
    }{
      \jul{\genv}{\gen{\allf{\tvar}{\ttm}}}{\ltyp}
    }
  \]
  To conclude, note that
  $\derivs{\deriv'}{\jul{\genv}{\pallfi{\tvar}{\gen{\ttm}}}{\ltyp}}$
  and
  $\size{\deriv}
  = 1 + \size{\deriv'}
  > \size{\deriv'}$.
\end{proof}

\begin{lemma}[Anti-substitution lemma]
\label{lem:lin:anti_substitution_lemma}
If $\derivs{\deriv}{\jul{\genv}{\utm\sub{\pvar}{\utmtwo}}{\ltyp}}$,
then there exist \linearEnvironments $\genv_0,\genvtwo$
and a \linearMultitype $\mtyp$
such that $\genv = \genv_0 + \genvtwo$
and $\derivs{\deriv_0}{\jul{\genv_0,\pvar:\mtyp}{\utm}{\ltyp}}$
and $\derivs{\derivtwo}{\jum{\genvtwo}{\utmtwo}{\mtyp}}$.
\end{lemma}
\begin{proof}
By induction on $\deriv$.
We start by considering the case in which $\utm = \pvar$.
In such case, we have by hypothesis $\derivs{\deriv}{\jul{\genv}{\utmtwo}{\ltyp}}$.
So it suffices to take
$\genv_0 := \emptyset$
and $\genvtwo := \genv$
and $\mtyp := \mset{\ltyp}$,
with the derivation
$\derivs{\deriv_0}{\jul{\pvar:\mset{\ltyp}}{\pvar}{\ltyp}}$
constructed applying $\ruleLVar$,
and
$\derivs{\derivtwo}{\jum{\genv}{\utmtwo}{\mset{\ltyp}}}$
constructed applying $\ruleLMulti$, with $\deriv$ as the single premise.
\\
\indent Now, in the remainder of the proof, we may assume that $\utm \neq \pvar$.
We proceed by case analysis on the last rule applied to construct $\deriv$.
The proof in each case is similar to the weighted substitution lemma~(\cref{lem:lin:weighted_substitution_lemma}),
but reasoning backwards.
For example, in the case of rule $\ruleLUnitE$,
we have that $\utm\sub{\pvar}{\utmtwo}$ is of the form $\eunit{\utmthree}{\utmfour}$,
and since $\utm \neq \pvar$
we have that $\utm = \eunit{\utmthree_0}{\utmfour_0}$
with $\utmthree_0\sub{\pvar}{\utmtwo} = \utmthree$
and $\utmfour_0\sub{\pvar}{\utmtwo} = \utmfour$.
The derivation $\deriv$ is of the form:
\[
  \indrule{\ruleLUnitE}{
    \inderiv{\deriv_1}{
      \jul{\genv_1}{\utmthree_0\sub{\pvar}{\utmtwo}}{\iunit}
    }
    \HS
    \inderiv{\deriv_2}{
      \jul{\genv_1}{\utmfour_0\sub{\pvar}{\utmtwo}}{\ltyp}
    }
  }{
    \jul{\genv_1+\genv_2}{
      \eunit{\utmthree_0\sub{\pvar}{\utmtwo}}{\utmfour_0\sub{\pvar}{\utmtwo}}
    }{\ltyp}
  }
\]
where $\genv = \genv_1+\genv_2$.
By \ih on $\deriv_1$ we have that there exist $\genv_{10},\genvtwo_1,\mtyp_1$
such that $\genv_1 = \genv_{10} + \genvtwo_1$
and there exist derivations
$\derivs{\deriv_{10}}{\jul{\genv_{10},\pvar:\mtyp_1}{\utmthree_0}{\iunit}}$
and $\derivs{\derivtwo_1}{\jum{\genvtwo_1}{\utmtwo}{\mtyp_1}}$.
Similarly, by \ih on $\deriv_2$ we have that there exist $\genv_{20},\genvtwo_2,\mtyp_2$
such that $\genv_2 = \genv_{20} + \genvtwo_2$
and there exist derivations
$\derivs{\deriv_{20}}{\jul{\genv_{20},\pvar:\mtyp_2}{\utmfour_0}{\ltyp}}$
and $\derivs{\derivtwo_2}{\jum{\genvtwo_2}{\utmtwo}{\mtyp_2}}$.
\\
\indent
To conclude, take
$\genv_0 := \genv_{10} + \genv_{20}$
and $\genvtwo := \genvtwo_1 + \genvtwo_2$
and $\mtyp := \mtyp_1 + \mtyp_2$.
Note that
$\genv
 = \genv_1 + \genv_2
 = \genv_{10} + \genvtwo_1 + \genv_{20} + \genvtwo_2
 = \genv_0 + \genvtwo$.
Applying $\ruleLUnitE$ on the derivations $\deriv_{10}$ and $\deriv_{20}$
we obtain a derivation
$\derivs{\deriv_0}{\jul{\genv_{10}+\genv_{20},\pvar:\mtyp_1+\mtyp_2}{\eunit{\utmthree_0}{\utmfour_0}}{\ltyp}}$,
that is,
$\derivs{\deriv_0}{\jul{\genv_0,\pvar:\mtyp}{\utm}{\ltyp}}$.
Moreover, the derivations $\derivtwo_1$ and $\derivtwo_2$ must be derived using
$\ruleLMulti$, so they can be joined to obtain a derivation
$\derivs{\derivtwo_1+\derivtwo_2}{\jum{\genvtwo_1+\genvtwo_2}{\utmtwo}{\mtyp_1+\mtyp_2}}$,
that is, $\derivs{\derivtwo}{\jum{\genvtwo}{\utmtwo}{\mtyp}}$.
\end{proof}

\begin{lemma}[Subject expansion]
\label{lem:lin:subject_expansion}
Let $\utm \to \utm'$.
If $\derivs{\deriv}{\jul{\genv}{\utm'}{\ltyp}}$,
then $\jul{\genv}{\utm}{\ltyp}$.
\end{lemma}
\begin{proof}
We proceed by induction on the size of $\deriv$.
We consider cases depending on the way in which the step $\utm \to \utm'$
is derived,
including cases for root reduction
(given by an application of a rewriting rule at the root of $\utm$)
as well as congruence closure
(given by reducing internally to a subterm).

\begin{enumerate}
\item[] \textbf{Root reduction cases.}
\item
  Suppose that
  $\utm = (\plamf{\pvar}{\utmtwo})\,\utmthree
      \to \utmtwo\sub{\pvar}{\utmthree}
      = \utm'$:
  Then by hypothesis
  $\derivs{\deriv}{\jul{\genv}{\utmtwo\sub{\pvar}{\utmthree}}{\ltyp}}$.
  By the anti-substitution lemma~(\cref{lem:lin:anti_substitution_lemma})
  there exist $\genv_0,\genvtwo,\mtyp$ such that
  $\jul{\genv_0,\pvar:\mtyp}{\utmtwo}{\ltyp}$
  and $\jum{\genvtwo}{\utmthree}{\mtyp}$.
  Hence:
  \[
    \indrule{\ruleLImpE}{
      \indrule{\ruleLImpI}{
        \jul{\genv_0,\pvar:\mtyp}{\utmtwo}{\ltyp}
      }{
        \jul{\genv_0}{\plamf{\pvar}{\utmtwo}}{\mtyp\limp\ltyp}
      }
      \HS
      \jum{\genvtwo}{\utmthree}{\mtyp}
    }{
      \jul{\genv_0+\genvtwo}{(\plamf{\pvar}{\utmtwo})\,\utmthree}{\ltyp}
    }
  \]
\item
  Suppose that
  $\utm = (\pallfi{\tvar}{\utmtwo})\,\ttm
      \to \utmtwo\sub{\tvar}{\ttm}
      = \utm'$:
  Then by hypothesis
  $\derivs{\deriv}{\jul{\genv}{\utmtwo\sub{\tvar}{\ttm}}}{\ltyp}$.
  Hence:
  \[
    \indrule{\ruleLAllE}{
      \indrule{\ruleLAllI}{
        \jul{\genv}{\utmtwo\sub{\tvar}{\ttm}}{\ltyp}
      }{
        \jul{\genv}{\pallfi{\tvar}{\utmtwo}}{\ttm\allimp\ltyp}
      }
    }{
      \jul{\genv}{(\pallfi{\tvar}{\utmtwo})\,\ttm}{\ltyp}
    }
  \]
\item
  Suppose that
  $\utm = \eunit{\iunit}{\utm'} \to \utm'$:
  Then by hypothesis
  $\derivs{\deriv}{\jul{\genv}{\eunit{\iunit}{\utm'}}{\ltyp}}$.
  Hence:
  \[
    \indrule{\ruleLUnitE}{
      \indrule{\ruleLUnitI}{
      }{
        \jul{}{\iunit}{\iunit}
      }
      \HS
      \jul{\genv}{\utm'}{\ltyp}
    }{
      \jul{\genv}{\eunit{\iunit}{\utm'}}{\ltyp}
    }
  \]
\item
  Suppose that
  $\utm = \verif{\teig}{\gen{\teig}} \to \iunit = \utm'$:
  Note that $\derivs{\deriv}{\jul{\genv}{\utm'}{\ltyp}}$
  has the judgment $\jul{\genv}{\iunit}{\ltyp}$ as its conclusion,
  which only be derived using $\ruleLUnitI$
  with $\genv = \emptyset$ and $\ltyp = \iunit$.
  To conclude, note that:
  \[
    \indrule{\ruleLVerifTeig}{
      \indrule{\ruleLGenTeig}{
      }{
        \jul{}{\gen{\teig}}{\gen{\teig}}
      }
    }{
      \jul{}{\verif{\teig}{\gen{\teig}}}{\iunit}
    }
  \]
\item
  Suppose that
  $\utm
   = \verif{\ttm\imp\ttmtwo}{\utmtwo}
   \to \verif{\ttmtwo}{(\utmtwo\,\gen{\ttm})}
   = \utm'$:
  Then by hypothesis
  $\derivs{\deriv}{\jul{\genv}{\verif{\ttmtwo}{(\utmtwo\,\gen{\ttm})}}{\ltyp}}$.
  Hence:
  \[
    \indrule{\ruleLVerifImp}{
      \jul{\genv}{\verif{\ttmtwo}{(\utmtwo\,\gen{\ttm})}}{\ltyp}
    }{
      \jul{\genv}{\verif{\ttm\imp\ttmtwo}{\utmtwo}}{\ltyp}
    }
  \]
\item
  Suppose that
  $\utm
   = \verif{\allf{\tvar}{\ttm}}{\utmtwo}
   \to \freshf{\teig}{\verif{\ttm\sub{\tvar}{\teig}}{(\utmtwo\,\teig)}}
   = \utm'$,
  where $\teig$ is fresh, \ie $\teig \notin \fv{\ttm,\utmtwo}$:
  Then by hypothesis
  $\derivs{\deriv}{\jul{\genv}{\freshf{\teig}{\verif{\ttm\sub{\tvar}{\teig}}{(\utmtwo\,\teig)}}}{\ltyp}}$.
  Hence:
  \[
    \indrule{\ruleLVerifAll}{
      \jul{\genv}{\freshf{\teig}{\verif{\ttm\sub{\tvar}{\teig}}{(\utmtwo\,\teig)}}}{\ltyp}
    }{
      \jul{\genv}{\verif{\allf{\tvar}{\ttm}}{\utmtwo}}{\ltyp}
    }
  \]
\item
  Suppose that
  $\utm
   = \gen{\ttm\imp\ttmtwo}
   \to \plamf{\pvar}{\eunit{\verif{\ttm}{\pvar}}{\gen{\ttmtwo}}}
   = \utm'$:
  Then by hypothesis
  $\derivs{\deriv}{\jul{\genv}{\plamf{\pvar}{\eunit{\verif{\ttm}{\pvar}}{\gen{\ttmtwo}}}}{\ltyp}}$.
  Hence:
  \[
    \indrule{\ruleLGenImp}{
      \jul{\genv}{\plamf{\pvar}{\eunit{\verif{\ttm}{\pvar}}{\gen{\ttmtwo}}}}{\ltyp}
    }{
      \jul{\genv}{\gen{\ttm\imp\ttmtwo}}{\ltyp}
    }
  \]
\item
  Suppose that
  $\utm
   = \gen{\allf{\tvar}{\ttm}}
   \to \pallfi{\tvar}{\gen{\ttm}}
   = \utm'$:
  Then by hypothesis
  $\derivs{\deriv}{\jul{\genv}{\pallfi{\tvar}{\gen{\ttm}}}{\ltyp}}$.
  Hence:
  \[
    \indrule{\ruleLGenAll}{
      \jul{\genv}{\pallfi{\tvar}{\gen{\ttm}}}{\ltyp}
    }{
      \jul{\genv}{\gen{\allf{\tvar}{\ttm}}}{\ltyp}
    }
  \]
\item
  Suppose that
  $\utm
   = \freshf{\teig}{\utm'}
   \to \utm'$,
  where $\teig \notin \fv{\utm'}$:
  Then by hypothesis
  $\derivs{\deriv}{\jul{\genv}{\utm'}{\ltyp}}$.
  By $\alpha$-conversion, we may suppose that
  $\teig \notin \fv{\genv,\ltyp}$.
  Hence:
  \[
    \indrule{\ruleLFresh}{
      \jul{\genv}{\utm'}{\ltyp}
    }{
      \jul{\genv}{\freshf{\teig}{\utm'}}{\ltyp}
    }
  \]
\item[]
  \textbf{Congruence closure cases.} \\
  Most congruence closure cases are straightforward by resorting to the \ih;
  we only include the proofs of the most illustrative or interesting cases.
\item
  \label{lem:lin:subject_expansion__case_cong_left_app}
  Left of a \proofApplication,
    $\utm = \utmtwo\,\utmthree
        \to \utmtwo'\,\utmthree
        = \utm'$
    with $\utmtwo \to \utmtwo'$.
  The derivation $\derivs{\deriv}{\jul{\genv}{\utm'}{\ltyp}}$ is of the form:
  \[
    \indrule{\ruleLImpE}{
      \jul{\genv_1}{\utmtwo'}{\mtyp\limp\ltyp}
      \HS
      \jum{\genv_2}{\utmthree}{\mtyp}
    }{
      \jul{\genv_1+\genv_2}{\utmtwo'\,\utmthree}{\ltyp}
    }
  \]
  where $\genv = \genv_1+\genv_2$.
  Since $\utmtwo \to \utmtwo'$, by \ih we have that
  $\jul{\genv_1}{\utmtwo}{\mtyp\limp\ltyp}$,
  hence applying the rule $\ruleLImpE$
  we obtain that $\jul{\genv_1+\genv_2}{\utmtwo\,\utmthree}{\ltyp}$.
\item
  Right of a \proofApplication,
    $\utm = \utmtwo\,\utmthree
        \to \utmtwo\,\utmthree'
        = \utm'$
    with $\utmthree \to \utmthree'$.
  The derivation $\derivs{\deriv}{\jul{\genv}{\utm'}{\ltyp}}$ is of the form:
  \[
    \indrule{\ruleLImpE}{
      \jul{\genv_1}{\utmtwo}{\mtyp\limp\ltyp}
      \HS
      \indrule{\ruleLMulti}{
        (\jum{\genv_{2,i}}{\utmthree'}{\ltyp_i})_\iI
      }{
        \jum{+_\iI\genv_{2,i}}{\utmthree'}{\mset{\ltyp_i}_\iI}
      }
    }{
      \jul{\genv_1+_\iI\genv_{2,i}}{\utmtwo\,\utmthree'}{\ltyp}
    }
  \]
  where $\genv = \genv_1+_\iI\genv_{2,i}$.
  Since $\utmthree \to \utmthree'$, applying the \ih once per each $\iI$
  we have that $\jul{\genv_{2,i}}{\utmthree}{\ltyp_i}$ holds for each $\iI$.
  Hence applying the rules $\ruleLMulti$ and $\ruleLImpE$
  we obtain that $\jul{\genv_1+_\iI\genv_{2,i}}{\utmtwo\,\utmthree}{\ltyp}$,
  as required.
\item
  Inside a \logicalAbstraction,
    $\utm = \pallfi{\tvar}{\utmtwo}
        \to \pallfi{\tvar}{\utmtwo'}
        = \utm'$
    with $\utmtwo \to \utmtwo'$.
  The derivation $\derivs{\deriv}{\jul{\genv}{\utm'}{\ltyp}}$ is of the form:
  \[
    \indrule{\ruleLAllI}{
      \jul{\genv}{\utmtwo'\sub{\tvar}{\ttm}}{\ltyp'}
    }{
      \jul{\genv}{\pallfi{\tvar}{\utmtwo'}}{\ttm\allimp\ltyp'}
    }
  \]
  Note that since $\utmtwo \to \utmtwo'$
  we have that $\utmtwo\sub{\tvar}{\ttm} \to \utmtwo'\sub{\tvar}{\ttm}$.
  By \ih we have that
    $\jul{\genv}{\utmtwo\sub{\tvar}{\ttm}}{\ltyp'}$,
  hence applying the rule $\ruleLAllI$
  we obtain that
    $\jul{\genv}{\pallfi{\tvar}{\utmtwo}}{\ttm\allimp\ltyp'}$,
  as required.
\item
  Verification of an implication,
    $\utm = \verif{\ttm\imp\ttmtwo}{\utmtwo}
        \to \verif{\ttm\imp\ttmtwo}{\utmtwo'}
        = \utm'$
    with $\utmtwo \to \utmtwo'$.
  The derivation $\derivs{\deriv}{\jul{\genv}{\utm'}{\ltyp}}$ is of the form:
  \[
    \indrule{\ruleLVerifImp}{
      \jul{\genv}{\verif{\ttmtwo}{(\utmtwo'\,\gen{\ttm})}}{\iunit}
    }{
      \jul{\genv}{\verif{\ttm\imp\ttmtwo}{\utmtwo'}}{\iunit}
    }
  \]
  Note that since $\utmtwo \to \utmtwo'$
  we have that 
  $\verif{\ttmtwo}{(\utmtwo\,\gen{\ttm})} \to \verif{\ttmtwo}{(\utmtwo'\,\gen{\ttm})}$.
  By \ih we have that
    $\jul{\genv}{\verif{\ttmtwo}{(\utmtwo\,\gen{\ttm})}}{\iunit}$,
  hence applying the rule $\ruleLVerifImp$
  we obtain that
    $\jul{\genv}{\verif{\ttm\imp\ttmtwo}{\utmtwo}}{\iunit}$
  as required.
\item
  Verification of a second-order universal quantifier,
    $\utm = \verif{\allf{\tvar}{\ttm}}{\utmtwo}
        \to \verif{\allf{\tvar}{\ttm}}{\utmtwo'}
        = \utm'$
    with $\utmtwo \to \utmtwo'$.
  The derivation $\derivs{\deriv}{\jul{\genv}{\utm'}{\ltyp}}$ is of the form:
  \[
    \indrule{\ruleLVerifAll}{
      \jul{\genv}{\freshf{\teig}{\verif{\ttm\sub{\tvar}{\teig}}{(\utmtwo'\,\teig)}}}{\iunit}
      \HS
      \teig\notin\fv{\ttm,\utmtwo'}
    }{
      \jul{\genv}{\verif{\allf{\tvar}{\ttm}}{\utmtwo'}}{\iunit}
    }
  \]
  Note that since $\utmtwo \to \utmtwo'$
  we have that 
  $\freshf{\teig}{\verif{\ttm\sub{\tvar}{\teig}}{(\utmtwo\,\teig)}}
  \to \freshf{\teig}{\verif{\ttm\sub{\tvar}{\teig}}{(\utmtwo'\,\teig)}}$.
  Furthermore, by $\alpha$-renaming we may assume that $\teig \notin \fv{\utmtwo}$.
  By \ih we have that
    $\jul{\genv}{\freshf{\teig}{\verif{\ttm\sub{\tvar}{\teig}}{(\utmtwo\,\teig)}}}{\iunit}$
  hence applying the rule $\ruleLVerifImp$
  we obtain that
    $\jul{\genv}{\verif{\allf{\tvar}{\ttm}}{\utmtwo}}{\iunit}$
  as required.
\item
  The remaning congruence cases are straightforward by \ih,
  similar to the proof of \caseref{lem:lin:subject_expansion__case_cong_left_app}
  in this lemma.
\end{enumerate}
\end{proof}

\begin{proposition}[Standardization]
\label{a:prop:f:standardization}
If $\utm \tos \iunit$ then $\utm \tofs \iunit$.
\end{proposition}
\begin{proof}
Suppose that $\utm \tos \iunit$.
Note that $\jul{}{\iunit}{\iunit}$.
Since $\utm \tos \iunit$,
by applying subject expansion~(\cref{lem:lin:subject_expansion})
repeatedly, we obtain a derivation $\deriv_0$
such that $\derivs{\deriv_0}{\jul{}{\utm}{\iunit}}$.
Consider a maximal, potentially infinite,
sequence of weak head reduction steps $\utm = \utm_0 \tow \utm_1 \tow \utm_2 \hdots$.
Applying weak head subject reduction~(\cref{lem:lin:weak_head_subject_reduction})
repeatedly, we construct a sequence of derivations
$\deriv_0,\deriv_1,\deriv_2,\hdots$
such that $\derivs{\deriv_i}{\jul{}{\utm_i}{\iunit}}$ for all $i$
and such that $\size{\deriv_i} > \size{\deriv_{i+1}}$ for all $i$.
Then the reduction sequence $\utm = \utm_0 \tow \utm_1 \tow \utm_2 \hdots$
must be finite because otherwise we would have an infinite decreasing sequence
of natural numbers
$\size{\deriv_0} > \size{\deriv_1} > \size{\deriv_2} > \hdots$.
Let $\utm_n$ be the last term of the sequence,
which must be a $\tow$-normal form.
Since $\utm \tos \iunit$ and $\utm \tos \utm_n$
by confluence~(\cref{prop:f:confluence}) we have that $\utm_n \tos \iunit$,
but $\utm_n$ is a weak head normal form, so $\utm_n = \iunit$.
\end{proof}

\subsection{Second-Order: Soundness (\cref{thm:f:soundness})}
  \label{a:sec:f:soundness}
  
The soundness theorem is proved in the main body (\cref{thm:f:soundness}).
In this section we prove all the required auxiliary lemmas: 
\begin{itemize}
\item
  \textsc{Correctness}
  (\cref{a:lem:f:correctness})
  corresponding to \cref{lem:f:correctness} in the main body.
\item
  \textsc{Adequacy for \logicalTerms}
  (\cref{a:lem:f:adequacy_for_propositions}),
  corresponding to \cref{lem:f:adequacy_for_propositions} in the main body.
\item
  \textsc{Irrelevance}
  (\cref{a:lem:f:irrelevance})
  corresponding to \cref{lem:f:irrelevance} in the main body.
\item
  \textsc{Substitution lemma}
  (\cref{a:lem:f:substitution})
  corresponding to \cref{lem:f:substitution} in the main body.
\end{itemize}

\begin{lemma}[Correctness]
\label{a:lem:f:correctness}
If $\ttm$ has no free \logicalVariables then
$\verif{\ttm}{\gen{\ttm}} \tofs \iunit$.
\end{lemma}
\begin{proof}
By induction on the size of $\ttm$:
\begin{enumerate}
\item
  \LogicalEigenvariable, $\ttm = \teig$:
  Then $\verif{\teig}{\gen{\teig}} \to \iunit$.
\item
  Implication, $\ttm = (\ttmtwo \imp \ttmthree)$:
  \[
    \begin{array}{rcll}
    &&
      \verif{\ttm\imp\ttmtwo}{\gen{\ttm\imp\ttmtwo}}
    \\
    & \tof &
      \verif{\ttmtwo}{(\gen{\ttm\imp\ttmtwo}\,\gen{\ttm})}
    \\
    & \tof &
      \verif{\ttmtwo}{((\plamf{\pvar}{\eunit{\verif{\ttm}{\pvar}}{\gen{\ttmtwo}}})\,\gen{\ttm})}
    \\
    & \tof &
      \verif{\ttmtwo}{(\eunit{\verif{\ttm}{\gen{\ttm}}}{\gen{\ttmtwo}})}
    \\
    & \tofs &
      \verif{\ttmtwo}{(\eunit{\iunit}{\gen{\ttmtwo}})}
      & \text{by \ih}
    \\
    & \tof &
      \verif{\ttmtwo}{\gen{\ttmtwo}}
    \\
    & \tofs &
      \iunit
      & \text{by \ih}
    \end{array}
  \]
\item
  Second-order universal quantifier, $\ttm = (\allf{\tvar}{\ttmtwo})$:
  \[
    \begin{array}{rcll}
    &&
      \verif{\allf{\tvar}{\ttmtwo}}{\gen{\allf{\tvar}{\ttmtwo}}}
    \\
    & \tof &
      \freshf{\teig}{\verif{\ttmtwo\sub{\tvar}{\teig}}{(\gen{\allf{\tvar}{\ttmtwo}}\,\teig)}}
    \\
    & \tof &
      \freshf{\teig}{\verif{\ttmtwo\sub{\tvar}{\teig}}{((\pallfi{\tvar}{\gen{\ttmtwo}})\,\teig)}}
    \\
    & \tof &
      \freshf{\teig}{\verif{\ttmtwo\sub{\tvar}{\teig}}{\gen{\ttmtwo\sub{\tvar}{\teig}}}}
    \\
    & \tofs &
      \freshf{\teig}{\iunit}
      & \text{by \ih}
    \\
    & \tof &
      \iunit
    \end{array}
  \]
\end{enumerate}
\end{proof}

\begin{lemma}[Adequacy for \logicalTerms]
\label{a:lem:f:adequacy_for_propositions}
$\semf{\ttm}{\asig} \in \VC{\subs{\ttm}{\asig}}$
\end{lemma}
\begin{proof}
We proceed by induction on $\ttm$.

If $\ttm$ is a \textbf{\logicalVariable} ($\ttm = \tvar$):
  Let $\asig(\tvar) = (\ttmtwo,\vcset)$, where
  we have that $\vcset \in \VC{\ttmtwo}$
  by definition of $\asig$ being a \candidateAssignment.
  Note that $\semf{\ttm}{\asig} = \semf{\tvar}{\asig} = \vcset$ and 
  $\subs{\ttm}{\asig} = \subs{\tvar}{\asig} = \ttmtwo$,
  so to check that $\semf{\ttm}{\asig} \in \VC{\subs{\ttm}{\asig}}$
  amounts to checking that $\vcset \in \VC{\ttmtwo}$, which holds by hypothesis.

If $\ttm$ is an \textbf{\logicalEigenvariable} ($\ttm = \teig$):
  Note that, by definition,
    $\semf{\teig}{\asig} \eqdef \set{\utm \ST \utm \tofs \gen{\teig}}$
  and $\subs{\teig}{\asig} = \teig$,
  because \logicalEigenvariables are not affected by substitutions.
  In order to show that $\semf{\teig}{\asig} \in \VC{\teig}$,
  we check each of the three conditions defining $\VC{\teig}$:
  \begin{enumerate}
  \item
    In order to show that $\gen{\teig} \in \semf{\teig}{\asig}$,
    note that $\gen{\teig} \tofs \gen{\teig}$ in zero reduction steps.
  \item
    Let $\utm \in \semf{\teig}{\asig}$.
    Then $\utm \tofs \gen{\teig}$,
    so $\verif{\teig}{\utm} \tofs \verif{\teig}{\gen{\teig}} \tof \iunit$.
  \item
    Let $\utm,\utm'$ be \metaterms such that $\utm \tow \utm'$.
    Then by confluence~(\cref{prop:f:confluence})
    we have that $\utm \tofs \gen{\teig}$
    if and only if $\utm' \tofs \gen{\teig}$,
    so $\utm \in \semf{\teig}{\asig}$
    if and only if $\utm' \in \semf{\teig}{\asig}$.
  \end{enumerate}

If $\ttm$ is an \textbf{implication} ($\ttm = (\ttmtwo\imp\ttmthree)$):
  Note first that $\semf{\ttmtwo}{\asig} \in \VC{\subs{\ttmtwo}{\asig}}$
  and $\semf{\ttmthree}{\asig} \in \VC{\subs{\ttmthree}{\asig}}$ by \ih,
  so in particular
  $\verif{\subs{\ttmtwo}{\asig}}{\gen{\subs{\ttmtwo}{\asig}}} \tofs \iunit$
  and
  $\verif{\subs{\ttmthree}{\asig}}{\gen{\subs{\ttmthree}{\asig}}} \tofs \iunit$
  by \cref{lem:f:correctness}.
  In order to show that $\semf{\ttmtwo\imp\ttmthree}{\asig} \in \VC{\subs{(\ttmtwo\imp\ttmthree)}{\asig}}$,
  we check each of the three conditions defining $\VC{\subs{(\ttmtwo\imp\ttmthree)}{\asig}}$:
  \begin{enumerate}
  \item
    In order to show that $\gen{\subs{(\ttmtwo\imp\ttmthree)}{\asig}} \in \semf{\ttmtwo\imp\ttmthree}{\asig}$,
    it suffices to show the two conditions defining $\semf{\ttmtwo\imp\ttmthree}{\asig}$.
    By correctness~\cref{lem:f:correctness}, we have that
      $\verif{\subs{(\ttmtwo\imp\ttmthree)}{\asig}}{\gen{\subs{(\ttmtwo\imp\ttmthree)}{\asig}}} \tofs
       \iunit$,
    so it suffices to show that for every $\utm \in \semf{\ttmtwo}{\asig}$ we have that
    $\gen{\subs{(\ttmtwo\imp\ttmthree)}{\asig}}\,\utm \in \semf{\ttmthree}{\asig}$.
    Indeed, let $\utm \in \semf{\ttmtwo}{\asig}$,
    and let us show that
    $\gen{\subs{(\ttmtwo\imp\ttmthree)}{\asig}}\,\utm \in \semf{\ttmthree}{\asig}$.
    Since $\semf{\ttmthree}{\asig} \in \VC{\subs{\ttmthree}{\asig}}$ by \ih,
    we know that $\gen{\subs{\ttmthree}{\asig}} \in \VC{\subs{\ttmthree}{\asig}}$
    and that $\VC{\subs{\ttmthree}{\asig}}$ is closed by $\tof$-convertibility,
    hence it suffices to show that
    $\gen{\subs{(\ttmtwo\imp\ttmthree)}{\asig}}\,\utm \tofs \gen{\subs{\ttmthree}{\asig}}$.
    Moreover, since $\semf{\ttmtwo}{\asig} \in \VC{\subs{\ttmtwo}{\asig}}$ by \ih,
    we know in particular that $\verif{\subs{\ttmtwo}{\asig}}{\utm} \tofs \iunit$.
    Then:
    \[
      \begin{array}{llll}
      &&
        \gen{\subs{(\ttmtwo\imp\ttmthree)}{\asig}}\,\utm
      \\
      & = &
        \gen{\subs{\ttmtwo}{\asig}\imp\subs{\ttmthree}{\asig}}\,
        \utm
      \\
      & \tof &
          (\plamf{\pvar}{
            \eunit{
              \verif{\subs{\ttmtwo}{\asig}}{\pvar}\,
            }{
              \gen{\subs{\ttmthree}{\asig}}
            }
          })\,
          \utm
      \\
      & \tof &
          \eunit{
            \verif{\subs{\ttmtwo}{\asig}}{\utm}
          }{
            \gen{\subs{\ttmthree}{\asig}}
          }
      \\
      & \tofs &
          \eunit{
            \iunit
          }{
            \gen{\subs{\ttmthree}{\asig}}
          }
        &
        \text{since $\verif{\subs{\ttmtwo}{\asig}}{\utm} \tofs \iunit$}
      \\
      & \tof &
        \gen{\subs{\ttmthree}{\asig}}
      \end{array}
    \]
  \item
    Let $\utm \in \semf{\ttmtwo\imp\ttmthree}{\asig}$
    and observe that $\verif{\subs{(\ttmtwo\imp\ttmthree)}{\asig}}{\utm} \tofs \iunit$
    holds by definition of $\semf{\ttmtwo\imp\ttmthree}{\asig}$.
  \item
    Let $\utm$, $\utm'$ be \metaterms such that $\utm \tof \utm'$.
    We argue that $\utm \in \semf{\ttmtwo\imp\ttmthree}{\asig}$
    if and only if $\utm' \in \semf{\ttmtwo\imp\ttmthree}{\asig}$.
    \begin{itemize}
    \item[$(\Rightarrow)$]
      Suppose that $\utm \in \semf{\ttmtwo\imp\ttmthree}{\asig}$.
      By definition, this means that two conditions hold,
      first $\verif{\subs{(\ttmtwo\imp\ttmthree)}{\asig}}{\utm} \tofs \iunit$
      and, second, $\utm\,\utmtwo \in \semf{\ttmthree}{\asig}$
      holds for every $\utmtwo \in \semf{\ttmtwo}{\asig}$.
      From the first condition and confluence~(\cref{prop:f:confluence})
      it is immediate to conclude that
      $\verif{\subs{(\ttmtwo\imp\ttmthree)}{\asig}}{\utm'} \tofs \iunit$.
      Moreover, for each $\utmtwo \in \semf{\ttmtwo}{\asig}$,
      we have that $\utm\,\utmtwo \tof \utm'\,\utmtwo$,
      so by the second condition and the fact that
      $\semf{\ttmthree}{\asig} \in \VC{\subs{\ttmthree}{\asig}}$ holds by \ih,
      we obtain that $\utm'\,\utmtwo \in \semf{\ttmthree}{\asig}$.
      Hence $\utm' \in \semf{\ttmtwo\imp\ttmthree}{\asig}$.
    \item[$(\Leftarrow)$]
      Similar to the proof of the $(\Rightarrow)$ direction.
    \end{itemize}
  \end{enumerate}

If $\ttm$ is a \textbf{universal quantification} ($\ttm = \allf{\tvar}{\ttmtwo}$):
  Given an arbitrary \logicalTerm $\ttmthree$ and an arbitrary $\vcset \in \VC{\ttmthree}$,
  we may consider the \candidateAssignment $\asigext \eqdef \asig\extsub{\tvar}{(\ttmthree,\vcset)}$.
  Note by \ih that in this case we have that
  $\semf{\ttmtwo}{\asigext} \in \VC{\subs{\ttmtwo}{\asigext}}$.

  Now, in order to show that
  $\semf{\allf{\tvar}{\ttmtwo}}{\asig} \in \VC{\subs{(\allf{\tvar}{\ttmtwo})}{\asig}}$,
  we check each of the three conditions defining $\VC{\subs{(\allf{\tvar}{\ttmtwo})}{\asig}}$:
  \begin{enumerate}
  \item
    To show that
    $\gen{\subs{(\allf{\tvar}{\ttmtwo})}{\asig}} \in \semf{\allf{\tvar}{\ttmtwo}}{\asig}$,
    it suffices to show the two conditions defining $\semf{\allf{\tvar}{\ttmtwo}}{\asig}$.
    By correctness~(\cref{lem:f:correctness}), we have that
      $\verif{\subs{(\allf{\tvar}{\ttmtwo})}{\asig}}{\gen{\subs{(\allf{\tvar}{\ttmtwo})}{\asig}}}
       \tofs \iunit$,
    so it suffices to show that
    for every \logicalTerm $\ttmthree$ and every $\vcset \in \VC{\ttmthree}$
    we have that
    $\gen{\subs{(\allf{\tvar}{\ttmtwo})}{\asig}}\,\ttmthree
     \in \semf{\ttmtwo}{\asig\extsub{\tvar}{(\ttmthree,\vcset)}}$.
    Indeed,
    let $\ttmthree$ be any \logicalTerm and let $\vcset \in \VC{\ttmthree}$,
    and let us write $\asigext \eqdef \asig\extsub{\tvar}{(\ttmthree,\vcset)}$.
    Our goal is to show that
    $
      \gen{\subs{(\allf{\tvar}{\ttmtwo})}{\asig}}\,\ttmthree
      \in \semf{\ttmtwo}{\asigext}
    $.
    Since by \ih we have that
    $\semf{\ttmtwo}{\asigext} \in \VC{\subs{\ttmtwo}{\asigext}}$,
    we know in particular that
    $\gen{\subs{\ttmtwo}{\asigext}} \in \semf{\ttmtwo}{\asigext}$
    and that $\semf{\ttmtwo}{\asigext}$ is closed by $\tof$-convertibility,
    hence it suffices to show that
    $\gen{\subs{(\allf{\tvar}{\ttmtwo})}{\asig}}\,\ttmthree
     \tofs \gen{\subs{\ttmtwo}{\asigext}}$.
    Let us write $\tsubst \eqdef \tilde{\asig}\extsub{\tvar}{\tvar}$
    be the \logicalSubstitution corresponding to the first projection
    of $\asig$ but fixing $\tvar$,
    in such a way that
    $\subs{(\allf{\tvar}{\ttmtwo})}{\asig}
     = \allf{\tvar}{\subs{\ttmtwo}{\tsubst}}$.
    Then we have:
    \[
      \begin{array}{rcll}
      &&
        \gen{\subs{(\allf{\tvar}{\ttmtwo})}{\asig}}\,\ttmthree
      \\
      & = &
        \gen{\allf{\tvar}{\subs{\ttmtwo}{\tsubst}}}\,\ttmthree
      \\
      & \tof &
        (\pallfi{\tvar}{\gen{\subs{\ttmtwo}{\tsubst}}})\,\ttmthree
      \\
      & \tof &
        \gen{\subs{\ttmtwo}{\tsubst}\sub{\tvar}{\ttmthree}}
      \\
      & = &
        \gen{\subs{\ttmtwo}{\asigext}}
      \end{array}
    \]
  \item
    Let $\utm \in \semf{\allf{\tvar}{\ttmtwo}}{\asig}$
    and observe that
    $\verif{\subs{(\allf{\tvar}{\ttmtwo})}{\asig}}{\utm} \tofs \iunit$
    holds by definition of $\semf{\allf{\tvar}{\ttmtwo}}{\asig}$.
  \item
    Let $\utm$, $\utm'$ be \metaterms such that $\utm \tof \utmtwo$.
    We argue that $\utm \in \semf{\allf{\tvar}{\ttmtwo}}{\asig}$
    if and only if $\utm' \in \semf{\allf{\tvar}{\ttmtwo}}{\asig}$.
    \begin{itemize}
    \item[$(\Rightarrow)$]
      Suppose that $\utm \in \semf{\allf{\tvar}{\ttmtwo}}{\asig}$.
      By definition,
      this means that two conditions hold,
      first $\verif{\subs{(\allf{\tvar}{\ttmtwo})}{\asig}}{\utm} \tofs \iunit$
      and, second, $\utm\,\ttmthree \in \semf{\ttmtwo}{\asig\extsub{\tvar}{(\ttmthree,\vcset)}}$
      holds for every \logicalTerm $\ttmthree$ and every $\vcset \in \VC{\ttmthree}$.
      From the first condition and confluence~(\cref{prop:f:confluence})
      it is immediate to conclude that
      $\verif{\subs{(\allf{\tvar}{\ttmtwo})}{\asig}}{\utm'} \tofs \iunit$.
      Moreover, for each \logicalTerm $\ttmthree$ and each $\vcset \in \VC{\ttmthree}$,
      and writing $\asigext \eqdef \asig\extsub{\tvar}{(\ttmthree,\vcset)}$,
      we have that $\utm\,\ttmthree \tof \utm'\,\ttmthree$,
      so by the second condition and the fact that
      $\semf{\ttmtwo}{\asigext} \in \VC{\subs{\ttmtwo}{\asigext}}$ holds by \ih
      we obtain that $\utm'\,\ttmthree \in \semf{\ttmtwo}{\asigext}$.
      Hence $\utm' \in \semf{\allf{\tvar}{\ttmtwo}}{\asig}$.
    \item[$(\Leftarrow)$]
      Similar to the proof of the $(\Rightarrow)$ direction.
    \end{itemize}
  \end{enumerate}
\end{proof}

\begin{lemma}[Irrelevance]
\label{a:lem:f:irrelevance}
If $\asig_1$ and $\asig_2$ agree on the free variables of $\ttm$,
then $\semf{\ttm}{\asig_1} = \semf{\ttm}{\asig_2}$.
\end{lemma}
\begin{proof}
Straightforward by induction on $\ttm$.
\end{proof}

\begin{lemma}[Substitution lemma]
\label{a:lem:f:substitution}
$\semf{\ttm\sub{\tvar}{\ttmtwo}}{\asig}
 = \semf{\ttm}{\asig\extsub{\tvar}{(\subs{\ttmtwo}{\asig},\semf{\ttmtwo}{\asig})}}$
\end{lemma}
\begin{proof}
By induction on $\ttm$.

If $\ttm$ is the \textbf{same \logicalVariable} as $\tvar$, \ie $\ttm = \tvar$:
  Then 
  $\semf{\tvar\sub{\tvar}{\ttmtwo}}{\asig}
   = \semf{\ttmtwo}{\asig}
   = \semf{\tvar}{\asig\extsub{\tvar}{(\subs{\ttmtwo}{\asig},\semf{\ttmtwo}{\asig})}}$.

If $\ttm$ is a \textbf{different \logicalVariable} than $\tvar$, \ie $\ttm = \tvartwo \neq \tvar$:
  Then 
  $\semf{\tvartwo\sub{\tvar}{\ttmtwo}}{\asig}
   = \semf{\tvartwo}{\asig}
   = \semf{\tvartwo}{\asig\extsub{\tvar}{(\subs{\ttmtwo}{\asig},\semf{\ttmtwo}{\asig})}}$.
\item
  \LogicalEigenvariable, $\ttm = \teig$.
  Then 
  $\semf{\teig\sub{\tvar}{\ttmtwo}}{\asig}
   = \semf{\teig}{\asig}
   = \set{\utm \ST \utm \tofs \gen{\teig}}
   = \semf{\teig}{\asig\extsub{\tvar}{(\subs{\ttmtwo}{\asig},\semf{\ttmtwo}{\asig})}}$.

If $\ttm$ is an \textbf{implication}, \ie $\ttm = (\ttm_1\imp\ttm_2)$:
  Our goal is to show that
  $\semf{(\ttm_1\imp\ttm_2)\sub{\tvar}{\ttmtwo}}{\asig}
   = \semf{\ttm_1\imp\ttm_2}{\asig\extsub{\tvar}{(\subs{\ttmtwo}{\asig},\semf{\ttmtwo}{\asig})}}$
  are equal as sets. We show the two inclusions:
  \begin{itemize}
  \item[$(\subseteq)$]
    Let $\utm \in \semf{(\ttm_1\imp\ttm_2)\sub{\tvar}{\ttmtwo}}{\asig}$,
    that is $\utm \in \semf{\ttm_1\sub{\tvar}{\ttmtwo}\imp\ttm_2\sub{\tvar}{\ttmtwo}}{\asig}$.
    This means that two conditions hold,
    first
    $\verif{\subs{(\ttm_1\sub{\tvar}{\ttmtwo}\imp\ttm_2\sub{\tvar}{\ttmtwo})}{\asig}}{\utm} \tofs \iunit$
    and, second,
    $\utm\,\utmtwo \in \semf{\ttm_2\sub{\tvar}{\ttmtwo}}{\asig}$
    holds for every $\utmtwo \in \semf{\ttm_1\sub{\tvar}{\ttmtwo}}{\asig}$.
    Note that
    $\subs{(\ttm_1\sub{\tvar}{\ttmtwo}\imp\ttm_2\sub{\tvar}{\ttmtwo})}{\asig}
    = \subs{(\ttm_1\imp\ttm_2)\sub{\tvar}{\ttmtwo}}{\asig}
    = \subs{(\ttm_1\imp\ttm_2)}{\asig\extsub{\tvar}{(\subs{\ttmtwo}{\asig},\semf{\ttmtwo}{\asig})}}$,
    so the first condition implies that
    $\verif{\subs{(\ttm_1\imp\ttm_2)}{\asig\extsub{\tvar}{(\subs{\ttmtwo}{\asig},\semf{\ttmtwo}{\asig})}}}{\utm} \tofs \iunit$.
    Moreover, by resorting to the \ih, the second condition implies that
    $\utm\,\utmtwo \in \semf{\ttm_2}{\asig\extsub{\tvar}{(\subs{\ttmtwo}{\asig},\semf{\ttmtwo}{\asig})}}$
    holds for every $\utmtwo \in \semf{\ttm_1}{\asig\extsub{\tvar}{(\subs{\ttmtwo}{\asig},\semf{\ttmtwo}{\asig})}}$.
    Hence
    $\utm \in \semf{(\ttm_1\imp\ttm_2)}{\asig\extsub{\tvar}{(\subs{\ttmtwo}{\asig},\semf{\ttmtwo}{\asig})}}$.
  \item[$(\supseteq)$]
    Symmetric to the proof of the $(\subseteq)$ case.
  \end{itemize}

If $\ttm$ is a \textbf{universal quantification}, \ie $\ttm = (\allf{\tvartwo}{\ttm'})$:
  Assume, using $\alpha$-conversion if necessary, that $\tvartwo \notin \set{\tvar} \cup \fv{\ttmtwo}$.
  Our goal is to show that
  $\semf{(\allf{\tvartwo}{\ttm'})\sub{\tvar}{\ttmtwo}}{\asig}
   = \semf{\allf{\tvartwo}{\ttm'}}{\asig\extsub{\tvar}{(\subs{\ttmtwo}{\asig},\semf{\ttmtwo}{\asig})}}$
  are equal as sets.
  We show the two inclusions:
  \begin{itemize}
  \item[$(\subseteq)$]
    Let $\utm \in \semf{(\allf{\tvartwo}{\ttm'})\sub{\tvar}{\ttmtwo}}{\asig}$,
    that is $\utm \in \semf{\allf{\tvartwo}{\ttm'\sub{\tvar}{\ttmtwo}}}{\asig}$.
    This means that two conditions hold,
    first $\verif{\subs{(\allf{\tvartwo}{\ttm'\sub{\tvar}{\ttmtwo}})}{\asig}}{\utm} \tofs \iunit$
    and, second,
    $\utm\,\ttmthree \in \semf{\ttm'\sub{\tvar}{\ttmtwo}}{\asig\extsub{\tvartwo}{(\ttmthree,\vcset)}}$
    holds for every \logicalTerm $\ttmthree$ and every $\vcset \in \VC{\ttmthree}$.
    Note that
    $\subs{(\allf{\tvartwo}{\ttm'\sub{\tvar}{\ttmtwo}})}{\asig}
     = \subs{(\allf{\tvartwo}{\ttm'})\sub{\tvar}{\ttmtwo}}{\asig}
     = \subs{(\allf{\tvartwo}{\ttm'})}{\asig\extsub{\tvar}{(\subs{\ttmtwo}{\asig},\semf{\ttmtwo}{\asig})}}$,
    so the first condition implies that
    $\verif{\subs{(\allf{\tvartwo}{\ttm'})}{\asig\extsub{\tvar}{(\subs{\ttmtwo}{\asig},\semf{\ttmtwo}{\asig})}}}{\utm}
     \tofs \iunit$.
    Consider an arbitrary \logicalTerm $\ttmthree$,
    an arbitrary $\vcset \in \VC{\ttmthree}$,
    and let us write $\asigext \eqdef \asig\extsub{\tvartwo}{(\ttmthree,\vcset)}$.
    By resorting to the \ih, the second condition
    implies that
    $\utm\,\ttmthree \in \semf{\ttm'}{\asigext\extsub{\tvar}{(\subs{\ttmtwo}{\asigext},\semf{\ttmtwo}{\asigext})}}$.
    Moreover, since $\tvartwo \notin \fv{\ttmtwo}$
    we have that
    $\subs{\ttmtwo}{\asigext} = \subs{\ttmtwo}{\asig}$
    and
    $\semf{\ttmtwo}{\asigext} = \semf{\ttmtwo}{\asig}$ by irrelevance~(\cref{lem:f:irrelevance}).
    From this we obtain that
    $\utm\,\ttmthree \in \semf{\ttm'}{\asig\extsub{\tvartwo}{(\ttmthree,\vcset)}\extsub{\tvar}{(\subs{\ttmtwo}{\asig},\semf{\ttmtwo}{\asig})}}$
    holds for every \logicalTerm $\ttmthree$ and every $\vcset \in \VC{\ttmthree}$.
    Hence
    $\utm \in \semf{(\allf{\tvartwo}{\ttm'})}{\asig\extsub{\tvar}{(\subs{\ttmtwo}{\asig},\semf{\ttmtwo}{\asig})}}$.
  \item[$(\supseteq)$]
    Symmetric to the proof of the $(\subseteq)$ case.
  \end{itemize}
\end{proof}

\begin{lemma}[Adequacy for terms]
\label{a:lem:f:adequacy_for_terms}
Let $\judgf{\penv}{\ptm}{\ttm}$ and $\compatf{\penv}{\psubst}$.
Then $\subs{\ptm}{\asig\psubst} \in \semf{\ttm}{\asig}$.
\end{lemma}
\begin{proof}
We proceed by induction on the derivation of the judgment $\judgf{\penv}{\ptm}{\ttm}$:

\Case[]{Axiom rule:}
  Let $\ptm = \pvar$ and $\judgf{\penv,\pvar:\ttm}{\pvar}{\ttm}$.
  Then since $\compatf{\penv}{\psubst}$
  we have that
  $\subs{\ptm}{\asig\psubst}
   = \subs{\pvar}{\asig\psubst}
   = \subs{\pvar}{\psubst}
   \in \semf{\ttm}{\asig}$.

\Case[]{Introduction of implication:}
  Let $\ptm = \plamf{\pvar}{\ptmtwo}$ and $\ttm = \ttmtwo\imp\ttmthree$,
  and let
    $\judgf{\penv}{\plamf{\pvar}{\ptmtwo}}{\ttmtwo\imp\ttmthree}$
  be derived from
    $\judgf{\penv,\pvar:\ttmtwo}{\ptmtwo}{\ttmthree}$.
  Note that by definition of $\semf{\ttmtwo\imp\ttmthree}{\asig}$,
  in order to show that
  $\subs{(\plamf{\pvar}{\ptmtwo})}{\asig\psubst} \in \semf{\ttmtwo\imp\ttmthree}{\asig}$
  we have to show two conditions,
  first that
  $\verif{\subs{(\ttmtwo\imp\ttmthree)}{\asig}}\subs{(\plamf{\pvar}{\ptmtwo})}{\asig\psubst} \tofs \iunit$
  and, second,
  that $\subs{(\plamf{\pvar}{\ptmtwo})}{\asig\psubst}\,\utm \in \semf{\ttmthree}{\asig}$
  holds for every $\utm \in \semf{\ttmtwo}{\asig}$.
  Let us show each of the two conditions:
  \begin{enumerate}
  \item
    To show the first condition,
    recall that
    $\semf{\ttmtwo}{\asig}$ is a $\subs{\ttmtwo}{\asig}$-\verificationCandidate
    by~\cref{lem:f:adequacy_for_propositions},
    so in particular $\gen{\subs{\ttmtwo}{\asig}} \in \semf{\ttmtwo}{\asig}$.
    Consider the \proofSubstitution
    $\psubstext \eqdef \psubst\extsub{\pvar}{\gen{\subs{\ttmtwo}{\asig}}}$
    and observe that
    $\compatf{\penv,\pvar:\ttmtwo}{\psubstext}$
    because $\subs{\pvar}{\psubstext} = \gen{\subs{\ttmtwo}{\asig}} \in \semf{\ttmtwo}{\asig}$.
    Applying the \ih, this implies that
      $\subs{\ptmtwo}{\asig\psubstext} \in \semf{\ttmthree}{\asig}$,
    and since
    $\semf{\ttmthree}{\asig}$ is a $\subs{\ttmthree}{\asig}$-\verificationCandidate
    by~\cref{lem:f:adequacy_for_propositions},
    this means in particular that
      $\verif{\subs{\ttmthree}{\asig}} \subs{\ptmtwo}{\asig\psubstext} \tofs \iunit$.
    Then:
    \[
      \begin{array}{rcl}
      &&
        \verif{\subs{(\ttmtwo\imp\ttmthree)}{\asig}}\subs{(\plamf{\pvar}{\ptmtwo})}{\asig\psubst}
      \\
      & = &
        \verif{\subs{\ttmtwo}{\asig}\imp\subs{\ttmthree}{\asig}}(\plamf{\pvar}{\subs{\ptmtwo}{\asig\psubst\extsub{\pvar}{\pvar}}})
      \\
      & \tof &
        \verif{\subs{\ttmthree}{\asig}}(
          (\plamf{\pvar}{\subs{\ptmtwo}{\asig\psubst\extsub{\pvar}{\pvar}}})
          \gen{\subs{\ttmtwo}{\asig}}
        )
      \\
      & \tof &
        \verif{\subs{\ttmthree}{\asig}}(
          \subs{\ptmtwo}{\asig\psubst\extsub{\pvar}{\pvar}}
            \sub{\pvar}{\gen{\subs{\ttmtwo}{\asig}}}
        )
      \\
      & = &
        \verif{\subs{\ttmthree}{\asig}}
          \subs{\ptmtwo}{\asig\psubstext}
      \\
      & \tofs &
        \iunit
      \end{array}
    \]
  \item
    To show the second condition, let $\utm$ be an arbitrary \metaterm
    such that $\utm \in \semf{\ttmtwo}{\asig}$.
    Consider the \proofSubstitution
    $\psubstext \eqdef \psubst\extsub{\pvar}{\utm}$
    and observe that
    $\compatf{\penv,\pvar:\ttmtwo}{\psubstext}$
    because $\subs{\pvar}{\psubstext} = \utm \in \semf{\ttmtwo}{\asig}$.
    Applying the \ih, this implies that
      $\subs{\ptmtwo}{\asig\psubstext} \in \semf{\ttmthree}{\asig}$.
    Then we have that:
    \[
      \begin{array}{rcl}
      &&
        \subs{(\plamf{\pvar}{\ptmtwo})}{\asig\psubst}\,\utm
      \\
      & = &
        (\plamf{\pvar}{\subs{\ptmtwo}{\asig\psubst\extsub{\pvar}{\pvar}}})\,\utm
      \\
      & \tof &
        \subs{\ptmtwo}{\asig\psubst\extsub{\pvar}{\pvar}}\sub{\pvar}{\utm}
      \\
      & = &
        \subs{\ptmtwo}{\asig\psubst\extsub{\pvar}{\utm}}
      \\
      & = &
        \subs{\ptmtwo}{\asig\psubstext}
      \end{array}
    \]
    Since
    $\semf{\ttmtwo}{\asig}$ is a $\subs{\ttmtwo}{\asig}$-\verificationCandidate
    by~\cref{lem:f:adequacy_for_propositions},
    in particular it is closed by $\tof$-convertibility,
    and since we have that
    $\subs{(\plamf{\pvar}{\ptmtwo})}{\asig\psubst}\,\utm
     \tofs
     \subs{\ptmtwo}{\asig\psubstext} \in \semf{\ttmthree}{\asig}$
    we conclude that
    $\subs{(\plamf{\pvar}{\ptmtwo})}{\asig\psubst}\,\utm \in \semf{\ttmthree}{\asig}$.
  \end{enumerate}

\Case[]{Elimination of implication:}
  Let $\ptm = \ptmtwo\,\ptmthree$,
  and let 
    $\judgf{\penv}{\ptmtwo\,\ptmthree}{\ttm}$
  be derived from
    $\judgf{\penv}{\ptmtwo}{\ttmtwo\imp\ttm}$
  and
    $\judgf{\penv}{\ptmthree}{\ttmtwo}$.
  By \ih we have that
  $\subs{\ptmtwo}{\asig\psubst} \in \semf{\ttmtwo\imp\ttm}{\asig}$
  and that
  $\subs{\ptmthree}{\asig\psubst} \in \semf{\ttmtwo}{\asig}$.
  Moreover, by definition of $\semf{\ttmtwo\imp\ttm}{\asig}$,
  the fact that
  $\subs{\ptmtwo}{\asig\psubst} \in \semf{\ttmtwo\imp\ttm}{\asig}$
  entails that
  $\subs{\ptmtwo}{\asig\psubst}\,\utm \in \semf{\ttm}{\asig}$
  holds for every $\utm \in \semf{\ttmtwo}{\asig}$.
  In particular,
  $\subs{\ptmtwo}{\asig\psubst}\,\subs{\ptmthree}{\asig\psubst} \in \semf{\ttm}{\asig}$,
  \ie
  $\subs{\ptm}{\asig\psubst} \in \semf{\ttm}{\asig}$.

\Case[]{Introduction of second-order universal quantifier:}
  Let $\ptm = \pallfi{\tvar}{\ptmtwo}$
  and $\ttm = \allf{\tvar}{\ttmtwo}$,
  and let
    $\judgf{\penv}{\pallfi{\tvar}{\ptmtwo}}{\allf{\tvar}{\ttmtwo}}$
  be derived from
    $\judgf{\penv}{\ptmtwo}{\ttmtwo}$,
  where $\tvar\notin\fv{\penv}$.
  Note that, by definition of $\semf{\allf{\tvar}{\ttmtwo}}{\asig}$,
  in order to show that
  $\subs{(\pallfi{\tvar}{\ptmtwo})}{\asig\psubst} \in \semf{\allf{\tvar}{\ttmtwo}}{\asig}$
  we have to show two conditions,
  first that
    $\verif{\subs{(\allf{\tvar}{\ttmtwo})}{\asig}}{
       \subs{(\pallfi{\tvar}{\ptmtwo})}{\asig\psubst}
     } \tofs \iunit$
  and, second, that
  $\subs{(\pallfi{\tvar}{\ptmtwo})}{\asig\psubst}\,\ttmthree
   \in \semf{\ttmtwo}{\asig\extsub{\tvar}{(\ttmthree,\vcset)}}$
  for every \logicalTerm $\ttmthree$
  and every $\vcset \in \VC{\ttmthree}$.
  \begin{enumerate}
  \item
    To show the first condition,
    begin by considering the \candidateAssignment
    $\asigext \eqdef \asig\extsub{\tvar}{(\teig,\semf{\teig}{\asig})}$.
    Recall that $\compatf[\asig]{\penv}{\psubst}$ holds by hypothesis,
    and observe that $\compatf[\asigext]{\penv}{\psubst}$ also holds,
    because if $\pvar:\ttmthree \in \penv$
    then by hypothesis
    we know that $\subs{\pvar}{\psubst} \in \semf{\ttmthree}{\asig}$
    and the application of the introduction rule for $\forall$
    ensures that $\tvar \notin \fv{\ttmthree}$;
    hence $\asig$ and $\asigext$ agree on the free variables of $\ttmthree$,
    so $\subs{\pvar}{\psubst} \in \semf{\ttmthree}{\asigext}$
    by irrelevance~(\cref{lem:f:irrelevance}).
    By \ih we obtain that $\subs{\ptmtwo}{\asigext\psubst} \in \semf{\ttmtwo}{\asigext}$,
    and since $\semf{\ttmtwo}{\asigext}$
    is a $\subs{\ttmtwo}{\asigext}$-\verificationCandidate
    by~\cref{lem:f:adequacy_for_propositions},
    this means in particular that
    $\verif{\subs{\ttmtwo}{\asigext}}{\subs{\ptmtwo}{\asigext\psubst}} \tofs \iunit$.
    Let us write $\tsubst \eqdef \tilde{\asig}\extsub{\tvar}{\tvar}$
    be the \logicalSubstitution corresponding to the first projection
    of $\asig$ but fixing $\tvar$,
    and note that
    $\subs{(\allf{\tvar}{\ttmtwo})}{\asig}
     = \allf{\tvar}{\subs{\ttmtwo}{\tsubst}}$
    and
    $\subs{(\pallfi{\tvar}{\ptmtwo})}{\asig}
     = \pallfi{\tvar}{\subs{\ptmtwo}{\tsubst}}$.
    Moreover, assume by $\alpha$-conversion that $\tvar$ does not occur
    free in the image of $\psubst$. Then:
    \[
      \begin{array}{rcll}
      &&
        \verif{\subs{(\allf{\tvar}{\ttmtwo})}{\asig}}{
          (\subs{(\pallfi{\tvar}{\ptmtwo})}{\asig\psubst})
        }
      \\
      & = &
        \verif{\allf{\tvar}{\subs{\ttmtwo}{\tsubst}}}{
          (\pallfi{\tvar}{\subs{\ptmtwo}{\tsubst\psubst}})
        }
      \\
      & \tof &
        \freshf{\teig}{
          \verif{\subs{\ttmtwo}{\tsubst}\sub{\tvar}{\teig}}{
            ((\pallfi{\tvar}{\subs{\ptmtwo}{\tsubst\psubst}})\,\teig)
          }
        }
      \\
      & \tof &
        \freshf{\teig}{
          \verif{\subs{\ttmtwo}{\tsubst\sub{\tvar}{\teig}}}{
            (\subs{\ptmtwo}{\tsubst\psubst}\sub{\tvar}{\teig})
          }
        }
      \\
      & \tof &
        \freshf{\teig}{
          \verif{\subs{\ttmtwo}{\asigext}}{
            (\subs{\ptmtwo}{\asigext\psubst})
          }
        }
      \\
      & \tofs &
        \freshf{\teig}{\iunit}
        \HS\text{since
             $\verif{\subs{\ttmtwo}{\asigext}}{\subs{\ptmtwo}{\asigext\psubst}} \tofs \iunit$
           }
      \\
      & \tof &
        \iunit
      \end{array}
    \]
  \item
    To show the second condition,
    let $\ttmthree$ be a \logicalTerm and $\vcset \in \VC{\ttmthree}$,
    and consider the \candidateAssignment
    $\asigext \eqdef \asig\extsub{\tvar}{(\ttmthree,\vcset)}$.
    Recall that $\compatf[\asig]{\penv}{\psubst}$ holds by hypothesis,
    and observe that $\compatf[\asigext]{\penv}{\psubst}$ also holds,
    because if $\pvar:\ttmfour \in \penv$
    then by hypothesis
    we know that $\subs{\pvar}{\psubst} \in \semf{\ttmfour}{\asig}$
    and the application of the introduction rule for $\forall$ ensures
    that $\tvar \notin \fv{\ttmfour}$;
    hence $\asig$ and $\asigext$ agree on the free variables of $\ttmfour$,
    so $\subs{\pvar}{\psubst} \in \semf{\ttmfour}{\asigext}$
    by irrelevance~(\cref{lem:f:irrelevance}).
    Applying the \ih, this implies that
    $\subs{\ptmtwo}{\asigext\psubst} \in \semf{\ttmtwo}{\asigext}$.
    Note that we may assume by $\alpha$-conversion
    that $\tvar$ does not occur free in the image of $\psubst$.
    Moreover, consider the \logicalSubstitution
    $\tsubst = \tilde{\asig}\extsub{\tvar}{\tvar}$
    and note that
    $\subs{(\pallfi{\tvar}{\ptmtwo})}{\asig}
    = \pallfi{\tvar}{\subs{\ptmtwo}{\tsubst}}$.
    Then we have that:
    \[
      \begin{array}{rcll}
      &&
        \subs{(\pallfi{\tvar}{\ptmtwo})}{\asig\psubst}\,\ttmthree
      \\
      & = &
        (\pallfi{\tvar}{\subs{\ptmtwo}{\tsubst\psubst}})\,\ttmthree
      \\
      & \tof &
        \subs{\ptmtwo}{\tsubst\psubst}\sub{\tvar}{\ttmthree}
      \\
      & = &
        \subs{(\subs{\ptmtwo}{\tsubst}\sub{\tvar}{\ttmthree})}{\psubst}
      \\
      & = &
        \subs{\ptmtwo}{\asigext\psubst}
      \end{array}
    \]
    Since $\semf{\ttmtwo}{\asigext}$ is a
    $\subs{\ttmtwo}{\asigext}$-\verificationCandidate
    by~\cref{lem:f:adequacy_for_propositions},
    in particular it is closed by $\tof$-convertibility,
    and since we have that
    $\subs{(\pallfi{\tvar}{\ptmtwo})}{\asig\psubst}\,\ttmthree
     \tofs
     \subs{\ptmtwo}{\asigext\psubst} \in \semf{\ttmtwo}{\asigext}$,
    we conclude that
    $\subs{(\pallfi{\tvar}{\ptmtwo})}{\asig\psubst}\,\ttmthree
     \in \semf{\ttmtwo}{\asigext}$.
  \end{enumerate}

\Case[]{Elimination of the second-order universal quantifier:}
  Let $\ptm = \ptmtwo\,\ttmthree$
  and let $\ttmtwo,\tvar,\ttmthree$
  be such that $\ttm = \ttmtwo\sub{\tvar}{\ttmthree}$
  and
    $\judgf{\penv}{\ptmtwo\,\ttmthree}{\ttmtwo\sub{\tvar}{\ttmthree}}$
  is derived from
    $\judgf{\penv}{\ptmtwo}{\allf{\tvar}{\ttmtwo}}$.
  By \ih we have that
    $\subs{\ptmtwo}{\asig\psubst} \in \semf{\allf{\tvar}{\ttmtwo}}{\asig}$.
  Moreover, by definition of $\semf{\allf{\tvar}{\ttmtwo}}{\asig}$,
  the fact that
    $\subs{\ptmtwo}{\asig\psubst} \in \semf{\allf{\tvar}{\ttmtwo}}{\asig}$
  entails that
    $\subs{\ptmtwo}{\asig\psubst}\,\ttmfour \in \semf{\ttmtwo}{\asig\extsub{\tvar}{(\ttmfour,\vcset)}}$
  for every \logicalTerm $\ttmfour$ and every $\vcset \in \VC{\ttmfour}$.
  In particular, taking $\ttmfour \eqdef \subs{\ttmthree}{\asig}$
  and $\vcset \eqdef \semf{\ttmthree}{\asig}$,
  which is a $\subs{\ttmthree}{\asig}$-\verificationCandidate
  by~\cref{lem:f:adequacy_for_propositions},
  we have that
    $\subs{\ptmtwo}{\asig\psubst}\,\subs{\ttmthree}{\asig} \in \semf{\ttmtwo}{\asig\extsub{\tvar}{(\subs{\ttmthree}{\asig},\semf{\ttmthree}{\asig})}}$.
  Remark that
  $\subs{(\ptmtwo\,\ttmthree)}{\asig\psubst}
  = \subs{\ptmtwo}{\asig\psubst}\,\subs{\ttmthree}{\psubst}$.
  Moreover, note that
  $\semf{\ttmtwo}{\asig\extsub{\tvar}{(\subs{\ttmthree}{\asig},\semf{\ttmthree}{\asig})}}
   =
   \semf{\ttmtwo\sub{\tvar}{\ttmthree}}{\asig}$
  by the substitution lemma~(\cref{lem:f:substitution}).
  Hence
  $\subs{(\ptmtwo\,\ttmthree)}{\asig\psubst} \in \semf{\ttmtwo\sub{\tvar}{\ttmthree}}{\asig}$,
  as required.
\end{proof}

\subsection{Second-Order: Completeness (\cref{thm:f:completeness})}
  \label{a:sec:f:completeness}
  
\begin{lemma}[Characterization of normal forms]
\label{lem:f:characterization_of_nfs}
A pure \metaterm $\utm$ is in $\beta$-normal form
if and only if it is a normal \proofTerm, according to
the grammar of normal forms given in \cref{sec:second_order:completeness}.
\end{lemma}
\begin{proof}
Both implications are routine by induction on $\utm$.
\end{proof}

\begin{remark}[Determinism of \argument matching]
\label{rem:f:match_deterministic}
The predicate $\match{\ttm}{\argus}{\ttms}{\ttm}$
depends functionally on $\ttm$ and $\argus$.
More precisely,
if $\match{\ttm}{\argus}{\ttms_1}{\ttm_1}$
and $\match{\ttm}{\argus}{\ttms_2}{\ttm_2}$
hold, then $\ttms_1 = \ttms_2$ and $\ttm_1 = \ttm_2$.
\end{remark}

\begin{remark}[\Argument matching preserves the length]
\label{rem:f:match_preserves_length}
If $\match{\ttm}{\argus}{\ttms}{\ttm}$ holds,
then $\size{\ttms} = \size{\argus}$,
\ie they are lists of the same length.
\end{remark}

\begin{lemma}[\Argument matching of a concatenation]
\label{lem:f:match_concat}
Let $\ttm,\ttm'$ be \logicalTerms and $\ttms$ be a sequence of \logicalTerms.
Moreover, let $\argus_1,\argus_2$ be two sequences of \arguments, and let us write
$\argus_1,\argus_2$ for their concatenation.
Then the following are equivalent:
\begin{enumerate}
\item
  $\match{\ttm}{(\argus_1,\argus_2)}{\ttms}{\ttm'}$
\item
  There exist sequences of \logicalTerms $\ttms_1,\ttms_2$
  and a \logicalTerm $\ttmthree$ such that
  $\ttms = \ttms_1,\ttms_2$
  and
  $\match{\ttm}{\argus_1}{\ttms_1}{\ttm''}$
  and
  $\match{\ttm''}{\argus_2}{\ttms_2}{\ttm'}$.
\end{enumerate}
\end{lemma}
\begin{proof}
By induction on the length of $\argus_1$.
We only prove the $(\Rightarrow)$ direction; the $(\Leftarrow)$ direction
is similar going backwards.
We consider three cases depending on whether $\argus_1$ is empty,
if it starts with a \metaterm, or it starts with a \logicalTerm:

If $\argus_1$ is \textbf{empty}, \ie $\argus_1 = \emptyset$:
  We know that $\match{\ttm}{\argus_2}{\ttms}{\ttm'}$.
  Take $\ttms_1 := \emptyset$, $\ttms_2 := \ttms$, and $\ttm'' := \ttm$.
  Then $\match{\ttm}{\argus_1}{\ttms_1}{\ttm''}$ 
  holds because $\match{\ttm}{\emptyset}{\emptyset}{\ttm}$ holds,
  and $\match{\ttm''}{\argus_2}{\ttms_2}{\ttm'}$
  holds because $\match{\ttm}{\argus_2}{\ttms}{\ttm'}$ holds.

If $\argus_1$ \textbf{starts with a \metaterm}, \ie $\argus_1 = (\utm,\argus'_1)$:
  Suppose that $\match{\ttm}{(\utm,\argus'_1,\argus_2)}{\ttms}{\ttm'}$.
  The only way to derive this judgment is with
  $\ttm = (\ttmtwo\imp\ttmthree)$
  and $\match{\ttmthree}{(\argus'_1,\argus_2)}{\ttmthrees}{\ttm'}$
  and $\ttms = (\ttmtwo,\ttmthrees)$.
  By \ih there exist $\ttmthrees_1,\ttmthrees_2,\ttm''$
  such that
  $\ttmthrees = (\ttmthrees_1,\ttmthrees_2)$
  and
  $\match{\ttmthree}{\argus'_1}{\ttmthrees_1}{\ttm''}$
  and
  $\match{\ttm''}{\argus_2}{\ttmthrees_2}{\ttm'}$.
  Taking $\ttms_1 := (\ttmtwo,\ttmthrees_1)$ and $\ttms_2 := \ttmthrees_2$,
  we have that
  $\ttms
   = (\ttmtwo,\ttmthrees)
   = (\ttmtwo,\ttmthrees_1,\ttmthrees_2)
   = (\ttms_1,\ttms_2)$.
  Note that
  $\match{\ttm}{\argus_1}{\ttms_1}{\ttm''}$
  holds,
  that is to say
  $\match{(\ttmtwo\imp\ttmthree)}{(\utm,\argus'_1)}{(\ttmtwo,\ttmthrees_1)}{\ttm''}$
  holds,
  because
  $\match{\ttmthree}{\argus'_1}{\ttmthrees_1}{\ttm''}$
  holds.
  Moreover, $\match{\ttm''}{\argus_2}{\ttms_2}{\ttm'}$ holds
  because $\match{\ttm''}{\argus_2}{\ttmthrees_2}{\ttm'}$ holds.

If $\argus_1$ \textbf{starts with a \logicalTerm}, \ie $\argus_1 = (\ttmtwo,\argus'_1)$:
  Then suppose that $\match{\ttm}{(\ttmtwo,\argus'_1,\argus_2)}{\ttms}{\ttm'}$.
  The only way to derive this judgment is with $\ttm = (\allf{\tvar}{\ttmthree})$
  and $\match{\ttmthree\sub{\tvar}{\ttmtwo}}{(\argus'_1,\argus_2)}{\ttmthrees}{\ttm'}$
  and $\ttms = (\ttmtwo,\ttmthrees)$.
  By \ih there exist $\ttmthrees_1,\ttmthrees_2,\ttm''$ such that
  $\ttmthrees = (\ttmthrees_1,\ttmthrees_2)$
  and
  $\match{\ttmthree\sub{\tvar}{\ttmtwo}}{\argus'_1}{\ttmthrees_1}{\ttm''}$
  and
  $\match{\ttmthree''}{\argus_2}{\ttmthrees_2}{\ttm'}$.
  Taking $\ttms_1 := (\ttmtwo,\ttmthrees_1)$ and $\ttms_2 := \ttmthrees_2$,
  we have that
  $\ttms
   = (\ttmtwo,\ttmthrees)
   = (\ttmtwo,\ttmthrees_1,\ttmthrees_2)
   = (\ttms_1,\ttms_2)$.
  Note that
  $\match{\ttm}{\argus_1}{\ttms_1}{\ttm''}$
  holds,
  that is to say
  $\match{(\allf{\tvar}{\ttmthree})}{(\ttmtwo,\argus'_1)}{(\ttmtwo,\ttmthrees_1)}{\ttm''}$
  holds,
  because
  $\match{\ttmthree\sub{\tvar}{\ttmtwo}}{\argus'_1}{\ttmthrees_1}{\ttm''}$
  holds.
  Moreover,
  $\match{\ttm''}{\argus_2}{\ttms_2}{\ttm'}$
  holds
  because
  $\match{\ttm''}{\argus_2}{\ttmthrees_2}{\ttm'}$ holds.
\end{proof}

\begin{definition}[Coherence of \arguments]
Let $\argu$ and $\argutwo$ be \arguments,
and define $\argu \coherent \argutwo$ to hold if one of the
two following conditions holds:
\begin{enumerate}
\item $\argu$ and $\argutwo$ are both \metaterms (but not necessarily equal).
\item $\argu$ and $\argutwo$ are the same \logicalTerms,
      \ie there exists a \logicalTerm $\ttm$ such that $\argu = \ttm = \argutwo$.
\end{enumerate}
Moreover, if $\argus = (\argu_1,\hdots,\argu_n)$
and $\argutwos = (\argutwo_1,\hdots,\argutwo_n)$,
we write $\argus \coherent \argutwos$
if $\argu_i \coherent \argutwo_i$ holds for every index $1 \leq i \leq n$.
\end{definition}

\begin{remark}
\label{rem:f:match_coherent}
The question of whether the predicate $\match{\ttm}{\argus}{\ttms}{\ttm'}$ holds
does \textbf{not} depend on the \metaterms appearing in the list of \arguments $\argus$.
More precisely, if $\argus \coherent \argutwos$
then $\match{\ttm}{\argus}{\ttms}{\ttm'}$
if and only if $\match{\ttm}{\argutwos}{\ttms}{\ttm'}$.
\end{remark}

\begin{remark}
$\sig[\penv,\pvar:\ttm]{\pvar}{\ttm'} = \minof{\size{\ttm}}{\size{\ttm'}}$.
\end{remark}

\begin{lemma}
\label{lemma:f:size_weakening}
If $\pvar \notin \fv{\unf}$
then $\sig[\penv]{\unf}{\ttm} = \sig[\penv,\pvar:\ttmtwo]{\unf}{\ttm}$
\end{lemma}
\begin{proof}
Straightforward by induction on $\unf$.
\end{proof}

\begin{lemma}[Proof size of application to \proofVariable]
\label{lem:f:size_proof_applications}
  Let $\unf$ be a normal \proofTerm,
  and let $\pvar \notin \fv{\unf}$. 
  Then there exists a $\beta$-normal form $\unftwo$
  such that $\unf\,\pvar \tobetas \unftwo$
  and $\sig{\unf}{\ttm\imp\ttmtwo} \geq \sig[\penv,\pvar:\ttm]{\unftwo}{\ttmtwo}$.
\end{lemma}
\begin{proof}
We consider three cases depending on the shape of $\unf$,
according to \cref{lem:f:characterization_of_nfs}.

If $\unf$ is a \textbf{\headForm}:
  Let $\unf = \uhd\,\argus$ where $\argus = (\argunf_1,\hdots,\argunf_n)$.
  Take $\unftwo := \uhd\,\argus\,\pvar$, and observe that it is in normal form.
  Note that $\unf\,\pvar = \unftwo$.
  It suffices to show that
  $\sig{\unf}{\ttm\imp\ttmtwo} \geq \sig[\penv,\pvar:\ttm]{\unftwo}{\ttmtwo}$.
  Let $\mathcal{P}$ and $\mathcal{Q}$ be the following conditions:
  \[
    \begin{array}{lll}
      \mathcal{P} & \equiv &
      \begin{array}[t]{l}
        \text{$\uhd$ is a variable, $\uhd = \pvartwo$, and there exist $\ttmthree,\ttmthrees,\ttmthree'$}
      \\
        \text{such that
          $\pvartwo : \ttmthree \in \penv$ and $\match{\ttmthree}{\argus}{\ttmthrees}{\ttmthree'}$.
        }
      \end{array}
    \\
      \mathcal{Q} & \equiv &
      \begin{array}[t]{l}
        \text{$\uhd$ is a variable, $\uhd = \pvartwo$, and there exist $\ttmthree,\ttmfours,\ttmfour'$}
      \\
        \text{such that
          $\pvartwo : \ttmthree \in \penv$ and $\match{\ttmthree}{(\argus,\pvar)}{\ttmfours}{\ttmfour'}$.
        }
      \end{array}
    \end{array}
  \]
  Note that if $\mathcal{Q}$ does not hold then it is immediate to conclude,
  since
  $\sig[\penv,\pvar:\ttm]{\unftwo}{\ttmtwo}
  = \sig[\penv,\pvar:\ttm]{\pvartwo\,\argus\,\pvar}{\ttmtwo}
  = 0$,
  so
  $\sig{\unf}{\ttm\imp\ttmtwo}
   \geq 0
   = \sig[\penv,\pvar:\ttm]{\unftwo}{\ttmtwo}$.
  Therefore, we may assume that $\mathcal{Q}$ holds.
  We claim that then $\mathcal{P}$ must also hold.
  Indeed, since $\mathcal{Q}$ holds
  we have that $\match{\ttmthree}{(\argus,\pvar)}{\ttmfours}{\ttmfour'}$,
  which implies that $\mathcal{P}$ holds by \cref{lem:f:match_concat}.

  Since $\mathcal{P}$ and $\mathcal{Q}$ both hold,
  the situation must be (by \cref{lem:f:match_concat} and \cref{rem:f:match_deterministic})
  that $\uhd = \pvartwo$ and $\pvartwo : \ttmthree \in \penv$
  and $\match{\ttmthree}{\argus}{\ttmthrees}{(\ttm'\imp\ttmtwo')}$
  and $\match{(\ttm'\imp\ttmtwo')}{\pvar}{\ttm'}{\ttmtwo'}$,
  where moreover we know that $\ttmthrees$ is of the form
  $\ttmthrees = (\ttmthree_1,\hdots,\ttmthree_n)$
  by \cref{rem:f:match_preserves_length}.
  Hence:
  \[
    \begin{array}{rcll}
    &&
      \sig{\unf}{\ttm\imp\ttmtwo}
    \\
    & = &
      \sig{\pvartwo\,\argus}{\ttm\imp\ttmtwo}
    \\
    & = &
      \minof{\size{\ttm\imp\ttmtwo}}{\size{\ttm'\imp\ttmtwo'}}
        + \sum_{i=1}^{n} \sig{\argunf_i}{\ttmthree_i}
    \\
    & = &
      \minof{\size{\ttm\imp\ttmtwo}}{\size{\ttm'\imp\ttmtwo'}}
        + \sum_{i=1}^{n} \sig[\penv,\pvar:\ttm]{\argunf_i}{\ttmthree_i}
    \\
    && \HS\text{by \cref{lemma:f:size_weakening}}
    \\
    & = &
      \minof{1+\size{\ttm}+\size{\ttmtwo}}{1+\size{\ttm'}+\size{\ttmtwo'}}
        + \sum_{i=1}^{n} \sig[\penv,\pvar:\ttm]{\argunf_i}{\ttmthree_i}
    \\
    & = &
      \minof{\size{\ttm}}{\size{\ttm'}}
        + \minof{\size{\ttmtwo}}{\size{\ttmtwo'}}
        + \sum_{i=1}^{n} \sig[\penv,\pvar:\ttm]{\argunf_i}{\ttmthree_i}
    \\
    & = &
      \minof{\size{\ttmtwo}}{\size{\ttmtwo'}}
        + \sum_{i=1}^{n} \sig[\penv,\pvar:\ttm]{\argunf_i}{\ttmthree_i}
        + \sig[\penv,\pvar:\ttm]{\pvar}{\ttm'}
    \\
    & = &
      \sig[\penv,\pvar:\ttm]{\pvartwo\,\argus\,\pvar}{\ttmtwo}
    \\
    & = &
      \sig[\penv,\pvar:\ttm]{\unftwo}{\ttmtwo}
    \end{array}
  \]

If $\unf$ is a \textbf{\proofAbstraction}:
  Let $\unf = \plamf{\pvar}{\unftwo}$.
  Note that $\unf\,\pvar \to \unftwo$ in exactly one step.
  Moreover
  $\sig{\unf}{\ttm\imp\ttmtwo}
   = \sig{\plamf{\pvar}{\unftwo}}{\ttm\imp\ttmtwo}
   = \sig[\penv,\pvar:\ttm]{\unftwo}{\ttmtwo}
  $.

If $\unf$ is a \textbf{\logicalAbstraction}:
  Let $\unf = \pallfi{\tvar}{\unfthree}$.
  Take $\unftwo := (\pallfi{\tvar}{\unfthree})\,\tvar$, and observe that it is in normal form.
  Note that $\unf\,\tvar = \unftwo$.
  Moreover
  $\sig{\unf}{\ttm\imp\ttmtwo}
   = \sig{\pallfi{\tvar}{\unfthree}}{\ttm\imp\ttmtwo}
   = 0
   = \sig[\penv,\pvar:\ttm]{(\pallfi{\tvar}{\unfthree})\,\pvar}{\ttmtwo}
   = \sig[\penv,\pvar:\ttm]{\unftwo}{\ttmtwo}
  $.
\end{proof}

\begin{lemma}[Substitution by \logicalEigenvariable]
\label{lem:f:subst_by_teig}
\quad
\begin{enumerate}
\item
  If
    $\match{\ttm}{\argus}{\ttms}{\ttm'}$
  then
    $\match{\ttm\sub{\tvar}{\teig}}{\argus\sub{\tvar}{\teig}}{\ttms\sub{\tvar}{\teig}}{\ttm'\sub{\tvar}{\teig}}$.
\item
  $\sig{\unf}{\ttm} = \sig[\penv\sub{\tvar}{\teig}]{\unf\sub{\tvar}{\teig}}{\ttm\sub{\tvar}{\teig}}$
\end{enumerate}
\end{lemma}
\begin{proof}
The first item is straightforward by induction on the derivation of the
judgment $\match{\ttm}{\argus}{\ttms}{\ttm'}$.
The second item is straightforward by induction on $\unf$,
considering each case in \cref{def:f:proof_size}
and resorting to the first item when $\unf$ is of the form $\pvar\argunf_1\hdots\argunf_n$.
\end{proof}

\begin{lemma}[Proof size of application to \logicalVariable]
\label{lem:f:size_logical_applications}
Let $\unf$ be a normal \proofTerm,
and suppose that $\tvar,\teig\notin\fv{\penv}$.
Then there exists a $\beta$-normal form $\unftwo$ 
such that $\unf\,\teig \tobetas \unftwo$,
and $\sig{\unf}{\allf{\tvar}{\ttmtwo}} \geq \sig[\penv]{\unftwo}{\ttmtwo\sub{\tvar}{\teig}}$.
\end{lemma}
\begin{proof}
We consider three cases depending on the shape of $\unf$,
according to \cref{lem:f:characterization_of_nfs}.

If $\unf$ is a \textbf{\headForm}:
  Let $\unf = \uhd\,\argus$
  where $\argus = (\argunf_1,\hdots,\argunf_n)$.
  Take $\unftwo := \uhd\,\argus\,\teig$, and observe that it is in normal form.
  Then we have that $\unf\,\teig = \unftwo$,
  and it suffices to show that
  $\sig{\unf}{\allf{\tvar}{\ttmtwo}}
   \geq \sig[\penv]{\unftwo}{\ttmtwo\sub{\tvar}{\teig}}$.
  Let $\mathcal{P}$ and $\mathcal{Q}$ be the following conditions:
  \[
    \begin{array}{lll}
      \mathcal{P} & \equiv &
      \begin{array}[t]{l}
        \text{$\uhd$ is a variable, $\uhd = \pvar$, and there exist $\ttmthree,\ttmthrees,\ttmthree'$}
      \\
        \text{such that
          $\pvar : \ttmthree \in \penv$ and $\match{\ttmthree}{\argus}{\ttmthrees}{\ttmthree'}$.
        }
      \end{array}
    \\
      \mathcal{Q} & \equiv &
      \begin{array}[t]{l}
        \text{$\uhd$ is a variable, $\uhd = \pvar$, and there exist $\ttmthree,\ttmfours,\ttmfour'$}
      \\
        \text{such that
          $\pvar : \ttmthree \in \penv$ and $\match{\ttmthree}{(\argus,\teig)}{\ttmfours}{\ttmfour'}$.
        }
      \end{array}
    \end{array}
  \]
  Note that if $\mathcal{Q}$ does not hold then it is immediate to conclude,
  since
  $\sig{\unftwo}{\ttmtwo\sub{\tvar}{\teig}}
  = \sig{\pvar\,\argus\,\teig}{\ttmtwo\sub{\tvar}{\teig}}
  = 0$,
  so
  $\sig{\unf}{\allf{\tvar}{\ttmtwo}}
   \geq 0
   = \sig{\unftwo}{\ttmtwo\sub{\tvar}{\teig}}$.
  Therefore, we may assume that $\mathcal{Q}$ holds.
  We claim that then $\mathcal{P}$ must also hold.
  Indeed, since $\mathcal{Q}$ holds
  we have that $\match{\ttmthree}{(\argus,\teig)}{\ttmfours}{\ttmfour'}$,
  which implies that $\mathcal{P}$ holds by \cref{lem:f:match_concat}.

  Since $\mathcal{P}$ and $\mathcal{Q}$ both hold,
  the situation must be (by \cref{lem:f:match_concat} and \cref{rem:f:match_deterministic})
  that $\uhd = \pvartwo$ and $\pvartwo : \ttmthree \in \penv$
  and $\match{\ttmthree}{\argus}{\ttmthrees}{(\allf{\tvar}{\ttmtwo'})}$
  and $\match{(\allf{\tvar}{\ttmtwo'})}{\teig}{\teig}{\ttmtwo'\sub{\tvar}{\teig}}$,
  where moreover we know that $\ttmthrees$ is of the form
  $\ttmthrees = (\ttmthree_1,\hdots,\ttmthree_n)$
  by \cref{rem:f:match_preserves_length}.
  Hence:
  \[
    \begin{array}{rcll}
    &&
      \sig{\unf}{\allf{\tvar}{\ttmtwo}}
    \\
    & = &
      \sig{\pvar\,\argus}{\allf{\tvar}{\ttmtwo}}
    \\
    & = &
      \minof{\size{\allf{\tvar}{\ttmtwo}}}{\size{\allf{\tvar}{\ttmtwo'}}}
      + \sum_{i=1}^{n} \sig{\argunf_i}{\ttmthree_i}
    \\
    & = &
      \minof{1+\size{\ttmtwo}}{1+\size{\ttmtwo'}}
      + \sum_{i=1}^{n} \sig{\argunf_i}{\ttmthree_i}
    \\
    & > &
      \minof{\size{\ttmtwo}}{\size{\ttmtwo'}}
      + \sum_{i=1}^{n} \sig{\argunf_i}{\ttmthree_i}
    \\
    & = &
      \minof{\size{\ttmtwo\sub{\tvar}{\teig}}}{\size{\ttmtwo'\sub{\tvar}{\teig}}}
      + \sum_{i=1}^{n} \sig{\argunf_i}{\ttmthree_i}
    \\
    & = &
      \minof{\size{\ttmtwo\sub{\tvar}{\teig}}}{\size{\ttmtwo'\sub{\tvar}{\teig}}}
      + \sum_{i=1}^{n} \sig{\argunf_i}{\ttmthree_i}
      + \sig{\teig}{\teig}
    \\
    & = &
      \sig{\pvar\,\argus\,\teig}{\ttmtwo\sub{\tvar}{\teig}}
    \\
    & = &
      \sig{\unftwo}{\ttmtwo\sub{\tvar}{\teig}}
    \end{array}
  \]

If $\unf$ is a \textbf{\proofAbstraction}:
  Let $\unf = \plamf{\pvar}{\unfthree}$.
  Take $\unftwo := (\plamf{\pvar}{\unfthree})\,\teig$ 
  and observe that it is in normal form.
  Note that $\unf\,\pvar = \unftwo$.
  Then
  $\sig{\unf}{\allf{\pvar}{\ttmtwo}}
  = \sig{\plamf{\pvar}{\unfthree}}{\allf{\pvar}{\ttmtwo}}
  = 0
  = \sig{(\plamf{\pvar}{\unfthree})\,\teig}{\ttmtwo\sub{\tvar}{\teig}}
  = \sig{\unftwo}{\ttmtwo\sub{\tvar}{\teig}}$.

If $\unf$ is a \textbf{\logicalAbstraction}:
  Let $\unf = \pallfi{\tvar}{\unfthree}$.
  Take $\unftwo := \unfthree\sub{\tvar}{\teig}$,
  where we remark that $\unfthree\sub{\tvar}{\teig}$ is in normal form
  because $\unfthree$ is pure.
  Note that $\unf\,\teig \to \unfthree\sub{\tvar}{\teig} = \unftwo$ 
  in exactly one step.
  Then
  $\sig{\unf}{\allf{\tvar}{\ttmtwo}}
   = \sig{\pallfi{\tvar}{\unfthree}}{\allf{\tvar}{\ttmtwo}}
   = \sig{\unfthree}{\ttmtwo}
   \eqby{\text{\cref{lem:f:subst_by_teig}}}
     \sig{\unfthree\sub{\tvar}{\teig}}{\ttmtwo\sub{\tvar}{\teig}}
   = \sig{\unftwo}{\ttmtwo\sub{\tvar}{\teig}}$.
\end{proof}

\begin{lemma}[Proof size of \arguments]
\label{lem:f:size_arguments}
Let $\argus = (\argu_1,\hdots,\argu_n)$ be a sequence of normal \arguments,
and let $\ttms = (\ttm_1,\hdots,\ttm_n)$ be a sequence of \logicalTerms.
Suppose that $\pvar:\ttm\in\penv$
and that $\match{\ttm}{\argus}{\ttms}{\teig}$.
Then for every index $1 \leq i \leq n$
we have that $\sig{\pvar\,\argus}{\teig} > \sig{\argu_i}{\ttm_i}$.
\end{lemma}
\begin{proof}
Immediate, since by definition:
\[
  \sig{\pvar\,\argus}{\teig}
  = \size{\teig} + \sum_{i=1}^{n} \sig{\argu_i}{\ttm_i}
  > \sig{\argu_i}{\ttm_i}
\]
\end{proof}

\begin{lemma}[Generator matching]
\label{lem:f:verif_gen_match}
Let $\argus = (\argu_1,\hdots,\argu_n)$ be a sequence of \arguments.
If $\verif{\teig}{(\gen{\ttm}\,\argus)} \tos \iunit$
then there exists a sequence of \logicalTerms $\ttms = (\ttm_1,\hdots,\ttm_n)$
such that $\match{\ttm}{\argus}{\ttms}{\teig}$.
Moreover, for every index $1 \leq i \leq n$
such that the \argument $\argu_i$ is a \metaterm,
we have that $\verif{\ttm_i}{\argu_i} \tos \iunit$.
\end{lemma}
\begin{proof}
We proceed by induction on the number of \arguments $n$,
\ie the length of $\argus$.
We consider four cases, depending on the shape of $\ttm$:

If $\ttm$ is a \textbf{\logicalVariable}, \ie $\ttm = \tvartwo$:
  We claim that this case is impossible.
  Indeed, we have that
  $\verif{\teig}{(\gen{\tvartwo}\,\argus)} \tos \iunit$.
  But the reducts of
  $\verif{\teig}{(\gen{\tvartwo}\,\argus)}$
  are of the form
  $\verif{\teig}{(\gen{\tvartwo}\,\argus')}$,
  which is a contradiction.
  Observe, in the particular case when $n = 0$,
  that $\verif{\teig}{\gen{\tvartwo}}$ is irreducible
  because $\tvartwo$ is not an \logicalEigenvariable, so $\tvartwo \neq \teig$.

If $\ttm$ is a \textbf{\logicalEigenvariable}, \ie $\ttm = \teigtwo$:
  There are two cases, depending on whether $\teig = \teigtwo$
  or $\teig \neq \teigtwo$.
  The proof when $\teig \neq \teigtwo$ is similar to the case
  in which $\ttm$ is an \logicalEigenvariable,
  \ie it is impossible to construct a reduction
  $\verif{\teig}{(\gen{\teigtwo}\,\argus)} \tos \iunit$.
  If $\teig = \teigtwo$,
  then we claim that $n = 0$.
  Indeed, if $n > 0$,
  all the reducts of
  $\verif{\teig}{(\gen{\teig}\,\argus)}$
  are of the form
  $\verif{\teig}{(\gen{\teig}\,\placeholder)}$,
  which is a contradiction.
  Hence $n = 0$
  and we have that $\match{\teig}{\emptyset}{\emptyset}{\teig}$.

If $\ttm$ is an \textbf{implication}, \ie $\ttm = (\ttmtwo\imp\ttmthree)$:
  Then
  $\verif{\teig}{(\gen{\ttmtwo\imp\ttmthree}\,\argus)} \tos \iunit$.
  Note that
  $\verif{\teig}{(\gen{\ttmtwo\imp\ttmthree}\,\argus)}
   \to
   \verif{\teig}{((\plamf{\pvar}{\eunit{\verif{\ttmtwo}{\pvar}}{\gen{\ttmthree}}})\,\argus)}
  $
  so by confluence~(\cref{prop:f:confluence}) 
  we have that
  $\verif{\teig}{((\plamf{\pvar}{\eunit{\verif{\ttmtwo}{\pvar}}{\gen{\ttmthree}}})\,\argus)}
   \tos \iunit$.
  We claim that $\argus$ must be of the form $\argus = (\utm,\argutwos)$ where
  the first element $\utm$ is a \metaterm, and $\argutwos$ is the tail.
  Indeed, note that $\argus$ cannot be empty
  because we would have that
  $\verif{\teig}{(\plamf{\pvar}{\eunit{\verif{\ttmtwo}{\pvar}}{\gen{\ttmthree}}})}
   \tos \iunit$,
  but this is impossible because all the reducts of
  $\verif{\teig}{(\plamf{\pvar}{\eunit{\verif{\ttmtwo}{\pvar}}{\gen{\ttmthree}}})}$
  are of the form
  $\verif{\teig}{(\plamf{\pvar}{\placeholder})}$, which is a contradiction.
  Moreover, $\argus$ cannot be of the form $\argus = (\ttmfour,\argutwos)$
  where $\ttmfour$ is a \logicalTerm, because we would have that
  $\verif{\teig}{((\plamf{\pvar}{\eunit{\verif{\ttmtwo}{\pvar}}{\gen{\ttmthree}}})\,\ttmfour\,\argutwos)}
   \tos \iunit$,
  but there are no rules to contract a pseudo-redex of the form
  $(\plamf{\pvar}{\placeholder})\,\ttmfour$,
  so all the reducts of
  $\verif{\teig}{((\plamf{\pvar}{\eunit{\verif{\ttmtwo}{\pvar}}{\gen{\ttmthree}}})\,\ttmfour\,\argutwos)}$
  are of the form
  $\verif{\teig}{((\plamf{\pvar}{\placeholder})\,\ttmfour\,\placeholder)}$,
  which is a contradiction.

  Therefore, $\argus = (\utm,\argutwos)$ where $\utm$ is a \metaterm,
  and we have that
  \[
    \verif{\teig}{((\plamf{\pvar}{\eunit{\verif{\ttmtwo}{\pvar}}{\gen{\ttmthree}}})\,\utm\,\argutwos)}
   \tos \verif{\teig}{((\eunit{\verif{\ttmtwo}{\utm}}{\gen{\ttmthree}})\,\argutwos)}
  \]
  so by confluence~(\cref{prop:f:confluence}) we have that
  $\verif{\teig}{((\eunit{\verif{\ttmtwo}{\utm}}{\gen{\ttmthree}})\,\argutwos)}
   \tos \iunit$.
  Then we necessarily have that
  $\verif{\ttmtwo}{\utm} \tos \iunit$
  and
  $\verif{\teig}{(\gen{\ttmthree}\,\argutwos)} \tos \iunit$.
  Since $\size{\argus} = 1 + \size{\argutwos}$,
  we may resort to the \ih, obtaining that there exists a sequence of types $\ttmthrees$ 
  such that $\match{\ttmthree}{\argutwos}{\ttmthrees}{\teig}$
  and such that $\verif{\ttmthree_i}{\argutwo_i} \tos \iunit$
  holds whenever $1 \leq i \leq n-1$ and $\argutwo_i$ is a \metaterm.
  To conclude, take $\ttms := (\ttmtwo,\ttmthrees)$
  where $\ttmtwo$ is the first element and $\ttmthrees$ is the tail.
  Observe that $\match{\ttm}{\argus}{\ttms}{\teig}$ holds
  because it means that
  $\match{(\ttmtwo\imp\ttmthree)}{(\utm,\argutwos)}{(\ttmtwo,\ttmthrees)}{\teig}$
  which in turn holds by definition,
  because we have that $\match{\ttmthree}{\argutwos}{\ttmthrees}{\teig}$.
  Moreover, the condition $\verif{\ttm_i}{\argu_i} \tos \iunit$
  holds whenever $1 \leq i \leq n$ and $\argu_i$ is a \metaterm.
  Indeed, when $i = 1$ this corresponds to the fact that $\verif{\ttmtwo}{\utm} \tos \iunit$,
  and when $i > 1$ this follows from the \ih.

If $\ttm$ is a \textbf{universal quantification}, \ie $\ttm = (\allf{\tvartwo}{\ttmtwo})$:
  Then we have that $\verif{\teig}{\gen{\allf{\tvartwo}{\ttmtwo}}\,\argus} \tos \iunit$,
  and we can note that
  $\verif{\teig}{\gen{\allf{\tvartwo}{\ttmtwo}}\,\argus}
   \to \verif{\teig}{((\pallfi{\tvartwo}{\gen{\ttmtwo}})\,\argus)}$
  so by confluence (\cref{prop:f:confluence}) we have that also
  $\verif{\teig}{((\pallfi{\tvartwo}{\gen{\ttmtwo}})\,\argus)} \tos \iunit$.
  We claim that $\argus$ must be of the form $\argus = (\ttmthree,\argutwos)$ where
  the first element $\ttmthree$ is a \logicalTerm, and $\argutwos$ is the tail.
  Indeed, note that $\argus$ cannot be empty because we would have that
  $\verif{\teig}{(\pallfi{\tvartwo}{\gen{\ttmtwo}})} \tos \iunit$,
  but this is impossible because all the reducts of
  $\verif{\teig}{(\pallfi{\tvartwo}{\gen{\ttmtwo}})}$
  are of the form $\verif{\teig}{(\pallfi{\tvartwo}{\placeholder})}$,
  which is a contradiction.
  Moreover, $\argus$ cannot be of the form $\argus = (\utm,\argutwos)$
  where $\utm$ is a \metaterm, because we would have that
  $\verif{\teig}{((\pallfi{\tvartwo}{\gen{\ttmtwo}})\,\utm\,\argutwos)} \tos \iunit$
  but there are no rules to contract a pseudo-redex of the form
  $(\pallfi{\tvartwo}{\placeholder})\,\utm$,
  so all the reducts of
  $\verif{\teig}{((\pallfi{\tvartwo}{\gen{\ttmtwo}})\,\utm\,\argutwos)}$
  are of the form
  $\verif{\teig}{((\pallfi{\tvartwo}{\placeholder})\,\utm'\,\placeholder)}$,
  which is a contradiction.

  Therefore, $\argus = (\ttmthree,\argutwos)$ where $\ttmthree$ is a \logicalTerm,
  and we have that
  \[
    \verif{\teig}{((\pallfi{\tvartwo}{\gen{\ttmtwo}})\,\ttmthree\,\argutwos)}
    \to
    \verif{\teig}{(\gen{\ttmtwo\sub{\tvartwo}{\ttmthree}}\,\argutwos)}
  \]
  so by confluence~(\cref{prop:f:confluence}) we have that
  $\verif{\teig}{(\gen{\ttmtwo\sub{\tvartwo}{\ttmthree}}\,\argutwos)} \tos \iunit$.
  Since $\size{\argus} = 1 + \size{\argutwos}$, we may resort to the \ih,
  obtaining that there exists a sequence of types $\ttmfours$
  such that $\match{\ttmtwo\sub{\tvartwo}{\ttmthree}}{\argutwos}{\ttmfours}{\teig}$,
  and such that $\verif{\ttmfour_i}{\argutwo_i} \tos \iunit$
  holds whenever $1 \leq i \leq n - 1$ and $\argutwo_i$ is a \metaterm.
  To conclude, take $\ttms := (\ttmthree,\ttmfours)$
  where $\ttmthree$ is the first element and $\ttmfours$ is the tail.
  Observe that $\match{\ttm}{\argus}{\ttms}{\teig}$ holds
  because it means that
  $\match{(\allf{\tvartwo}{\ttmtwo})}{(\ttmthree,\argutwos)}{(\ttmthree,\ttmfours)}{\teig}$
  which in turn holds by definition,
  because we have that
  $\match{\ttmtwo\sub{\tvartwo}{\ttmthree}}{\argutwos}{\ttmfours}{\teig}$.
  Moreover, the condition $\verif{\ttm_i}{\argu_i} \tos \iunit$
  holds whenever $1 \leq i \leq n$ and $\argu_i$ is a \metaterm,
  which follows directly from the \ih.
\end{proof}

\begin{lemma}[Reconstruction of typing derivation]
\label{lem:f:typing_derivation_reconstruction}
Let $\argus = (\argu_1,\hdots,\argu_n)$ be a sequence of \arguments
such that all \metaterms appearing in the sequence are pure,
and let $\ttms = (\ttm_1,\hdots,\ttm_n)$ be a sequence of \logicalTerms.
Suppose that $\judgf{\penv}{\ptm}{\ttm}$
and that $\match{\ttm}{\argus}{\ttms}{\ttmtwo}$.
Suppose moreover that $\judgf{\penv}{\argu_i}{\ttm_i}$
holds whenever $1 \leq i \leq n$ and $\argu_i$ is a \metaterm.
Then $\judgf{\penv}{\ptm\,\argus}{\ttmtwo}$.
\end{lemma}
\begin{proof}
By induction on the derivation of the judgment $\match{\ttm}{\argus}{\ttms}{\ttmtwo}$:
\begin{enumerate}
\item
  Suppose that the judgment is of the form $\match{\ttm}{\emptyset}{\emptyset}{\ttm}$,
  where $\argus = \emptyset$ and $\ttms = \emptyset$
  and $\ttmtwo = \ttm$.
  Then immediately we have that $\judgf{\penv}{\ptm}{\ttm}$ holds by hypothesis.
\item
  Suppose that the judgment $\match{(\ttmthree\imp\ttmfour)}{(\ptmtwo,\argutwos)}{(\ttmthree,\ttmfours)}{\ttmtwo}$
  is derived from $\match{\ttmfour}{\argutwos}{\ttmfours}{\ttmtwo}$,
  where $\ttm = (\ttmthree\imp\ttmfour)$
  and $\argus = (\ptmtwo,\argutwos)$ and $\ttms = (\ttmthree,\ttmfours)$.
  Then we know by hypothesis that $\judgf{\penv}{\ptm}{\ttmthree\imp\ttmfour}$
  and also that $\judgf{\penv}{\ptmtwo}{\ttmthree}$,
  so by the elimination of the implication we have $\judgf{\penv}{\ptm\,\ptmtwo}{\ttmfour}$.
  Applying the \ih on the judgments $\judgf{\penv}{\ptm\,\ptmtwo}{\ttmfour}$
  and $\match{\ttmfour}{\argutwos}{\ttmfours}{\ttmtwo}$,
  we obtain that $\judgf{\penv}{\ptm\,\ptmtwo\,\argutwos}{\ttmtwo}$,
  that is, $\judgf{\penv}{\ptm\,\argus}{\ttmtwo}$.
\item
  Suppose that the judgment $\match{(\allf{\tvar}{\ttmthree})}{(\ttmfour,\argutwos)}{(\ttmfour,\ttmfives)}{\ttmtwo}$
  is derived from $\match{\ttmthree\sub{\tvar}{\ttmfour}}{\argutwos}{\ttmfives}{\ttmtwo}$,
  where $\ttm = (\allf{\tvar}{\ttmthree})$
  and $\argus = (\ttmfour,\argutwos)$
  and $\ttms = (\ttmfour,\ttmfives)$.
  Then we know by hypothesis that $\judgf{\penv}{\ptm}{\allf{\tvar}{\ttmthree}}$
  so by the elimination of the second-order universal quantifier we have
  $\judgf{\penv}{\ptm\,\ttmfour}{\ttmthree\sub{\tvar}{\ttmfour}}$.
  Applying the \ih on the judgments $\judgf{\penv}{\ptm\,\ttmfour}{\ttmthree\sub{\tvar}{\ttmfour}}$
  and $\match{\ttmthree\sub{\tvar}{\ttmfour}}{\argutwos}{\ttmfives}{\ttmtwo}$,
  we obtain that $\judgf{\penv}{\ptm\,\ttmfour\,\argutwos}{\ttmtwo}$,
  that is, $\judgf{\penv}{\ptm\,\argus}{\ttmtwo}$.
\end{enumerate}
\end{proof}

\begin{theorem}[Completeness]
\label{a:thm:f:completeness}
If $\verif{\ttm}{\subs{\utm}{\penv}} \tos \iunit$
and $\utm$ is good,
then its $\beta$-normal form is a \proofTerm $\utm^\downarrow$ such
that $\judgf{\penv}{\utm^\downarrow}{\ttm}$.
\end{theorem}
\begin{proof}
Let $\unf$ denote the $\beta$-normal form of $\utm$ and proceed by induction
on the lexicographic pair $(\sig{\unf}{\ttm},\size{\ttm})$,
where $\size{\ttm}$ denotes the size of $\ttm$.
We consider four cases, depending on the shape of $\ttm$.

If $\ttm$ is an \textbf{implication}, \ie $\ttm = (\ttmtwo\imp\ttmthree)$:
  Note that
  $\verif{\ttm}{\subs{\utm}{\penv}}
   = \verif{\ttmtwo\imp\ttmthree}{\subs{\utm}{\penv}}
   \to \verif{\ttmthree}{(\subs{\utm}{\penv}\,\gen{\ttmtwo})}$,
  so by confluence~(\cref{prop:f:confluence}) we have that
   $\verif{\ttmthree}{(\subs{\utm}{\penv}\,\gen{\ttmtwo})}
   \tos \iunit$.
  Let $\pvar$ be a fresh \proofVariable,
  consider the \proofEnvironment $\penvtwo := (\penv,\pvar:\ttmtwo)$,
  and the \metaterm $\utm\,\pvar$,
  and note by \cref{lem:f:size_proof_applications}
  that there exists a pure \metaterm $\unftwo$ in $\beta$-normal form
  such that $\utm\,\pvar \tos \unf\,\pvar \tos \unftwo$,
  and
  $\sig{\unf}{\ttm}
   = \sig{\unf}{\ttmtwo\imp\ttmthree}
   \geq \sig[\penv,\pvar:\ttmtwo]{\unftwo}{\ttmthree}
   = \sig[\penvtwo]{\unftwo}{\ttmthree}$.
  We claim that we are now in a position to apply the \ih
  on $\verif{\ttmthree}{(\subs{(\unftwo)}{\penvtwo})}$:
  \begin{itemize}
  \item
    First, from the fact that $\utm\,\pvar \tos \unftwo$,
    we obtain that $\subs{(\utm\,\pvar)}{\penvtwo} \tos \subs{(\unftwo)}{\penvtwo}$,
    and since we know that
     $\verif{\ttmthree}{(\subs{(\utm\,\pvar)}{\penvtwo})}
     = \verif{\ttmthree}{(\subs{(\utm\,\pvar)}{\penv,\pvar:\ttmtwo})}
     = \verif{\ttmthree}{(\subs{\utm}{\penv}\,\gen{\ttmtwo})}
     \tos \iunit$,
    by confluence~(\cref{prop:f:confluence})
    we obtain that
    $\verif{\ttmthree}{(\subs{(\unftwo)}{\penvtwo})}
     = \verif{\ttmthree}{(\subs{(\unftwo)}{\penv,\pvar:\ttmthree})}
     \tos \iunit$.
  \item
    Second, note that the pair
    $(\sig{\unf}{\ttm},\size{\ttm}) = (\sig{\unf}{\ttm},\size{\ttmtwo\imp\ttmthree})$
    is strictly greater than
    $(\sig[\penvtwo]{\unftwo}{\ttmthree},\size{\ttmthree})$
    according to the lexicographic order,
    because we have already shown that
    $\sig{\unf}{\ttm} \geq \sig[\penvtwo]{\unftwo}{\ttmthree}$
    and moreover
    $\size{\ttmtwo\imp\ttmthree} > \size{\ttmthree}$.
  \end{itemize}
  Hence by \ih we obtain that $\judgf{\penvtwo}{\unftwo}{\ttmthree}$.
  This means that $\judgf{\penv,\pvar:\ttmtwo}{\unftwo}{\ttmthree}$,
  so using the introduction rule for the implication
  we have that $\judgf{\penv}{\plamf{\pvar}{\unftwo}}{\ttmtwo\imp\ttmthree}$.
  To conclude,
  note that
  $\lam{\pvar}{\utmtwo^\downarrow} \tobetainvs \lam{\pvar}{\utm^\downarrow\,\pvar} \toeta \utm^\downarrow$.
  Since $\tobeta$ and $\toeta$ commute, and $\utm^\downarrow$
  is in $\beta$-normal form, we have that
  $\lam{\pvar}{\utmtwo^\downarrow} \toetas \utm^\downarrow$.
  Hence, by subject reduction $\judgf{\penv}{\unf}{\ttmtwo\imp\ttmthree}$.

If $\ttm$ is a \textbf{universal quantification}, \ie $\ttm = (\allf{\tvar}{\ttmtwo})$:
  Assume by $\alpha$-conversion that $\tvar \notin \fv{\penv}$.
  Note that
  $\verif{\ttm}{\subs{\utm}{\penv}}
  = \verif{\allf{\tvar}{\ttmtwo}}{\subs{\utm}{\penv}}
  \to \freshf{\teig}{\verif{\ttmtwo\sub{\tvar}{\teig}}{(\subs{\utm}{\penv}\,\teig)}}$,
  where $\teig$ is assumed to be fresh in the sense that $\teig \notin\fv{\ttmtwo,\utm,\penv}$,
  so by confluence~(\cref{prop:f:confluence}) we have that
  $\freshf{\teig}{\verif{\ttmtwo\sub{\tvar}{\teig}}{(\subs{\utm}{\penv}\,\teig)}}
   \tos \iunit$.
  This in turn implies that
  $\verif{\ttmtwo\sub{\tvar}{\teig}}{(\subs{\utm}{\penv}\,\teig)} \tos \iunit$.
  Consider the \metaterm $\utm\,\teig$,
  and note by \cref{lem:f:size_logical_applications}
  that there exists a pure \metaterm $\unftwo$ in $\beta$-normal form
  such that $\utm\,\teig \tos \unf\,\teig \tos \unftwo$,
  and
  $\sig{\unf}{\ttm}
   = \sig{\unf}{\allf{\tvar}{\ttmtwo}}
   \geq \sig{\unftwo}{\ttmtwo\sub{\tvar}{\teig}}$.
  We claim that we are now in a position to apply the \ih on
  $\verif{\ttmtwo\sub{\tvar}{\teig}}{(\subs{(\unftwo)}{\penv})}$.
  \begin{itemize}
  \item
    First, from the fact that $\utm\,\teig \tos \unftwo$
    we obtain that $\subs{(\utm\,\teig)}{\penv} \tos \subs{(\unftwo)}{\penv}$,
    and since we know that
    $\verif{\ttmtwo\sub{\tvar}{\teig}}{(\subs{(\utm\,\teig)}{\penv})}
     = \verif{\ttmtwo\sub{\tvar}{\teig}}{(\subs{\utm}{\penv}\,\teig)}
     \tos \iunit$, by confluence~(\cref{prop:f:confluence})
    we obtain that
    $\verif{\ttmtwo\sub{\tvar}{\teig}}{(\subs{(\unftwo)}{\penv})}
     \tos \iunit$.
  \item
    Second, note that the pair
    $(\sig{\unf}{\ttm},\size{\ttm})
     = (\sig{\unf}{\allf{\tvar}{\ttmtwo}},\size{\allf{\tvar}{\ttmtwo}})$
    is strictly greater than the pair
    $(\sig{\unftwo}{\ttmtwo\sub{\tvar}{\teig}},\size{\ttm\sub{\tvar}{\teig}})$
    according to the lexicographic order,
    because we have already shown that
    $\sig{\unf}{\allf{\tvar}{\ttmtwo}} \geq \sig{\unftwo}{\ttmtwo\sub{\tvar}{\teig}}$
    and moreover
    $\size{\allf{\tvar}{\ttmtwo}}
    > \size{\ttmtwo}
    = \size{\ttmtwo\sub{\tvar}{\teig}}$.
  \end{itemize}
  Hence by \ih we obtain that
  $\judgf{\penv}{\unftwo}{\ttmtwo\sub{\tvar}{\teig}}$.
  Since we know that $\tvar\notin\fv{\penv}$,
  renaming $\teig$ by $\tvar$, this means that
  $\judgf{\penv}{\unftwo\sub{\teig}{\tvar}}{\ttmtwo}$,
  and using the introduction rule for the second-order universal quantifier
  we have that $\judgf{\penv}{\pallfi{\tvar}{\unftwo\sub{\teig}{\tvar}}}{\allf{\tvar}{\ttmtwo}}$.
  To conclude,
  note that
  $\pallfi{\tvar}{\utmtwo\sub{\teig}{\tvar}}
   \tobetainvs \pallfi{\tvar}{\utm^\downarrow\,\tvar}
   \toeta \unf$.
  Since $\tobeta$ and $\toeta$ commute, and $\unf$
  is in $\beta$-normal form, we have that
  $\pallfi{\tvar}{\unftwo\sub{\teig}{\tvar}} \toetas \unf$.
  Hence, by subject reduction $\judgf{\penv}{\unf}{\allf{\tvar}{\ttmtwo}}$.

If $\ttm$ is a \textbf{\logicalVariable}, \ie $\ttm = \tvar$:
  We claim that this case is impossible.
  Indeed, by hypothesis, $\verif{\tvar}{\subs{\utm}{\penv}} \tos \iunit$,
  but the outermost $\tvar$-verifier is persistent,
  \ie all the reducts of $\verif{\tvar}{\subs{\utm}{\penv}}$
  are of the form $\verif{\tvar}{\placeholder}$, which is a contradiction.

If $\ttm$ is an \textbf{\logicalEigenvariable}, $\ttm = \teig$:
  Then $\verif{\teig}{\subs{\utm}{\penv}} \tos \iunit$
  so by confluence
  $\verif{\teig}{(\subs{(\unf)}{\penv})} \tos \iunit$.
  According to the characterization of normal forms given by \cref{lem:f:characterization_of_nfs},
  $\unf$ might be of three possible shapes,
  namely either $\unf$ is a \proofAbstraction ($\unf = \plamf{\pvar}{\unftwo}$)
  or a \logicalAbstraction ($\unf = \pallfi{\tvar}{\unftwo}$)
  or a neutral \proofTerm ($\unf = \uhd\,\argunf_1\hdots\argunf_n$).
  We claim that $\unf$ cannot be a \proofAbstraction,
  as this implies that
  $\verif{\teig}{(\plamf{\pvar}{(\subs{(\unftwo)}{\penv})})} \tos \iunit$,
  but there are no rewriting rules that allow forming a redex that involves
  the outermost $\teig$-verifier and the \proofAbstraction,
  so all the reducts of $\verif{\teig}{(\plamf{\pvar}{(\subs{(\unftwo)}{\penv})})}$
  are of the form $\verif{\teig}{(\plamf{\pvar}{\placeholder})}$,
  which is a contradiction.
  Similarly, $\unf$ cannot be a \logicalAbstraction,
  as we would have that
  $\verif{\teig}{(\pallfi{\tvar}{(\subs{(\unftwo)}{\penv})})} \tos \iunit$,
  but this is impossible.

  Therefore, $\unf$ must be a neutral \proofTerm,
  \ie $\unf = \uhd\,\argus$,
  where $\uhd$ is a \headForm
  and $\argus = (\argu_1,\hdots,\argu_n)$ is a sequence of \arguments\footnote{Note that
  each \argument $\argu_i$ is a normal \argument; we write just $\argu_i$ rather than
  the more explicit $\argunf_i$ to alleviate the notation.}.
  Again, we have three possibilities,
  either $\uhd$ is formed by a \proofAbstraction and a \logicalApplication
  ($\uhd = (\plamf{\pvar}{\unftwo})\,\ttm$),
  or by a \logicalAbstraction and a \proofApplication
  ($\uhd = (\pallfi{\tvar}{\unftwo})\,\unfthree$),
  or $\uhd$ is a \proofVariable ($\uhd = \pvar$).
  We claim that $\uhd$ must be a variable.
  Indeed, the case $\uhd = (\plamf{\pvar}{\unftwo})\,\ttm$
  is impossible, as we would have that
  $\verif{\teig}{(
    (\plamf{\pvar}{\subs{(\unftwo)}{\penv}})\,\ttm\,\subs{\argus}{\penv}} \tos \iunit$,
  but there are no rewriting rules that allow contracting a
  pseudo-redex of the form $(\plamf{\pvar}{\placeholder})\,\ttm$,
  so all the reducts of
  $\verif{\teig}{(
    (\plamf{\pvar}{\subs{(\unftwo)}{\penv}})\,\ttm\,\subs{\argus}{\penv})}$
  must necessarily be of the form
  $\verif{\teig}{(
    (\plamf{\pvar}{\placeholder})\,\ttm\,\placeholder
   )}$,
  which is a contradiction.
  Similarly, the case $\uhd = (\pallfi{\tvar}{\unftwo})\,\unfthree$
  is impossible.

  Therefore, $\uhd$ must be a a variable,
  and $\unf = \pvar\,\argus$. Note that we have that
  $\verif{\teig}{(
     \subs{\pvar}{\penv}\,\subs{\argus}{\penv}
   )} \tos \iunit$.
  We claim that $\pvar \in \dom{\penv}$.
  Indeed, if $\pvar \notin \dom{\penv}$
  we would have that
  $\verif{\teig}{(
     \pvar\,\subs{\argus}{\penv}
   )} \tos \iunit$,
  but all the reducts of
  $\verif{\teig}{(
     \pvar\,\subs{\argus}{\penv}
   )}$
  must be of the form
  $\verif{\teig}{(
     \pvar\,\placeholder
   )}$, a contradiction.
  Hence there exists a \logicalTerm $\ttmtwo$
  such that $\pvar:\ttmtwo \in \penv$,
  and we have that
  $\verif{\teig}{(
     \gen{\ttmtwo}\,\subs{\argus}{\penv}
   )} \tos \iunit$.
  By \cref{lem:f:verif_gen_match}
  this implies that there exists a sequence of \logicalTerms
  $\ttmtwos = (\ttmtwo_1,\hdots,\ttmtwo_n)$
  such that
  $\match{\ttmtwo}{\subs{\argus}{\penv}}{\ttmtwos}{\teig}$
  and such that
  for every index $1 \leq i \leq n$ 
  such that the \argument $\subs{\argu_i}{\penv}$ is a \metaterm
  we have that
  $\verif{\ttmtwo_i}{\subs{\argu_i}{\penv}} \tos \iunit$.
  Remark, for each index $1 \leq i \leq n$, that
  $\subs{\argu_i}{\penv}$ is a \metaterm
  if and only if $\argu_i$ is a \metaterm.
  Now note that
  $\match{\ttmtwo}{\argus}{\ttmtwos}{\teig}$
  holds by \cref{rem:f:match_coherent}
  since $\argus \coherent \subs{\argus}{\penv}$.
  Hence, for every index $1 \leq i \leq n$ such that $\argu_i$ is a \metaterm
  we have that
  $\sig{\unf}{\teig}
  = \sig{\pvar\,\argus}{\teig}
  > \sig{\argu_i}{\ttmtwo_i}$
  by \cref{lem:f:size_arguments}.
  In particular, the pair
  $(\sig{\unf}{\teig}, \size{\teig})$
  is strictly greater than the pair
  $(\sig{\argu_i}{\ttmtwo_i}, \size{\ttmtwo_i})$
  according to the lexicographic order,
  because the first component strictly decreases.

  We may apply the \ih to obtain that $\judgf{\penv}{\argunf_i}{\ttmtwo_i}$
  holds for every $1 \leq i \leq n$ such that $\argu_i$ is a \metaterm.
  Observe that $\argu_i$ must be in $\beta$-normal form
  because $\unf$ is a $\beta$-normal form and $\unf = \var\argus$,
  so $\argunf_i = \argu_i$.
  Finally, since we have $\match{\ttmtwo}{\argus}{\ttmtwos}{\teig}$,
  and $\judgf{\penv}{\pvar}{\ttmtwo}$ (since $\pvar:\ttmtwo \in \penv$),
  and $\judgf{\penv}{\argu_i}{\ttmtwo_i}$
  holds for each \metaterm in the sequence $\argus$,
  by \cref{lem:f:typing_derivation_reconstruction}
  we conclude that $\judgf{\penv}{\pvar\,\argus}{\teig}$.
  This means that $\judgf{\penv}{\unf}{\teig}$, as required.
\end{proof}

\section{Second-Order Logic --- Properties of $\lambdaCheckFw$}

\subsection{Higher-Order: Soundness (\cref{thm:hol:soundness})}
  \label{a:sec:hol:soundness}
  
In this section, we show that $\lambdaFw$ is sound with respect to the
\epistemic realizability semantics induced by $\lambdaCheckFw$.

A straightforward observation is that \logicalTerms have unique \kinds under the global
\kindAssignment, \ie if $\tjudg{\kiasig}{\ttm}{\ki_1}$ and $\tjudg{\kiasig}{\ttm}{\ki_2}$
then $\ki_1 = \ki_2$.

\begin{lemma}[Reduction preserves well-kindedness]
If $\wkjudg{\utm}$ and $\utm \to \utm'$
then $\wkjudg{\utm'}$.
\end{lemma}

\begin{definition}[Sets of valid expressions]
We write $\ValidKinds$ for the set of all \kinds.
A \logicalTerm is \emph{ground} if it has no free \logicalVariables
(but note that it may have free \logicalEigenvariables).
We write $\ValidTypes$ for the set of \emph{valid \logicalTerm/\kind pairs},
defined as $\ValidTypes \eqdef \set{(\ttm,\ki) \ST \tjudg{\kiasig}{\ttm}{\ki}}$.
The subset of \emph{valid ground \logicalTerm/\kind pairs}
is defined by
$\GValidTypes := \set{(\ttm,\ki) \ST \ttm \text{ is ground} \land (\ttm,\ki) \in \ValidTypes}$.
We write $\ValidTypesK{\ki}$ for the set of \emph{valid \logicalTerms of \kind $\ki$},
defined as $\ValidTypesK{\ki} \eqdef \set{\ttm \ST \tjudg{\kiasig}{\ttm}{\ki}}$.
The subset of \emph{valid ground \logicalTerms of \kind $\ki$}
is defined by $\GValidTypesK{\ki} \eqdef \set{\ttm \ST \ttm \text{ is ground} \land \tjudg{\kiasig}{\ttm}{\ki}}$.
Sometimes we write ``$\ttm:\ki$'' instead of ``$(\ttm,\ki)$'' for
elements of $\ValidTypes$ and $\GValidTypes$.
We write $\ValidMetaterms$ for the set of all well-kinded \metaterms,
\ie $\set{\utm \ST \wkjudg{\utm}}$.
\end{definition}

\begin{definition}[Dependent function space]
If $I$ is a set of indices and $(J_i)_\iI$ is an indexed family of sets,
we write $\Prod{i \in I} J_i$
for the set of all functions $f : I \to \cup_{i \in I} J_i$
such that $f(i) \in J_i$.
\end{definition}
Sometimes we implicitly use an uncurried notation for functions;
so for example we write $f : \Prod{i \in I} J_i \to K_i$
and $f(i, j) = k_{i,j}$
for the function $f$ such that
$f(i)$ is again a function $f(i) : J_i \to K_i$ for every $i \in I$,
and such that
$f(i)(j) = k_{i,j}$ holds for each $i \in I$ and each $j \in J_i$.

\begin{definition}[\CandidateFamilies]
\label{def:hol:candidateFamilies}
A family of sets $(\VC{i})_{i \in \GValidTypes}$
indexed by valid ground \logicalTerm/\kind pairs
is a \defn{\candidateFamily} if the following conditions hold:
\begin{enumerate}
\item
  \textbf{Stability under $\beta$-conversion.}
  For all $\ki \in \ValidKinds$
  and all ground $\ttm_1,\ttm_2 \in \GValidTypesK{\ki}$,
  if $\ttm_1 \eqbeta \ttm_2$
  then $\VCK{\ttm_1}{\ki} = \VCK{\ttm_2}{\ki}$.
\item
  \textbf{Inhabited.}
  \label{def:hol:candidateFamilies__prop_nonempty}
  For every ground $\ttm \in \GValidTypesK{\ki}$ 
  there is an element $\semgbase{\ttm} \in \VCK{\ttm}{\ki}$.
  Moreover, if $\ttm_1,\ttm_2 \in \GValidTypesK{\ki}$ and $\ttm_1 \eqbeta \ttm_2$
  then $\semgbase{\ttm_1} = \semgbase{\ttm_2}$.
\item
  \textbf{Candidate condition for base \kinds.}
  For every base \kind $\bki$ and every ground \logicalTerm $\ttm\in\GValidTypesK{\bki}$
  we have that $\VCK{\ttm}{\bki}$ is a singleton.
  We write $\singleElem$ for its unique element, so that
  $\VCK{\ttm}{\bki} = \set{\singleElem}$.
\item
  \textbf{Candidate condition for propositions.}
  For every ground proposition $\ttm\in\GValidTypesK{\Prop}$,
  we have that $\VCK{\ttm}{\Prop}$ is a set of sets of well-kinded \metaterms,
  that is, $\VCK{\ttm}{\Prop} \subseteq \powerset{\ValidMetaterms}$.
  Moreover, for every $\vcset \in \VCK{\ttm}{\Prop}$ we have that:
  \begin{enumerate}
  \item $\gen{\ttm} \in \vcset$
  \item If $\utm \in \vcset$ then $\verif{\ttm}{\utm} \tows \iunit$.
  \item If $\utm_1,\utm_2$ are \metaterms such that $\utm_1 \tow \utm_2$
        then $\utm_1 \in \vcset$ if and only if $\utm_2 \in \vcset$.
  \end{enumerate}
\item
  \textbf{Candidate condition for higher-kinded \logicalTerms.}
  \label{def:hol:candidateFamilies:case_arrow}
  For all $\ki_1,\ki_2 \in \ValidKinds$
  and every ground $\ttm\in\GValidTypesK{\ki_1\to\ki_2}$,
  we have that
  $\VCK{\ttm}{\ki_1\to\ki_2}$ is the set of all functions $f$
  such that:
  \[
      f : \Prod{\ttmtwo \in \GValidTypesK{\ki_1}}
          \VCK{\ttmtwo}{\ki_1}
          \to
          \VCK{\ttm\ttmtwo}{\ki_2}
  \]
  and such that
  for every $\ttmtwo_1,\ttmtwo_2 \in \GValidTypesK{\ki_1}$,
  and every $\vcset \in \VCK{\ttmtwo}{\ki_1}$
  if $\ttmtwo_1 \eqbeta \ttmtwo_2$
  then $f(\ttmtwo_1,\vcset) = f(\ttmtwo_2,\vcset)$.
\end{enumerate}
\end{definition}

\begin{lemma}[Correctness]
\label{lem:hol:correctness}
For any ground \logicalTerm $\ttm \in \GValidTypesK{\Prop}$
we have $\verif{\ttm}{\gen{\ttm}} \tows \iunit$.
\end{lemma}
\begin{proof}
By induction on the size of $\ttm$.
Similar to the proof for the second-order case (\cref{lem:f:correctness}).
\end{proof}

\begin{definition}[\CandidateAssignments]
Let $\mathcal{F} = (\VC{i})_{i \in \GValidTypes}$ be a \candidateFamily.
A \defn{$\mathcal{F}$-\candidateAssignment} is a function $\asig$
mapping each \logicalVariable $\tvar$
to a pair $(\ttm,\vcset)$ where,
if we let $\ki := \kiasig(\tvar)$
(that is, $\ki$ is the \kind of $\tvar$ in the global \kindAssignment $\kiasig$),
then we have that $\ttm \in \GValidTypesK{\ki}$ is a ground \logicalTerm
and $\vcset \in \VCK{\ttm}{\ki}$.
Usually the ambient \candidateFamily $\mathcal{F}$ is clear from the context, and we
speak just of \candidateAssignments, rather than of $\mathcal{F}$-\candidateAssignments.
We write $\tilde{\asig}$ for the \defn{first projection} of $\asig$,
defined as the \logicalSubstitution such that
$\tilde{\asig}(\tvar) = \ttm$ whenever $\asig(\tvar) = (\ttm,\vcset)$.
By abuse of notation, sometimes we write $\asig$ for the first projection of
the \candidateAssignment $\asig$ instead of $\tilde{\asig}$, if there is little danger of confusion.
If $\asig_1,\asig_2$ are \candidateAssignments,
we write $\asig_1 \eqbeta \asig_2$
if, for every \logicalVariable $\tvar$,
letting $\asig_1(\tvar) = (\ttm_1,\vcset_1)$
and $\asig_2(\tvar) = (\ttm_2,\vcset_2)$,
we have that $\ttm_1 \eqbeta \ttm_2$ and $\vcset_1 = \vcset_2$.
\end{definition}

\begin{definition}[Candidate interpretation of \logicalTerms]
Let $(\VC{i})_{i \in \GValidTypes}$ be a \candidateFamily.
For every valid \logicalTerm/\kind pair $\ttm:\ki \in \ValidTypes$
and every \candidateAssignment $\asig$
we define the interpretation $\semg{\ttm}{\ki}{\asig}$
by recursion on $\ttm$ according to the equations below.
The construction is well-defined in the sense that $\semg{\ttm}{\ki}{\asig} \in \VCK{\subs{\ttm}{\asig}}{\ki}$.
Technically speaking, well-definedness has to be established
simultaneously with its definition (in inductive-recursive fashion),
but we state it as an adequacy lemma below.
\begin{enumerate}
\item
    $\semg{\tvar}{\ki}{\asig}
    \eqdef
    \vcset$
    if $\asig(\tvar) = (\ttm,\vcset)$.
\item
    $\semg{\teig}{\ki}{\asig}$
    is the set of well-kinded \metaterms
    $\utm \in \ValidMetaterms$
    such that $\utm \tows \gen{\teig}$.
\item
    $\semg{(\tlam{\tvar}{\ki_1}{\ttm})}{\ki_1\to\ki_2}{\asig} \eqdef f$,
    where $f$ is the following function:
    \[
    \begin{array}{l}
      f : \Prod{\ttmtwo \in \GValidTypesK{\ki_1}}
          \VC{\ttmtwo:\ki_1}
          \to
          \VC{\subs{({\tlam{\tvar}{\ki_1}{\ttm}})}{\asig}\ttmtwo:\ki_2}
    \\
      f(\ttmtwo,\vcset) =
        \semg{\ttm}{\ki_2}{\asig\extsub{\tvar}{(\ttmtwo,\vcset)}}
    \end{array}
    \]
\item
    $\semg{\ttm\,\ttmtwo}{\ki_2}{\asig}
     \eqdef
     \semg{\ttm}{\ki_1\to\ki_2}{\asig}\,(\subs{\ttmtwo}{\asig},\semg{\ttmtwo}{\ki_2}{\asig})$
\item
    $\semg{(\ttm\imp\ttmtwo)}{\Prop}{\asig}$
    is the set of well-kinded \metaterms $\utm$ such that:
    \begin{itemize}
    \item[]
        $\verif{\subs{(\ttm\imp\ttmtwo)}{\asig}}{\utm} \tows \iunit$, and
    \item[]
        for every $\utmtwo\in\semg{\ttm}{\Prop}{\asig}$
        we have $\utm\,\utmtwo\in\semg{\ttmtwo}{\Prop}{\asig}$.
    \end{itemize}
\item
    $\semg{(\all{\tvar}{\ki}{\ttm})}{\Prop}{\asig}$
    is the set of well-kinded \metaterms $\utm$ such that:
    \begin{itemize}
    \item[]
        $\verif{\subs{(\all{\tvar}{\ki}{\ttm})}{\asig}}{\utm} \tows \iunit$, and
    \item[]
        for every well-kinded \logicalTerm
          $\ttmtwo\in\ValidTypesK{\ki}$
        and every
          $\vcset\in\VCK{\ttmtwo}{\ki}$ \\
        \hphantom{\HS} we have
          $\utm\,\ttmtwo\in\semg{\ttm}{\Prop}{\asig\extsub{\tvar}{(\ttmtwo,\vcset)}}$.
    \end{itemize}
\end{enumerate}
\end{definition}

\begin{lemma}[Adequacy for types]
\label{lem:hol:adequacy_for_types}
Let $(\VC{i})_{i \in \GValidTypes}$ be a \candidateFamily.
\begin{enumerate}
\item
  Let $\asig$ be a \candidateAssignment and $\ttm\in\ValidTypesK{\ki}$.
  Then $\semg{\ttm}{\ki}{\asig} \in \VCK{\subs{\ttm}{\asig}}{\ki}$.
\item
  Let $\asig_1,\asig_2$ be \candidateAssignments and $\ttm\in\ValidTypesK{\ki}$.
  If $\asig_1 \eqbeta \asig_2$ then
  $\semg{\ttm}{\ki}{\asig_1} = \semg{\ttm}{\ki}{\asig_2}$.
\end{enumerate}
\end{lemma}
\begin{proof}
We prove both items simultaneously by induction on $\ttm$.
The second item is straightforward.
To prove the first item:
\begin{enumerate}
\item
  Variable, $\ttm = \tvar$:
  Let $\asig(\tvar) = (\ttmtwo,\vcset)$
  where we know that $\vcset \in \VCK{\ttmtwo}{\ki}$
  Then $\subs{\ttm}{\asig} = \subs{\tvar}{\asig} = \ttmtwo$
  and by definition
  $\semg{\ttm}{\ki}{\asig}
  = \semg{\tvar}{\ki}{\asig}
  = \vcset
  \in \VCK{\ttmtwo}{\ki}
  = \VCK{\subs{\ttm}{\asig}}{\ki}
  $.
\item
  Implication, $\ttm = (\ttmtwo\imp\ttmthree)$:
  Similar to the proof for the second-order case
  (\cref{lem:f:adequacy_for_propositions}).
\item
  Universal quantifier,
  $\ttm = (\all{\tvar}{\ki}{\ttmtwo})$:
  Similar to the proof for the second-order case
  (\cref{lem:f:adequacy_for_propositions}).
\item
  Abstraction, $\ttm = \tlam{\tvar}{\ki_1}{\ttmtwo}$ with $\ki = (\ki_1\to\ki_2)$:
  Then by definition
  $\semg{\tlam{\tvar}{\ki}{\ttmtwo}}{\ki_1\to\ki_2} = f$ is a function
  $f : \Prod{\ttmthree\in\GValidTypesK{\ki_1}}
       \VCK{\ttmthree}{\ki_1}
       \to \VCK{\subs{(\tlam{\tvar}{\ki_1}{\ttmtwo})}{\asig}\,\ttmthree}{\ki_2}$
  such that
  $f(\ttmthree,\vcset) = \semg{\ttmtwo}{\ki_2}{\asig\extsub{\tvar}{(\ttmthree,\vcset)}}$.
  Since $(\VC{i})_{i\in\GValidTypes}$ is a \candidateFamily,
  to show that
  $f \in \VCK{\subs{(\tlam{\tvar}{\ki}{\ttmtwo})}{\asig}}{\ki_1\to\ki_2}$,
  we need to show, first, that the domain and codomain of $f$
  correspond to
  $f : \Prod{\ttmthree\in\GValidTypesK{\ki_1}}
       \VCK{\ttmthree}{\ki_1}
       \to \VCK{(\subs{(\tlam{\tvar}{\ki}{\ttmtwo})}{\asig})\,\ttmthree}{\ki_2}$,
  which is indeed the case.
  Second, we need to show that
  if $\ttmthree_1,\ttmthree_2 \in \GValidTypesK{\ki_1}$
  and $\vcset \in \VCK{\ttmthree_1}{\ki_1}$
  and $\ttmthree_1 \eqbeta \ttmthree_2$
  then $f(\ttmthree_1,\vcset) = f(\ttmthree_2,\vcset)$.
  Indeed,
  $\semg{\ttmtwo}{\ki_2}{\asig\extsub{\tvar}{(\ttmthree_1,\vcset)}}
   = \semg{\ttmtwo}{\ki_2}{\asig\extsub{\tvar}{(\ttmthree_2,\vcset)}}$
  holds by resorting to the \ih, using the second item in the statement
  of this lemma,
  noting that
  $\asig\extsub{\tvar}{(\ttmthree_1,\vcset)}
   \eqbeta
   \asig\extsub{\tvar}{(\ttmthree_2,\vcset)}$.
\item
  Application, $\ttm = \ttmtwo\,\ttmthree$,
  where $\ttmtwo \in \ValidTypesK{\ki_1\to\ki_2}$
  and $\ttmthree \in \ValidTypesK{\ki_1}$
  and $\ki = \ki_2$:
  Then by \ih we have that
  $\semg{\ttmtwo}{\ki_1\to\ki_2}{\asig} \in \VCK{\subs{\ttmtwo}{\asig}}{\ki_1\to\ki_2}$
  and
  $\semg{\ttmthree}{\ki_1}{\asig} \in \VCK{\subs{\ttmthree}{\asig}}{\ki_1}$.
  By the fact that $(\VC{i})_{i\in\GValidTypes}$ is a \candidateFamily,
  we know that $\semg{\ttmtwo}{\ki_1\to\ki_2}{\asig}$ is a function
  $f : \Prod{\ttmfour\in\GValidTypesK{\ki_1}}
       \VCK{\ttmfour}{\ki_1}
       \to \VCK{\subs{\ttmtwo}{\asig}\,\ttmfour}{\ki_2}$.
  In particular,
  $\semg{\ttmtwo\,\ttmthree}{\ki_2}{\asig}
   =
   \semg{\ttmtwo}{\ki_1\to\ki_2}{\asig}(
     \subs{\ttmthree}{\asig},
     \semg{\ttmthree}{\ki_1}{\asig}
   ) \in \VCK{\subs{\ttmtwo}{\asig}\,\subs{\ttmthree}{\asig}}{\ki_2}
       = \VCK{\subs{(\ttmtwo\,\ttmthree)}{\asig}}{\ki}$.
\end{enumerate}
\end{proof}

\begin{lemma}[Irrelevance]
\label{lem:hol:irrelevance}
Let $(\VC{i})_{i \in \GValidTypes}$ be a \candidateFamily.
If $\ttm:\ki \in \ValidTypes$
and $\asig_1$ and $\asig_2$ are \candidateAssignments that
agree on the free variables of $\ttm$,
then $\semg{\ttm}{\ki}{\asig_1} = \semg{\ttm}{\ki}{\asig_2}$.
\end{lemma}
\begin{proof}
Straightforward by induction on $\ttm$.
\end{proof}

\begin{lemma}[Substitution lemma]
\label{lem:hol:substitution}
Let $(\VC{i})_{i \in \GValidTypes}$ be a \candidateFamily.
Let $\ttm:\ki \in \ValidTypes$ and $\ttmtwo:\kitwo \in \ValidTypes$
and $\kiasig(\tvar) = \kitwo$.
Moreover, suppose that $\asig$ is a \candidateAssignment.
Then
$\semg{\ttm\sub{\tvar}{\ttmtwo}}{\ki}{\asig}
 = \semg{\ttm}{\ki}{\asig\extsub{\tvar}{(\subs{\ttmtwo}{\asig},\semg{\ttmtwo}{\kitwo}{\asig})}}$.
\end{lemma}
\begin{proof}
By induction on $\ttm$:
\begin{enumerate}
\item
  Same \logicalVariable, $\ttm = \tvar$:
  Then $\ki = \kitwo$,
  and
  $\semg{\ttm\sub{\tvar}{\ttmtwo}}{\ki}{\asig}
  = \semg{\tvar\sub{\tvar}{\ttmtwo}}{\ki}{\asig}
  = \semg{\ttmtwo}{\ki}{\asig}
  = \semg{\ttmtwo}{\kitwo}{\asig}
  = \semg{\tvar}{\ki}{\asig\extsub{\tvar}{(\subs{\ttmtwo}{\asig},\semg{\ttmtwo}{\kitwo}{\asig})}}
  = \semg{\ttm}{\ki}{\asig\extsub{\tvar}{(\subs{\ttmtwo}{\asig},\semg{\ttmtwo}{\kitwo}{\asig})}}
  $.
\item
  Different \logicalVariable, $\ttm = \tvartwo \neq \tvar$:
  Note that $\kiasig(\tvartwo) = \ki$,
  and let $\asig(\tvartwo) = (\ttmthree,\vcset)$
  where $\ttmthree \in \GValidTypesK{\ki}$
  and $\vcset \in \VCK{\ttmthree}{\ki}$.
  Then
  $\semg{\ttm\sub{\tvar}{\ttmtwo}}{\ki}{\asig}
  = \semg{\tvartwo\sub{\tvar}{\ttmtwo}}{\ki}{\asig}
  = \semg{\tvartwo}{\ki}{\asig}
  = \vcset
  = \semg{\tvartwo}{\ki}{\asig\extsub{\tvar}{(\subs{\ttmtwo}{\asig},\semg{\ttmtwo}{\kitwo}{\asig})}}
  = \semg{\ttm}{\ki}{\asig\extsub{\tvar}{(\subs{\ttmtwo}{\asig},\semg{\ttmtwo}{\kitwo}{\asig})}}
  $.
\item
  \LogicalEigenvariable, $\ttm = \teig$:
  Note that $\kiasig(\teig) = \ki$.
  Then
  $\semg{\ttm\sub{\tvar}{\ttmtwo}}{\ki}{\asig}
  = \semg{\teig\sub{\tvar}{\ttmtwo}}{\ki}{\asig}
  = \semg{\teig}{\ki}{\asig}
  = \set{\utm \in \ValidMetaterms \ST \utm \tows \gen{\teig}}
  = \semg{\teig}{\ki}{\asig\extsub{\tvar}{(\subs{\ttmtwo}{\asig},\semg{\ttmtwo}{\kitwo}{\asig})}}
  = \semg{\ttm}{\ki}{\asig\extsub{\tvar}{(\subs{\ttmtwo}{\asig},\semg{\ttmtwo}{\kitwo}{\asig})}}
  $.
\item
  Implication $\ttm = (\ttmthree\imp\ttmfour)$:
  Similar to the proof for the second-order case
  (\cref{lem:f:substitution}).
\item
  Universal quantifier $\ttm = (\all{\tvartwo}{\kithree}{\ttmthree})$:
  Similar to the proof for the second-order case
  (\cref{lem:f:substitution}).
\item
  Abstraction, $\ttm = \tlam{\tvartwo}{\ki_1}{\ttmthree}$
    with $\ki = \ki_1\to\ki_2$:
  Assume, by $\alpha$-conversion,
  that $\tvartwo \notin \set{\tvar}\cup\fv{\ttmtwo}$
  then the left-hand side of the equation is
  $\semg{\ttm\sub{\tvar}{\ttmtwo}}{\ki_1\to\ki_2}{\asig}
  = \semg{(\tlam{\tvartwo}{\ki_1}{\ttmthree})\sub{\tvar}{\ttmtwo}}{\ki_1\to\ki_2}{\asig}
  = \semg{\tlam{\tvartwo}{\ki_1}{\ttmthree\sub{\tvar}{\ttmtwo}}}{\ki_1\to\ki_2}{\asig}$,
  which is a function $f$ such that:
  \[
    \begin{array}{l}
      f
      :
      \Prod{\ttmfour\in\ValidTypesK{\ki_1}}
      \VCK{\ttmfour}{\ki_1}
      \to
      \VCK{(\subs{(\tlam{\tvartwo}{\ki_1}{\ttmthree\sub{\tvar}{\ttmtwo}})}{\asig}\ttmfour)}{\ki_2}
    \\
      f(\ttmfour,\vcset)
      =
      \semg{\ttmthree\sub{\tvar}{\ttmtwo}}{\ki_2}{\asig\extsub{\tvartwo}{(\ttmfour,\vcset)}}
    \end{array}
  \]
  On the other hand, 
  let $\asigext = \asig\extsub{\tvar}{(\subs{\ttmtwo}{\asig},\semg{\ttmtwo}{\kitwo}{\asig})}$.
  The right-hand side of the equation is
  $\semg{\ttm}{\ki_1\to\ki_2}{\asig\extsub{\tvar}{(\subs{\ttmtwo}{\asig},\semg{\ttmtwo}{\kitwo}{\asig})}}
  = \semg{\tlam{\tvartwo}{\ki_1}{\ttmthree}}{\ki_1\to\ki_2}{\asigext}
  $
  is a function $g$ such that:
  \[
    \begin{array}{l}
      g
      :
      \Prod{\ttmfour\in\ValidTypesK{\ki_1}}
      \VCK{\ttmfour}{\ki_1}
      \to
      \VCK{(\subs{(\tlam{\tvartwo}{\ki_1}{\ttmthree})}{\asigext}\,\ttmfour)}{\ki_2}
    \\
      g(\ttmfour,\vcset)
      =
      \semg{\ttmthree}{\ki_2}{\asigext\extsub{\tvartwo}{(\ttmfour,\vcset)}}
    \end{array}
  \]
  To show that $f = g$,
  note first that they are have the same domain and codomain.
  Indeed, it suffices to observe that for each
  $
    \subs{(\tlam{\tvartwo}{\ki_1}{\ttmthree\sub{\tvar}{\ttmtwo}})}{\asig}
    =
    \subs{(\tlam{\tvartwo}{\ki_1}{\ttmthree})}{\asigext}
  $.
  To conclude, by functional extensionality, it suffices to show that
  $f$ and $g$ are pointwise equal.
  So let $\ttmfour \in \ValidTypesK{\ki_1}$
  and $\vcset \in \VCK{\ttmfour}{\ki_1}$,
  and note that:
  \[
    \begin{array}{rcll}
    &&
      f(\ttmfour,\vcset)
    \\
    & = &
      \semg{\ttmthree\sub{\tvar}{\ttmtwo}}{\ki_2}{\asig\extsub{\tvartwo}{(\ttmfour,\vcset)}}
    \\
    & = &
      \semg{\ttmthree}{\ki_2}{
        \asig
          \extsub{\tvar}{(\subs{\ttmtwo}{\asig},\semg{\ttmtwo}{\kitwo}{\asig})}
          \extsub{\tvartwo}{(\ttmfour,\vcset)}
      }
      & \text{by \ih}
    \\
    & = &
      \semg{\ttmthree}{\ki_2}{\asigext\extsub{\tvartwo}{(\ttmfour,\vcset)}}
    \\
    & = &
      g(\ttmfour,\vcset)
    \end{array}
  \]
\item
  Application, $\ttm = \ttmthree\,\ttmfour$
    with $\ttmthree \in \ValidTypesK{\ki_1\to\ki_2}$
    and $\ttmfour \in \ValidTypesK{\ki_1}$
    and $\ki = \ki_2$:
  Let $\asigext = \asig\extsub{\tvar}{(\subs{\ttmtwo}{\asig},\semg{\ttmtwo}{\kitwo}{\asig})}$.
  Then:
  \[
    \begin{array}{rcll}
    &&
      \semg{\ttm\sub{\tvar}{\ttmtwo}}{\ki}{\asig}
    \\
    & = &
      \semg{(\ttmthree\,\ttmfour)\sub{\tvar}{\ttmtwo}}{\ki_2}{\asig}
    \\
    & = &
      \semg{\ttmthree\sub{\tvar}{\ttmtwo}\,\ttmfour\sub{\tvar}{\ttmtwo}}{\ki_2}{\asig}
    \\
    & = &
      \semg{\ttmthree\sub{\tvar}{\ttmtwo}}{\ki_1\to\ki_2}{\asig}(
        \subs{\ttmfour\sub{\tvar}{\ttmtwo}}{\asig}
        ,
        \semg{\ttmfour\sub{\tvar}{\ttmtwo}}{\ki_1}{\asig}
      )
    \\
    & = &
      \semg{\ttmthree}{\ki_1\to\ki_2}{\asigext}(
        \subs{\ttmfour}{\asigext}
        ,
        \semg{\ttmfour}{\ki_1}{\asigext}
      )
      \HS\text{applying the \ih twice}
    \\
    & = &
      \semg{\ttmthree\,\ttmfour}{\ki_2}{\asigext}
    \\
    & = &
      \semg{\ttm}{\ki}{\asigext}
    \end{array}
  \]
\end{enumerate}
\end{proof}

\begin{lemma}[Stability under $\eqbeta$]
\label{lem:hol:segm_up_to_beta}
Let $(\VC{i})_{i \in \GValidTypes}$ be a \candidateFamily.
Let $\ttm_1,\ttm_2:\ki \in \ValidTypes$
be such that $\ttm_1 \eqbeta \ttm_2$,
and let $\asig$ be a \candidateAssignment.
Then $\semg{\ttm_1}{\ki}{\asig} = \semg{\ttm_2}{\ki}{\asig}$.
\end{lemma}
\begin{proof}
By induction on the derivation of $\ttm_1 \eqbeta \ttm_2$.
The interesting case is when there is a $\beta$-reduction or expansion
step at the root.
Indeed, let $\ttmtwo\in\ValidTypesK{\ki}$
and $\ttmthree \in \ValidTypesK{\ki_1}$,
and suppose that $\kiasig(\tvar) = \ki_1$
and $\ttm_1 = (\tlam{\tvar}{\ki_1}{\ttmtwo})\,\ttmthree$
and $\ttm_2 = \ttmtwo\sub{\tvar}{\ttmthree}$.
Then we claim that
$\semg{\ttm_1}{\ki}{\asig} = \semg{\ttm_2}{\ki}{\asig}$.
Indeed:
\[
  \begin{array}{rcll}
  &&
    \semg{\ttm_1}{\ki}{\asig}
  \\
  & = &
    \semg{((\tlam{\tvar}{\ki_1}{\ttmtwo})\,\ttmthree)}{\ki}{\asig}
  \\
  & = &
    \semg{(\tlam{\tvar}{\ki_1}{\ttmtwo})}{\ki_1\to\ki}{\asig}(
      \subs{\ttmthree}{\asig},
      \semg{\ttmthree}{\ki_1}{\asig}
    )
  \\
  & = &
    \semg{\ttmtwo}{\ki}{
      \asig
        \extsub{\tvar}{
          (\subs{\ttmthree}{\asig},\semg{\ttmthree}{\ki_1}{\asig})
        }
    }
  \\
  & = &
    \semg{\ttmtwo\sub{\tvar}{\ttmthree}}{\ki}{\asig}
    & \text{by \cref{lem:hol:substitution}}
  \\
  & = &
    \semg{\ttm_2}{\ki}{\asig}
  \end{array}
\]
\end{proof}

\begin{definition}[\ProofSubstitutions and compatibility]
A \defn{\proofSubstitution} is a function $\psubst$
mapping each \proofVariable $\pvar$
to a well-kinded \metaterm $\psubst(\pvar) \in \ValidMetaterms$.
If $\utm \in \ValidMetaterms$ is a well-kinded \metaterm,
we write $\subs{\utm}{\psubst}$
for the well-kinded \metaterm that results from replacing each occurrence of each
free variable $\pvar$ in $\utm$ by $\psubst(\pvar)$, avoiding capture.

If $\asig$ is a \candidateAssignment, $\penv$ is a \proofEnvironment,
and $\psubst$ is a \proofSubstitution,
we say that $\psubst$ is \defn{compatible} with $\penv$ under $\asig$,
written $\compatf{\penv}{\psubst}$,
if for every $\pvar:\ttm \in \penv$
such that $\ttm \in \ValidTypesK{\ki}$
we have that $\subs{\pvar}{\psubst} \in \semg{\ttm}{\ki}{\asig}$.
\end{definition}

\begin{lemma}[Adequacy for terms]
\label{lem:hol:adequacy_for_terms}
Let $(\VC{i})_{i \in \GValidTypes}$ be a \candidateFamily,
and let $\asig$ be a \candidateAssignment.
Let $\pjudg{\tenv}{\penv}{\ptm}{\ttm}$ where $\ttm \in \ValidTypesK{\ki}$
and $\compatf{\penv}{\psubst}$.
Assume moreover that $\tenv$ is compatible with the global \kindAssignment,
\ie $\compat{\tenv}{\kiasig}$ (which can always be fulfilled by appropriately
renaming variables).
Then $\subs{\ptm}{\asig\psubst} \in \semg{\ttm}{\ki}{\asig}$.
\end{lemma}
\begin{proof}
By induction on the derivation of the judgment $\pjudg{\tenv}{\penv}{\ptm}{\ttm}$.
The proof is exactly analogous to the proof in the second-order case,
but now using the higher-order versions of the
irrelevance lemma (\cref{lem:hol:irrelevance}),
the substitution lemma (\cref{lem:hol:substitution}),
and the adequacy lemma for types (\cref{lem:hol:adequacy_for_types}).
All rules in the higher-order case follow the same structure reasoning
than in the second order case, except for the conversion rule,
which is new.
In the case of the conversion rule,
let $\pjudg{\tenv}{\penv}{\ptm}{\ttmtwo}$
be derived  from $\pjudg{\tenv}{\penv}{\ptm}{\ttm}$
and $\ttm \eqbeta \ttmtwo$.
By \ih, $\subs{\ptm}{\asig\psubst} \in \semg{\ttm}{\Prop}{\asig}$
so by \cref{lem:hol:segm_up_to_beta}
we conclude that $\subs{\ptm}{\asig\psubst} \in \semg{\ttmtwo}{\Prop}{\asig}$.
\end{proof}

The following result constructs one specific \candidateFamily:

\begin{proposition}
\label{prop:hol:verif_family}
The following indexed family $(\Verif{i})_{i \in \GValidTypes}$,
defined by induction on the \kind, is a \candidateFamily:
\begin{enumerate}
\item $\VerifK{\ttm}{\bki} \eqdef \set{\vcset_\bki}$
      where $\vcset_\bki = \set{\singleElem}$.
\item $\VerifK{\ttm}{\Prop} \eqdef \set{\vcset_{\ttm}}$,
      where $\vcset_{\ttm} \eqdef \set{\utm \in \ValidMetaterms \ST \verif{\ttm}{\utm} \tows \iunit}$.
\item $\VerifK{\ttm}{\ki_1\to\ki_2}$ is the set
      of all functions defined according to
      the definition of \candidateFamilies
      (\cref{def:hol:candidateFamilies},~\cref{def:hol:candidateFamilies:case_arrow}).
\end{enumerate}
\end{proposition}
\begin{proof}
Let us check that this definition fulfills the
five conditions in \cref{def:hol:candidateFamilies}:
\begin{enumerate}
\item
  Stability under $\beta$-conversion:
  Let $\ttm_1,\ttm_2 \in \GValidTypesK{\ki}$
  such that $\ttm_1 \eqbeta \ttm_2$,
  and let us show that
  $\VerifK{\ttm_1}{\ki} = \VerifK{\ttm_2}{\ki}$.
  We proceed by induction on $\ki$:
  \begin{enumerate}
  \item
    If $\ki$ is a base \kind,
    then
    $\VerifK{\ttm_1}{\ki} = \set{\vcset_\bki} = \VerifK{\ttm_2}{\ki}$.
  \item
    If $\ki = \Prop$,
    then it suffices to note that $\vcset_{\ttm_1} = \vcset_{\ttm_2}$
    because $\verif{\ttm_1}{\utm} \tows \iunit$
    if and only if $\verif{\ttm_2}{\utm} \tows \iunit$
    given that $\ttm_1 \eqbeta \ttm_2$ 
    so $\verif{\ttm_1}{\utm}$ and $\verif{\ttm_2}{\utm}$
    are the same \metaterm.
  \item
    If $\ki = (\ki_1\to\ki_2)$ is a base \kind,
    then $\VerifK{\ttm_1}{\ki_1\to\ki_2}$ is the set of all functions
    \[
      f
      :
      \Prod{\ttmtwo\in\GValidTypesK{\ki_1}}
      \VerifK{\ttmtwo}{\ki_1}
      \to
      \VerifK{\ttm_1\ttmtwo}{\ki_2}
    \]
    such that $f(\ttmtwo_1,\vcset) = f(\ttmtwo_2,\vcset)$ whenever $\ttmtwo_1 \eqbeta \ttmtwo_2$.
    Similarly, $\VerifK{\ttm_1}{\ki_1\to\ki_2}$ is the set of all functions
    \[
      g
      :
      \Prod{\ttmtwo\in\GValidTypesK{\ki_1}}
      \VerifK{\ttmtwo}{\ki_1}
      \to
      \VerifK{\ttm_2\ttmtwo}{\ki_2}
    \]
    such that $g(\ttmtwo_1,\vcset) = g(\ttmtwo_2,\vcset)$ whenever $\ttmtwo_1 \eqbeta \ttmtwo_2$.
    It suffices to note by \ih that 
    $\VerifK{\ttm_1\ttmtwo}{\ki_2} = \VerifK{\ttm_2\ttmtwo}{\ki_2}$,
    so $f$ and $g$ range over functions of the same codomain.
  \end{enumerate}
\item
  Inhabited:
  Let $\ttm \in \GValidTypesK{\ki}$
  and let us construct an element
  $\semgbase{\ttm} \in \VerifK{\ttm}{\ki}$
  in such a way that if $\ttm_1 \eqbeta \ttm_2$
  then $\semgbase{\ttm_1} = \semgbase{\ttm_2}$.
  We proceed by induction on $\ki$.
  \begin{enumerate}
  \item
    If $\ki$ is a base \kind, $\ki = \bki$,
    then $\VerifK{\ttm}{\bki} = \set{\vcset_\bki}$
    so it suffices to take
    $\semgbase{\ttm} \eqdef \vcset_\bki$.
  \item
    If $\ki = \Prop$,
    then $\VerifK{\ttm}{\bki} = \set{\vcset_\ttm}$,
    so it suffices to take
    $\semgbase{\ttm} \eqdef \vcset_\ttm$.
    It is easy to observe that if $\ttm_1 \eqbeta \ttm_2$
    then $\semgbase{\ttm_1} = \semgbase{\ttm_2}$,
    as we already observed for stability under $\beta$-conversion.
  \item
    If $\ki = (\ki_1\to\ki_2)$,
    observe that for every $\ttmtwo\in\GValidTypesK{\ki_1}$
    there exists an element
      $\semgbase{\ttm\,\ttmtwo} \in \VerifK{\ttm\,\ttmtwo}{\ki_2}$
    by \ih. Let:
    \[
      \begin{array}{l}
        \semgbase{\ttm} :
        \Prod{\ttmtwo\in\GValidTypesK{\ki_1}}
        \VerifK{\ttmtwo}{\ki_1}
        \to
        \VerifK{\ttm\ttmtwo}{\ki_2}
      \\
        \semgbase{\ttm}(\ttmtwo,\vcset)
        =
        \semgbase{\ttm\,\ttmtwo}
      \end{array}
    \]
    Note that $\semgbase{\ttm} \in \VerifK{\ttm}{\ki_1\to\ki_2}$
    because if $\ttmtwo_1 \eqbeta \ttmtwo_2$
    and $\vcset \in \VerifK{\ttmtwo_1}{\ki_1}$
    then
    $\semgbase{\ttm}(\ttmtwo_1,\vcset)
     = \semgbase{\ttm\,\ttmtwo_1}
     = \semgbase{\ttm\,\ttmtwo_2}
     = \semgbase{\ttm}(\ttmtwo_2,\vcset)$
    by \ih.
    Moreover, note that if $\ttm_1 \eqbeta \ttm_2$
    then $\semgbase{\ttm_1} = \semgbase{\ttm_2}$
    because for every $\vcset \in \VerifK{\ttmtwo}{\ki_1}$
    we have that
    $\semgbase{\ttm_1}(\ttmtwo,\vcset)
    = \semgbase{\ttm_1\,\ttmtwo}
    = \semgbase{\ttm_2\,\ttmtwo}
    = \semgbase{\ttm_2}(\ttmtwo,\vcset)$
    by \ih.
  \end{enumerate}
\item
  Candidate condition for base \kinds:
  Immediate by definition.
\item
  Candidate condition for propositions:
  Let $\ttm \in \GValidTypesK{\Prop}$.
  Let us check that
  $\vcset_\ttm = \set{\utm \in \ValidMetaterms \ST \verif{\ttm}{\utm} \tows \iunit}$
  has the three required conditions:
  \begin{enumerate}
  \item
    $\gen{\ttm} \in \vcset_\ttm$
    because
    $\verif{\ttm}{\gen{\ttm}} \tos \iunit$.
  \item
    If $\utm \in \vcset_\ttm$
    then $\verif{\ttm}{\gen{\ttm}} \tos \iunit$.
    by definition.
  \item
    Let $\utm_1,\utm_2$ be \metaterms such that $\utm_1 \tow \utm_2$.
    Then $\verif{\ttm}{\utm_1} \tows \iunit$
    if and only if $\verif{\ttm}{\utm_2} \tows \iunit$
    because $\tow$ is confluent.
  \end{enumerate}
\item
  Candidate condition for higher-kinded \logicalTerms:
  Immediate by definition.
\end{enumerate}
\end{proof}

\begin{theorem}[Soundness]
\label{a:thm:hol:soundness}
Let $\penv,\tm,\ttm$ be such that they have no free \logicalVariables.
If $\pjudg{\kiasig}{\penv}{\tm}{\ttm}$
then $\verif{\ttm}{\subs{\tm}{\penv}} \tos \iunit$.
\end{theorem}
\begin{proof}
For each \kind $\ki$ suppose that $\teig_\ki$ is a
\logicalEigenvariable such that $\teig_\ki \in \GValidTypesK{\ki}$.
Let $\mathcal{F} = (\Verif{i})_{i \in \GValidTypes}$
be the \candidateFamily of \cref{prop:hol:verif_family}.
Define the following mapping:
\[
  \asig(\tvar) \eqdef (\teig_\ki, \semgbase{\teig_\ki})
  \HS\text{whenever $\kiasig(\tvar) = \ki$}
\]
where $\semgbase{\teig_\ki} \in \VerifK{\teig_\ki}{\ki}$
is given by the construction in~\cref{prop:hol:verif_family}.
Observe that $\asig$ is a $\mathcal{F}$-\candidateAssignment.

Note that $\compat{\penv}{\gensubst{\penv}}$,
where $\gensubst{\penv}$ is the generative substitution for $\penv$,
because if $\pvar:\ttm\in\penv$ then $\verif{\ttm}{\gen{\ttm}} \tows \iunit$
by correctness~(\cref{lem:hol:correctness}).

Note that the fact that the judgment $\pjudg{\kiasig}{\penv}{\tm}{\ttm}$
holds implies that $\tjudg{\kiasig}{\ttm}{\Prop}$,
and moreover $\ttm$ is ground by hypothesis,
so $\ttm \in \GValidTypesK{\Prop}$.
Then by \cref{lem:hol:adequacy_for_terms}
we have that $\subs{\ptm}{\asig\penv} \in \semg{\ttm}{\Prop}{\asig}$,
where moreover $\semg{\ttm}{\Prop}{\asig} \in \VerifK{\subs{\ttm}{\asig}}{\Prop}$
by \cref{lem:hol:adequacy_for_types}.
By the definition of the family $\mathcal{F}$ in \cref{prop:hol:verif_family},
this means that
$\verif{\subs{\ttm}{\asig}}{\subs{\ptm}{\asig\penv}} \tows \iunit$.
Since $\ttm$ and $\ptm$ have no free \logicalVariables,
$\subs{\ttm}{\asig} = \ttm$ and $\subs{\ptm}{\asig} = \ptm$.
Hence $\verif{\ttm}{\ptm} \tows \iunit$,
as required.
\end{proof}

\end{document}
\endinput